\documentclass[11pt,a4paper]{article}

\usepackage{aa_macros}

\title{Algebraicity and Asymptotics: An explosion of BPS indices from algebraic generating series}

\abstract{It is an observation of Kontsevich and Soibelman that generating series that produce certain (generalized) Donaldson Thomas invariants are secretly algebraic over $\mathbb{Q}(x)$.  From a physical perspective this observation arises naturally for DT invariants that appear as BPS indices in theories of class $S[A]$: explicit algebraic equations (that completely determine these series) can be derived using (degenerate) spectral networks.  In this paper, we conjecture an algebraic equation associated to DT invariants for the Kronecker $3$-quiver with dimension vectors $(3n,2n),\, n \in \mathbb{Z}_{>0}$ in the non-trivial region of its stability parameter space.  Using a functional equation due to Reineke, we show algebraicity of generating series for Euler characteristics of stable moduli for the Kronecker $m$-quiver assuming algebraicity of generating series for DT invariants.  In the latter part of the paper we deduce very explicit results on the asymptotics of DT invariants/Euler characteristics under the assumption of algebraicity of their generating series; explicit asymptotics are deduced for dimension vectors $(3n,2n),\, n \in \mathbb{Z}_{>0}$ for the Kronecker $3$-quiver.  The algebraic equation is derived using spectral network techniques developed in \cite{sn} and \cite{wwc}, but the main results can be understood without knowledge of spectral networks.}

\author{Tom Mainiero}
\affiliation
{New High Energy Theory Center, Rutgers University, Piscataway, NJ 08854, USA}
\emailAdd{tom.mainiero@rutgers.edu}

\date{\today}

\begin{document}

\maketitle

\section{Introduction}

\subsection{Main Results Part I: A short-lived exercise in succinctness}

In this paper, we concern ourselves with some results relating to a property we will call \textit{algebraicity}: suitably defined generating series of Donaldson-Thomas Invariants, BPS indices, or even just plain Euler characteristics of (stable) quiver moduli---a priori defined as formal series---are secretly algebraic functions over $\mathbb{Q}$.  Kontsevich and Soibelman have an understanding of algebraicity in the context of DT invariants \cite{kont_soib} associated to a large class of 3CY categories; however their proof uses rather indirect methods: explicit algebraic relations are not produced. 

On the other hand, in a large class of (four dimensional $\mathcal{N}=2$) supersymmetric field theories called \textit{theories of class $S[A]$}, there exists algorithmic machinery for computing generating series of BPS indices: \textit{spectral networks}.  Using this machinery, algebraicity can be seen directly, and one can algorithmically construct explicit algebraic relations.  Roughly speaking, a spectral network is a directed, decorated graph associated to a $\mathbb{Z}_{>0}$-family of BPS states; algebraic relations for generating series of the associated BPS indices follows by a system of algebraic relations determined by the edges and vertices of this graph.

The main results of this paper are threefold:
\begin{enumerate}
	\labitem{(A)}{list:DT_equations} Algebraic equations that produce DT invariants for the $3$-Kronecker quiver:
		\begin{center}
\begin{tikzpicture}
\tikzstyle{block} = [rectangle, draw=blue, thick, fill=blue!10,
text width=16em, text centered, rounded corners, minimum height=2em]

\node at (1.2, 0)[circle,draw=blue,very thick] (second) {$q_{2}$};

\node at (-1.2, 0)[circle,draw=blue,very thick] (first) {$q_{1}$};

\draw[-latex] (second) to  [bend right = 35] (first);
\draw[-latex] (second)  to  [bend right = 0] (first);
\draw[-latex] (second)  to  [bend right = -35] (first);
	\end{tikzpicture}
\end{center}	
	
	 associated to the family of collinear dimension vectors $\left(3n q_{1} + 2n q_{2} \right)_{n = 1}^{\infty}$ (see \eqref{eq:func_eqs} and \eqref{eq:3_2_dt_relation});
	
	\labitem{(B)}{list:euler_equations} As a corollary of \ref{list:DT_equations} and a functional equation due to Reineke: an algebraic equation determining associated Euler characteristics of stable quiver moduli (see \eqref{eq:3_2_eul_relation} and accompanying definitions in \S \ref{sec:stab_eul_setup});
	
	\labitem{(C)}{list:asymptotics} Explicit constraints on the asymptotics of BPS indices/DT invariants/($m$-Kronecker) Euler characteristics due to algebraicity: ``\textit{Either there are finitely many BPS indices or exponentially many}".  Formulae are produced that allow for the extraction of asymptotics from algebraic equations; these formulae are applied to the algebraic equations of \ref{list:DT_equations} and \ref{list:euler_equations}.
\end{enumerate}

There are many reasons why physicists should be excited about algebraicity, one being the likelihood of deep implications for the geometry of moduli spaces of vacua for four-dimensional $\mathcal{N}=2$ field theories (in particular, the generalized cluster-like structures whose Coulomb branch manifestation was discovered by Gaiotto-Moore-Neitzke). As already mentioned in \ref{list:asymptotics}, a less subtle implication---one that is treated with more attention in this paper---is the fact that \textit{algebraicity leads to exponential growth of BPS indices with the charge}. Only in the simplest of possible scenarios--in particular, when the generating series is a rational function with zeros and poles at roots of unity---do BPS indices \textit{not} grow exponentially with the charge.  The complicated nature of spectral networks in theories of class $S[A]$ with rank $K>1$ (e.g. $SU(K)$ super Yang Mills) suggests that typically algebraic generating series are not polynomials or rational functions; in this sense, the nice behaviour of the BPS spectrum for theories of rank $1$ may be rather exceptional.\footnote{In this paper we do not mention theories of class $S[D]$ or $S[E]$ only because spectral network techniques are currently best understood for class $S[A]$; nevertheless, one should expect algebraicity (and the corollary of exponential growth) should be follow for types $D$ and $E$ using an appropriate generalization of the spectral network machinery (algebraicity follows from the combinatorial nature of spectral network machinery).  See work of Longhi-Park \cite{Longhi:2016rjt} for such a generalization to types $D_{n},\,E_{6},\,$ and $E_{7}$.} In the final part of this paper we provide a constructive proof of how algebraic generating series lead to exponential growth; in the process we give an explicit way to extract precise asymptotics of BPS indices given the algebraic equations satisfied by generating series.  (Note that, in early work \cite{Kol:1998zb}, Barak Kol predicted exponential growth of the density of certain BPS states with charge in four-dimensional $\mathcal{N}=4,\, SU(3)$ SYM---a rank $K=2$ theory of class $S[A]$ associated to a torus; this is a corollary of the exponential growth in the BPS index: see, e.g., \S \ref{sec:alg_asymptotics}.)

In order to make the results of the paper more accessible to non-experts, and to compensate for its size, we provide a detailed and slow introduction.

\subsection{Motivation: BPS indices, DT invariants, and Generating Series}

Let us begin with some context and motivation that will hopefully put algebraicity and the results of this paper into a larger context.  This section is written with the non-expert in mind.

We begin with a diagram that encodes various relationships between between BPS indices (in four-dimensional $\mathcal{N} = 2$ supersymmetric field theories), DT invariants, and Euler characteristics of stable moduli.  Dotted arrows indicate partially-defined correspondences (with the conditions defining the domain of definition indicated); the squigglyness of an arrow is negatively correlated with the author's (subjective and biased) notion of rigour.  Reineke's arrow is magenta coloured to stress that the relationship is not a simple equality of invariants, but a relationship between generating series.  (This diagram should not be considered complete, having inevitable omissions; notably, one can include a squiggly red arrow labelled ``Denef-Moore" from the ``BPS Indices" box to the ``DT invariants box" in reference to the work \cite{Denef-Moore}.)

\begin{center}
	\begin{tikzpicture}
		\tikzstyle{block} = [rectangle, draw=blue, thick, fill=blue!10,
		text width=15em, text centered, rounded corners, minimum height=2em]
		
		\tikzstyle{blockred} = [rectangle, draw=blue, thick, fill=red!10,
		text width=15em, text centered, rounded corners, minimum height=2em]
		
		\node at (4.5,4.5)[blockred,draw=red,very thick] (BPS) {BPS indices in four-dimensional $\mathcal{N} = 2$ supersymmetric field theories};
		
		\node at (-4.5,4.5)[blockred,draw=red,very thick] (sn) {Spectral Networks};
		
		\node at (-4.5, 0)[block,draw=blue,very thick] (exp) {Exponents of certain commutators of Poisson automorphisms};
		
		\node at (4.5, 0)[block,draw=blue,very thick] (steuler) {Euler characteristics of stable moduli};
		
		\node at (-4.5, -4.5)[block,draw=blue,very thick] (DT) {Donaldson-Thomas Invariants};
		
		\node at (4.5, -4.5)[block,draw=blue,very thick] (sseuler) {Stacky (and Weighted) Euler Characteristics of semistable moduli};
		
		\draw[->, densely dotted, thick, decorate,decoration=snake,color=red] (BPS) -- (steuler) node[sloped,above,midway]{\tiny Denef} node[sloped,below,midway]{\parbox{7em}{\tiny \centering Semistable=stable, Type IIB inspired}};
		
		\draw[<->, densely dotted, thick, decorate,decoration=snake, segment length=8mm, color=red] (sn) -- (BPS) node[sloped,above,midway]{\parbox{5em}{\tiny \centering Gaiotto-Moore-Neitzke}} node[sloped,below,midway]{\parbox{5em}{\tiny \centering Theories of class $S[A]$}};
		
		\draw[<-, decorate,decoration=snake, segment length=10mm, thick, color=red] (exp) -- (sn) node[sloped,above,midway]{\parbox{5em}{\tiny \centering Gaiotto-Moore-Neitzke}};
		
		\draw[->,decorate, thick, decoration=snake, segment length=8mm, color=red] (BPS) -- (exp)  node[sloped,above,midway]{\tiny Gaiotto-Moore-Neitzke};
		
		\draw[latex-latex, thick, densely dotted,color=magenta] (exp) -- (steuler) node[above,midway]{\tiny Reineke} node[below,midway]{\tiny Acyclic Quivers \footnotemark};
		
		\draw[latex-latex,thick, color=blue] (DT) -- (exp)  node[sloped,above,midway]{\tiny Kontsevich-Soibelman};
		
		\draw[latex-latex, thick, color=blue] (DT) -- (sseuler) node[above,midway]{\parbox{5em}{\tiny \centering Joyce-Song (JS), Behrend (B); Kontsevich-Soibelman (KS)}} node[below,midway]{\parbox{6em}{\tiny \centering  Quivers with Potential, Coherent Sheaves (JS,B); 3CY categories (KS)}} ;
		
		\draw[latex-latex, thick, color=blue,densely dotted] (steuler) -- (sseuler) node[below,midway,sloped]{\tiny Semistable=Stable};
	
	\end{tikzpicture}
\footnotetext{Joyce-Song's reproduction \cite[Second Proof of Thm. 7.29]{js} of Reineke's proof \cite{reineke:integrality} indicates that Reineke's functional equations relating DT invariants and BPS indices holds for an arbitrary finite quiver without relations (when equipped with the appropriate stability condition).}
\end{center}

\subsubsection{BPS Indices} \label{sec:intro_BPS_indices}
  A prevalent philosophy in modern physics is a notion familiar to geometers: it is better to study families of objects as opposed to individual objects.  This idea is natural to those who study theories with extended supersymmetry, where one never really has a single theory, but a whole family parametrized by the choices of supersymmetric vacua.  Even if one attempts to focus on a single theory (specified by a choice of vacuum), at low energies this theory will look like a sigma model from spacetime into the moduli of supersymmetric vacua.  Hence, there is no escape from families and one is naturally forced to ask questions about the geometry of the moduli of vacua in order to understand any particular theory.

To probe this geometry, one might try something Yonedaesque: look at maps into the moduli space and try asking for certain ``invariants" that can be associated to such maps, e.g. quantities that are well-behaved in families of theories. Particularly, as a physicist, one might ask ``What part of the particle spectrum of the theory parallel transports nicely as I vary moduli?"  One answer to this question is provided by BPS states: states invariant under half\footnote{Here we are thinking of $\mathcal{N}=2$ supersymmetry; in theories with, e.g. $\mathcal{N}=4$ supersymmetry one may also speak of BPS states invariant under 1/4 of the supercharges.} of the supersymmetries.  Indeed, representation-theoretic arguments suggest that such states should parallel transport nicely as any moduli or parameters of a theory are deformed slightly.  A careful study of this story leads one to the definition of the \textit{BPS index}  (a.k.a. the second helicity supertrace): a $\mathbb{Z}$-valued quantity that is a weighted count of BPS states; it is defined as a weighted-supertrace over a finite dimensional Hilbert space (see \cite{Lerche:1988zy,Lerche:1988zz,Harvey:1995fq,Kiritsis:1997gu} for some early references).  Under sufficiently small perturbations, the BPS index remains constant under parallel transport.  For larger perturbations, e.g. over a sufficiently large region in the moduli of vacua, it has been long known that the BPS index is not constant under parallel transport, but \textit{piecewise constant}: jumps occur along the union of (real) hypersurfaces in the moduli of vacua (a.k.a. the wall-crossing phenomenon). Fortunately, these jumps are intimately connected to the geometry of the moduli space of vacua.

To make the notions above more explicit, we restrict our attention to four-dimensional $\mathcal{N}=2$ supersymmetric field theories on $\mathbb{R}^{3,1}$ such that the notion of a Coulomb branch $\CB$---a component of the moduli space of supersymmetric vacua---is sensible.  $\CB$ is a complex affine space equipped with the following data\footnote{Note that this data equivalently describes the Coulomb branch as a \textit{special K\"{a}hler} manifold away from a collection of singular points $\CV$ (where the K\"{a}hler metric degenerates).}
\begin{enumerate}
	\labitem{(D1)}{list:lattice_dat}  A sheaf $\widehat{\Gamma}$ over $\CB$ valued in lattices (finitely generated free abelian groups) equipped with a $\mathbb{Z}$-valued antisymmetric pairing---i.e. given any open $U \subset \CB$, then $\widehat{\Gamma}(U)$ (the space of sections of $\widehat{\Gamma}$ over $U$) is a finitely generated free abelian group equipped with an antisymmetric pairing $\langle \cdot, \cdot \rangle_{U} \in \widehat{\Gamma}(U)^{\otimes 2} \rightarrow \mathbb{Z}$.  The stalk $\widehat{\Gamma}_{u}$ at a point $u \in \CB$ is referred to as \textit{the charge lattice} at $u$.\footnote{ To put this lattice into a physical context, if the low energy effective field theory at $u$ is a $U(1)^{r}$ gauge theory with $f$ flavours, then $\widehat{\Gamma}_{u}$ is the rank $2r + f$ the lattice of possible electric, magnetic, and flavour charges of states in the theory defined by the vacuum $u$.}
	
	\labitem{(D2)}{list:cc_dat} The \textit{central charge function} $\mathcal{Z}$: a global section of $\mathrm{Hom}_{\mathbb{Z}}(\widehat{\Gamma},\underline{\mathbb{C}}) \otimes_{\underline{\mathbb{C}}} \mathcal{O}_{\CB}$ over $\CB$.
\end{enumerate}
Satisfying the conditions that, there is a complex codimension 1 (Zariski closed) subset $\CV \subset \CB$ of ``singular points" such that:
\begin{enumerate}
	\labitem{(C1)}{list:lattice_cond} $\widehat{\Gamma}$ is a local system (locally constant sheaf) on the complement $\CB' := \CB \backslash \CV$.  Moreover, the  quotient of $\widehat{\Gamma}(U)$ by the radical $\widehat{\Gamma}_{\mathrm{rad}}(U) := \{\gamma \in \widehat{\Gamma}(U): \langle \gamma, \cdot \rangle \equiv 0\}$ is a lattice of rank $2\dim_{\mathbb{C}}(\CB)$;
	
	\labitem{(C2)}{list:cc_cond} for each open $U \subset \CB'$, $\mathcal{Z}$ defines an embedding $U \hookrightarrow \mathrm{Hom}_{\mathbb{Z}}(\widehat{\Gamma}(U),\mathbb{C})$ whose image is a Lagrangian subspace of a symplectic slice of $\mathrm{Hom}_{\mathbb{Z}}(\widehat{\Gamma}(U),\mathbb{C})$.
\end{enumerate}

Choosing a point $u \in \CB '$, the BPS index---which is defined using the data of the $\mathcal{N}=2$ field theory on $\mathbb{R}^{3,1}$ (hence, more than just the data of $\CB$)---defines a function
\begin{align}
	\Omega^{\mathrm{BPS}}( \cdot, u): \widehat{\Gamma}_{u} \backslash \{0\} \rightarrow \mathbb{Z}.
	\label{eq:Omega_BPS}
\end{align}

As mentioned above, as $\gamma$ is parallel transported along a path on $\CB'$, $\Omega^{\mathrm{BPS}}$ exhibits the behaviour of a piecewise constant function on $\CB'$: jumps occur along the real codimension 1 locus defined by
\begin{align*}
	\{u \in \CB':\text{$\exists \gamma_{1}, \gamma_{2} \in \widehat{\Gamma}_{u}$ non-collinear charges of BPS states with $\arg \left[\mathcal{Z}_{u}(\gamma_{1}) \right] = \arg \left[\mathcal{Z}_{u}(\gamma_{2}) \right]$} \}
\end{align*}
often called the collection of \textit{walls of marginal stability}.  Of course, at this point it is natural to ask ``Can we quantify the piecewise constant behaviour?".  We will come back to this question; for now let us consider the following question: suppose that one is supplied with a point $u \in \CB$ and a charge $\gamma \in \widehat{\Gamma}_{u}$, then what is the BPS index $\Omega^{\mathrm{BPS}}(\gamma,u)$?  Trying to calculate this index directly from its definition as a weighted supertrace over some Hilbert space (encoding BPS particles of charge $\gamma$ in the theory defined by the vacuum $u$) would be easy if the Hilbert space were divinely gifted.  More often than not, however, the higher powers that be take pleasure in our suffering, and only in rare instances---particularly when both the theory of interest and the BPS states of interest admit a perturbative description---is it possible to easily construct the Hilbert space.  More generally, one must rely on indirect methods to compute the BPS index---much like the way one can sometimes compute Euler characteristics without understanding the full cohomology of the space. 

\subsubsection{DT invariants}

Typically in supersymmetric field/string theories, BPS states take on semi-classical descriptions closely related to interesting geometric quantities; in such situations BPS indices can be interpreted as counts of:
	\begin{enumerate}
			\labitem{(a)}{list:slag} special Lagrangians in some (Calabi-Yau) Calabi-Yau threefold $X$ (``$A$-branes");
			
			\labitem{(b)}{list:shf} coherent sheaves (``$B$-branes") on a Calabi-Yau threefold $X^{\vee}$ (the mirror-dual description to \ref{list:slag});
	
			\labitem{(c)}{list:quiv} complex quiver representations;
		\end{enumerate}
sometimes equivalent (or dual) descriptions of the same theory allow all three interpretations to play a role: e.g. Denef \cite{denef:qqhh} describes a physical derivation of the relationship between special Lagrangians and quivers, and---when $X = \mathbb{P}^{n}$---mathematicians may recognize the passage from \ref{list:shf} to \ref{list:quiv} in the context of Beilinson's theorem \cite{beilinson}.  The interpretations \ref{list:slag} and \ref{list:shf} are rather natural for string theories, or field theories with string-theoretic constructions. It is thus, reasonable to expect that the BPS index, a priori defined as a trace over a finite-dimensional Hilbert space, should be given in terms of an appropriate ``count" of geometric objects.  Indeed, in the proper context, BPS indices are conjecturally\footnote{Here the word ``conjecture" is subtle and can be used in two ways. (1): For physicists, one can ask whether or not the definition of DT invariant is the appropriate notion of BPS index when there is an underlying semiclassical geometric description of the physics---at least for four-dimensional $\mathcal{N}=2$ field theories, Gaiotto-Moore-Neitzke's construction may satisfy most physicist's level of rigour regarding the equivalence between Kontsevich and Soibelman's definition of DT invariants and BPS indices.  (2): For mathematicians, the word ``conjecture" should be taken to mean that the very definition of a BPS index relies on the construction of Hilbert spaces in quantum field theories: which may not be possible at the most stringent levels of rigour.} the same as Donaldson-Thomas (DT)-invariants \cite{dt,thomas,behrend,ks:wcf,js}, which are defined terms of geometric quantities; the appropriate catchphrase being that ``DT invariants are virtual counts of semistable objects in a 3-Calabi-Yau (3CY) Category".  

Indeed, a significant amount of evidence\footnote{At least from field theories obtained by geometric engineering.} suggests (see \cite[\S 1.1]{cddmms:geom_eng} and references therein) that we should conjecture that BPS states in four-dimensional $\mathcal{N}=2$ field theories are in correspondence with a subcollection of objects in a 3CY category: a triangulated ($A_{\infty}$-)category equipped with data and conditions abstracted from the derived category of coherent sheaves on a Calabi-Yau 3-fold. Corresponding to \ref{list:slag} - \ref{list:quiv}, this category is usually: a Fukaya category, ($A_{\infty}$ enhancement of the derived) category of coherent sheaves \ref{list:quiv}, or a ($A_{\infty}$-enhanced derived) category of quiver representations.  Assuming such a 3CY category $\mathcal{C}$ is handed to us, distinguished triangles play the role of bound state formations\footnote{See, e.g. \cite[\S 5.5.1]{mirror_sym_2}.}; hence, it is natural to expect that the charge of a BPS state--thought of as an object of $\mathcal{C}$---should be identified with its class in the Grothendieck group $K_{0}(\mathcal{C})$.  From this perspective, then the ``count" of BPS states of a particular charge $\gamma \in K_{0}(\mathcal{C})$ should be a ``count" of isomorphism classes of objects with class $\gamma$; this count can be realized by replacing the \textit{set} of isomorphism classes of objects with class $\gamma$ with a moduli \textit{space}/stack of objects (usually something algebraic, e.g. a scheme or an Artin stack) and computing its (weighted) Euler characteristic.  For those worried about where the choice of vacuum $u \in \CB$ has gone, have no fear: an important algebro-geometric lesson is that a notion of ``semistability" is typically required to form interesting moduli stacks; roughly, it is the choice of stability condition that replaces the choice of vacuum.  The notion of stability that has gained acceptance among physicists and practitioners of Kontsevich and Soibelmans motivic DT invariants is provided by (an $A_{\infty}$ version of) Bridgeland stability conditions on $\mathcal{C}$.  For the purposes of this vague discussion there is no harm in thinking of $\mathcal{C}$ as an ordinary triangulated category (i.e. by passing to the derived category $H^{0}(\mathcal{C})$) and thinking of a stability condition as the data of a pair $(\mathsf{A}, Z)$ of an abelian subcategory $\mathsf{A}$ (arising as the heart of a bounded t-structure) along with a linear function $Z \in \mathrm{Hom}_{\mathbb{Z}} \left(K_{0}(\mathsf{A}), \mathbb{C} \right)$ satisfying some conditions; from this data one can define the notion of a (semi)stable object.  The ability to vary a stability condition in a nice manner is aided by the observation that the space of stability conditions forms a Hausdorff complex manifold \cite[\S 3.4]{ks:wcf} $\mathrm{Stab}(\mathcal{C})$.  

To translate this into more physical language: for a fixed $u \in \CB'$ the choice of abelian category $\mathsf{A}$ can be thought of as the choice of a set of elementary BPS states (in the theory defined by vacuum $u$) from which all other BPS states can be generated as ``bound" states\footnote{In order to make sense of such a statement, one must choose a ``locally constant" family of theories---parametrized by $[0, \infty)$---where the theory over $t=0$ is the theory of interest and the limit $t \rightarrow \infty$ is a ``free" limit where all BPS states (or at least a subset of the ones of interest) split into a set of ``generating" constituents.  This is, for example, what happens to Halo bound states for as one approaches a ``$\mathcal{K}$-wall" (see, e.g. \cite[\S 3]{gmn:framed}).}; hence, $\widehat{\Gamma}_{u} \cong K_{0}(\mathsf{A}) \cong K_{0}(\mathcal{C})$ and the function $Z$ is then immediately given by the central charge function $\mathcal{Z}_{u}$.  Given that $\mathrm{Stab}(\mathcal{C})$ is a complex manifold, one may hope that the choice of abelian subcategory may be done in such a manner that one can embed $\CB$ into $\mathrm{Stab}(\mathcal{C})$ in a holomorphic manner.  As it turns out, this is only possible locally due to monodromy around points in the complex codimension 1 locus $\CV$ (mentioned above the conditions \ref{list:lattice_cond}-\ref{list:cc_cond}): slowly varying the theory along a homotopically non-trivial closed loop on $\CB'$ will typically act non-trivially on the generating set of BPS states (at the level of the charge lattice, this is the monodromy of the local system $\widehat{\Gamma}$ over $\CB'$); more precisely, monodromy acts as an autoequivalence on $\mathcal{C}$ that does not preserve the abelian subcategory $\mathsf{A}$ \cite{douglas:2000ah,douglas:2001,douglas:2002fj}. Hence, at best we have an embedding of the universal cover of $\CB'$ into $\mathrm{Stab}(\mathcal{C})$ (c.f. the conjecture in \cite[\S 1.1]{cddmms:geom_eng}).

Working with a fixed stability condition, DT invariants define a map
\begin{align*}
	\Omega^{\mathrm{DT}}( \cdot, Z): K_{0} \left(\mathsf{A} \right) \backslash \{0 \} \rightarrow \mathbb{Q}
\end{align*}
such that $\Omega^{\mathrm{DT}}(\gamma; Z)$ has the interpretation of a virtual count of semistable objects with fixed class $\gamma \in K_{0}(\mathsf{A})$.  As mentioned above, traditionally this count is viewed as the (weighted) Euler characteristic of an appropriately-defined moduli stack of objects with class $\gamma$---it is the stackiness that makes this count a priori rational-valued; however, it is conjectured to be an integer for many 3CY categories \cite{soib}. (Of course a DT invariant must be an integer if it is to agree with a BPS index, as the latter is always an integer.)

  In the simplest of situations $\Omega^{\mathrm{DT}}(\gamma,Z)$ can be defined as an honest Euler characteristic of a moduli space of semistable objects with class $\gamma \in K_{0} \left(\mathsf{A} \right)$---particularly when the moduli space of semistable objects with class $\gamma \in K_{0} \left(\mathsf{A} \right)$ is a smooth (complex) projective variety.  But this is not always the case.  For concreteness, let us specialize our attention to DT invariants that count quiver moduli---particularly quiver moduli for an acyclic quiver $Q$ (i.e. without oriented cycles).  In this case our 3CY category of interest is $\mathsf{D^{b}Rep}_{\mathbb{C}}(Q)$; restricting our focus to stability conditions ($Z, \mathsf{A}$) such that $\mathsf{A}= \mathsf{Rep}_{\mathbb{C}}(Q)$ (identified as the full subcategory defined by chain complexes of representations supported in degree zero), then $K_{0} \left(\mathsf{A} \right)$ is canonically identified with the free abelian group $\mathbb{Z}Q_{0}$ generated by the set of vertices $Q_{0}$ of $Q$. The class of any object $A$ of $\mathrm{Rep}_{\mathbb{C}}(Q)$ is an element of $\mathbb{Z}_{\geq 0} Q_{0}$: the dimension vector of the quiver representation $A$.  Now, if we choose $\gamma \in \mathbb{Z}_{> 0} Q_{0}$ to be primitive\footnote{I.e. $\gamma \neq n \alpha$ for some $n \in \mathbb{Z}$ such that $n \neq 1$, and $\alpha \in \mathbb{Z} Q_{0}$.}, then any semistable representation is necessarily stable and, wonderfully, one can define a moduli space of stable representations of dimension vector $\gamma$ as a smooth projective variety \cite{king:quiv_rep}; in this case $\Omega^{\mathrm{DT}}(\gamma,Z)$ is simply the Euler characteristic of this variety.  

On the other hand, for non-primitive $\gamma$, the moduli spaces of semistable and stable representations no longer coincide.  In addition, although the moduli space of stable representations is always a smooth variety, the moduli space of semistable representations---which really should be thought of as parametrizing \textit{polystable} representations (direct sums of stable representations)---is a singular projective variety.  Not surprisingly, the proper count of semistable objects is no longer a na\"{i}ve Euler characteristic of semistable moduli and, as the presence of singularities in our moduli space suggests, one should really be computing stacky invariants. In this situation we can define DT invariants as (weighted) Euler characteristics of an appropriately defined moduli \textit{stack} of semistable representations (in the spirit of Joyce-Song, and Behrend); from this perspective it is not surprising that $\Omega^{\mathrm{DT}}$ is a priori rational (as opposed to integer) valued.  

Despite this fuss being made about DT invariants \textit{not} always being Euler characteristics of stable moduli, for quivers without potential there are stability conditions such that the corresponding DT invariants are related to Euler characteristics of stable moduli.  This relationship occurs at the level of generating series: roughly a generating series that encodes Euler characteristics of stable moduli is the inverse series for a generating series that encodes the corresponding DT invariants, and was first pointed out by Reineke \cite{reineke:integrality} in the context of the $m$-Kronecker quiver equipped with what we will call ``an ($m$)-\textit{wild stability}" condition due to its connection to wild wall crossing and the exponential growth phenomenon studied in \cite{wwc} (c.f. appendices \ref{app:quiv_rep} and \ref{app:m_wild} for details and references regarding wild stability conditions and their associated DT invariants).

A different approach to DT invariants---one that includes generating series naturally---has enjoyed more success in the physics literature: Kontsevich and Soibelman's definition of (numerical) DT invariants associated to 3CY categories.  The key idea is that DT invariants are encoded into a single group element of a pronilpotent group given by the completion of a subgroup of (birational) Poisson automorphisms of an algebraic torus. The group element is defined a priori by the stability data, and DT invariants are recovered as the exponents that appear in the decomposition of this group element into a distinguished set of generators---i.e. a decomposition an infinite ordered composition of distinguished (birational) Poisson automorphisms.\footnote{See \S \ref{sec:alg_geom} for an explicit representation of these distinguished birational Poisson automorphisms.} The advantage of this description is that the group elements (and their decompositions) associated to different choices of stability condition are naturally related by \textit{wall-crossing} formulae---leading to a precise description of the piecewise-linear behaviour of DT invariants.  The power of these formulae is the following: if $(Z_{0}, \mathsf{A}_{0})$ is a stability condition, all of whose DT invariants are known, then for any other stability condition $(Z_{1}, \mathsf{A}_{1})$, connected by a piecewise smooth path $\underline{Z}: [0,1] \rightarrow \mathrm{Stab}(\mathcal{C})$, wall-crossing formulae---which effectively reduce to the computation of commutators of Poisson automorphisms---allow for the extraction of the DT invariants associated to $(Z_{1}, \mathsf{A}_{1})$.

Kontsevich and Soibelman's approach indicates that DT invariants should be encoded in generating series: the distinguished Poisson automorphisms mentioned above are determined by infinite series (see \S \ref{sec:alg_geom}).  In particular, one should think of DT invariants/BPS indices  in $\mathbb{Z}_{> 0}$-families: Let $\Gamma$ be the lattice $K_{0}(\mathsf{A})$ (or, in a physical context, the stalk/charge lattice $\widehat{\Gamma}_{u}$); then for each primitive dimension vector $\gamma_{r} \in \Gamma$ such that $\mathbb{Z}_{>0} \gamma_{r}$ represents the ray $r \in (\Gamma \backslash \{0\})/\mathbb{Z}_{> 0}$, we should encode the sequence of DT invariants/BPS indices $\Omega_{n} := \Omega(n \gamma_{r})$ into a generating series of the form
\begin{align}
	\gdt_{r} = \prod_{n = 1}^{\infty} (1 - (\pm z)^{n})^{n \Omega_{n}} \in \formal{\mathbb{Q}}{z}
	\label{eq:gen_series_intro}
\end{align}
where $z$ is a formal parameter and the sign ``$\pm$" depends on the context under consideration.  

The property of \textit{algebraicity} concerns the generating series $\gdt_{r}$: although defined as a formal series, $\gdt_{r}$ is secretly an algebraic function over $\mathbb{Q}$, i.e. $\gdt_{r}$ is the root of a polynomial with coefficients in the field of rational functions $\mathbb{Q}(z)$.  	

\subsubsection{Algebraicity and Asymptotics} \label{sec:alg_asymptotics}

As we will see, an interesting side effect of algebraicity is the feature of exponential asymptotics:
suppose $\ggen$ is a formal series in a variable $z$ that encodes a sequence of integers $(\beta_{n})_{n = 1}^{\infty} \subset \mathbb{Z}$ via an ``Euler product" factorization
\begin{align}
	\ggen = \prod_{n=1}^{\infty} (1 - (\pm z)^{n})^{n \beta_{n}}.
	\label{eq:ggen_intro}
\end{align}
If $\ggen$ is the $z = 0$ series expansion of an algebraic function over $\mathbb{Q}$ then---as a corollary of the results in \S \ref{sec:generating_asymptotics}---the large $n$ growth of $\beta_{n}$ will be of the form
\begin{align}
	\beta_{n}= \Theta(n) n^{-2 - \alpha} \mathtt{R}^{-n} + \mathcal{O}(n^{-2 - \alpha'} \mathtt{R}^{-n})
	\label{eq:hag_asymp}
\end{align}
as $n \rightarrow \infty$; where $\Theta \in \mathcal{O}(1)$ is a bounded (``oscillatory") function, $\mathtt{R} \in \conj{\mathbb{Q}} \cap (0,1]$, and $\alpha,\alpha'$ are non-negative rational numbers satisfying $\alpha' > \alpha$.  This behaviour can be summarized by the slogan:
\begin{center}
	\textit{There are only finitely many non-vanishing $\beta_{n}$ (when $\mathtt{R} = 1)$, or an exponential growth with $n$ (when $\mathtt{R} < 1$).}
\end{center}
In this paper, this result is applicable in two situations: when $\ggen = \gdt_{r}$ (c.f. \eqref{eq:gen_series_intro}) so that the $\beta_{n} = \Omega_{n}$ are DT invariants/BPS indices, and when $\ggen$ is the generating series for Euler characteristics of stable $m$-Kronecker moduli (for an $m$-wild stability condition).

 The details about how to compute $\alpha$ and $\Theta(n)$ can be found in \S \ref{sec:alg_asymp}, for now it is worth mentioning that $\mathtt{R}$ is the smallest of the radii of convergences of the formal series defined by $\ggen$ and $\ggen^{-1}$---equivalently, the radius of convergence of the series $\log(\ggen)$.  To see the connection between radii of convergences and growth rate is not hard: if $\sum_{n = 0 }^{\infty} h_{n} z^{n}$ is the series expansion of any function, holomorphic at $z =0$, then the radius of convergence $\mathtt{R}_{h}$ of this series is given by
\begin{align*}
	\mathtt{R}_{h}^{-1} = \limsup_{n \rightarrow \infty} (h_{n})^{1/n} ;
\end{align*}
at this point it is only a mild insult to your calculus teacher to expect that the $h_{n}$ should behave roughly like $\mathtt{R}_{\ggen}^{-n}$ in the large $n$ limit.  Now, heuristically speaking, looking at \eqref{eq:ggen_intro} it is not too far-fetched to believe that the large $n$ behaviour of $n \beta_{n}$ is given by the coefficients $[z^{n}] \log(\ggen)$ (see Lemma \ref{lem:omega_to_log}).  Because $\ggen$ is an algebraic series (in particular holomorphic), then $\log(\ggen)$ is series expansion of a holomorphic function about $z = 0$; applying our mild insult, we end up with the na\"{i}ve expectation that $\beta_{n}$ behaves like $n^{-1} \mathtt{R}_{\log(\ggen)}^{-n} = n^{-1} \mathtt{R}^{-n}$.  This heuristic analysis does not get the ``subexponential" behaviour $n^{-2-\alpha} \Theta(n)$ correct, but it is instructive for seeing the ``exponential" behaviour $\mathtt{R}^{-n}$.

To justify the slogan---in particular the statement about ``exponential growth"---we need to exclude the possibility that $\mathtt{R}>1$; indeed, when we apply the assumption that the $\beta_{n}$ must be integers, the constraint $\mathtt{R} \leq 1$ follows (see Prop. \ref{prop:growth}).  Given this, we have a justification for the slogan above:  if $\mathtt{R} = 1$ there must be finitely many $\beta_{n}$---due to the subexponential $n^{-2-\alpha}$ suppression---and when $\mathtt{R} < 1$ we have exponential growth.  

The slightly more subtle observation that $\mathtt{R}$ is an algebraic number follows from the fact that singularities of the analytic continuation of $\log(\ggen)$: occurring due to zeros, branch points, and poles of the algebraic function $\ggen$---which occur are all at algebraic numbers.  For example, if there is a branch point on the boundary of the origin-centred disk of radius $\mathtt{R}$, then $\mathtt{R}$ is the absolute value of a root of the discriminant polynomial of any polynomial equation obeyed by $\ggen$ (when this latter polynomial is thought of as a polynomial with coefficients in $\mathbb{Q}(z)$).  If the failure is due to an (unbranched) zero or a pole of $\ggen$, then the polynomial equation obeyed by $\mathtt{R}$ is given by one of the coefficients of the polynomial equation obeyed by $\ggen$.

In independent work, Clay Cordova and Shu-Heng Shao conjectured that the growth rate $\mathtt{R}$ of sequences of $m$-Kronecker DT invariants are algebraic numbers: see \cite[\S 5]{clay_shao}.  As just described, this conjecture is affirmed via a corollary of algebraicity for the associated generating series for Kronecker DT invariants.  Moreover, Cordova-Shao produce explicit (irreducible, integer coefficient) polynomials whose positive roots are growth rates of the $m$-Kronecker invariants $d(n,2n,m)$ by numerical analysis for $m \in \{3, \cdots ,8\}$ (the notation $d(a,b,m)$ is defined in the beginning of \S \ref{sec:intro_partII}).  When $m=3$, we have $d(n,2n,3) = d(n,n,3)$ (c.f. App. \ref{app:mod_equivs} and references therein); the large $n$ growth rate of $d(n,n,3)$ is a known rational number $\mathtt{R} = \frac{256}{27}$ (given by the exponential of \eqref{eq:c_m} for $m = 3$), so satisfies the easy first degree (irreducible) polynomial: $256-27x$.  For $m>3$, however the corresponding polynomials are no longer first degree (the growth rates are no longer rational) and have surprisingly large integer coefficients.  Using our technique it should be possible to verify Cordova-Shao's polynomials by studying the generating series associated to the $(1,2|m)$-herd spectral networks (which produce the generating series for the $m$-Kronecker DT invariants $d(n,n2,m)$) that may be constructed according to the procedure outlined in \cite[App. D]{glm:sn_spin}.  In this paper, as a corollary of the algebraic equations for generating series deduced for the slope 3/2 invariants $d(3n,2n,3)$, we deduce an irreducible polynomial \eqref{eq:d_2} that divides the discriminant polynomial of \eqref{eq:3_2_3_poly} (see \eqref{eq:Disc_F_3_2}) and whose smallest magnitude root provides the growth rate; just as with the Cordova-Shao polynomials, our polynomial has amusingly large integer coefficients.

A curious corollary of exponential growth of BPS indices is the existence of Hagedorn temperatures: temperatures at which the $\log$ of the partition function---or more accurately correlation functions given by its derivatives---for the theory fails to be analytic (as a function of temperature) due to an exponential growth in the density of states.  Heuristically, if the density of states $\rho(E)$ of a theory grows as $E^{-\alpha} \exp(c E)$ for large energies, then one expects that the partition function $\mathpzc{Z}(\beta) = \int_{0}^{\infty} dE \rho(E) e^{-\beta E}$ (and its derivatives) diverges at the inverse temperature $\beta = c$.  This sort of argument can be made precise in our situation; so let us say a few more relevant words.  First, suppose that we are presented an $\mathcal{N}=2$ field theory on $\mathbb{R}^{3,1}$ and fix a Coulomb branch modulus $u \in \CB'$ off of any walls of marginal stability.   For each fixed $\vartheta \in \mathbb{R}/(2\pi \mathbb{Z})$, the partition function of the full theory---regarded as a function of the inverse temperature $\beta$---includes a summand $\mathpzc{Z}_{\vartheta}$ resulting from the contribution of BPS states with central charge phase $\vartheta$; because we are off any walls of marginal stability, there is an injective mapping
\begin{align*}
	\left \{\text{$\vartheta \in \mathbb{R}/(2 \pi \mathbb{Z})$: $\vartheta$ is the central charge phase of some BPS state } \right\} &\longrightarrow (\Gamma \backslash \{0\})/\mathbb{Z}_{> 0}\\
	\vartheta & \longmapsto r(\vartheta)
\end{align*}
so that all BPS states of central charge phase $\vartheta$ have charges living along a single ray $\mathbb{Z}_{>0} \gamma_{\vartheta}$ for the unique primitive charge vector $\gamma_{\vartheta} \in \Gamma$ associated to the ray $r(\vartheta)$ (if we were faithful to the notation of the previous section we would write $\gamma_{r(\vartheta)}$ instead).  As a result,
\begin{align*}
	\mathpzc{Z}_{\vartheta}(\beta) = \sum_{n=1}^{\infty} N(n) e^{-n \beta |\mathcal{Z}_{u}(\gamma_{\vartheta})|}
\end{align*}
where $N(n)$ is the number of BPS states of charge $n \gamma_{\vartheta}$ (given as the trace of the identity operator on a finite-dimensional Hilbert space), and $\mathcal{Z}_{u}(\gamma_{\vartheta})$ is the central charge of a BPS state of charge $\gamma_{\vartheta}$.  Luckily for us, the BPS index---a weighted trace over the same finite-dimensional Hilbert space used to calculate $N(n)$---provides a lower-bound\footnote{This follows most easily from the definition \eqref{eq:BPS_superdim} of the BPS index in App. \ref{app:BPS_index_quiver} and the fact that the half-hypermultiplet representation is four (complex) dimensional.} on $N(n)$:
\begin{align*}
	N(n) \geq \frac{1}{4} |\Omega(n \gamma_{\vartheta};u)|.
\end{align*}
So if we think of $\mathpzc{Z}_{\vartheta}$ as a series in $\exp(-\beta |\mathcal{Z}_{u}(\gamma_{\vartheta})|)$, and assume an exponential growth rate for BPS indices ($\mathtt{R} < 1$), then it follows that $\mathpzc{Z}_{\vartheta}$ has radius of convergence $< \mathtt{R}$.  This is highly suggestive that the full theory has a Hagedorn temperature 
\begin{align*}
	\tau^{H}_{\vartheta} &\leq - |\mathcal{Z}_{\gamma_{\vartheta}}| \log(\mathtt{R})^{-1}.
\end{align*} 
Under the assumption that $N(n) \sim C |\Omega(n \gamma_{\vartheta};u)|$ as $n \rightarrow \infty$ for some constant $C \geq 1/4$, the inequality above is actually an equality; this seems to be the case when the spectrum of BPS states is controlled by the $3$-Kronecker quiver, where numerical evidence (based on computations of protected spin characters) suggests that $N(n)$ differs from the corresponding BPS index by an overall sign \cite[\S 7.4 and A.1]{wwc}.  What happens at and beyond Hagedorn temperature at this point is unclear; it is expected that the theory undergoes a phase transition whose nature is described by the analytic continuation of the partition function.

There may be an even more direct handle on the Hagedorn temperature---and the possible phase transition that occurs---when we consider the theory in the presence of a (timelike) BPS line defect associated to the magnetically dual charge $\gamma^{\vee}_{\vartheta} \in \Gamma$, i.e. a line defect corresponding to $\gamma^{\vee}_{\vartheta}$ such that $\langle \gamma^{\vee}_{\vartheta}, \gamma_{\vartheta} \rangle = 1$.   From the perspective of a fixed spatial slice, the insertion of this line defect creates the presence of an infinitely massive ``magnetic monopole" if we regard $\gamma_{\vartheta}$ as an electric charge.  In the theory with the insertion, there exists a subset of the BPS states---dubbed \textit{halo bound states}---that have semiclassical descriptions as bound states consisting of a a nucleus/core given by the infinitely massive monopole surrounded by a ``halo" of BPS states of the theory without the insertion, each of charge $\gamma_{\vartheta}$.  The generating series $\gdt_{r(\vartheta)}$ is related to the partition function that counts these halo bound states; an understanding of precisely how $\gdt_{r(\vartheta)}$ fits into the finite-temperature partition function of the theory will requires careful sign-counting and, once again, a better understanding of the protected spin-character.  Nevertheless, identifying the formal parameter $z$ with $\pm \exp(-\beta |\mathcal{Z}_{\gamma_{\vartheta}} |)$ the finite radius of convergence $\mathtt{R}$ of $\gdt_{r}$ suggests that something interesting may happen at the temperature
\begin{align*}
	- |\mathcal{Z}_{\gamma_{\vartheta}}| \log(\mathtt{R})^{-1}.
\end{align*}
	It is interesting to speculate if the algebraic curve defined by the minimal polynomial of $\gdt_{r}$---which defines the domain of the maximal analytic continuation of $\gdt_{r}$---can relay any information about possible phase transitions of the theory in the presence of the line defect.  As described in \S \ref{sec:alg_geom}, this curve may play an important role in the geometry of the moduli space of vacua for the three dimensional effective field theory obtained by compactifying the theory (without the line defect insertion) on a circle.

Of course, here we have only focused on the part of the partition function receiving contributions from BPS states with fixed central charge phase $\vartheta$.  However, for theories whose BPS spectrum is partially controlled by the $m$-Kronecker quiver ($m \geq 3$), it is known that within $\mathbb{R}/(2 \pi \mathbb{Z})$, there exist two intervals $A_{+} \subset \mathbb{R}/(2 \pi \mathbb{Z})$ and $A_{-} = - A_{+}$ along with dense subsets $D_{+} \subset A_{+}$ and $D_{-} = - D_{+}$ such that for each $\vartheta \in D := D_{+} \cup D_{-}$ there exists a BPS state with central charge phase $\vartheta$.  Moreover, numerical evidence seems to support the conjecture that for each $\vartheta \in D$, the associated $\mathbb{Z}_{> 0}$-family of BPS indices $(\Omega(n \gamma_{\vartheta};u))_{n = 1}^{\infty}$ have exponential growth.  With this in mind, one should expect (countably) infinitely many Hagedorn temperatures $(\tau_{\vartheta}^{H})_{\vartheta \in D}$ each temperature $\tau_{\vartheta}^{H}$ corresponding to a possible phase transition related to the exponentially growing degeneracies of BPS states of central charge phase $\vartheta \in D$.

Hagedorn temperatures have been familiar among string theorists for quite some time (see e.g. the work of  Atick and Witten \cite{atick_witten}), with exponential growths in the density of single particle states with the mass seeming to arise as a natural consequence of ``stringy" degrees of freedom. However, string theorists should note that the novel feature of our context is that Hagedorn temperatures arise from field theoretic degrees of freedom.  On the other hand, it may be an interesting to ask\footnote{A question posed to the author quite frequently by Willy Fischler.} if there is an intrinsically stringy explanation for such temperatures, in a manner that is faithful to the Type IIB/M-theoretic origins of theories of class $S[A]$.  Work of Barak Kol \cite{Kol:1998zb} provides heuristics for such an explanation in the context of four dimensional $\mathcal{N}=4,\, SU(3)$ SYM.

It is also interesting to note that the bound states that contribute to the exponential degeneracy have semiclassical sizes that grow at least linearly with their charge \cite[\S 7]{wwc} (or equivalently their mass).  Thermodynamic arguments suggest that growth (by an unbounded function) in (semiclassical) size with mass is a general feature of any four-dimensional field theory with an exponential growth in the density of states.  This behaviour should be familiar to phemenologists in the context of QCD, where the Hagedorn temperature is a phase transition temperature into the quark-gluon plasma due to the exponential degeneracy of resonant states (Hagedorn's ``fireballs") of ever growing size \cite{hag, er:hag}.  In our case, the states contributing to the Hagedorn temperature are not resonances, but honest mass-eigenstates (albeit this is only sensible in the infinite volume theory on $\mathbb{R}^{3,1}$ with unbroken supersymmetry).  In particular, they are ``threshold" bound states (or when quiver language is sensible: polystable representations), whose finely balanced stability is due to the magic of supersymmetry.  Using spectral networks---a tool of theories of class $S[A]$---one may obtain explicit results regarding Hagedorn temperatures; this begs the question if anything interesting can be learned about Hagedorn temperatures , and their associated phase transitions, in field theory by studying theories of class $S[A]$.

\subsubsection{Algebraicity and Geometry} \label{sec:alg_geom}
An enlightening way to view wall-crossing formulae (WCF) is via Kontsevich and Soibelman's suggestion that the collections of DT invariants satisfying WCF are in one-to-one correspondence with a certain class of ``formal\footnote{A formal variety/scheme is essentially a space where all locally defined functions are given by formal series (i.e. not necessarily convergent series).  Among the plethora of references for a precise definition, one can see e.g. \cite[\S II.9]{hartshorne}, or \cite[Ch.1]{demazure} (for those with functorial inclinations).}  Poisson varieties": the data of DT invariants and WCF can be used to construct the formal Poisson variety, and conversely, one can recover the DT invariants from this (formal) variety.   Not surprisingly, this geometric perspective may shed new light on algebraicity.  Indeed, algebraicity hints that the formal variety may not be so formal: its local functions may be series expansions of regular functions on some branched cover of the affine plane---it is as if the variety had been wearing a tuxedo t-shirt all along.

Physicists should be pleased to note that Kontsevich and Soibelman's geometric suggestion lies at the heart of the work of Gaiotto, Moore, and Neitzke (GMN) who connect BPS indices to the geometry of moduli spaces of supersymmetric vacua, a physical manifestation of (a symplectic leaf of) Kontsevich and Soibelman's formal Poisson varieties.  The connection is the following: suppose one is interested in the BPS indices of an $\mathcal{N}=2$ field theory on $\mathbb{R}^{3,1}$, then GMN tell us that these BPS indices should be studied by looking at the $\mathcal{N}=2$ field theory on $\mathbb{R}^{2,1} \times S^{1}$ formed by compactifying our original theory on a circle.  The Coulomb branch of this latter theory is now a hyperk\"{a}hler manifold $\mathpzc{M}$---perhaps appropriately called the ``Seiberg-Witten (moduli) space"\footnote{As suggested to me by Greg Moore.} in reference to early work by Seiberg and Witten \cite{Seiberg:1996nz}; it is the geometry of $\mathpzc{M}$---which can be probed by distinguished atlases of coordinates---that encodes our BPS indices.  In particular, recall that from a hyperk\"{a}hler manifold $\mathpzc{M}$ we can recover a $\mathbb{P}^{1}$-family of holomorphic symplectic manifolds $(\mathpzc{M}_{\zeta},\omega_{\zeta})_{\zeta \in \mathbb{P}^{1}}$ (each element of the family diffeomorphic to one another, but not necessarily holomorphic) if we fix one such complex structure $\zeta$, then GMN provide us with a distinguished atlas of holomorphic-symplectic\footnote{Here we are glossing over important details: the coordinates that GMN consider are only holomorphic on the complement of a distinguished set of codimension 1 loci.  To recover a holomorphic symplectic atlas we can either restrict each coordinate chart to an open domain in the complement of the codimension 1 loci, or take such a restriction and consider its maximal analytic continuation.} coordinate charts that cover $\mathpzc{M}_{\zeta}$.  The BPS indices of the theory are encoded in the transition functions of this atlas.  In particular, the transition functions of the GMN atlas are holomorphic symplectomorphisms of a special type: those completely determined by generating series of BPS indices of the form \eqref{eq:gen_series_intro}.

This is precisely the technique for constructing a Kontsevich-Soibelman formal Poisson variety from some given DT invariants satisfying the constraints of wall-crossing formulae.\footnote{More precisely, the holomorphic-symplectic manifold $\mathpzc{M}_{\zeta}$ for $\zeta \neq 0$ appears to be a symplectic slice of an analytic version of the Kontsevich-Soibelman formal Poisson variety.  For example, in theories of class $S[A]$ the space $\mathpzc{M}_{\zeta}$ (for fixed $\zeta \neq 0$) can be identified \cite{GMN2} with a symplectic slice of Fock and Goncharov's $\mathcal{X}$-spaces \cite{FG,FG2}: explicit examples of what we are calling ``Kontsevich-Soibelman" spaces.}  Actually, because it turns out that the Kontsevich-Soibelman formal Poisson variety is often more than just a ``formal variety", we will refer to it as the ``Kontsevich-Soibelman space" (associated to a set of DT invariants) when emphasizing this point.  To emphasize the importance of algebraicity of our generating series on the construction of the Kontsevich-Soibelman space, we must say a few more words about the glueing transformations alluded to above.

Suppose we start with the data of $\Lambda$ is a lattice and $\langle \cdot, \cdot \rangle$ is an integer valued antisymmetric form on $\Lambda$.  Then the group ring $\mathbb{C}[\Lambda]$---thought of as the complex vector space $\bigoplus_{\lambda \in \Lambda} \mathbb{C}X_{\lambda}$ with product $X_{\alpha} X_{\beta} := X_{\alpha + \beta}$---has an induced Poisson bracket 
	\begin{align}
		\{X_{\alpha}, X_{\beta} \} := \langle \alpha, \beta \rangle X_{\alpha} X_{\beta} = \langle \alpha, \beta \rangle X_{\alpha + \beta}
		\label{eq:Poisson_bracket}
	\end{align}
and may be thought of as the ring of functions on the space of characters 
\begin{align*}
\mathbb{T} = \mathrm{Hom}_{\mathrm{Grp}}(\Lambda, \mathbb{C}^{\times}) \cong (\mathbb{C}^{*})^{\mathrm{rank}(\Lambda)}
\end{align*}
thought of as a complex variety:\footnote{Actually, depending on what sort of algebraic geometers we aspire to be---we would say that $\mathbb{C}[\Lambda]$ is the ring of functions on $\spec \mathbb{C}[\Lambda]$ whose $\mathbb{C}$-valued points are given by $\mathrm{Hom}_{\mathrm{CRing}}(\mathbb{C}[\Lambda], \mathbb{C}) \cong \mathrm{Hom}_{\mathrm{Grp}}(\Lambda, \mathbb{C}^{\times})$.} $X_{\alpha}$ is simply the map that evaluates a character on $\alpha \in \Lambda$.  

To construct a Kontsevich-Soibelman space, our dream is to glue together copies of $\mathbb{T}$.  In the algebraic world, glueing is done along Zariski open subsets, so our glueing transformations should really just be Poisson isomorphisms between Zariski open subsets of copies of $\mathbb{T}$---or if we think of all copies of $\mathbb{T}$ as canonically identified with a single space, abusively denoted by $\mathbb{T}$, we seek Poisson automorphisms between two Zariski open subsets of $\mathbb{T}$: this is precisely the set of \textit{birational} Poisson automorphisms.  In the language of rings of functions, a birational Poisson automorphism is just a ring automorphism of the quotient field $\mathbb{C}(\Lambda)$ that preserves the Poisson bracket on $\mathbb{C}(\Lambda)$ induced\footnote{If we think of the quotient field as consisting of fractions $f/g$ for $f,g \in \mathbb{C}[\Lambda]$ with $g \neq 0$, then the Poisson bracket on the integral domain $\mathbb{C}[\Lambda]$ induces a Poisson bracket on its quotient field by utilizing the Leibniz rule.} from the bracket on $\mathbb{C}[\Lambda]$.

Our favourite such automorphisms will be given by first choosing a quadratic refinement\footnote{If we want to rid ourselves of the choice of a quadratic refinement, we can instead think about copies of tori given by $\spec$ of a twisted version of the group ring $\mathbb{C}[\Lambda]$: the product rule is twisted compared to the usual group rule by $X_{\alpha} X_{\beta} = (-1)^{\langle \alpha, \beta \rangle} X_{\alpha+\beta}$.  Taking the spectrum of the twisted group ring, one obtains a variety that is isomorphic to $\mathbb{T}$ as defined above, but not canonically so: the choice of quadratic refinement specifies an isomorphism.  Working with these twisted tori is advantageous when compared to the annoyance of dealing with global issues of choice for quadratic refinement over various copies of glued tori.} $q: \Lambda \rightarrow \{\pm 1\}$ of $\langle \cdot, \cdot \rangle$ and defining
\begin{align}
	K_{\gamma} : X_{\alpha} \longmapsto  X_{\alpha} (1-q(\gamma) X_{\gamma})^{\langle \gamma, \alpha \rangle}
	\label{eq:KS_transf}
\end{align}
for some given $\gamma \in \Lambda$; this action may be familiar to some readers as a cluster $\mathcal{X}$ transformation (see \cite{FG2005,FG2}): it is a birational automorphism that is a regular automorphism from the Zariski open set $q(\gamma) X_{\gamma} \neq 1$ to itself.  Next, suppose that $(\Lambda, \langle \cdot, \cdot \rangle)$ is the data of a lattice and antisymmetric pairing that appears in the context of our discussion about DT invariants in the previous section: e.g. $\Lambda$ could be the lattice of dimension vectors for some quiver, or (if we are physicists) the charge lattice $\widehat{\Gamma}(U)$ associated to some simply connected $U \subset \CB'$.  Imagine, further, that we are handed:
\begin{enumerate}
	\item a ray $r \in (\Lambda \backslash \{0 \}) /\mathbb{Z}_{> 0}$, generated by $\gamma_{r} \in \Lambda$ a primitive lattice vector, and
	
	\item an associated \textit{finite} sequence of DT invariants $(\Omega(n \gamma_{r}))_{n = 1}^{N}$;
\end{enumerate}	
then we will encode these into the geometry of two algebraic tori glued together by the birational Poisson automorphism:
\begin{align}
	S_{r} := \left(K_{N\gamma_{r}} \right)^{\Omega(N \gamma_{r})} \circ \left(K_{(N-1) \gamma_{r}} \right)^{\Omega([N-1] \gamma_{r})} \cdots \circ \left( K_{2\gamma_{r}} \right)^{\Omega(2 \gamma_{r})} \circ \left(K_{\gamma_{r}} \right)^{\Omega(\gamma_{r})}.
	\label{eq:S_r_finite}
\end{align}
On the other hand, we can can recover the sequence $\left(\Omega(n \gamma_{r}) \right)_{n = 1}^{N}$ if one is presented with the glueing automorphism $S_{r}$: in the subgroup of rational Poisson automorphisms generated by the basic automorphisms \eqref{eq:KS_transf} there happens to be a unique factorization of $S_{r}$ into generators, which allows the extraction of the DT invariants by reading off exponents.  When there are only finitely many DT invariants associated to a ray $r \in \Lambda/\mathbb{Z}_{\geq 0}$, one may continue in this manner: associating to each such ray an automorphism $S_{r}$ which can be used to glue together copies of $\mathbb{T}$ to construct a complex algebraic variety (a cluster $\mathfrak{X}$-variety by construction).

However, as emphasized in previous sections, there are examples where this finiteness condition along each ray does not hold: e.g. DT invariants/BPS indices associated to the $m$-Kronecker quiver ($m \geq 3$) or certain theories of class $S[A_{\geq 2}]$ (e.g. pure $SU(3)$ SYM).  Nevertheless, Kontsevich, Soibelman, and (on the physics side) the GMN crew have found success in encoding infinite sequences of DT invariants in compositions of Poisson automorphisms via the same procedure as above; that is, we form \textit{infinite} compositions
\begin{align}
	S_{r} = \cdots \left(K_{n \gamma_{r}} \right)^{\Omega(n \gamma_{r})} \circ \cdots \circ \left( K_{2\gamma_{r}} \right)^{\Omega(2 \gamma_{r})} \circ \left(K_{\gamma_{r}} \right)^{\Omega(\gamma_{r})}.
	\label{eq:S_r_infinite}
\end{align}
which can be thought of as the limit of a sequence of finite truncations.  If we are physicists, and we believe that the GMN functions \textit{always} provide an honest holomorphic atlas on the space $\mathpzc{M}_{\zeta}$ described above, then we need not worry too much about the meaning of the automorphism $S_{r}$ that glues different coordinate charts together: its finite truncations are required to converge to a holomorphic Poisson automorphism.  From this perspective, if we are to take physicists seriously (wise advice), then we should expect that Kontsevich and Soibelman's space---constructed by glueing together tori via automorphisms of the form \eqref{eq:S_r_finite}---is at least an analytic space for certain (physical) classes of DT invariants.

Suppose for the moment that we do not believe in GMN, or---alternatively---we are militant algebraists that abhor the idea of working with analytic spaces.  We can then ask if there is still a geometric interpretation of the infinite composition \eqref{eq:S_r_infinite} without any assumption on the behaviour of $(\Omega(n \gamma_{r}))_{n = 1}^{\infty}$.  To aid us in the quest toward this answer, we first note that if we truncate the infinite composition \eqref{eq:S_r_infinite} to the composition of birational Poisson automorphisms associated to the first $L \geq 1$ DT invariants, we can compute the resulting transformation $S_{r}^{L}$ to be equivalent to
\begin{align}
	S_{r}^{L}: X_{\alpha} \mapsto X_{\alpha} \left[ \prod_{n=1}^{L} (1 - q(n \gamma_{r}) X_{n \gamma_{r}})^{n \Omega(n \gamma_{r})} \right]^{\langle \gamma_{r}, \alpha \rangle}.
	\label{eq:S_r_trunc}
\end{align}
So, as $L \rightarrow \infty$ the right hand side of \eqref{eq:S_r_trunc} can be expressed as a product of $X_{\alpha}$ with a formal series in $X_{\gamma_{r}}$---already we can see the appearance of generating series, but we will come back to this later.  Thus, in order for $S_{r}$ to be a morphism of a geometric object, that corresponding geometric object must admit formal series in $X_{\gamma_{r}}$ as local functions.  In the world of algebraic geometry, this is indicative that the appropriate space we should be dealing with is a formal scheme.  In this paper, it would not be fair to the reader to induce deliberate coma by delving into a precise theory of formal schemes; instead we will simply mention that a formal scheme is a special type of (locally ringed) space that may include formal series as global functions.  One way way of obtaining a formal scheme is by taking formal completions or, more suggestively, ``formal neighbourhoods" of ordinary varieties.   As the nickname suggests, a formal neighbourhood is a formal version of a normal bundle in differential geometry; namely a formal neighbourhood of a subvariety $A \subset V$ has underlying topological space given simply by $A$, but with functions that are regular along the directions of $A$ and formal series in the normal directions to $A$: one might like to visualize the formal neighbourhood of $A \subset V$ as $A$ sitting inside of $V$ with infinitesimal ``$\infty$-jet" fuzz pointing along the normal directions to $A$.

In our situation, we seek a formal scheme that admits functions that are formal along the ``$\gamma_{r}^{*}$-direction" (to make the dual vector $\gamma_{r}^{*}$ precise, one must choose a splitting of $\mathbb{Z}\gamma_{r} \hookrightarrow \Lambda \rightarrow \Lambda/\mathbb{Z} \gamma_{r}$), i.e. that admits formal series in $X_{\gamma_{r}}$.  Let us try to define this formal scheme by restricting to a formal neighbourhood of a subtorus of $\mathrm{Hom}(\Lambda, \mathbb{C}^{\times})$ that is isomorphic to $\mathrm{Hom}(\Lambda/(\mathbb{Z} \gamma_{r}), \mathbb{C}^{\times})$; such a subtorus has normal directions in the ``$\gamma_{r}^{*}$-direction".  In fact, there is a whole $\mathbb{C}^{\times}$-family of such subtori given by the loci $X_{\gamma_{r}} \equiv \alpha$ for some fixed choice of $\alpha \in \mathbb{C}^{\times}$, but taking formal completions along the subtorus defined by $\alpha$, we get formal series in $(X_{\gamma_{r}} - \alpha)$---not quite what we want.  This suggests that, in order to get formal series in $X_{\gamma_{r}}$ by taking formal neighbourhoods, we need to consider a partial compactification of $\mathbb{T}$ by adding in a divisor corresponding to the missing $X_{\gamma_{r}} \equiv 0$ locus.  The easiest partial compactifications to consider are those coming from toric varieties. Namely, one can embed $\mathbb{T}$ as the dense torus of a toric variety $\mathbb{X}$ which can be thought of as a partial compactification of $\mathbb{T}$ obtained by glueing in a divisor ``at $X_{\gamma_{r}}=0$"; the resulting space is isomorphic to $\mathbb{C} \times (\mathbb{C}^{*})^{\mathrm{rank}(\Lambda) - 1}$ when thought of as a complex variety.  Taking a formal neighbourhood of our added divisor, we obtain a formal scheme---call it $\widehat{\mathbb{X}}$---with the functions we desire: its ring of global functions is isomorphic to $\formal{\mathbb{C}}{X_{\gamma_{r}}} \otimes_{\mathbb{C}} \mathbb{C}[\Lambda/\mathbb{Z} \gamma_{r}]$.  Moreover a careful construction shows that the resulting formal scheme happens to be Poisson with Poisson bracket given by \eqref{eq:Poisson_bracket} extended to formal series via the obvious bilinearity condition. 

Now for the punchline: on this formal scheme there is an interesting family of Poisson automorphisms given by embedding the group (under multiplication) of formal series with constant coefficient 1
\begin{align*}
	A = \{f \in \formal{\mathbb{Q}}{z}: [z^{0}]f =1 \}
\end{align*}
into the group of Poisson automorphisms of  $\widehat{\mathbb{X}}$ via
\begin{equation}
	\begin{aligned}
		\theta: A &\longrightarrow \operatorname{Aut}\left[(\widehat{\mathbb{X}},\{ \cdot, \cdot \}) \right]\\
			    f &\longmapsto (\theta_{f}: X_{\alpha} \mapsto X_{\alpha} f^{\langle \gamma, \alpha \rangle}).
	\end{aligned}
	\label{eq:theta_action}
\end{equation}
We conclude that the infinite composition of Poisson automorphisms $S_{r}$ can be thought of as the automorphism $\theta_{\gdt_{r}}$ associated to the generating series
\begin{align*}
	\gdt_{r} =\prod_{n = 1}^{\infty}(1- (q(\gamma_{r})X_{\gamma_{r}})^{n})^{n \Omega(n \gamma_{r})}.
\end{align*}
In this sense, $S_{r}$ can be thought of as the glueing transformation between copies of various formal (Poisson) schemes formed by the above construction: the result is Kontsevich and Soibelman's formal Poisson variety.

Of course, even the laziest of geometers prefer spaces with better functions than formal series.  This is why the property of \textit{algebraicity}---the property that every generating series $\gdt_{r}$ is an algebraic function over $\mathbb{Q}$ for each $r \in (\Lambda \backslash \{0\})/\mathbb{Z}_{> 0}$---is so fascinating.  Certainly algebraic functions are holomorphic; so it is no surprise that $\theta_{\gdt_{r}}$ secretly ``extends" to a Poisson automorphism defined on a little tubular neighbourhood of the $X_{\gamma_{r}} = 0$ locus sitting inside of $\mathbb{X}$, now thought of as a holomorphic manifold.  In this manner, the Konstevich-Soibelman space can be constructed as an analytic space, just as we expected with physical assumptions.  However, one can go further: the minimal polynomial of $\gdt_{r}$ (a polynomial with coefficients in $\mathbb{Q}[z]$) defines an algebraic curve  $\twid{\mathbb{X}}$ that forms a finite degree cover of $\mathbb{X}$ (the degree of the cover being given by the degree of the minimal polynomial).  As it so happens, one can make $\twid{\mathbb{X}}$ into a Poisson variety and, moreover, $\theta_{\gdt_{r}}$ ``extends" to an honest (global) Poisson automorphism of this cover.  In this manner, one can construct the Kontsevich-Soibelman space as a legitimate variety by glueing together covers of $\mathbb{X}$---if we so wish, we can throw out the $z=0$ locus (and its preimage in the cover) and restrict ourselves to glueing together covers of the original torus $\mathbb{T}$ (the resulting space is birationally equivalent to the construction with $\mathbb{X}$).

Of course, the above all still holds for generating series that are algebraic over $\mathbb{C}$.  Because we have generating series that are algebraic over $\mathbb{Q}$, we could have replaced all occurrences of $\mathbb{C}$ with any $\mathbb{Q}$-algebra---specifically, $\mathbb{Q}$ itself.  E.g. the torus $\mathbb{T}$ could have been defined as $\spec \mathbb{Q}[\Lambda]$ and the toric variety $\mathbb{X}$ with the corresponding construction for a toric variety over $\mathbb{Q}$.  The advantage of such a construction is that we may be able to talk about points in the resulting Kontsevich-Soibelman space in fields other than $\mathbb{C}$.  It is interesting to speculate if such points have any relevance to physics.

To emphasize the geometric constraints that algebraicity provides, we will summarize our discussion in the following manner.  First, let:
\begin{itemize}
	\item $\formal{\mathbb{Q}}{z}^{\mathrm{rat}}_{1}$ denote the subgroup (under multiplication) of series expansions (about $z=0$) of elements of the subgroup of rational functions (under multiplication) generated by $\{(1-z^{k})^{\pm 1}: k \in \mathbb{Z}\}$,

	\item $\formal{\mathbb{Q}}{z}^{\mathrm{rat}}$ denote the subalgebra of formal series that are series expansions of rational functions about $z=0$,
	
	\item $\formal{\mathbb{Q}}{z}^{\mathrm{alg}}$ the subalgebra of formal series that are algebraic (i.e. appear roots of some polynomial in with coefficients in $\mathbb{Q}[z]$), and
	
	\item $\formal{\mathbb{Q}}{z}^{\mathrm{hol}}$ the subalgebra of formal series that have finite radius of convergence (hence are series expansions of holomorphic functions around $z = 0$). 
	
\end{itemize}	
	We have the inclusions (growing increasingly horrified as we read from left to right)
\begin{align*}
	\formal{\mathbb{Q}}{z}_{1}^{\mathrm{rat}} \subset \formal{\mathbb{Q}}{z}^{\mathrm{rat}} \subset \formal{\mathbb{Q}}{z}^{\mathrm{alg}} \subset \formal{\mathbb{Q}}{z}^{\mathrm{hol}} \subset \formal{\mathbb{Q}}{z}.
\end{align*}

With this in mind, the geometry of the Kontsevich-Soibelman space is determined by the associated generating series of DT invariants via the following sequence of progressively more general situations:
\begin{enumerate}
	\item $\gdt \in \mathbb{Q}[z]_{1} \Rightarrow \theta_{\gdt}$ is a finite sequence of cluster $\mathcal{X}$- transformations between algebraic tori ($\cong \mathbb{T}$) $\Rightarrow$ the Kontsevich-Soibelman space is obtained by glueing together tori along the complement of complement of a finite collection of codimension 1 loci (given by the zeros and poles of $\gdt$---which are roots of unity).
	
	\item $\gdt \in \formal{\mathbb{Q}}{z}^{\mathrm{rat}} \Rightarrow \theta_{\gdt}$ is a birational transformation between algebraic tori $\Rightarrow$ the Kontsevich-Soibelman space is obtained by glueing together tori along the complement of a finite collection of codimension 1 loci (given by the zeros and poles of $\gdt$);
	
	\item $\gdt \in \formal{\mathbb{Q}}{z}^{\mathrm{alg}} \Rightarrow \theta_{\gdt}$ is a birational transformation between finite algebraic \textit{covers} of tori $\Rightarrow$ the Kontsevich-Soibelman space is obtained by glueing together covers of tori along the complement of the ramification locus)
	
	\item $\gdt \in \formal{\mathbb{Q}}{z}^{\mathrm{hol}} \Rightarrow \theta_{\gdt}$ is an analytic transformation along open sets in the analytic topology $\Rightarrow$ the Kontsevich-Soibelman space is obtained by glueing together holomorphic tori along open sets to form an analytic space.
	
	\item $\gdt \in \formal{\mathbb{Q}}{z}  \Rightarrow \theta_{\gdt}$ is a Poisson automorphism of formal varieties $\Rightarrow$ the Kontsevich-Soibelman space is obtained by glueing together together partial compactifications (add a codimension 1 divisor ``at $z=0$") of tori along formal neighbourhoods of the added divisor (a formal Poisson variety).
\end{enumerate}

\subsection{Main Results Part II: The Return of the Main Results} \label{sec:intro_partII}

The main examples of this paper center around spectral networks that compute DT invariants associated to representations of the $m$-Kronecker quiver $K_{m}$ shown below.
	\begin{center}
\begin{tikzpicture}
\tikzstyle{block} = [rectangle, draw=blue, thick, fill=blue!10,
text width=16em, text centered, rounded corners, minimum height=2em]

\node at (1.2, 0)[circle,draw=blue,very thick] (second) {$q_{2}$};

\node at (-1.2, 0)[circle,draw=blue,very thick] (first) {$q_{1}$};

\node at (0,0.1) {{\Huge \vdots}};

\draw [decorate,decoration={brace, amplitude=6pt, mirror},xshift=0pt,yshift=-1pt,red,thick]
(-1.2,-0.9) -- (1.2,-0.9) node[black, midway,below,yshift = -6pt]{{\tiny $m$ arrows}};
\draw[-latex] (second) to  [bend right = 35] (first);
\draw[-latex] (second)  to  [bend right = 60] (first);
\draw[-latex] (second)  to  [bend right = -35] (first);
\draw[-latex] (second)  to  [bend right = -60] (first);
	\end{tikzpicture}
\end{center}

Specifically, let $d(a,b,m)$ denote the quiver DT invariants $\Omega^{\mathrm{DT}}(a q_1 + b q_2;Z)$ defined with respect to a \textit{wild} stability condition: a homomorphism $Z: \mathbb{Z}[q_{1}, q_{2}] \rightarrow \mathbb{C}$ such that $\arg(Z(q_{1})) < \arg(Z(q_{2}))$.  Let us begin by recalling some results about the ``slope-1" invariants $d(n,n,m)$:  If we encode these invariants via the formal series
\begin{align*}
	\gdt_{1} := \prod_{n = 1}^{\infty} \left(1 - (-1)^{n} z^{n} \right)^{n d(n,n,m)}
\end{align*}
then there is an algebraic relation satisfied by $\gdt_{1}$; specifically, $\gdt_{1} = P^{m}$ where $P$ is the unique solution to the algebraic relation
\begin{align}
	0 := P - z P^{(m-1)^2} - 1.
	\label{eq:P_rel}
\end{align}
which is a formal series in $z$ with constant coefficient 1.  The algebraic relation \eqref{eq:P_rel} was first postulated in the mathematics literature by Kontsevich and Soibelman.  Later it was proven by Reineke using a functional equation that relates DT invariants for the $m$-Kronecker and Euler characteristics of stable moduli:  using known results about Euler characteristics of the moduli space of $m$-Kronecker stable representations for dimension vectors $n(q_{1} + q_{2})$, the functional equation determines the algebraic relation \eqref{eq:P}.   

On the other hand, in \cite{wwc}, the algebraic relation \eqref{eq:P_rel} was independently produced by utilizing spectral networks: specifically, a family of spectral networks called $m$-herds.  Far from pure mathematical curiousity, $m$-herds are realized as WKB networks (a.k.a $\CW$-networks) at certain regions on the Coloumb branch of pure four-dimensional $\mathcal{N}=2,\, SU(3)$ SYM.  With this explicit physical realization, one can use BPS quiver techniques to argue that the BPS indices associated to $m$-herds are, indeed, the DT invariants $d(n,n,m)$.

This result was rather surprising from a physical perspective due to the resulting exponential growth rate of BPS indices when $m \geq 3$: $|d(n,n,m)| \sim \text{const} \cdot n^{-5/2} \exp(c_{m} n)$ where the growth constant $c_{m} >0$ is given by
\begin{align}
	c_{m} := (m-1)^2 \log \left[(m-1)^2 \right] - m(m-2) \log \left[m(m-2) \right];
	\label{eq:c_m}
\end{align}
because the actual number of BPS states of a particular charge is (up to a constant) bounded below by the absolute value of the BPS index, the result is an exponential growth in the density of states.

In this paper our main example is concerned with the 3-Kronecker DT invariants $d(3n,2n,3)$; by encoding them into the generating series
	\begin{align*}
		\gdt_{3/2} &= \prod_{n=1}^{\infty}(1 - (-1)^{n} z^{n})^{n d(3n,2n,3)} \in \formal{\mathbb{Z}}{z}
	\end{align*}
we show, using a spectral network we call the $(3,2|3)$-herd---a generalization of the $m$-herd family---that $\gdt_{3/2}$ is an algebraic function over $\mathbb{Q}$.  Specifically, $\gdt_{3/2} = M V W$ for three formal series $M,V,W \in \formal{\mathbb{Z}}{z}$ that satisfy the algebraic relations
\begin{equation}
	\begin{aligned}
		M &= 1 + z M^{4} \left\{ (1 + V) (1 + V - W)^2 [V^2(1+W) - 1]^{3} \right\},\\
		0 &= (-1 + V) (1 + V)^2 + (1 + V^3) W - V (M + V) W^2,\\
		0 &=  V \left(V^2 -1 \right) - \left[M (V + 1) + V(V-2) - 1 \right] W.
	\end{aligned}
	\label{eq:func_eqs}
\end{equation}	 
With a bit of elimination theory, one can, moreover determine a $39$th degree polynomial $\mathcal{F}_{(3/2|3)} \in (\mathbb{Z}[z])[t]$ (see \eqref{eq:3_2_3_poly}) with coefficients in $\mathbb{Z}[z]$ that is irreducible as a polynomial in $\mathbb{C}[z,t]$, and has $\gdt_{3/2}$ as the unique root with series representation of the form $1 + \sum_{n = 1}^{\infty} t_{n} z^{n}$.

When trying to determine the slope 3/2 DT invariants, the associated Euler characteristics of stable moduli are not (to the author's current awareness) a priori known; so one cannot provide a derivation of the DT invariant generating series using Reinke's functional equation as in the slope 1 situation.  On the other hand, one can utilize Reineke's functional equation to determine the generating series $\geul_{3/2} \in (\mathbb{Z}[z])[e]$ of the associated Euler characteristics of stable moduli; in fact, the slope 3/2 Euler characteristic generating series is an algebraic function over $\mathbb{Q}$, satisfying an irreducible 9th degree polynomial with coefficients in $\mathbb{Z}[z]$ and irreducible as a polynomial in $\mathbb{C}[z,t]$ (see \eqref{eq:3_2_eul_relation}).  This result is an example of Thm.~\ref{thm:dt_euler_alg} below (an easy corollary of Reinke's functional equation): generating series $\gdt_{r}$ for slope $r$ DT-invariant generating series is an algebraic function (over $\mathbb{Q}$) if and only if the generating series $\geul_{r}$ for the associated Euler characteristics of stable moduli is an algebraic function (over $\mathbb{Q}$).  Moreover, an explicit Euler characteristic algebraic relation can be determined from a given DT invariant algebraic relation and vice versa.

\subsection{Protected Spin Characters and Functional Equations}
One future direction of this work---one that experts will note we have left out of most of the introduction---involves protected spin characters: a powerful generalization of the BPS index that includes the spin content of BPS states.  Indeed, using a generalization of the spectral network machinery, the work of Galakhov-Moore-Longhi in \cite{glm:sn_spin} suggests that the algebraic equations satisfied by generating series of BPS indices may actually be specializations of \textit{functional} equations for generating functions of protected spin characters.  In particular, they show that the generalization of the generating series $P$ of \eqref{eq:P_rel}---now a function $\mathsf{P}$ of two variables ``$z$" and ``$y$"---satisfies a functional equation of the form
\begin{align*}
	\mathsf{P}(z,y) = 1 + z \prod_{s=-(m-2)}^{m-2} \mathsf{P}(z y^{2s},y)^{m-1-|s|}
\end{align*} 
where the variable $y$ helps encode spin content, and we should have $P = \mathsf{P}(z,-1)$; indeed, the specialization to $y = -1$ recovers the algebraic equation \eqref{eq:P_rel}.  It is tempting to imagine what horrors (or simplifications) may be in store for the functional equation generalizations of the algebraic equations stated in this paper.

\subsection{Organization of the Paper}
The paper is organized in the manner listed below.
\begin{enumerate}
	\item \ref{sec:m_herds}: A brief review of some relevant results in \cite{wwc}.

	\item  \ref{sec:non_diag}: Continuation of the work in \cite{wwc} to generalizations of the $m$-herd networks.  This section culminates in the statement of the result \ref{eq:func_eqs} (derived in Appendix \ref{sec:derivation}), the algebraic relation \eqref{eq:3_2_dt_relation} satisfied by $\gdt_{3/2}$, and the resulting analysis of the corresponding BPS index asymptotics. 
	
	\item \ref{sec:stab_eul}: A review Reineke's functional equation that relates DT invariants and Euler characteristics of stable moduli for the $m$-Kronecker quiver.
	
	\item \ref{sec:alg_sn}: An in-depth discussion of how spectral networks lead to algebraicity.  The precise statement of algebraicity is included in Claim \ref{claim:func_eqn} and Corollary \ref{cor:alg_BPS}.
	
	\item \ref{sec:alg_asymp}: The derivation of how algebraicity of generating series lead to exponential asymptotics.
\end{enumerate}

The statements of (most) of the main results can be rephrased purely in terms of DT invariants/BPS indices and so can be understood without any knowledge of the spectral network machinery.

\section*{Acknowledgements}
I would like to thank Andrew Neitzke and Gregory Moore for useful discussions and reviews of preliminary drafts; C. Cordova, D. Galakhov, P. Longhi, Y. Soibelman, and  T. Weist for their useful discussions; and W. Fischler for (occasionally educational) amusement.  This work was supported by NSF Research Training Group award DMS-0636557 and DOE grant DOE-SC0010008.  The majority of this paper was written while the author was a graduate student at the University of Texas at Austin.

\section{Global Notation}
	If $F \in \formal{\mathbb{C}}{z}$ is a formal series in the variable $z$
	\begin{align*}
		F &= \sum_{n = 0}^{\infty} f_{n} z^{n}
	\end{align*}
	we will frequently use the notation
	\begin{align*}
		[z^{n}] F := f_{n}.
	\end{align*}

	 The term \textit{generating series} $\mathsf{G}$ for a sequence of rational numbers $(\beta_{n})_{n = 1}^{\infty} \subset \mathbb{Q}$ (which are expected to be integers in all of our examples) will mean a formal series 
	\begin{align*}
		\mathsf{G} &= 1 + \sum_{n = 1}^{\infty} g_{n} x^{n} \in \formal{\mathbb{Q}}{x}
	\end{align*}	
	such that the sequence $\beta_{n}$ appears in the ``Euler-product factorization" of $\mathsf{G}$:
	\begin{align}
		\mathsf{G} &= \prod_{n = 1}^{\infty}(1 - x^{n})^{n \beta_{n}}.
		\label{eq:gen_gen_func}
	\end{align}
		For any formal series with constant coefficient $1$, such a factorization exists and is unique.  Indeed, the $\beta_{n}$ can be extracted from $\mathsf{G}$ by taking the logarithm of both sides of \eqref{eq:gen_gen_func} and applying M\"{o}bius inversion:
	\begin{align}
		\beta_{n} = -\frac{1}{n^2} \sum_{k|n} k \mu \left(\frac{n}{k} \right) \left([x^{k}] \log(\mathsf{G}) \right),
		\label{eq:beta_from_ggen}
	\end{align}	
	 where $\mu: \mathbb{Z}_{>0} \rightarrow \{0,1,-1\}$ is the M\"{o}bius-mu function, $\log(\mathsf{G}) \in \formal{\mathbb{Q}}{x}$ is defined by composing the Taylor series expansion of $\log(1+x)$ around $x=0$ with $\mathsf{G}-1$, and $[x^{k}] \log(\mathsf{G})$ is the coefficient of $x^{k}$ in the formal series $\log(\mathsf{G})$; in terms of (formal) derivatives,
	\begin{align*}
		[x^{d}] \log(\mathsf{G}) &= \frac{1}{d!} \left[ \frac{d^n}{dx^{d}} \log \left(\mathsf{G} \right) \right]_{x = 0}.
	\end{align*}
	Note that \textit{any} formal series in $\formal{\mathbb{Q}}{z}$ with constant coefficient 1 is a generating series for a sequence of rational numbers defined by \eqref{eq:beta_from_ggen}.
	
	\begin{remark}
		In \S \ref{sec:asymptotics} we will also allow generating series $\ggen \in \formal{\conj{\mathbb{Q}}}{x}$ of algebraic numbers $(\beta_{n})_{n = 1}^{\infty} \subset \conj{\mathbb{Q}}$.
	\end{remark}

 In practice the generating series of interest will be series in a distinguished variable we will denote by $z$.  Two important generating series will be denoted by symbols (in sans-serif font) that we make an effort to protect.
	\begin{itemize}
		\item The generating series for DT invariants or BPS indices will be denoted by $\gdt$; it will be a generating series in the variable $\twid{z} = \pm z$ (where the sign depends on the context).
	
		\item The generating series for Euler characteristics of stable quiver moduli will be denoted by $\geul$; it will be a generating series in the variable $z$.
	\end{itemize}
	
	Other formal series in $z$, that are closely related to generating series, will also appear throughout with protected symbols:
	\begin{itemize}		
		\item the symbol $P$ will denote the formal series that controls various sorts of generating series (street-factors, soliton generating series, and the generating series for the corresponding BPS indices) related to $m$-herds;
		
		\item the symbols $M,\, V,$ and $W$ will denote the formal series that control various sorts of generating series related to $(3,2|3)$-herds (introduced below). 
	\end{itemize}
		
The $m$-Kronecker quiver will be denoted by the quiver with vertices $q_{1}$ and $q_{2}$ and $m$ arrows from $q_{2}$ to $q_{1}$.  In its physical embodiment, the $m$-Kronecker quiver will arise as a sub-quiver of the BPS quiver, and the vertices $q_{1},\, q_{2}$ will be identified with distinguished ``charges" $\gamma_{1}$ and $\gamma_{2}$ respectively.  DT invariants associated to the $m$-Kronecker quiver (equipped with a non-trivial stability condition) with dimension vector $(a,b) \in \mathbb{Z}^{2}$ will be denoted in two ways:
		\begin{enumerate}
			\item $d(a,b,m)$ if we are referring to results purely from the mathematical literature.
		
			\item $\Omega \left(a \gamma_{1} + b \gamma_{2} \right)$ if we are interpreting the DT invariants as BPS indices in some four-dimensional $\N = 2$ field theory.  The parameter $m$ will be clear from context.
		\end{enumerate}

\section{A Brief Review of \texorpdfstring{$m$-Herds}{m-herds}} \label{sec:m_herds}
We first roughly sketch the wall-crossing motivation behind the definition of an $m$-herd, as precisely defined in \cite{wwc}.  For the specifics of spectral networks the reader is referred to \cite{sn,snn}; for a detailed account of the definitions and techniques used to obtain the results stated in this paper, the reader is referred to \cite{wwc}: Section 2 and Appendices B-D.

\subsection{Setting, Terminology, and Conventions} \label{sec:setting}
We begin our journey in the realm of theories of class $S[A_{K-1}]$ for $K \geq 3$.  For concreteness we take our canonical example to be pure $SU(3)$ SYM which is given by the theory $S[A_{2}, C, D]$ where $C = \mathbb{CP}^1$, and $D$ is a specific pair of defect operators located at $0$ and $\infty$ (providing irregular punctures).  The discussion in this paper does not refer to the details of the punctures so nothing is lost by imagining $C$ as the cylinder $S^{1} \times \mathbb{R}$.  We now recall some essential notation and terminology.

\begin{definition}[Definitions]\
	\begin{enumerate}
		\item The \textit{Coulomb branch} $\CB$ of $S[A_{K-1},C,D]$ is the set of tuples $(\phi_{2}, \cdots, \phi_{K})$ of holomorphic $r$-differentials $\phi_{r}$ with singularities at $\mathfrak{s}_{1},\cdots, \mathfrak{s}_{n} \in C$ prescribed by the defect operators $D$.

		\item Let $u = (\phi_{2}, \, \cdots, \, \phi_{r} ) \in \CB$ and denote the holomorphic cotangent bundle of $C$ as $\CT^* C$.  Then the \textit{spectral cover} (a.k.a Seiberg-Witten curve) is a $K$-sheeted branched cover $\pi_{u}: \Sigma_{u} \rightarrow C$, where $\Sigma_{u}$ is the subvariety
	\begin{equation}
		\Sigma_{u} := \{\lambda \in \CT^*C: \lambda^{K} + \sum_{r=2}^{K} \phi_{r} \lambda ^{K-r} = 0\} \subset \CT^* C,
	\end{equation}
	and the projection $\pi_{u}$ is the restriction of the standard projection $\CT^* C \rightarrow C$.
	
		\item  On $\CB$ there may be (complex) codimension-1 loci where a cycle of $\Sigma_{u}$ degenerates.  Let $\CB^{*} = \CB - \text{\{degeneration loci\}}$.  Then  $\widehat{\Gamma} \rightarrow \CB^{*}$ is the local system of charge lattices.  In the theories of type $A_{K-1}$, the fibre $\widehat{\Gamma}_{u}$ is a sublattice of $H_{1}(\Sigma_{u};\mathbb{Z})$, but for our purposes, nothing will be lost if we take $\widehat{\Gamma}_{u}$ to be the full lattice:
	\begin{align*}
		\widehat{\Gamma}_{u} := H_{1}(\Sigma_{u};\mathbb{Z}).
	\end{align*}
	For each $u \in \CB^{*},\, \widehat{\Gamma}_{u}$ is equipped with the skew-symmetric pairing on homology: $\langle \cdot, \cdot \rangle_{u}: \widehat{\Gamma}_{u}^{\otimes 2} \rightarrow \mathbb{Z}$.

	\item $Z \in \widehat{\Gamma}^{*} \otimes_{\mathbb{Z}} \mathbb{C}$ is the central charge function.\footnote{It is given fibrewise via period integrals on $\Sigma_{u}$: $Z_{u}: \gamma \mapsto \int_{\gamma} \lambda_{u}$, where $\lambda_{u}$ is the pullback of the tautological 1-form on $\CT^{*}C$ to $\Sigma_{u}$ ($\lambda_{u}$ is the Seiberg-Witten differential).}
	
	\end{enumerate}
\end{definition}

\begin{definition}[Notation]
Unless otherwise noted, from this point on we will work over fixed $u \in \CB^{*}$ and use the streamlined notation $\Sigma := \Sigma_{u}$ and $\Gamma := \widehat{\Gamma}_{u}$.
\end{definition}

\begin{definition}\
\begin{enumerate}
	
	\item  Denoting the unit tangent bundle to $\Sigma$ via $UT \Sigma$, and defining $H$ to be the homology class represented by a 1-chain that wraps once around some fibre of $UT \Sigma \rightarrow \Sigma$, then 
		\begin{align*}
			\twid{\Gamma} := H_{1}(UT \Sigma; \mathbb{Z})/(2H);
		\end{align*}	
		it is a $\mathbb{Z}/2\mathbb{Z}$-extension of the charge lattice $\Gamma$.  We abuse notation and denote the image of $H$ in the quotient by $H$ again.
		
	\item $\twid{\left( \cdot \right)}: \Gamma \rightarrow \twid{\Gamma}$, taking $\gamma$ to $\twid{\gamma}$, is the \textit{standard lift} defined in App. \ref{app:standard_lift} and \cite[App. E]{wwc}.
	
	\item Let $L$ be a finitely generated commutative monoid, then $\formal{\mathbb{Z}}{L}$ is the ring whose underlying set is given by formal expressions of the form
	\begin{align*}
		\sum_{\alpha \in L} c_{\alpha} X_{\alpha}
	\end{align*}	
	such that
	\begin{enumerate}
		\item $c_{\alpha} \in \mathbb{Z},\, \forall \alpha \in L$;
		
		\item $\{\alpha \in \Gamma: c_{\alpha} \neq 0\}$ is well-ordered\footnote{This condition ensures there is a least-element among the $c_\alpha$; it allows us to define the product of two such formal series in the obvious way.  If we did not have this condition, the obvious multiplication would fail for the same reason that it only makes sense to take products of Laurent series (series infinite in only one direction) as opposed to series that are infinite in ``two directions".} with respect to a total order such that for any $l \in L$ we have $l < n l$ if $n > 1$.  In particular, if $L$ is a finitely generated submonoid of $\twid{\Gamma}$ and we are equipped with a central charge function $Z: \Gamma \rightarrow \mathbb{C}$, we have in mind a total order induced by (a possibly small perturbation of) the mass function:
		\begin{align*}
			L \hookrightarrow \twid{\Gamma} \overset{\mod H} {\longrightarrow} \Gamma \overset{|Z|}{\rightarrow} \mathbb{R}_{ \geq 0}.
		\end{align*}
	\end{enumerate}
	The ring structure is given (with identity $1 = X_{0}$) by the sum and product structure:
	\begin{align*}
		\sum_{\alpha \in L} c_{\alpha} X_{\alpha} + \sum_{\beta \in L} d_{\beta} X_{\beta} &= \sum_{\alpha \in L} \left( c_{\alpha} + d_{\alpha} \right) X_{\alpha} ; \\
		\left(\sum_{\alpha \in L} c_{\alpha} X_{\alpha} \right) \left(\sum_{\beta \in L} d_{\beta} X_{\beta} \right) &= \sum_{\alpha + \beta = \delta} c_{\alpha} d_{\beta} X_{\delta}.
	\end{align*}		
	$\formal{\mathbb{Z}}{L}$ contains the group-ring $\mathbb{Z}[L]$ as all such formal series with only finitely many non-vanishing coefficients.
	\end{enumerate}
\end{definition}

\begin{numrmk} \label{rmk:gen_choice}
	We will mostly be concerned with $\formal{\mathbb{Z}}{L}$ when $L \leq \twid{\Gamma}$ is a one-dimensional sub-lattice.  In such a situation, then $L$ has two possible generators (differing by a sign) and $\formal{\mathbb{Z}}{L}$ has two corresponding identifications, each corresponding to a choice of generator, with formal series in a single variable.
\end{numrmk}

\begin{definition}[Terminology]
	We will refer to WKB spectral networks as $\CW$-networks.  When referring to the WKB network at a particular choice of $u \in \CB^{*}$ and $\vartheta \in S^{1}$, we call the network $\CW_{\vartheta}$ (the point $u \in \CB^{*}$ will be clear from context).  Abstract spectral networks (which do not necessarily have any WKB realization) will be referred to as either ``spectral networks" or simply ``networks".
\end{definition}

\begin{definition}\
	\begin{enumerate}
		\item Recall to each street $p$ of a spectral network, we associate a formal series $Q(p) \in \formal{\mathbb{Z}}{\twid{\Gamma}}$ called a \textit{street-factor}.  The constant term (the coefficient in front of $X_{0} = 1$) of a street-factor is always 1.
	
		\item A street $p$ is \textit{two-way} if $Q(p) \neq 1$.  A spectral network is \textit{degenerate} if it contains a two-way street.
		
		\item The collection of two-way streets of a spectral network is its \textit{degenerate skeleton}.
	\end{enumerate}
\end{definition}

\begin{remark}[Notes]\

	\begin{itemize}

	\item The spectral networks of interest in this paper hypothetically correspond to BPS states and, hence, are all degenerate.  Furthermore, all networks will have the property that (up to a choice of sign) there exists a unique primitive $\gamma_{c} \in \Gamma$ such that every street-factor $Q(p)$ is an element of the subring $\formal{\mathbb{Z}}{\twid{\Gamma}_{c}} \leq \formal{\mathbb{Z}}{\twid{\Gamma}}$ generated by $\twid{\Gamma}_{c} := \mathbb{Z} \twid{\gamma}_{c}$.  Off of walls of marginal stability in $\CB$, every degenerate $\CW$-network satisfies this property.  
	
	\item Only three sheets of $\Sigma$ will be relevant for any network that we will consider; following the conventions of \cite{wwc}, in some local coordinate chart we will trivialize the spectral cover and denote two-way streets of type $12$ by red lines, streets of type $13$ by blue lines, and streets of type $13$ by fuchsia lines.
	
	\item Some network diagrams will be drawn with only the two-way streets.  Following the argument in \cite{wwc} Appendix C.3, one may decorate any such network with an arbitrarily complicated ``background" of one-way streets without modifying any computations relevant to the generating series attached to the two-way streets.  This is a particularly convenient observation as any such degenerate network that happens to be realized as a $\CW$-network will usually have such complicated backgrounds.
	
	\item Some diagrams in this paper that also appear in \cite{wwc} are mirror images of their original versions in \cite{wwc} (e.g. red streets run northwest instead of northeast and blue streets run northeast instead of northwest).  As a result, the statements about $m$-herds are can be superficially translated via the map \cite{wwc} $\rightarrow$ ``this paper" given by $\gamma' \mapsto \gamma_{1}$ and $\gamma \mapsto \gamma_{2}$.
	
	\end{itemize}

\end{remark}

  \begin{definition}\
 	 \begin{enumerate}
 		 \item	We will say a formal series $\gdt \in \formal{\mathbb{Z}}{\twid{\Gamma}_{c}} \cong \formal{\mathbb{Z}}{X_{\twid{\gamma_{c}}}}$ \textit{generates} $\Omega(n \gamma_{c}) $ (alternatively: $\gdt$ is the \textit{generating series for} $(\Omega(n \gamma_{c}) )_{n = 1}^{\infty}$) if 
  	\begin{align}
  		\gdt = \prod_{n = 1}^{\infty} (1 - X_{n \twid{\gamma_{c}}})^{n \Omega(n \gamma_{c})}.
  		\label{eq:gen_omega_def}
  	\end{align}
  	For notational convenience (when $\twid{\gamma}_{c}$ is known from context), it will be useful to define the variable
  	\begin{align}
  		\twid{z} :=  X_{\twid{\gamma_{c}}};
  		\label{eq:twid_z_def}
  	\end{align}
  	so that we may simply write,
  	\begin{align*}
  		\gdt = \prod_{n = 1}^{\infty} (1 - \twid{z}^{n})^{n \Omega(n \gamma_{c})}.
  	\end{align*}
  	
		\item If $u \in \CB$ is an $m$-wild point, and $\gamma_{1},\, \gamma_{2} \in \widehat{\Gamma}_{u}$ are the\footnote{The use of ``the" here may suggest that there is only one $m$-Kronecker BPS subquiver, with non-trivial stability condition, at $u \in \CB$.  This may not be the case, but implicitly we will always restrict our attention to a single $m$-Kronecker BPS subquiver (corresponding to a particular choice of $\gamma_{1}$ and $\gamma_{2}$).} two hypermultiplet charges such that $\langle \gamma_{1}, \gamma_{2} \rangle = m \in \mathbb{Z}_{>0}$, then the generating series for the BPS indices $\left( \Omega\left[ n \left(a \gamma_{1} + b \gamma_{2} \right) \right] \right)_{n=1}^{\infty}$ will be denoted by $\gdt_{a/b}$. 

	\end{enumerate}
  \end{definition}
Note that by taking the logarithm of both sides of \eqref{eq:gen_omega_def}, applying M\"{o}bius inversion, and (for simplicity of notation) defining $\twid{z} := X_{\twid{\gamma_{c}}}$ , one finds:
\begin{align}
	\Omega(n \gamma_{c}) = -\frac{1}{n^2} \sum_{d|n} \mu \left(\frac{n}{d} \right) \frac{1}{(d-1)!} \left[ \frac{d^n}{d\twid{z}^n} \log \left(\gdt \right) \right]_{\twid{z} = 0},
	\label{eq:omega_invert}
\end{align}
where $\mu: \mathbb{Z}_{>0} \rightarrow \{0,1,-1\}$ is the M\"{o}bius-mu function.

\begin{numrmk} \label{rmk:reineke_vs_me}
	When the BPS indices $\Omega(n \gamma_{c})$ are computed at an $m$-wild point, and are expected to coincide with DT invariants associated to the $m$-Kronecker quiver, it is useful to compare our definition of generating series in \eqref{eq:gen_omega_def} to the generating series appearing in the work of Reineke \cite{reineke:integrality}.  Indeed, if $\gamma_{c} = a \gamma_{1} + b \gamma_{2}$ as argued in Appendix \ref{app:signs}, it is natural to define the variable
	\begin{align}
		z := (-1)^{m a b + a + b} X_{\twid{a \gamma_{1} + b \gamma_{2}}}.
		\label{eq:z_var_def}
	\end{align}
Thus,
\begin{align*}
  		\gdt_{a/b} &= \prod_{n = 1}^{\infty} (1 - (-1)^{n \left(m a b + a + b \right)} z^{n})^{n \Omega(n \gamma_{c})} \in \formal{\mathbb{Z}}{z}\\
  		&=\prod_{n = 1}^{\infty} (1 - \left((-1)^{N} z \right)^{n})^{n \Omega(n \gamma_{c})},
\end{align*}
where $N := mab - a^2 - b^2$.  It follows that $\gdt_{a/b}$ coincides precisely with the generating series ``$G_{\mu}(t)$" of \cite[\S 6]{reineke:integrality} with $\mu = a/b$ and $t = z$.
\end{numrmk}

\subsection{Generating \texorpdfstring{$m$-herds via Wall-Crossing}{m-herds via Wall-Crossing}} \label{sec:m_herds_wallcrossing}
Suppose there exists a point $u_{0} \in \CB^{*}$ and a path $U: [0,1] \rightarrow \CB^{*}$ beginning at $u_{0}$ ($U(0) =u_{0}$) satisfying the following:

\begin{enumerate}
	\item At $u_{0}$ there are two BPS states with respective charges $\gamma_{1}, \gamma_{2} \in \Gamma_{u(0)}$, intersection pairing $\langle \gamma_{1}, \gamma_{2} \rangle = m \geq 1$, and BPS indices $\Omega(\gamma_{i}) = 1$ for $i=1,2$ (i.e. BPS hypermultiplets). 
	
	\item $U$ crosses a wall of marginal stability for this pair of hypermultiplets at some time $t_{*} \in (0,1)$.

	\item For $t \in (t_{*}, 1]$, $u$ does not cross any other walls of marginal stability involving the bound state of charge $\gamma_{1} + \gamma_{2}$ (produced by the wall-crossing at $U(t_{*})$).
\end{enumerate}

  Suppose further that the degenerate networks corresponding to the hypermultiplets at $u_{0}$ (i.e. the degenerate $\CW_{\vartheta_{i}}$ networks at phases $\vartheta_{i} = \arg \left(Z_{\gamma_{i}} \right)$ for $i = 1,2$) appear as simple saddle-connections.  Then the quasi-imaginative reader may be able to visualise a process by which $m$-herds are generated.  Indeed, let $\gamma_{i}(t)$ denote the parallel transport of $\gamma_{i}(t), i=1,2$ along the path $U$ from $U(0)$ to $U(t)$; further, for $t \in [0,1]$ define
\begin{align*}
	\vartheta_{i}(t) = 
	\left\{	
	\begin{array}{ll}
		\arg \left(Z_{\gamma_{i}(t)} \right) & \text{for $t \leq t_{w}$}\\
		\arg\left(Z_{\gamma_{1}(t) + \gamma_{2}(t)} \right) & \text{for $t > t_{w}$}
	\end{array}
	\right.,\, i = 1,2.
\end{align*}
Then $\{\vartheta_{1}(t)\}_{t \in [0,1]}$ and $\{\vartheta_{2}(t)\}_{t \in [0,1]}$ define two families of phases, equal to the central charge phases of the BPS states of charges $\gamma_{1}$ and $\gamma_{2}$, that are distinct for $ t \in [0, t_{*})$ but are equal (to the central-charge phase of the bound state $\gamma_{1} + \gamma_{2}$) for $t \in [t_{*},1]$.  Then looking at the family of networks $\{\CW_{\vartheta_{i}(t)} \}_{t \in [0,1]}$ we should see two distinct saddle connections for $t \in [0, t_{*})$ that combine into a single network on the wall of marginal stability at $t = t_{*}$ (c.f. Fig.~\ref{fig:saddle_conns}).  This network, which supports the linearly independent charges $\gamma_{1}(t_{*})$ and $\gamma_{2}(t_{*})$, should just be the superposition of two saddle connections with $m$ distinct, transverse intersections. Now the claim is that as $t$ continues to increase, the network at $t_{*}$ will resolve via the growth of a new two-way street growing from each of the intersection points.  Specifically, if we choose a local trivialization of our cover around each intersection point, and let the intersecting saddle connections be of types $12$ and $23$, then a new two-way street of type $13$ should grow from the intersection.  A visualization of this process is shown in Fig.~\ref{fig:3-herd_motivation} for the case $m=3$.

\begin{figure}[t!]
	\begin{center}
		 \includegraphics[scale=0.55]{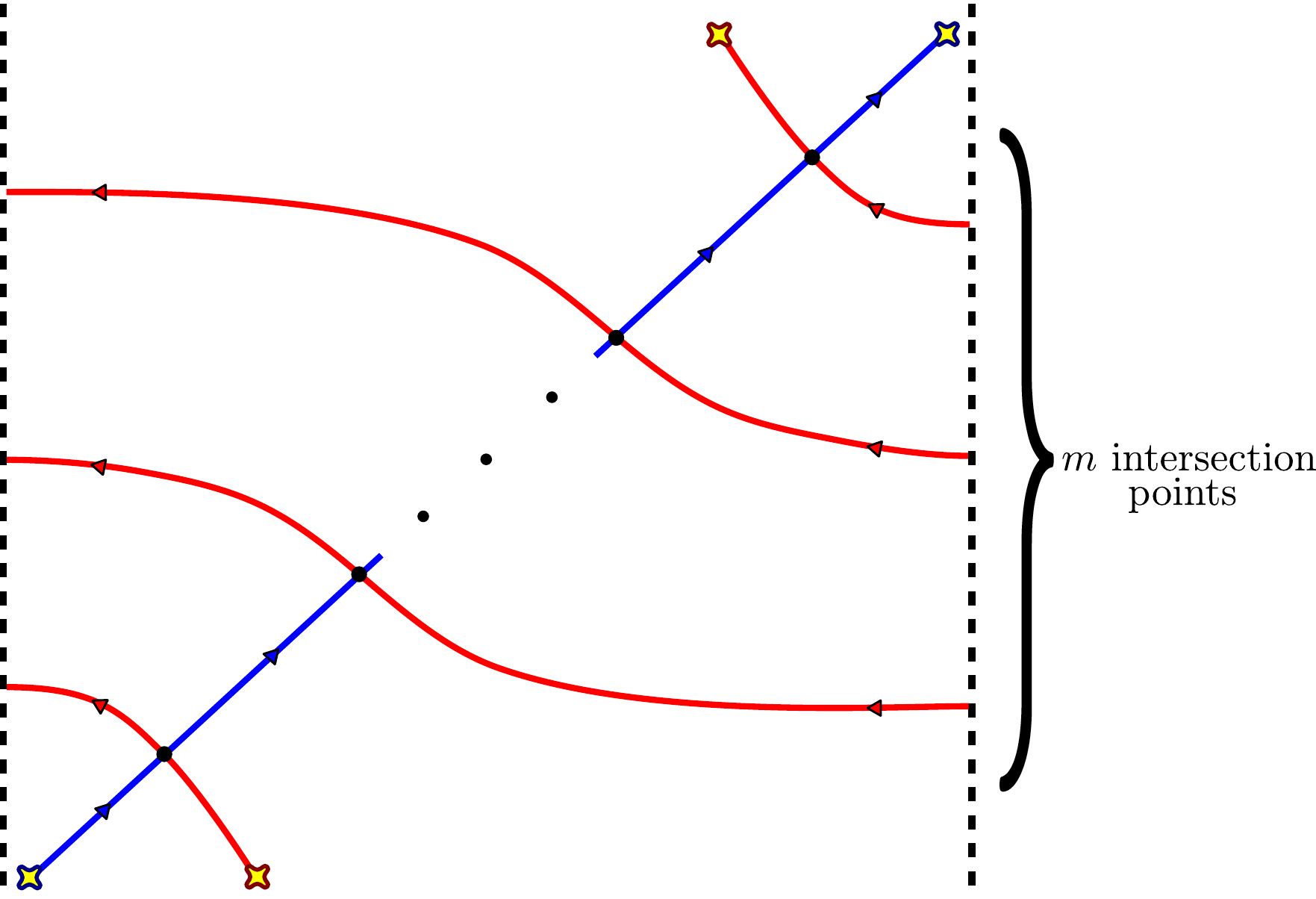}
		\caption{Two saddle connections appearing simultaneously on a wall of marginal stability.  The dotted lines are identified to form the cylinder.  Letting $\gamma_{1}$ be the charge supported by the red saddle connection and $\gamma_{2}$ the charge supported by the blue saddle connection, then we have the intersection pairing $\langle \gamma_{1}, \gamma_2 \rangle = m$ (using the orientation on the cylinder induced by the standard orientation on the plane).  \label{fig:saddle_conns}}
	\end{center}
\end{figure}

\begin{figure}[t!]
	\begin{center}
		 \includegraphics[scale=0.55]{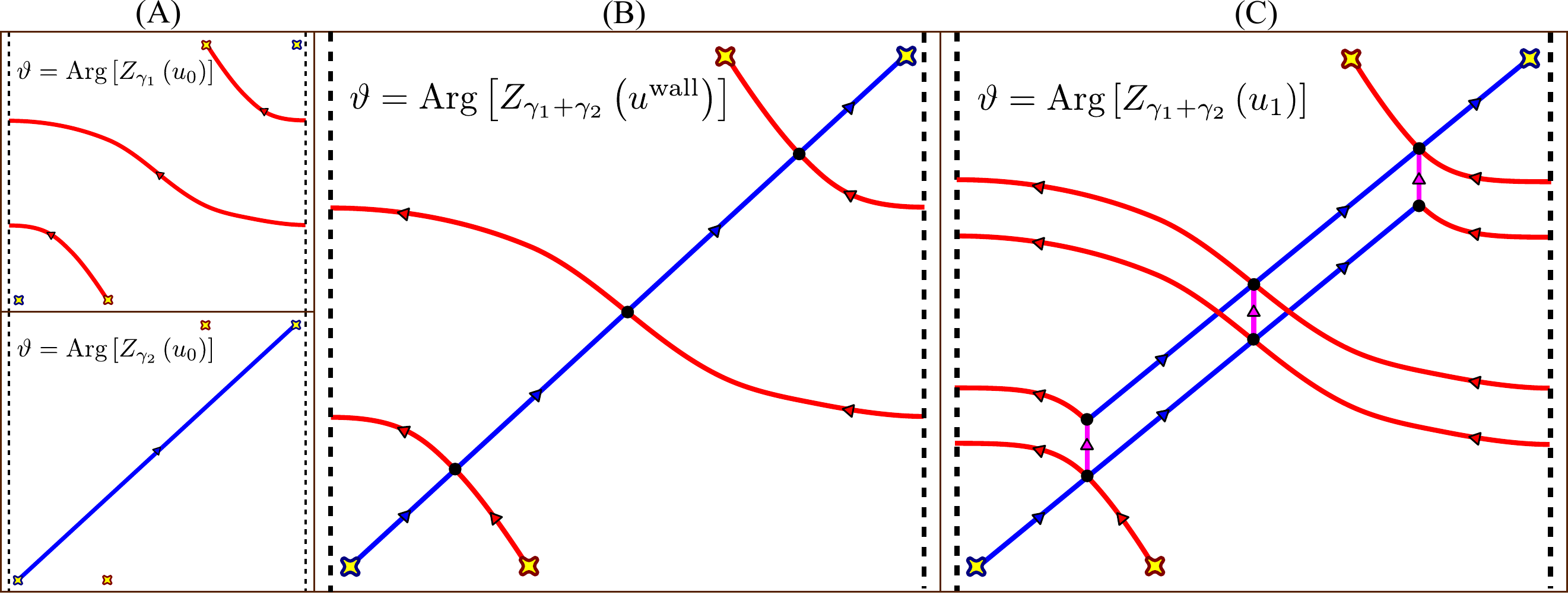}
		\caption{Snapshots of the family of hypothetical $\CW$-networks described in \S \ref{sec:m_herds_wallcrossing} that depict a wall-crossing of two hypermultiplets with charges $\gamma_{1}$ and $\gamma_{2}$ such that $\langle \gamma_{1}, \gamma_{2} \rangle =3$.  Streets of type $12$ are shown in red, $23$ in blue, and $13$ in fuchsia; only two-way streets are depicted.  Arrows denote street orientations according to the conventions set in \cite{wwc}; yellow crosses denote branch points.  The black dotted lines are identified to form the cylinder. $(A)$: The two hypermultiplet networks at a point $u^{0} = U(t=0)$ just ``before" the wall of marginal stability. $(B)$: The $\CW_{\vartheta}$-network at the point $u^{\text{wall}} = u(t_{*})$ on the wall of marginal stability and at phase $\vartheta = \arg\left[Z_{\gamma_{1}}(u^{\text{wall}}) \right] = \arg\left[Z_{\gamma_{2}}(u^{\text{wall}}) \right] = \arg\left[Z_{\gamma_{1} + \gamma_{2}}(u^{\text{wall}}) \right]$. $(C)$: Slightly ``after" the wall at a point $u_{1} = U(t=1)$, a BPS bound state of charge $\gamma_{1} + \gamma_{2}$ is born and a two-way street of type $13$ ``grows" as one proceeds away from the wall.
\label{fig:3-herd_motivation}}
	\end{center}
\end{figure}

\begin{figure}[t!]
	\begin{center}
		 \includegraphics[scale=0.55]{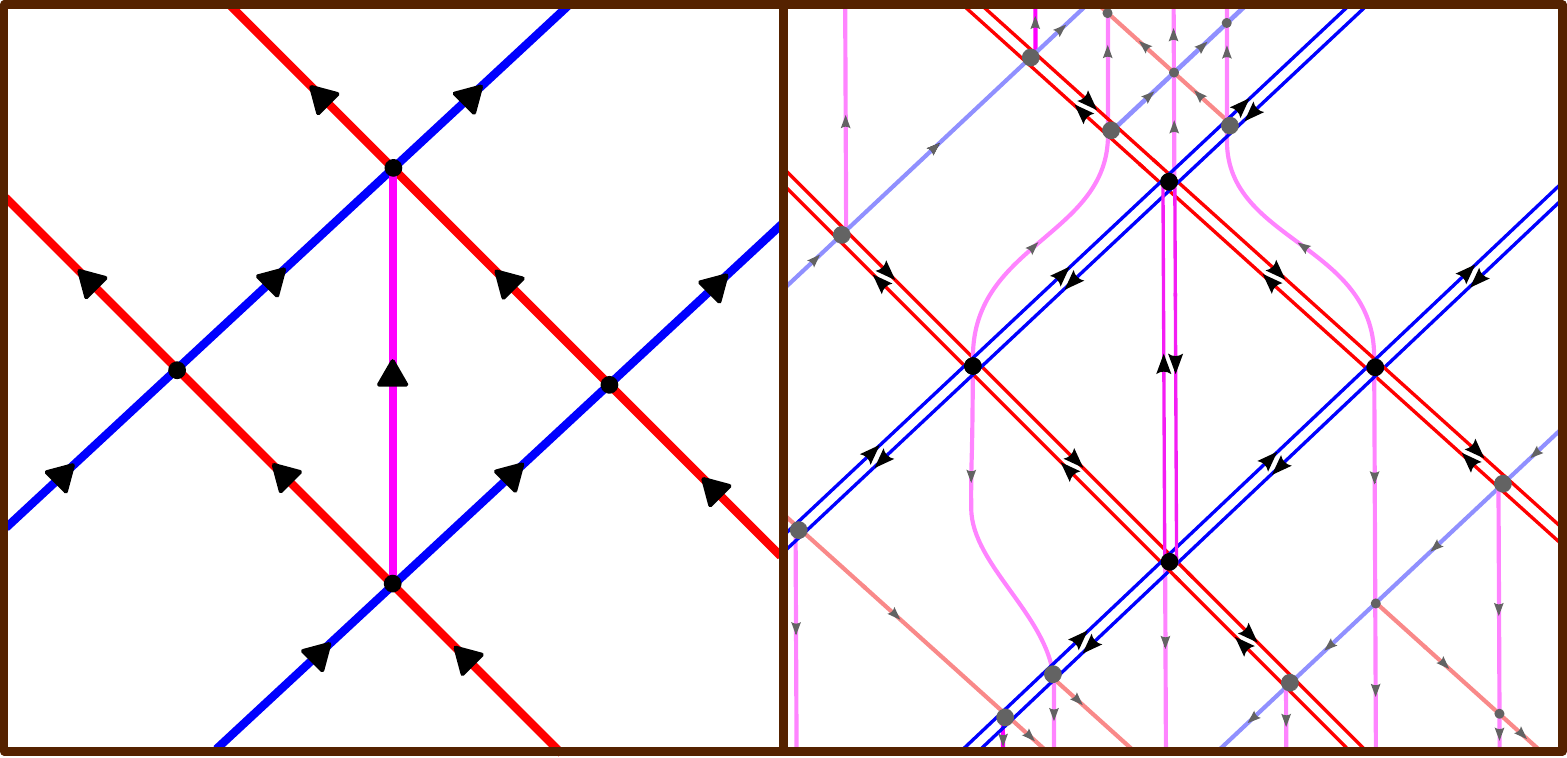}
		\caption{\textit{Left Frame}: The building-block for $m$-herds, named a \textit{horse}.  The solid streets depicted are capable of being two-way.  The sheets of the cover $\Sigma \rightarrow C$ are (locally) labelled from $1$ to $K \geq 3$.  Red streets are of type $12$, blue streets are of type $23$, and fuchsia streets are of type $13$.  We choose an orientation for this diagram such that all streets ``flow up".  \textit{Right Frame}: A horse as it may appear in an actual spectral network.  One-way streets shown as partially transparent and two-way streets are resolved using the ``British resolution" (c.f \cite{sn}).  One can imagine horses with increasingly intricate ``backgrounds" of one-way streets. \label{fig:horse}}
	\end{center}
\end{figure}

As it turns out, these hypothetical degenerate networks can be formed by glueing together copies of the local-patch ``building-block" of Fig.~\ref{fig:horse}, named a \textit{horse}.  Roughly speaking, a $\CW$-network is called an \textit{$m$-herd} if its sub-network consisting of only two-way streets is given by glueing together $m$ horses in a chain, and then coupling the two end horses to branch points.\footnote{After this coupling, the six-way street equations constrain some of the streets on a horse to be only one-way.}  The building-block construction provides the appropriate machinery for understanding the BPS spectrum associated to any $m$-herd.

\begin{proposition}[c.f. Prop. 3.1 of \cite{wwc}]  \label{prop_Q}\
	Let $N$ be an $m$-herd, then the street-factors $Q(p)$ for all two-way streets $p$ on $N$ are given in terms of powers of a single generating series $P \in \mathbb{Z} [ \! [ z ] \! ]$ satisfying the functional equation
	\begin{equation}
		P = 1 + z  P^{(m-1)^2},
		\label{eq:P}
	\end{equation}
	where $z = (-1)^{m} X_{\twid{\gamma_{c}}}$ with $\gamma_{c} = \gamma_{1} + \gamma_{2}$ for some $\gamma_{1},\, \gamma_{2} \in \Gamma = H_{1}(\Sigma;\mathbb{Z})$ such that $\langle \gamma_{1}, \gamma_{2} \rangle = m$.  Furthermore, the series $P^m$ generates the BPS indices $\Omega(n \gamma_{c})$. 
	\end{proposition}

\section{Seeking Non-diagonal Herds} \label{sec:non_diag}
Now we pursue a natural question: what are the degenerate networks associated to slope $a/b \neq 1$ BPS states; particularly, we are interested in slopes $a/b$ such that $\arg \left(Z_{a\gamma_{1} + b\gamma_{2}} \right)$ lies inside of the densely populated arc of phases, i.e. the arc of phases corresponding to integers $a$ and $b$ satisfying (see App. \ref{app:dense_arc}):
\begin{align}
	\frac{m - \sqrt{m^2 - 4}}{2} < \frac{a}{b} < \frac{m + \sqrt{m^2 + 4}}{2}.
	\label{eq:dense_arc}	
\end{align}
Galakhov-Longhi-Moore have a constructive definition of candidates for such slope-$a/b$ networks \cite[App. D]{glm:sn_spin}; however, we will not use their construction---instead, in order to aid in our definitions, we will proceed by looking at actual $\CW$-networks associated to an $m$-wild point.

\subsection{A 3-wild Point} \label{sec:wild_point}
Let us turn our attention to pure $SU(3)$ SYM; as mentioned in \S \ref{sec:setting} this is just $S[A_{2},C,D]$ where $C = \mathbb{P}^{1}$ and $D$ denotes a particular set of defect operators at $0$ and $\infty$.  Let $z$ denote the coordinate on $\mathbb{P}^{1} \backslash \{0, \infty \}$, given by the restriction of the standard coordinate patch on $\mathbb{P}^{1} \backslash \{\infty \} \cong \mathbb{C}$ to $\mathbb{P}^{1} \backslash \{0, \infty \}$.   A point on the Coulomb branch $\CB$ is a pair $(\phi_{2},\phi_{3})$ of a quadratic (meromorphic) differential $\phi_{2}$ and cubic (meromorphic) differential $\phi_{3}$ on $\mathbb{P}^{1}$ with singularities at $0$ and $\infty$ fixed by our defect operators.   Explicitly, there is an identification $\mathbb{C}^{2} \overset{\sim}{\rightarrow} \CB$ sending the point $(u_2, u_3) \in \mathbb{C}^{2}$ to the pair $(\phi_{2}, \phi_{3})$ defined by:
\begin{align*}
	\phi_{2} &= \frac{u_{2} dz^{\otimes 2}}{z^2},\\
	\phi_{3} &= \left(\frac{\Lambda}{z^{2}}+\frac{u_{3}}{z^{3}}+\frac{\Lambda}{z^{4}} \right) dz^{\otimes 3}.
\end{align*}
Using this identification, we will abuse notation and say $u = (u_{2}, u_{3})$ is a point in $\CB$.  The spectral cover (Seiberg-Witten curve) at the point $u = (u_{2}, u_{3}) \in \CB$ of the pure $SU(3),\, \mathcal{N} = 2$ theory with dynamical scale $\Lambda$ is given by
\begin{equation*}
	\Sigma_{u} = \left \{ \lambda \in \CT^* C:  \lambda^{3} + \lambda \left(\frac{u_{2}}{z^{2}} \right)  \left(dz \right)^{\otimes 2} +\left(\frac{\Lambda}{z^{2}}+\frac{u_{3}}{z^{3}}+\frac{\Lambda}{z^{4}} \right) \left(dz\right)^{\otimes 3} = 0 \right \};
\end{equation*}
as usual we will work in units where $\Lambda = 1$.  Our focus will be on the point $u^{w} \in \CB$ given by $(u^{w}_{2},u^{w}_{3}) = \left(3, \frac{95}{10} \right)$ (where the superscript $w$ is for ``wild") where numerical observations indicate that the $3$-wild BPS spectrum presents itself. 

\begin{figure}[t!]
	\begin{center}
		\hspace*{\fill}
		 	\includegraphics[scale=0.6]{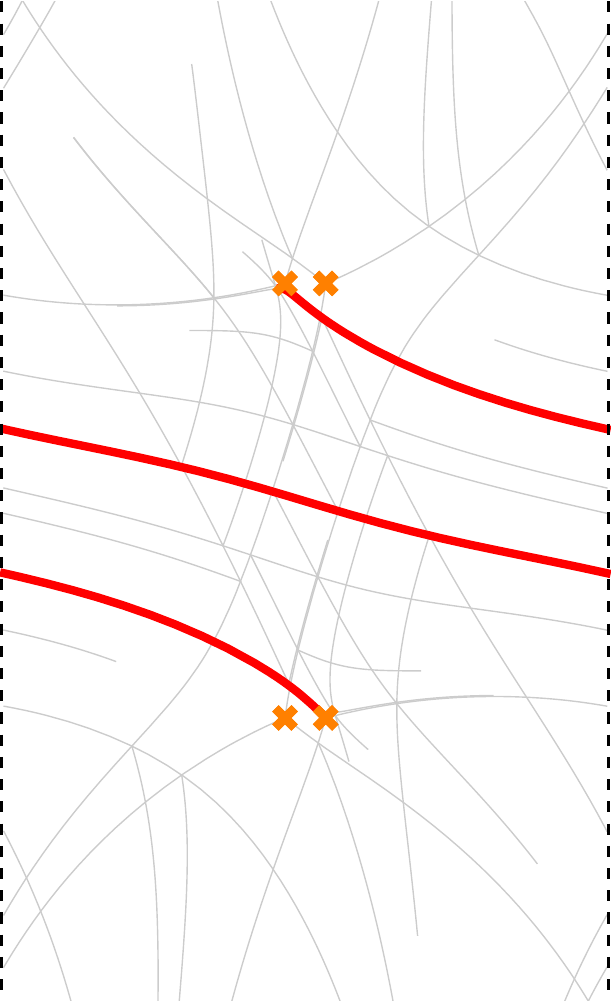}
		 \hfill
		  	\includegraphics[scale=0.6]{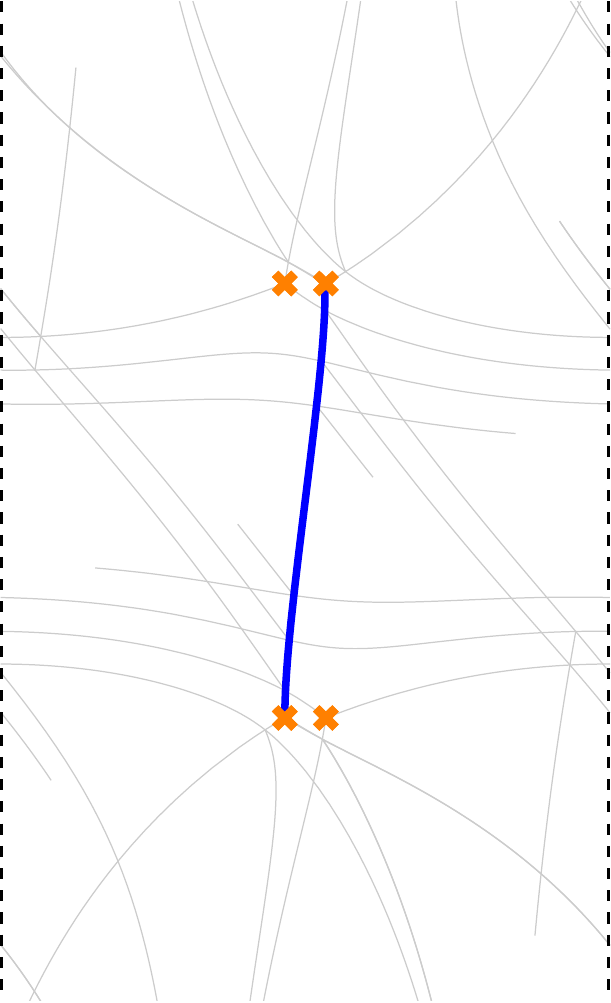}
		 \hspace*{\fill}
		\caption{Degenerate saddle-connection spectral networks at the point $u^{w}$ (defined in \S \ref{sec:wild_point}).  Streets that become two-way at a critical phase are thickened and highlighted in colour; in each figure, the dotted lines are identified in order to form a cylinder.  Both networks are produced at mass-cutoff $50$ in $\Lambda = 1$ units. \textit{Left}: $\CW$-network computed close to the critical phase $\arg(Z_{\gamma_{1}})$.  \textit{Right}: $\CW$-network computed close to the critical phase $\arg(Z_{\gamma_{2}})$. \label{fig:u_w_saddles}}
	\end{center}
\end{figure}

\begin{claim}
	$u^{w}$ is a $3$-wild point.
\end{claim}

	To support this claim, we note that at this point there are two BPS states with linearly independent charges $\gamma_{1}, \gamma_{2} \in \Gamma = \hat{\Gamma}_{u}$ and central charges
\begin{equation}
	\begin{aligned}
		Z_{\gamma_{1}} &= 20.980 - 40.148 i,\\
		Z_{\gamma_{2}} &= 7.244 - 9.083 i.
	\end{aligned}
	\label{eq:3-wild_central_charge}
\end{equation}
The existence of such states is established by observing the occurrence of two saddle-connection $\CW$-networks at the phases $\arg(Z_{\gamma_{1}})$ and $\arg(Z_{\gamma_{2}})$, both depicted in Fig.~\ref{fig:u_w_saddles}; from this figure it can be deduced that $\langle \gamma_{1}, \gamma_{2} \rangle = 3$.  Moreover, the fact that both BPS states are represented by saddle-connections implies that these states are BPS-hypermultiplets.\footnote{One can verify this claim directly by computing the protected spin-character \cite{glm:sn_spin} for a saddle-connection, orindirectly by noting that a saddle connection corresponds to a special Lagrangian 3-sphere in the type IIB origin of theories of class S.  Such 3-spheres correspond to BPS hypermultiplets in the field theory limit.}  This makes it reasonable to conjecture that there exists a BPS quiver at the point $u^{w}$ such that $\gamma_{1}$ and $\gamma_{2}$ label two of its nodes; under this conjecture the 3-Kronecker quiver is a subquiver of a BPS quiver at $u^{w}$. Furthermore, the central charge phases \eqref{eq:3-wild_central_charge} impose a wild stability condition on this subquiver  (c.f. Appendices \ref{app:stab_m_kron}-\ref{app:bridgeland_to_slope}).  Indeed, for any $\theta$ such that $e^{-i\theta} Z_{\gamma_{1}}$ and $e^{-i\theta} Z_{\gamma_{2}}$ lie in the half-space $\mathbb{H} \subset \mathbb{C}$, we have
\begin{align*}
	\arg \left(e^{-i\theta}Z_{\gamma_{1}} \right)  > \arg \left(e^{-i\theta} Z_{\gamma_{2}} \right)
\end{align*}
as elements of $(0, \pi)$.  So, via Prop. \ref{prop:kron_stdform} and the related discussion in Apps. \ref{app:quiv_rep} and \ref{app:BPS_quiv}, $u^{w}$ is indeed a 3-wild point.

The $\CW_{\vartheta}$-network at the phase $\vartheta = \arg \left(Z_{\gamma_{1} + \gamma_{2}} \right)$ is a 3-herd.  Now our interest lies in discovering the form of the degenerate networks inside of the densely populated arc of phases: i.e. networks associated to BPS states of charges $a \gamma_{1} + b \gamma_{2}$ for $a,b$ coprime and satisfying \eqref{eq:dense_arc}. The process of drawing such networks becomes progressively more computationally taxing as $a$ and $b$ grow: for slope-$a/b$, in order to see the production of two way streets, the mass cutoff of the $\CW$-network must be at least the mass of the predicted BPS state: $M_{a/b} := |a Z_{\gamma_{1}} + b Z_{\gamma_{2}} |$.  As a result, in this paper we will restrict our attention to results obtained by studying the two ``low mass" cases: slopes $1/2$ and $2/3$ (at $m=3$).  As mentioned in the introduction, the network at slope $2/3$ produces novel functional equations for the BPS spectrum.  On the other hand, the network at slope $1/2$, and its generalization to slope $1/(m-1)$ networks in the $m$-wild spectrum, produces the same spectrum as the slope $1$ ($m$-herd) situation (and, in fact, produces essentially the same algebraic equations as for the $m$-herd); nevertheless it merits discussion as it may lead to future insight toward the general structure for networks at general slopes.
 
For the time-being it is helpful to fix some terminology when referring to these ``off-diagonal" slope networks.  As generalizations of $m$-herds, plenty of names come to mind to a sufficiently mischievous author.  However, in a rare triumph of science over good fun, rather than parsing through the full spectrum of farm animals, we settle on a more descriptive terminology.

\begin{definition}[Terminology]\footnote{In \cite[App. D]{glm:sn_spin}, the authors use the terminology ``$m$-$(a,b)$-herd" instead of what we would call an $(a,b|m)$ herd.  The author rationalizes his notation due to a deep-rooted derision of multiple hyphens.}
	Let $a$ and $b$ be coprime integers.  Then by an $(a,b|m)$-herd we will mean a degenerate spectral network such that:
	\begin{enumerate}
		\item If $a = b = 1$, it is an $m$-herd;
		
		\item One can continuously deform its degenerate skeleton into two saddle-connections with $m$-intersection points: if (in a trivialization of the spectral cover) the branch points are of type $ij$ and $jk$, respectively, then in the limit that one shrinks the degenerate streets of type $ik$ to length zero, one recovers two saddle-connections that intersect $m$-times (c.f. Fig.~\ref{fig:saddle_conns}).
		
		\item The corresponding BPS indices $\Omega(n \gamma_{c})$ coincide with the DT invariants $d(na,nb,m)$ (for all $n \in \mathbb{Z}_{\geq 1}$) of the Kronecker $m$-quiver equipped with a non-trivial stability condition.
		
		\item Its degenerate skeleton is constructed by glueing together $m$ basic building-blocks.
	\end{enumerate}
\end{definition}

The last condition is what prevents this terminology from being a definition: we have not provided what the building-blocks should be for general $(a,b)$.  This condition is motivated by our experiences with $m$-herds, in particular the expectation that---as in the case for $m$-herds---the BPS spectrum of an $(a,b|m)$-herd can (hopefully) be calculated inductively for all $m$.

Speaking more vagely: an $(a,b|m)$-herd is a type of degenerate spectral network (with sufficiently simple skeleton) that can appear as a $\CW_{\vartheta}$ network at phase $\vartheta = \arg(Z_{a \gamma_1 + b \gamma_2})$ in an $m$-wild region (i.e. $\gamma_{1}$ and $\gamma_{2}$ are two hypermultiplet charges with $\langle \gamma_1, \gamma_2 \rangle = m$).

\begin{remark}
 Not all $\CW_{\vartheta}$ networks for $\vartheta = \arg(Z_{\gamma_1 + \gamma_2})$ at an $m$-wild point may be $m$-herds.  Indeed, during the work on \cite{wwc}, examples of spectral networks supporting the charge $\gamma_{1} + \gamma_{2}$ in an $m$-wild region, but were \textit{not} $m$-herds were found by D. Galakhov and P. Longhi.  Hence, one should not expect that all spectral networks producing the DT invariants $d(a,b,m)$ should be $(a,b|m)$-herds.
\end{remark}

In this paper, we will give definitions for $(m-1,1|m),(1,m-1|m),(2,3|3)$ and $(3,2|3)$-herds.
 As mentioned in the beginning of this section, in \cite[App. D]{glm:sn_spin}, the authors provide an excellent candidate for a rigorous definition of the degenerate skeleton of an $(a,b|m)$-herd for general $(a,b)$; in that paper, the building-blocks are referred to as ``fat-horses".  The reader should be warned, however, that our method of construction for $(m-1,1|m)$ and $(1,m-1|m)$-herds use different building-blocks than the ``fat-horses", but the resulting degenerate skeletons appear to be the same.

\begin{remark}
	In all $(a,b|m)$-herds considered, the BPS indices $\Omega_{m}(n(a \gamma_{1} + b \gamma_{2}))$ generated by an $(a,b|m)$-herd are identical to $\Omega_{m}(n(b \gamma_{1} + a \gamma_{2}))$.  As described in \ref{app:m_wild}, for the $m$-Kronecker quiver $K_{m}$ this is a consequence of the ``transposition" autofunctor acting on the category of representations of $K_{m}$. At the level of all the spectral networks in this paper, this symmetry is immediate consequence of the fact that the $(a,b|m)$-herd degenerate skeleton is the mirror image of the $(b,a|m)$-herd degenerate skeleton.  The networks of \cite{glm:sn_spin} also appear to have this mirror-image property.
\end{remark}

 
\subsection{\texorpdfstring{$(m-1,1|m)$-herds (for $m \geq 2$)}{(m-1,1|m)-herds (for m >= 2)}}
First, we turn our attention toward the $(2,1|3)$-herd, realized as a $\CW_{\vartheta}$-network at the point $u^{w} \in \CB$ described above in \S \ref{sec:wild_point} and with central charge phase $\vartheta = \arg \left(Z_{\gamma_{1} + 2\gamma_{2}} \right)$.   A simplified version of this network, with only the two-way streets shown, is depicted in Fig.~\ref{fig:(2_1)_diagram}.  Just as for the $m$-herds, Fig.~\ref{fig:(2_1)_diagram} suggests a rather simple generalization for diagrams corresponding to any $m \geq 2$.  Indeed,  one can check that the diagram in Fig.~\ref{fig:(2_1)_diagram} is reproduced by glueing together three copies of the building-block\footnote{This building block looks roughly like a horse (see Fig.~\ref{fig:horse}) tipped 45 degrees: the key difference between the two diagrams is that a horse has only one ``secondary" street (fuchsia street), while the building block shown in Fig.~\ref{fig:(m-1_1)_(1_m-1)_blocks} has two secondary streets.} shown in the left panel of Fig.~\ref{fig:(m-1_1)_(1_m-1)_blocks}, and coupling the end building blocks to four distinct branch points in the appropriate way.\footnote{In this procedure, just as with horses, some of the streets in the building block are constrained to be one-way streets after coupling to the branch points.}

\begin{definition}
	Let $m \geq 2$.  A network is an $(m-1,1|m)$-herd (respectively $(1,m-1|m)$-herd) if:
	\begin{enumerate}	
		\item Its two-way streets agree with the diagram formed by glueing $m$-copies of the building-block in the left panel (resp. right panel) of Fig.~\ref{fig:(m-1_1)_(1_m-1)_blocks} in the following manner: 
	\begin{align*}
		a_{1}^{(l)} &= a_{2}^{(l-1)}\\
		\conj{a}_{2}^{(l)} &= \conj{a}_{1}^{(l-1)}\\
		b_{1}^{(l)} &= \conj{b}_{2}^{(l-1)}\\
		b_{2}^{(l)} &= \conj{b}_{1}^{(l-1)},\\
	\end{align*}
	for every $l \in \{2, \cdots, m\}$ .
	
	\item Each of the two-way streets labelled by $a_{1}^{(1)},\, b_{2}^{(1)}, \, \conj{a}_{1}^{(m)},\,$ and $\conj{b}_{2}^{(m)}$ are connected to distinct branch points. 
	
	\item A version of the \textit{No Holes} condition\footnote{This is a technical, yet important, condition that allows one to define an $(m-1,m|m)$-herd (resp. $(m,m-1|m)$-herd)  on a general curve $C$ and insure that all street factors are formal series in a single variable $X_{\twid{\gamma}}$ for a primitive $\gamma \in \Gamma$, unique up to a choice of sign (c.f. Rmk.~\ref{rmk:gen_choice} in \S \ref{sec:setting}).  Without this condition, street factors may be formal series in a collection of formal variables, associated to linearly independent, primitive, elements of $\Gamma$.}, described in \cite[App. C.3]{wwc} , with the horse replaced by the building-block in the left (resp. right) panel of Fig.~\ref{fig:(m-1_1)_(1_m-1)_blocks}.
\end{enumerate}
\end{definition}

The name $(m-1,1|m)$-herd suggests that these diagrams are related to the slope $(m-1)/1$ BPS states. Indeed, the proofs and analyses of Appendices C-D in \cite{wwc} can all be appropriately modified with either one of the building-blocks of Fig.~\ref{fig:(m-1_1)_(1_m-1)_blocks} replacing the horse (Fig.~\ref{fig:horse}); with these simple modifications, we arrive at the following.

\begin{proposition} \label{prop_(m-1,1)_herds}
	The BPS indices $\{\Omega(n \gamma_{c})\}_{n=1}^{\infty}$ for the $(m-1,1|m)$-herds (resp. $(1,m-1|m)$-herds) are generated by a formal series $P$ satisfying the algebraic equation
	\begin{align}
		P = 1 + z P^{(m-1)^2}
		\label{eq:(m-1,1)_funceq}
	\end{align}
	with $z := - X_{\twid{\gamma_{c}}}$ for some $\gamma_{c} \in \Gamma$.  Furthermore, there exist charges $\gamma_{1}, \gamma_{2} \in \Gamma$ such that $\langle \gamma_{1}, \gamma_{2} \rangle = m$ and $\gamma_{c} = (m-1) \gamma_{1} +  \gamma_{2}$ (resp. $\gamma_{c} = \gamma_{1} + (m-1) \gamma_{2}$).
\end{proposition}

\noindent The reason for the negative sign appearing in $z = -X_{\twid{\gamma_{c}}}$ is explained in Appendix \ref{app:signs}.

Note that, modulo the identification of the formal variable $z$ with an element of $\formal{\mathbb{Z}}{\twid{\Gamma}}$, the algebraic equation (\ref{eq:(m-1,1)_funceq}) is identical to the algebraic equation (\ref{eq:P}) for $m$-herds.  In particular, the BPS indices $\{\Omega(n\gamma_{c})\}_{n=1}^{\infty}$ are identical to the $m$-herd situation.  This should be expected under the assumption of equivalence of BPS indices $\Omega\left[n(a \gamma_{1} + b \gamma_{2})\right]$ for $(a,b|m)$-herds  with the DT invariants $d(a,b,m)$ associated to the $m$-Kronecker quiver equipped with dimension vector $(a,b)$.  In particular, the observation that $\Omega \left[n(m \gamma_{1} + (m-1) \gamma_{2}) \right]$ is the same as $\Omega \left[n(\gamma_{1} + \gamma_{2})\right]$ is expected from the more general equivalence of $m$-Kronecker DT invariants
\begin{align*}
	d(a,b,m) = d(ma-b,b,m).
\end{align*}
As described in Appendix \ref{app:m_wild}, this equivalence follows as a consequence of ``reflection" endofunctor \cite{weist:loc} on the $m$-Kronecker quiver representation category---which takes representations with dimension vector $(a,b)$ to representations with dimension vector $(ma -b, b,m)$.  In the physics, this equivalence is expected to manifest itself as a monodromy of the local system $\hat{\Gamma} \rightarrow \CB^{*}$ (c.f. \cite{wwc} \S 6.2 and Appendix \ref{app:m_wild}).

\begin{figure}[t!]
	\begin{center}
		 \includegraphics[scale=0.25]{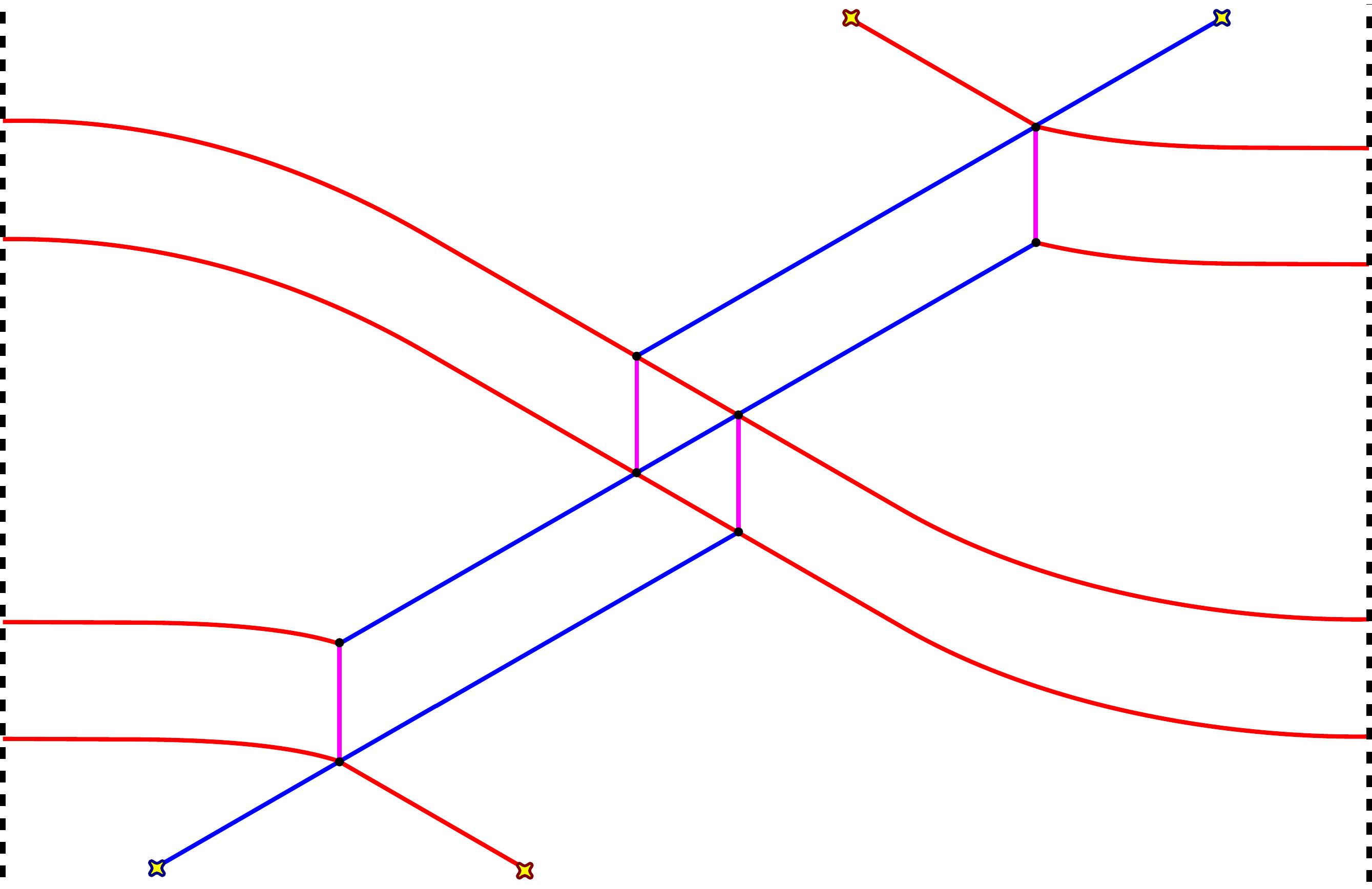}
		\caption{Two-way streets of a $(2,1|3)$-herd on the cylinder. \label{fig:(2_1)_diagram}}
	\end{center}
\end{figure}

\begin{figure}[t!]
	\begin{center}
		 \includegraphics[scale=0.5]{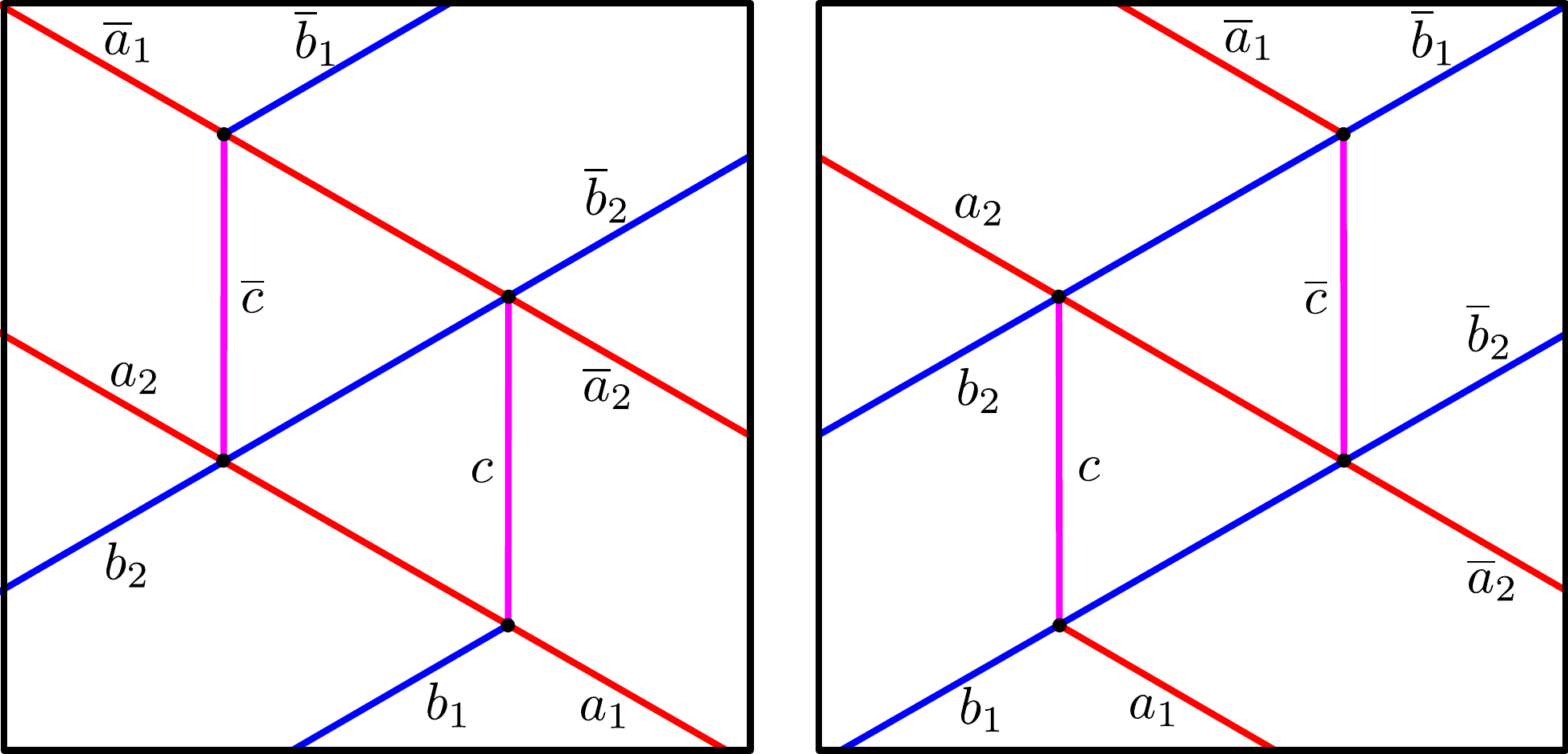}
		\caption{Left: Building-block for an $(m-1,1|m)$-herd.  Right: Building-block for a $(1,m-1|m)$-herd (which is the mirror image of the building-block in the left panel). \label{fig:(m-1_1)_(1_m-1)_blocks}}
	\end{center}
\end{figure}

\subsection{The \texorpdfstring{$(2,3|3)$ and $(3,2|3)$-herds}{(2,3|3) and (3,2|3)-herds}} \label{sec:2_3_herds}

To achieve a spectral network with the potential to produce a more interesting BPS spectrum, we turn our attention toward the $\CW$-network associated to bound states of charges $\{2n \gamma_{1} + 3n \gamma_{2}\}_{n=1}^{\infty}$: this is the network $\CW_{\vartheta}$ at phase $\vartheta = \arg \left(Z_{2 \gamma_{1} + 3 \gamma_{2}} \right)$ lying in the densely populated arc. At the point $u^{w} \in \CB$, the bound state of charge $\gamma_{c} = 2 \gamma_{1} + 3 \gamma_{2}$ has central charge $Z_{\gamma_{c}} = 63.692- 107.545 i$ and mass $M = |Z_{\gamma_{c}}| = 124.99$.  The corresponding $\CW_{\vartheta}$ network at $\vartheta = \arg(Z_{\gamma_{c}})$ is shown in Fig.~\ref{fig:(2_3|3)_WKB} with the relevant two-way streets highlighted using our colour-coding convention.\footnote{In practice, one way to determine the two-way streets by perturbing the phase of interest slightly, and carefully (read painfully) watching for near collisions (that become \textit{actual} collisions for the unperturbed phase) of anti-parallel streets as the mass-cutoff is increased.  This technique also helps rule out the myriad of one-way streets that run close and parallel to one another.}  A simplified picture, showing only two-way streets, is shown in Fig.~\ref{fig:(3_2)_diagram}.

\begin{figure}[t!]
	\begin{center}
		\includegraphics[scale=0.8]{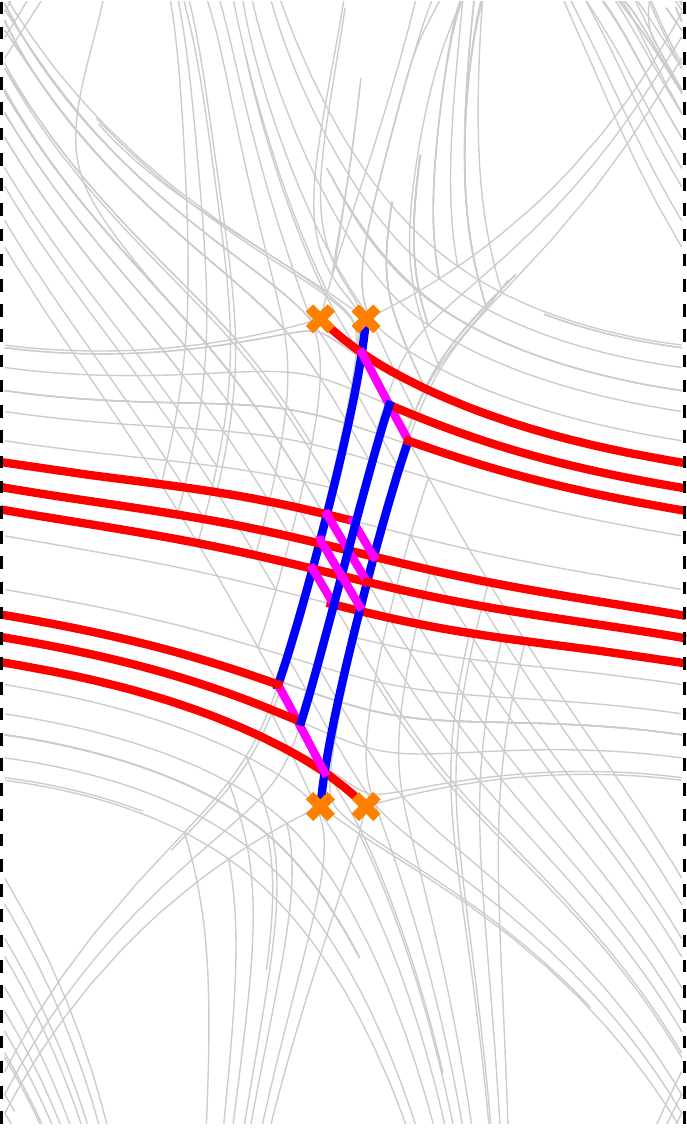}
		\caption{$\CW$-network demonstrating a $(2,3|3)$-herd at the point $u^{w}$.  The network is produced at mass cutoff $140$ close to the critical phase $\arg(Z_{2\gamma_{1} + 3 \gamma_{2}})$.  Streets that become two-way at the critical phase are thickened and highlighted in colour; in each figure, the dotted lines are identified in order to form a cylinder. \label{fig:(2_3|3)_WKB}}
	\end{center}
\end{figure}

One may hope that the process of solving for the street-factors reduces to a simple recursive/inductive building-block procedure as with the $m$-herds and $(m-1,1|m)$-herds.  Indeed, there is an obvious candidate for a building-block, shown in Fig.~\ref{fig:(3_2)_building_block}. However, the presence of the two coupled joints of valence-six (where all streets are two-way) leads to a set of highly-coupled non-linear equations with no simple explicit solution for the incoming soliton data in terms of the outgoing soliton data;\footnote{Perhaps, a more clever reader may be able to find such a solution.} so, we reluctantly abandon this technique.

Instead, we will pursue a more ad-hoc approach specific to the case $m=3$.  First, we recall the ``abelian" six-way junction rules: the equations that impose conditions on the street-factors at a general type of joint (where six streets, all possibly two-way, meet).\footnote{These equations can be derived either from the ``non-abelian" six-way junction rules for soliton generating series, or directly using the homotopy invariance techniques of \cite{sn} (which is how the non-abelian rules are derived).}  Referring to the labelled streets of such a joint in Fig.~\ref{fig:six_way_streetfactors}, we have:
\begin{equation}
	\begin{aligned}
	1 &= Q \left(q_{12} \right) Q \left(q_{13} \right) Q \left(p_{12} \right)^{-1} Q \left( p_{13} \right)^{-1}\\
	1 &= Q \left(q_{23} \right) Q \left(q_{12} \right) Q \left(p_{23} \right)^{-1} Q \left( p_{12} \right)^{-1}\\
	1 &= Q \left(q_{13} \right) Q \left(q_{23} \right) Q \left(p_{13} \right)^{-1} Q \left( p_{23} \right)^{-1}.
	\label{eq:abelian_rules}
	\end{aligned}
\end{equation}

\begin{figure}[t!]
	\begin{center}
		 \includegraphics[scale=0.8]{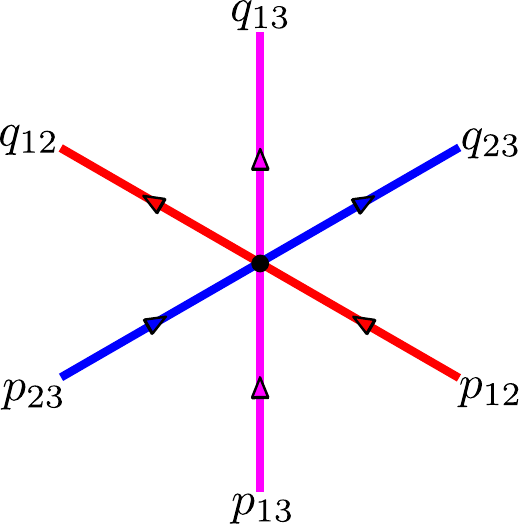}
		\caption{A ``six-way junction": joint where six (possibly) two-way streets meet. \label{fig:six_way_streetfactors}}
	\end{center}
\end{figure}

Any two of these equations are algebraically independent---however, by taking the inverse of both sides of one equation, and multiplying by another, produces the third.  Hence, the abelian rules can only eliminate two effective street-factor degrees of freedom from each joint.  Nevertheless they are quite powerful in reducing the number of unknown street-factors.  From these equations, along with the 180-degree rotation symmetry of the diagram, it is a simple exercise to show that all street factors can be expressed as Laurent monomials in terms of the three street-factors, labeled $M,V,$ and $W$, shown next to their corresponding streets in Fig.~\ref{fig:(3_2)_diagram}; in particular, the street factors attached to the streets extending from the branch points are given as powers of $\gdt_{3/2} := M V W$.

\begin{claim} \
	\begin{enumerate}
		\item All street factors are power series in the variable $z := -X_{\twid{\gamma}_{c}}$, where $\gamma_{c} = 3 \gamma_{1} + 2 \gamma_{2}$ with $\gamma_{1}$ and $\gamma_{2}$ the charges represented by the lifts of the saddle-connections of Fig.~\ref{fig:saddle_conns} to the spectral cover (i.e. the hypothetical hypermultiplet charges before their wall-crossing).
	
		\item	$M V W \in \formal{\mathbb{Z}}{z}$ is the generating series $\gdt_{3/2}$ for the BPS indices $\{\Omega(n \gamma_{c})\}_{n=1}^{\infty}$ in the sense of \eqref{eq:gen_omega_def}.
	\end{enumerate}
\end{claim}

To argue the first claim we recall that a $(3,2|3)$-herd appears as a $\CW$-network off of a wall of marginal stability; thus, all generating series are formal series in the formal variable $X_{\twid{\gamma_{c}}}$ assigned to the standard lift $\twid{\gamma_{c}} \in \twid{\Gamma}$ of a unique (up to sign) charge $\gamma_{c} \in \Gamma$.  (For reasons expressed in Appendix \ref{app:signs}, it is more convenient express the street factors as formal series in the formal variable $z = -X_{\twid{\gamma_{c}}}$).  Using the six-way junction equations (c.f. \cite{wwc}), an order-by-order computation of street factors is consistent with this fact and shows that $\gamma_{c} = 3\gamma_{1} + 2\gamma_{2}$.  This shows the first claim.

To argue the second claim we recall some essential features of the machinery that produces the BPS indices from a spectral network.  Let $N$ be a spectral network, $\text{str}(N)$ be the set of streets of $N$, and  $\ell: \mathrm{str}(N) \rightarrow C_{1}(\Sigma; \mathbb{Z})$ the map describing the ``lift" of each street to a 1-chain on $\Sigma$.  Assume each street-factor is a formal series in some variable $X_{\twid{\gamma_{c}}}$ with integer coefficients.  Now for each $p \in \mathrm{str}(N)$ define a sequence of integers $\left(\alpha_{n}(p)\right)_{n =1}^{\infty} \subset \mathbb{Z}$ by a factorization of the street-factor $Q(p)$:

	\begin{equation}
		Q(p) = \prod_{n =1}^{\infty} \left(1 - \left(X_{\twid{\gamma}_{c}} \right)^{n} \right)^{\alpha_{n}(p)},
		\label{eq:Q-exp}
	\end{equation}
then we define a collection of 1-chains $\{L(n \gamma_{c})\}_{n=1}^{\infty} \subset C_{1}(\Sigma; \mathbb{Z})$ via
\begin{align}
	L(n \gamma_{c}) &:= \sum_{p \in \mathrm{str}(N)} \alpha_{n}(p) \ell(p).
	\label{eq:L_def}
\end{align}
It can be shown that $L(n \gamma_{c})$ is closed and its homology class is a multiple of $n \gamma_{c}$---in fact, the BPS index is related to the homology class $[L(n \gamma_{c})] \in H_{1}(\Sigma; \mathbb{Z})$ via
\begin{align}
	[L(n \gamma_{c})] &= n \gamma_{c} \Omega(n \gamma_{c}).
	\label{eq:L_omega}
\end{align}

Turning back to the case of the $(3,2|3)$-herd, define $c_{n}$ as the exponents in the following expansion
\begin{align*}
	MVW &= \prod_{n = 1}^{\infty} \left( 1 - \left(X_{\twid{\gamma}_{c}} \right)^{n} \right)^{c_{n}}\\
	&= \prod_{n = 1}^{\infty} \left( 1 - (-1)^{n} z^{n} \right)^{c_{n}}.
\end{align*}
Note that any 1-chain supported on the lift of the spectral network and representing the homology class $\gamma_{1}$ must project down to a 1-chain passing through the branch points of type $12$ (emanating from the red-streets in Fig.~\ref{fig:(3_2)_diagram}); similarly, any 1-chain on the lift of the spectral network representing the homology class $\gamma_{2}$ must project down to a 1-chain passing through the branch points of type $23$ (emanating from the blue-streets in Fig.~\ref{fig:(3_2)_diagram}).  Hence, any 1-chain representing the cohomology class $\gamma_{c} = 3\gamma_{1} + 2 \gamma_{2}$ must project down to a path passing three times through the branch points of type 12, and two times through the branch points of type 13.  With this observation and the fact that the street-factors attached to the red and blue streets emanating from branch points are given by $\left(MVW\right)^3$ and $\left(MVW \right)^2$ respectively, then from  \eqref{eq:L_def} and \eqref{eq:L_omega} it follows that
\begin{align*}
	c_{n} &= n \Omega(n \gamma_{c});
\end{align*}
the second claim follows.

\begin{figure}[t!]
	\begin{center}
		 \includegraphics[scale=0.55]{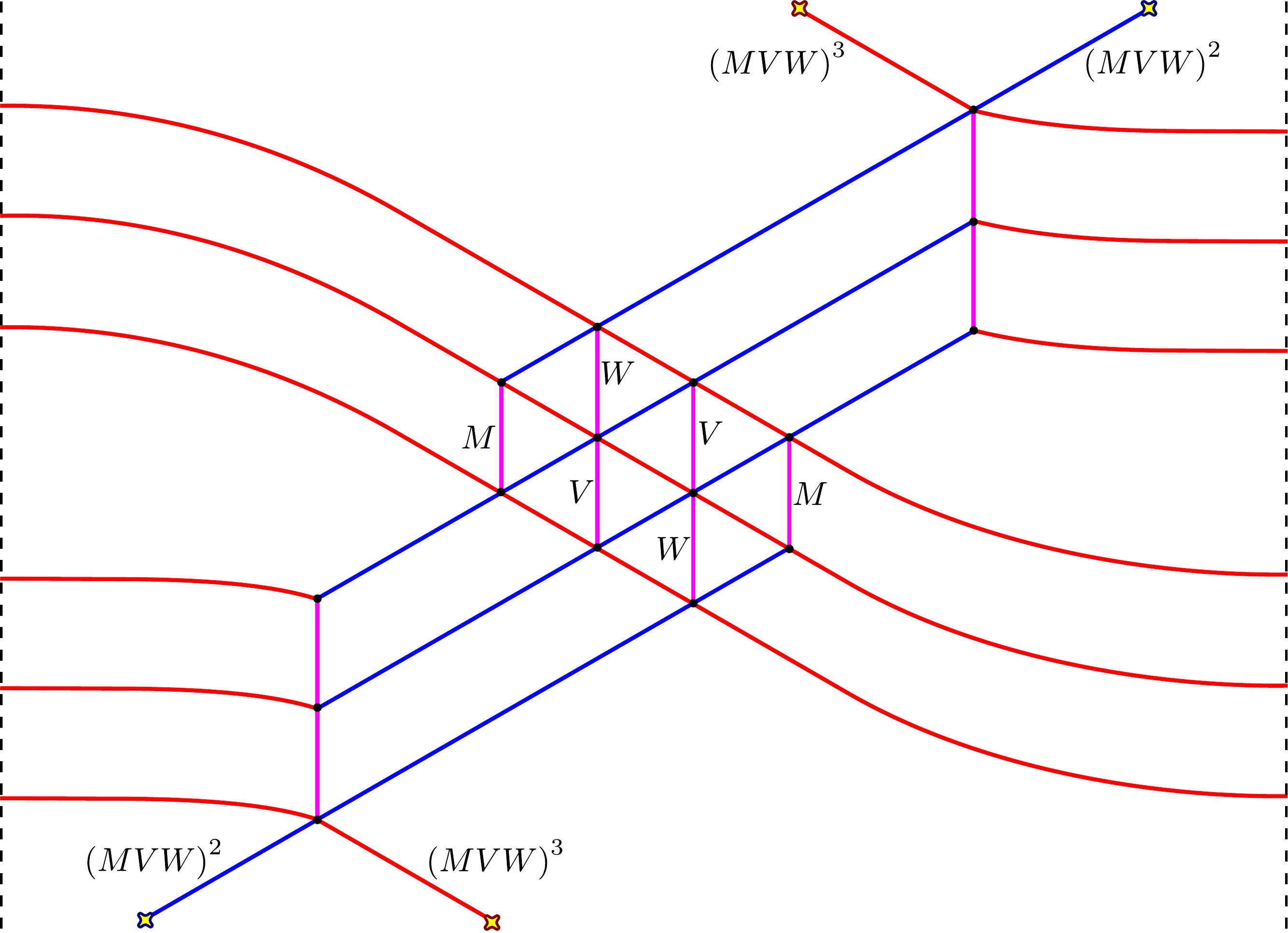}
		\caption{Two-way streets of a $(3,2|3)$-herd. \label{fig:(3_2)_diagram}}
	\end{center}
\end{figure}

\begin{figure}[t!]
	\begin{center}
		 \includegraphics[scale=0.55]{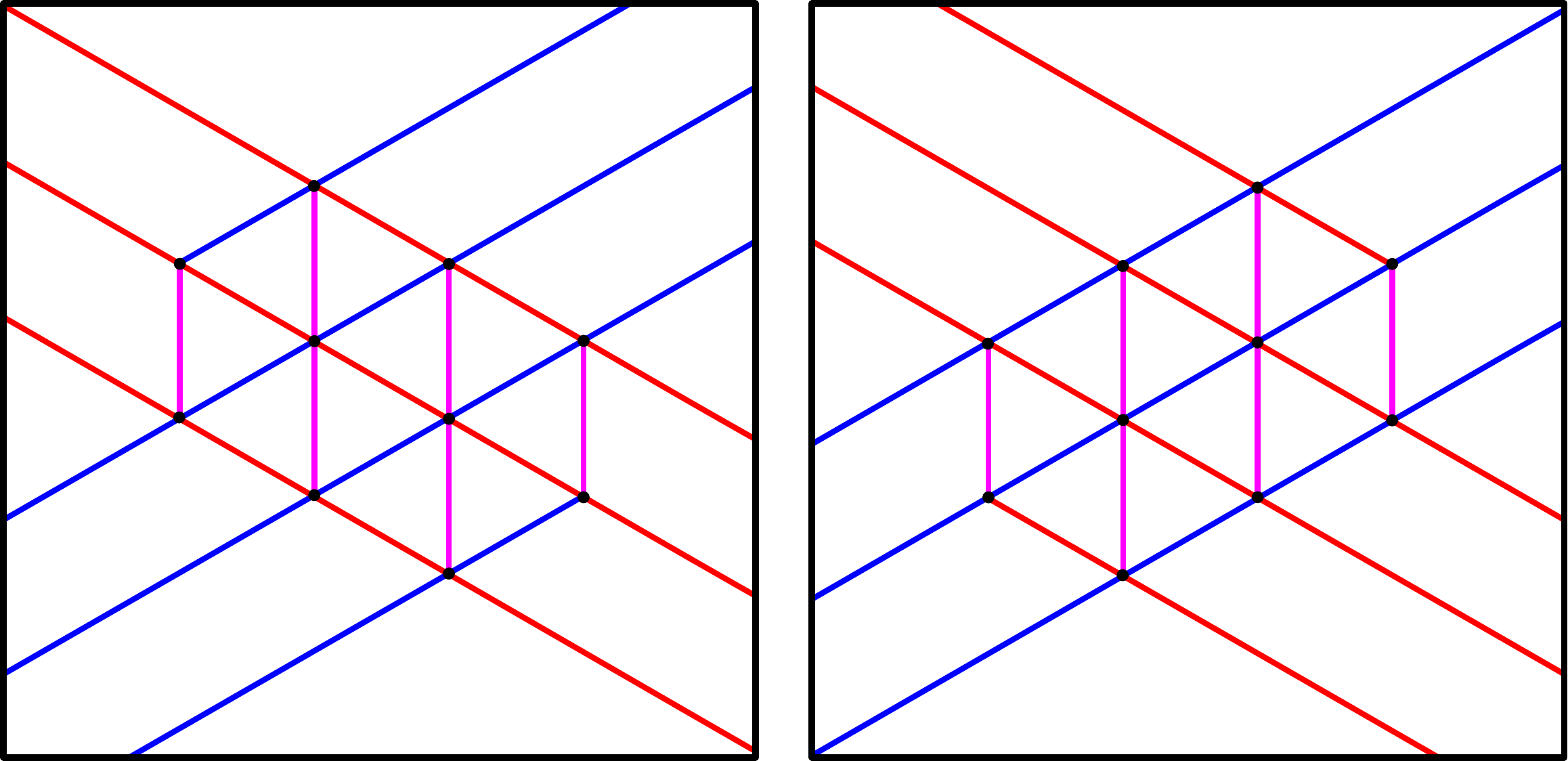}
		\caption{Candidate building-blocks for ``higher" networks. \label{fig:(3_2)_building_block}}
	\end{center}
\end{figure}

As sketched in the appendix, using the six-way junction equations of \cite{wwc} along with some order-by-order numerical calculations, one can derive the following system of algebraic relations:
\begin{equation}
	\begin{aligned}
	M &= 1 + z M^{4} \left\{ (1 + V) (1 + V - W)^2 [V^2(1+W) - 1]^{3} \right\}\\
	0 &= (-1 + V) (1 + V)^2 + (1 + V^3) W - V (M + V) W^2\\
	0 &=  V \left(V^2 -1 \right) - \left[M (V + 1) + V(V-2) - 1 \right] W.
	\end{aligned}
	\tag{\ref{eq:func_eqs}}
\end{equation}
Of course, there is a plethora of ways of rewriting \eqref{eq:func_eqs}, the above is a result of what the author found the most pleasing.  In other words, if we define three polynomials\footnote{polynomials in the variables $m,v,$ and $w$ and coefficients in $\mathbb{Z}[z]$.} in $\left(\mathbb{Z}[z] \right)[m,v,w]$
\begin{align*}
	\mathcal{G} &:= 1 + z m^{4} \left\{ (1 + v) (1 + v - w)^2 [v^2(1+w) - 1]^{3} \right\} - m\\
	\mathcal{M} &:= (-1 + v) (1 + v)^2 + (1 + v^3) w - v (m + v) w^2\\
	\mathcal{N} &:= v \left(v^2 -1 \right) - \left[m (v + 1) + v(v-2) - 1 \right] w
\end{align*}
then $(m,v,w) = (M,V,W)$ must provide a simultaneous root of $\mathcal{G},\, \mathcal{M}$, and $\mathcal{N}$.  There are, in fact, forty-two\footnote{Making a tempting relation to the meaning of life, the universe, and everything.} simultaneous root tuples; only one of which can be identified as the tuple of street-factors $(M,V,W)$.  Indeed, any tuple of simultaneous roots that has the interpretation as a street-factor must be a tuple of formal power series, each with constant coefficient 1.  One can check that there exists precisely one such simultaneous root tuple; we can produce it by substituting such a formal power series ansatz (e.g. $M = 1 + \sum_{n=1}^{\infty} m_{n} z^{n}$) into \eqref{eq:func_eqs} and solving order-by-order in $z$.
The first few terms are given by,

\begin{equation}
	\begin{aligned}
	M &= 1 + 2 z + 146 z^2 + 15824 z^3 + 2025066 z^4 + 284232734 z^5 +  42316425168 z^6 + \mathcal{O}(z^{7}),\\
	V &= 1+4 z+ 324 z^2+36224 z^3+4704404 z^4+665965148 z^5+99703601696 z^6 + \mathcal{O}(z^{7}),\\
	W &= 1+7 z+ 514 z^2+55685 z^3+7121694 z^4+999071727 z^5+148683258448 z^6 + \mathcal{O}(z^{7}).
	\end{aligned}
	\label{eq:MVW_expansion}
\end{equation}
Hence, 
\begin{align*}
	\gdt_{3/2} = 1+13z + 1034 z^2 + 115395 z^3 + 14986974 z^4 + 2122315501 z^5 + 317853709072 z^6 + \mathcal{O}(z^{7}).
\end{align*}
	Using \eqref{eq:omega_invert}, the first few BPS indices are
\begin{align*}
	\left(\Omega \left[n \left(3\gamma_{1} + 2 \gamma_{2} \right) \right] \right)_{n =1}^{6} &= \left(13, -478, 34227, -3279848, 367873950, -45602813070 \right)
\end{align*}
The first three values of this sequence are consistent with BPS indices obtained from the \texttt{HiggsBranchFormula} function\footnote{As mentioned in the documentation for \texttt{CoulombHiggs.m}, the \texttt{HiggsBranchFormula} function is based on Reineke's work in \cite{reineke:hn_system}.  Values of \eqref{eq:omega_invert} were not checked with this package for $n \geq 4$ due to the large computation-time required.} in the Mathematica package \texttt{CoulombHiggs.m} \cite{mps:gen_quiv}, based on formulae developed by Reineke \cite{reineke:hn_system} and Manschot-Pioline-Sen \cite{mps:1,mps:2,mps:3}.

Using an expansion of $\gdt_{3/2}$ out to the coefficient multiplying $z^{K}$, we can reliably extract the first $K$ BPS indices from \eqref{eq:omega_invert}.  This was performed for $K = 1360$ and the resulting BPS indices can be found in the file $\mathtt{DTList1360.csv}$ included with the arXiv preprint version of this paper.

In the next section we discuss the large $n$ asymptotics of BPS indices using an algebraic equation satisfied by $\gdt_{3/2}$.

\subsection{Asymptotics of BPS indices associated to dimension vectors \texorpdfstring{$(3n,2n)$}{(3n,2n)}} \label{sec:asymp_DT_3_2}
With a bit of elimination theory, the algebraic relations \eqref{eq:M_eq_prelim} and the relation $\gdt_{3/2} = MVW$ can be reduced to a single algebraic relation between $z$ and $\gdt := \gdt_{3/2}$:
\begin{equation}
	\begin{aligned}
 0 	=& - (1 + z) + \gdt (4 - 5 z) +  \gdt^2( - 6 + z) + \gdt^3(4 + 21 z) +\\
 	& \gdt^4 ( - 1 - 34 z)  - \gdt^5 \left(7 z \right)  + \gdt^6 \left(76 z + z^2\right)  + \gdt^7 \left( - 64 z - 13 z^2\right) +\\
 	& \gdt^8 \left(6 z - 114 z^2\right)  + \gdt^9 \left(7 z - 80 z^2\right)  + \gdt^{10} \left(6 z^2 \right)  + \gdt^{11} \left(119 z^2 \right)  +\\
 	& \gdt^{12} \left(53 z^2 + z^3\right)  + \gdt^{13} \left( - 55 z^2 + 44 z^3\right) +  \gdt^{14} \left( - 21 z^2 - 38 z^3\right) + \gdt^{15} \left( 77 z^3 \right)  -\\
 	& \gdt^{16} \left(382 z^3\right)  +  \gdt^{17} \left(270 z^3 \right)  +  \gdt^{18} \left(80 z^3 - z^4\right) + \gdt^{19} \left(35 z^3 + 7 z^4\right)  +\\
 	& \gdt^{20} \left(39 z^4\right)  - \gdt^{21} \left( 367 z^4 \right)  - \gdt^{22} \left(173 z^4 \right) -
  \gdt^{23} \left(30 z^4 \right) -\\
  	&  \gdt^{24} \left(35 z^4\right) +  \gdt^{25} \left(3 z^5 \right)-  \gdt^{26} \left( 17 z^5 \right) -  \gdt^{27} \left(77 z^5\right) -\\
  	& \gdt^{28}  \left(14 z^5 \right) +  \gdt^{29}  \left(21 z^5 \right) -  \gdt^{32} \left(3 z^6 \right) + \gdt^{33} \left(9 z^6\right) -\\
 	&  \gdt^{34} \left(7 z^6 \right) +\gdt^{39} z^7.
 	\end{aligned}
 	\label{eq:3_2_dt_relation}
\end{equation}

With this algebraic relation, using the techniques developed in \S \ref{sec:alg_asymp} we can determine an explicit form for the $n \rightarrow \infty$ asymptotics of $\Omega \left[n \left(3\gamma_{1} + 2 \gamma_{2} \right) \right]$.  Indeed, using Corollary \ref{cor:single} and \eqref{eq:3_2_dt_relation} we have the $n \rightarrow \infty$ asymptotics\footnote{Written in a different form, the ``exponential part" of the asymptotics can be expressed as $\rho^{-n} = e^{Kn}$ for $K = - \log(\rho) \approx 5.27187$.}
\begin{align}
	\Omega \left[n \left(3\gamma_{1} + 2 \gamma_{2} \right) \right] \sim C (-1)^{n+1} n^{-5/2} \rho^{-n} + \mathcal{O}(n^{-7/2} \rho^{-n})
	\label{eq:3_2_herd_asymptotics}
\end{align}
where
\begin{align*}
	\rho &\approx 0.005134\\
	C &\approx 0.075084.
\end{align*}
The derivation of these asymptotics is elaborated on in Example \ref{ex:3_2_asy} at the end of \S \ref{sec:generating_asymptotics}.  More precisely, $\rho$ is an algebraic number given by the smallest-magnitude root of the 10th degree polynomial \eqref{eq:d_2}.

Note that, defining the sequence,
\begin{align}
	R_{n} := 1 - \frac{\Omega \left[n \left(3\gamma_{1} + 2 \gamma_{2} \right) \right]}{C (-1)^{n+1} n^{-5/2} \rho^{-n}}
	\label{eq:R_n_def_DT}
\end{align}
the asymptotics \eqref{eq:3_2_herd_asymptotics} are equivalent to the statement that $R_{n} \in \mathcal{O}(n^{-1})$.  A plot of $R_{n}$---displaying the $\mathcal{O}(n^{-1})$ behaviour---is shown in Fig.~\ref{fig:DT_asymptotics_error} for $1 \leq n \leq 1360$.  In fact, for the range of $n$ shown in this plot, $0 < R_{n} \leq c n^{-1}$ for any constant $c \geq 0.12$.

\begin{figure}[t!]
	\begin{center}
		 \includegraphics[scale=0.6]{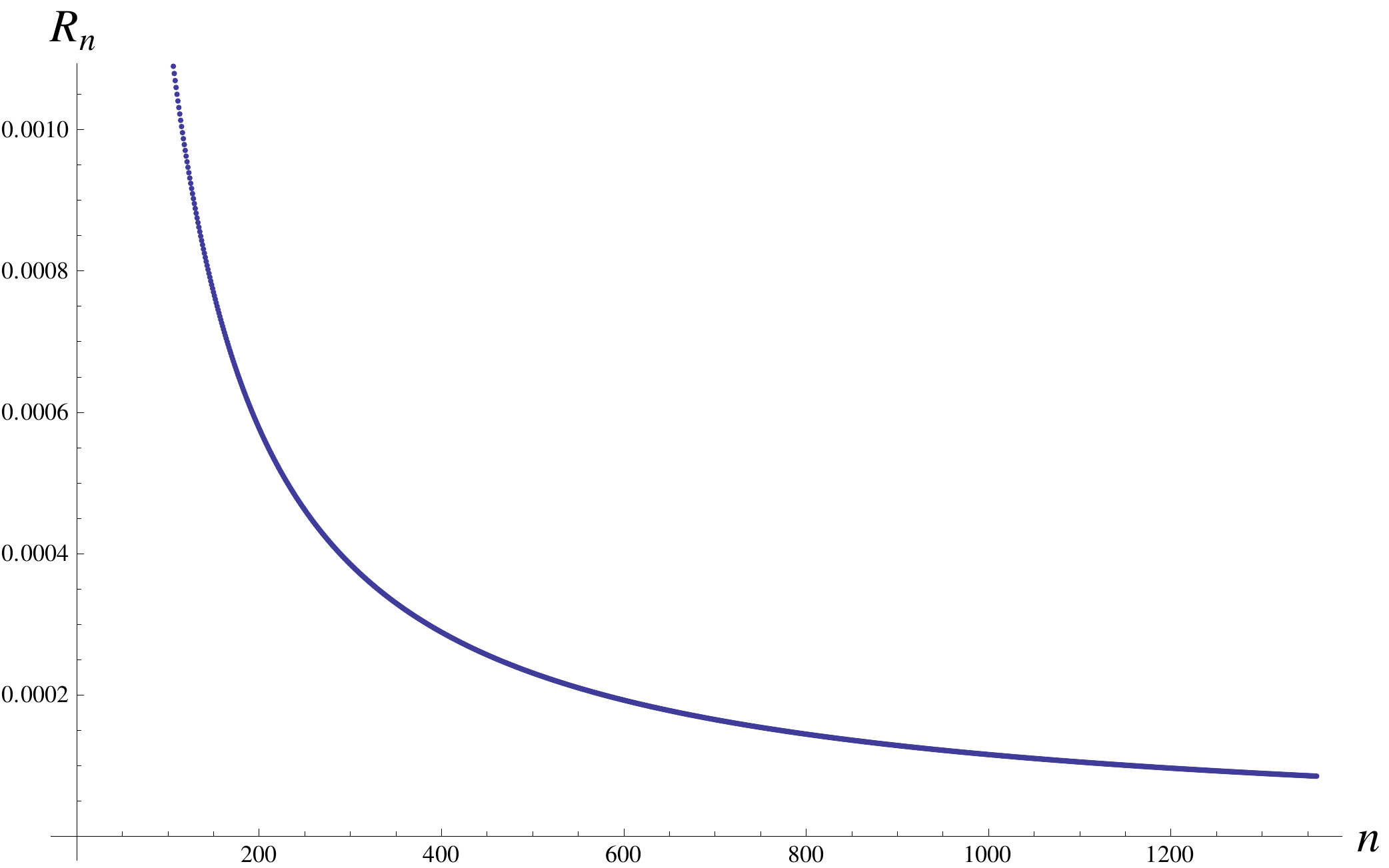}
		\caption{Plot of $R_{n}$ (defined in \eqref{eq:R_n_def_DT}) vs. $n$ for $1 \leq n  \leq 1360$. \label{fig:DT_asymptotics_error}}
	\end{center}
\end{figure}

\section{Euler Characteristics of Stable Moduli and their Asymptotics} \label{sec:stab_eul}

\subsection{Setup} \label{sec:stab_eul_setup}
Fix $m \geq 1$ an integer, then $m$-Kronecker quiver will be represented as the quiver with two vertices $q_{1}, q_{2}$ and $m$ arrows between them.
\begin{center}
\begin{tikzpicture}
\tikzstyle{block} = [rectangle, draw=blue, thick, fill=blue!10,
text width=16em, text centered, rounded corners, minimum height=2em]

\node at (1.2, 0)[circle,draw=blue,very thick] (second) {$q_{2}$};

\node at (-1.2, 0)[circle,draw=blue,very thick] (first) {$q_{1}$};

\node at (0,0.1) {{\Huge \vdots}};

\draw [decorate,decoration={brace, amplitude=6pt, mirror},xshift=0pt,yshift=-1pt,red,thick]
(-1.2,-0.9) -- (1.2,-0.9) node[black, midway,below,yshift = -6pt]{{\tiny $m$ arrows}};
\draw[-latex] (second) to  [bend right = 35] (first);
\draw[-latex] (second)  to  [bend right = 60] (first);
\draw[-latex] (second)  to  [bend right = -35] (first);
\draw[-latex] (second)  to  [bend right = -60] (first);
	\end{tikzpicture}
\end{center}
As described in \cite[\S 5.1]{reineke:quiv_rep}, and appendix \ref{app:quiv_rep}, there is a single \textit{non-trivial} stability condition that one can put on representations of this quiver.  Fix this non-trivial condition, define $\Lambda := \mathbb{Z} \langle q_{1}, q_{2} \rangle$, and denote the corresponding moduli space of stable representations with dimension vector $\alpha q_{1}^{*} + \beta q_{2}^{*} \in \Lambda^{*}$ by $\CM^{m}_{\st}(\alpha,\beta)$---it is a smooth, complex, quasi-projective variety.

Now, fix $(a,b) \in \mathbb{Z}^{2}$ coprime.  We are interested in studying the sequence of Euler characteristics
\begin{align*}
	\chi_{n} &:= \chi \left( \CM^{m}_{\st}(na,nb) \right)
\end{align*}
and DT invariants
\begin{align*}
	d_{n} &:= d(na,nb,m).
\end{align*}
To do so, we encapsulate these sequences' generating series in such a way that they arise as Euler-product factorizations.  For the Euler characteristics we define,
\begin{align*}
	\geul_{a/b} &:= \prod_{n = 1}^{\infty} (1 - z^{n})^{-n \chi_{n}} \in \formal{\mathbb{Z}}{z}
\end{align*}
and for the DT invariants:
\begin{align}
	\gdt_{a/b} &:= \prod_{n =1}^{\infty} \left(1 - \left((-1)^{N} z \right) \right)^{n d_{n}} \in \formal{\mathbb{Z}}{z}
	\label{eq:T_a_b}
\end{align}
where
\begin{align*}
	N &:= m a b - a^2 - b^2 = \dim \left(\CM^{(a,b)}_{\st} \right) - 1
\end{align*}
Because $(a,b)$ is fixed, without fear of confusion, we will write $\gdt = \gdt_{a/b}$ and $\geul = \geul_{a/b}$ for the remainder of this section.  The reason for the sign $(-1)^{N}$ in \eqref{eq:T_a_b} is at first sight an unnecessary complication, unless the formal variable $z$ has some sort of meaning in the appropriate context.  Indeed, secretly $z$ is leading a double-life as a generator of the group-ring $K_{0}(\mathsf{Rep}_{\mathbb{C}}(K_{m})) \cong \mathbb{Z}[\Lambda]$---but we will not expose $z$ for what it is, as we will not need to use this fact in any meaningful way.

Of course, we have intentionally confused the notation used for the generating series of BPS indices $\Omega \left[n \left(a \gamma_1 + b \gamma_2 \right) \right]$ at an $m$-wild point, and the generating series for the $d(na,nb,m)$; as mentioned in the introduction (c.f. \S \ref{sec:intro_BPS_indices}), BPS indices are a priori defined via super-traces over Hilbert spaces.  When the theory has a BPS quiver, however, strong physical evidence indicates that the BPS indices are precisely the DT invariants associated to the quiver (where the stability condition is derived from the central-charge function, and a quiver-potential may be present).   Hence, if $m$-Kronecker quiver arises as a BPS subquiver at a point $u \in \CB$, the nodes $q_{1}$ and $q_{2}$ can be identified with two charges $\gamma_1$ and $\gamma_2$ in $\Gamma = \widehat{\Gamma}_{u}$ such that $\langle \gamma_{1}, \gamma_{2} \rangle = m$; in particular, the sub-lattice $\mathbb{Z}\langle \gamma_{1}, \gamma_{2} \rangle$ of the full charge lattice $\Gamma$ can be identified with $\Lambda$.

On the other hand, the existence of a physical role for Euler characteristics (of stable moduli) associated to non-primitive dimension vectors is an interesting open question; nevertheless there appears to be a relationship between DT invariants and Euler characteristics of stable-moduli.

\subsection{Reineke's Functional Equation and Algebraicity for Stable Kronecker Moduli} \label{sec:reineke_func}

The key element of Reineke's proof of the integrality of the DT invariants $d(a,b,m)$ is the functional equation \cite{reineke:integrality} relating Euler characteristics and DT invariants:
\begin{align}
	\gdt &= \geul \circ \left(z \gdt^{N} \right)
	\label{eq:reineke_func}
\end{align}
where $\circ$ indicates composition of formal series.  As written, \eqref{eq:reineke_func} is a functional equation with a ``recursive flavour"; however, it is equivalent to a statement about compositional inverses of formal series.  To see this we will reverse-engineer some of Reineke's work, beginning by recalling some facts about composition of formal series.

	First, recall that formal power series $\formal{\mathbb{Z}}{z}$ form a ring under addition ($+$) and multiplication $(\cdot)$ that is, furthermore, equipped with a partially-defined composition $(\circ)$ operation (analogous the composition of functions) given by substitution of formal series.  Specifically, letting $(z) \subset \formal{\mathbb{Z}}{z}$ denote the ideal of formal series with vanishing constant coefficient, composition is a map
\begin{align*}
	\circ: \formal{\mathbb{Z}}{z} \times (z) \rightarrow \formal{\mathbb{Z}}{z}
\end{align*}
satisfying (for $A,B$ and $C$ taken in the proper domain of definition for $\circ$):
\begin{itemize}
	\item $0 \circ C = 0$ and $A \circ 0 = 0$;
	\item $(A + B) \circ C = (A \circ C) + (B \circ C)$;
	\item $(A \cdot B) \circ C = (A \circ C) \cdot (B \circ C)$;
	\item $(A \circ B) \circ C = A \circ (B \circ C)$;
	\item $A \circ z = A$, and $z \circ C = C$.
\end{itemize}
It is defined by the property that, for any $g \in (z)$, the map $(\cdot) \circ g: \formal{\mathbb{Z}}{z} \rightarrow \formal{\mathbb{Z}}{z}$ is the unique morphism of rings such that $(\cdot) \circ g: z \mapsto g$; in other words, $f \circ g$ is the formal series defined by substituting every instance of $z$ in $f$ with $g$.  Note that, in order for such a substitution to give a well-defined formal power series for general $f \in \formal{\mathbb{Z}}{z}$, the series $g$ must have zero constant coefficient\footnote{Otherwise the resulting series $f \circ g$ would have constant coefficient given by an infinite sum of elements of whatever ring we are working over.}---which is why we only pre-compose by formal series in $(z)$.

Now, the restricted composition $\circ: (z)^{\times 2} \rightarrow (z)$ equips $(z)$ with the structure of a composition ring with compositional identity given by $z \in (z)$---so it makes sense to ask whether or not a formal series in $S \in (z)$ has a compositional inverse.\footnote{Every left inverse must be a right inverse and vice-versa.}  It can be checked that such an inverse exists if and only if $[z^{1}] S = \pm 1$ (more generally, $[z^1] S$ should be a unit in the the coefficient ring).  Every such $S$ can be decomposed as $S = zF$ for some $F \in \formal{\mathbb{Z}}{z}$ with $[z^{0}] F = \pm 1$ and the existence of a compositional inverse is the statement that there exists $G \in (z)$ satisfying
\begin{align*}
	(z F) \circ G = z.
\end{align*}
We can manipulate this equation into a form with a more ``recursive flavour": because $F$ has constant coefficient in the ring of units of $\mathbb{Z}$, then $F$ admits multiplicative inverse $F^{-1} \in \formal{\mathbb{Z}}{z}$.  Using this fact, it follows that $(z F) \circ G = z$ is equivalent to the functional equation
\begin{align*}
	G = z \cdot \left(F^{-1} \circ G  \right).
\end{align*}
With this observation in mind, we are ready to rewrite Reineke's functional equation.  Indeed, define the compositionally invertible series
\begin{align*}
	D := z \gdt^{N};
\end{align*}
then \eqref{eq:reineke_func} can be written as
\begin{align*}
	D &= z \cdot \left(\geul^{N} \circ D \right);
\end{align*}
so via our discussion above, \eqref{eq:reineke_func} is equivalent to the statement that $D = z \gdt^{N}$ is the compositional inverse of $z \geul^{-N}$.  We repeat this fact in the proof of the theorem below.

\begin{theorem} \label{thm:dt_euler_alg} \
	\begin{enumerate}
		\item If $\gdt$ obeys the relation $f(z, \gdt) =0$ in $\formal{\mathbb{Z}}{z}$ for some $f \in \mathbb{Z}[z,t]$, then $\geul$ obeys the relation $f(z \geul^{-N}, \geul) = 0$.
		\item If $\geul$ obeys the relation $f(z, \geul) = 0$ in $\formal{\mathbb{Z}}{z}$ for some $f \in \mathbb{Z}[z,e]$, then $\gdt$ obeys the relation $f(\gdt, z \gdt^{N}) = 0$.
	\end{enumerate}
\end{theorem}
\begin{proof}
For the first statement; taking $D = z \gdt^{N}$ as above, we have
	\begin{align*}
		f(z \geul^{-N}, \geul) \circ D &= f \left[D \cdot \left( \geul \circ D \right)^{-N}, \geul \circ D\right]\\
		&= f \left (z  \gdt^{N} \cdot \left( \gdt \right)^{-N}, \gdt \right)\\
		&= f(z, \gdt)\\
		&= 0.
	\end{align*}
	where on the second line we used Reineke's functional equation \eqref{eq:reineke_func}.  As $D$ is compositionally invertible, then $f(w \geul^{-N}, \geul) = 0$.  The second statement follows by a similar computation.
\end{proof}

This results in the obvious corollary.
\begin{corollary}
	$\geul$ is algebraic over $\mathbb{Q}(z)$ if and only if $\gdt$ is algebraic over $\mathbb{Q}(z)$.
\end{corollary}

Moreover, if we know the polynomial relation for either $\gdt$ or $\geul$, then the proposition provides us with an explicit polynomial relation: e.g. if $f(z,\gdt) = 0$ for some two-variable polynomial $f$, then a polynomial relation $g(z,\geul) = 0$ for $\geul$ is given by the polynomial $g(z,e) = e^{d} f(z e^{-N}, e) \in \mathbb{Z}[z,e]$ for sufficiently large $d$.

\begin{remark}[Notes]\
	\begin{itemize}	
		\item Reineke's functional equation generalizes to any finite quiver (with no relations) for certain stability conditions \cite[Second proof of Thm 7.29]{js}.  The corresponding statement about algebraicity follows.
	
		\item The interpretation of Reineke's functional equation as a statement that the two functions
	\begin{align*}
		z \gdt^{N} &= \prod_{n = 1}^{\infty} \left(1 - \left((-1)^{N} z \right)^{n} \right)^{N n d_{n}} 
	\end{align*}
	and
	\begin{align*}
		z \geul^{-N} &= 	\prod_{n = 1}^{\infty} \left(1 - z^{n} \right)^{N n \chi_{n}} 
	\end{align*}
		are compositional inverses seems to suggest that we should have a priori defined the generating series for DT invariants and Euler characteristics (at least for the $m$-Kronecker quiver) directly in terms of $D = z \gdt^{N}$ and $C := z \geul^{-N}$ to make Reineke's functional equation seem more symmetric; doing so would make the statement of algebraicity also symmetric: if there exists a $g \in \mathbb{Z}[z,d]$ such that $g(z,D) = 0$, then by composing this relation with $C$ we have $g(C,z) = 0$.  Thus $D$ is algebraic if and only if $C$ is algebraic.
		\end{itemize}
\end{remark}

\begin{example}
	Let $a = b = 1$ so that $N = m-2$.  Defining $f_{m}(z,t) := -1 + t (1 - z t^{(m-2)})^{m}$ (see \eqref{eq:m_herd_relation}), then $f(z,\gdt_{1}) = 0$; hence, by Thm. \eqref{thm:dt_euler_alg}, (using the fact that $N = m-2$ when $a = b = 1$) we have a degree 1 polynomial satisfied by $\geul = \geul_{1}$:
\begin{align*}
	0 &=  -1 + \geul_{1} (1 - z \geul_{1}^{-(m-2)} \geul_{1}^{(m-2)})^{m} = -1 + \geul_{1} ( 1 - z)^{m}
\end{align*}
	Hence $\geul_{1} = (1-z)^{-m}$ corresponding to the Euler characteristics $\chi_{1} = -m$ and $\chi_{n} = 0$ for all $n > 1$.  This is effectively a reverse-engineering of Reineke's proof of the algebraic relation satisfied by $\gdt_{1}$.  From the point of view of spectral networks, however, it is easiest to start from the algebraic relation for $\gdt_{1}$ to derive the algebraic relation (and corresponding Euler characteristics) for $\geul_{1}$.  Indeed, if one accepts that the BPS generating series for $m$-herds is the generating series $\gdt_{1}$ for DT invariants $d(n,n,m)$, then the spectral-network derivation of the relation $f_{m}(z,\gdt_{1}) = 0$ in \cite{wwc}, along with the discussion above, constitutes an independent physical derivation of the Euler characteristics for stable quiver moduli with dimension vectors $(an,bn) = (n,n)$. 
\end{example}

In \S \ref{sec:Euler_3_2} we compute the Euler characteristics in the case $(a,b) = (3,2)$ using the algebraic relation \eqref{eq:3_2_dt_relation} for $\gdt_{3/2}$, which was derived through spectral network techniques applied to the $(3,2|3)$-herd.

\subsection{Determining Euler characteristics from DT invariants and Vice-Versa} \label{sec:euler_dt_formulae}
Using \eqref{eq:reineke_func}, Reineke uses Lagrange-inversion to write down a rather explicit formula expressing DT invariants in terms of a weighted sum over Euler characteristics of stable moduli (see Lemmata 5.2, 5.3 and Prop. 5.4 of \cite{reineke:integrality}):
	\begin{align*}
		k^2 d_{k} &= \sum_{l|k} \mu \left(\frac{k}{l}\right) (-1)^{l} \sum_{r = 1}^{l} \left( \sum_{(\lambda_{1}, \cdots \lambda_{r}) \in \mathcal{P}(l,r)} (-1)^{\lambda_{1}} \prod_{j = 1}^{r} \binom{j N \chi_{j}}{\lambda_{j} - \lambda_{j+1}} \right)
	\end{align*}
	where
	\begin{align*}
		\mathcal{P}(l,r) &:= \left \{ \text{Partitions of $l$ of length $r$} \right \} = \left \{(\lambda_{1}, \cdots, \lambda_{r}) \in \mathbb{Z}^{r}: \sum_{j = 1}^{r} \lambda_{j} = l \right \}.
	\end{align*}
	On the other hand, by composing \eqref{eq:reineke_func} with $z \geul^{-N}$, we have (immediate from our discussion about compositional inverses)
	\begin{align*}
		\geul &= \gdt \circ \left(z \geul^{-N} \right)	
	\end{align*}	
	so we can apply the same Lagrange-inversion procedure to derive an explicit formula expressing Euler characteristics in terms of a weighted-sum over DT invariants:
	\begin{align*}
		k^2 \chi_{k} &= \sum_{l|k} \mu \left(\frac{k}{l}\right) (-1)^{l} \sum_{r = 1}^{l} \left( \sum_{(\lambda_{1}, \cdots \lambda_{r}) \in \mathcal{P}(l,r)} (-1)^{\lambda_{1}} \prod_{j = 1}^{r} \binom{-j N d_{j}}{\lambda_{j} - \lambda_{j+1}} \right).
	\end{align*}

An alternative method, which is quite easy to implement into computer algebra software, is provided by the following.  Suppose we are given the BPS generating series $\gdt = \gdt_{a/b}$; then the corresponding Euler characteristics can be extracted via the following recursive procedure:
 \begin{align}
 	\chi_{k} &=  \frac{1}{k} [z^{k}] \left\{\frac{\gdt}{\prod_{l = 0}^{k-1} \left[1 - (z \gdt^{N})^{l}) \right]^{-l \chi_{l}} } \right\},
 	\label{eq:euler_rec}
 \end{align}
 which can be derived from \eqref{eq:reineke_func}.  If one knows the BPS generating series accurately up to the coefficient multiplying $z^{n}$, then we can reliably extract the first $n$ Euler characteristics in this manner.  Similarly, if one were provided with the Euler characteristic generating series $\geul = \geul_{a/b}$, then
 \begin{align*}
 	d_{k} &=  \frac{(-1)^{Nk}}{k} [z^{k}] \left\{\frac{\geul}{\prod_{l = 0}^{k-1} \left[1 - (z \geul^{-N})^{l}) \right]^{l d_{l}} } \right\}
 \end{align*}

 
\begin{numrmk}
 Note that from \eqref{eq:euler_rec}, we have
 \begin{align*} 
 	\chi_{1} = [z] \gdt = (-1)^{N+1} d(a,b,m) = (-1)^{\dim \CM^{(a,b)}_{\st}} d(a,b,m).
\end{align*} 	
Indeed, for primitive dimension vectors (i.e. coprime pairs $(a,b)$), all semi-stable representations are stable and the DT invariants $d(a,b,m)$ are weighted Euler characteristics of the smooth complex projective variety $\CM_{\st}^{m}(a,b) = \CM_{\sst}^{m}(a,b)$ parameterizing stable representations; in this situation the weighting function is just an overall sign: (see \cite[\S 4]{js}):
\begin{align*}
	d(a,b,m) = (-1)^{\dim \mathcal{M}^{(a,b)}_{\st}} \chi \left[\CM_{\st}^{m}(a,b) \right] = (-1)^{N + 1} \chi \left[\CM_{\st}^{m}(a,b) \right] .
\end{align*}
 On the other hand, for non-primitive dimension vectors, we should expect 
\begin{align*} 
	|d(ka,kb,m)| \neq \chi_{k}.
\end{align*}  
Indeed, from the perspective of Joyce-Song/Behrend, the DT invariant \footnote{In the notation of \cite{js}, the DT invariants we are interested in are called ``BPS-invariants" and denoted by the symbol $\hat{DT}$; the DT invariants denoted $\bar{DT}$ are related to the coefficients of generating series for DT invariants.}
$d(ka,kb,m)$, for $k > 1$, is a weighted Euler characteristic (with non-trivial weighting function: see\cite[Equation 7.56]{js}) of the (typically singular) projective variety $\CM_{\sst}^{m}(ka,kb)$ that parametrizes polystable representations.
\end{numrmk}

\subsection{Euler Characteristics of Stable Moduli with Dimension vectors \texorpdfstring{$(3n,2n)$}{(3n,2n)}} \label{sec:Euler_3_2}
In this section we specialize to the case $(a,b) = (3,2)$.  By applying Thm. \ref{thm:dt_euler_alg} to the conjectured algebraic relation \eqref{eq:3_2_dt_relation} for the generating series $\gdt = \gdt_{3/2}$, we find that $\geul = \geul_{3/2}$ must obey the algebraic relation:
\begin{equation}
	\begin{aligned}
 	0=& -z + \geul (-5 z + z^2) + \geul^2 (z - 13 z^2 + z^3) + \\
	& \geul^3 (21 z - 114 z^2 + 44 z^3 - z^4) + \\
 	& \geul^4 (-34 z - 80 z^2 - 38 z^3 + 7 z^4) + \\
 	& \geul^5 (-1 - 7 z + 6 z^2 + 77 z^3 + 39 z^4 + 3 z^5) + \\
	& \geul^6 (4 + 76 z + 119 z^2 - 382 z^3 - 367 z^4 - 17 z^5) + \\
	& \geul^7 (-6 - 64 z + 53 z^2 + 270 z^3 - 173 z^4 - 77 z^5 - 3 z^6) + \\
 	& \geul^8 (4 + 6 z - 55 z^2 + 80 z^3 - 30 z^4 - 14 z^5 + 9 z^6) + \\
 	& \geul^9 (-1 + 7 z - 21 z^2 + 35 z^3 - 35 z^4 + 21 z^5 - 7 z^6 + z^7).
 	 \end{aligned}
 	\label{eq:3_2_eul_relation}
\end{equation}

\begin{remark}
	Note that, while \eqref{eq:3_2_dt_relation} defines $\gdt$ as a root of a $39$th degree polynomial (with coefficients in $\mathbb{Z}[z]$, the relation \eqref{eq:3_2_eul_relation} defines $\geul$ as a root of a $9$th degree polynomial; in this sense, the relation satisfied by $\geul$ is much easier than that satisfied by $\gdt$.  
\end{remark}

By substituting in a formal power series (with constant coefficient 1) ansatz for $\geul$, one can determine the coefficients order by order.  The first few values of the expansion are,
\begin{align*}
	\geul &= 1 + 13 z + 189 z^2 + 1645 z^3 - 6611 z^4 - 429139 z^5 + \mathcal{O}(z^{6}).
\end{align*}
If we determine this expansion out to the term multiplying $z^K$, we can reliably determine the first $K$ Euler characteristics (of stable moduli).  Section \ref{sec:euler_dt_formulae} describes various other procedures for extracting the Euler characteristics $\chi_{n}$ given the DT invariant generating series $\gdt$.  The first few Euler characteristics are given by
\begin{align*}
	\left(\chi_{n} \right)_{n = 1}^{8} = \left(13, 49, -28, -5277, -50540, 546995, 19975988, 103632164 \right).
\end{align*}
For leisure reading, the first 1360 $\chi_{n}$---corresponding to dimension-vectors $(3n,2n)$---can be found in the file \texttt{EulerList1360.csv} included with the arXiv preprint version of this paper.\footnote{These were extracted using \eqref{eq:euler_rec} and the solution $\gdt_{3/2}$ expanded out to the coefficient multiplying $z^{1360}$.}
 
 As shown in \S \ref{sec:asymptotics}, one can determine the $n \rightarrow \infty$ asymptotics of the $\chi_{n}$ directly from the algebraic relation \eqref{eq:3_2_eul_relation}.  In particular, via Corollary~\ref{cor:multi}, we find
 \begin{align}
	\chi_{n} = n^{-5/2} \left[C \rho^{-n} + \conj{C} \left( \conj{\rho} \right)^{-n} \right] +  \mathcal{O}\left(n^{-7/2} |\rho|^{-n} \right)
	\label{eq:chi_3_2_asymps_cx}
\end{align} 
where 
\begin{align*}
	\rho \approx 0.0151352 + 0.0373931 i
\end{align*}
and its complex conjugate $\conj{\rho}$ are roots of the polynomial
\begin{equation}
	\begin{aligned}
	&-84375-86320988 z+1929101336 z^2-60820581960 z^3\\
	&+219678053524 z^4-230881412176 z^5+38310552576 z^6\\
	&-946356992 z^7+4917248 z^8
	\end{aligned}
	\label{eq:eul_discrim_factor}
\end{equation} 
(which arises as a factor of the discriminant of \eqref{eq:3_2_eul_relation} thought of as a polynomial in with coefficients in $\mathbb{Z}[z]$),
and
\begin{align*}
	C &\approx -0.077407+0.200149 i.
\end{align*}
A numerical study of the first 1360 such Euler characteristics indicates a very strong fit to these asymptotics.  Indeed, note that we may rewrite \eqref{eq:chi_3_2_asymps_cx} as
\begin{align}
	\chi_{n} = 2|C|  \Theta(n)  n^{-5/2} |\rho|^{-n}  + \mathcal{O}\left(n^{-7/2} |\rho|^{-n} \right)
	\label{eq:chi_n_asymps}
\end{align}
where
\begin{align}
	\Theta(n) &:=  \cos \left[n \arg(\rho) - \arg(C) \right].
	\label{eq:eul_oscillation}
\end{align}
Now, defining
\begin{align}
	R_{n}:= \frac{\chi_{n}}{2|C| n^{-5/2} |\rho|^{-n}} - \Theta(n)
	\label{eq:R_n_def}
\end{align}
then via \eqref{eq:chi_3_2_asymps_cx} we have $R_{n} \in \mathcal{O}(n^{-1})$.  In Fig.~\ref{fig:stable_asymptotics_error}, the first 1360 values of the sequence $R_{n}$ are plotted; visual inspection (and a bit of curve-fitting) suggests that $R_{n}$ is snugly bounded by the curves $\pm c n^{-1}$, where $c \approx .36$.

\begin{figure}[t!]
	\begin{center}
		 \includegraphics[scale=0.7]{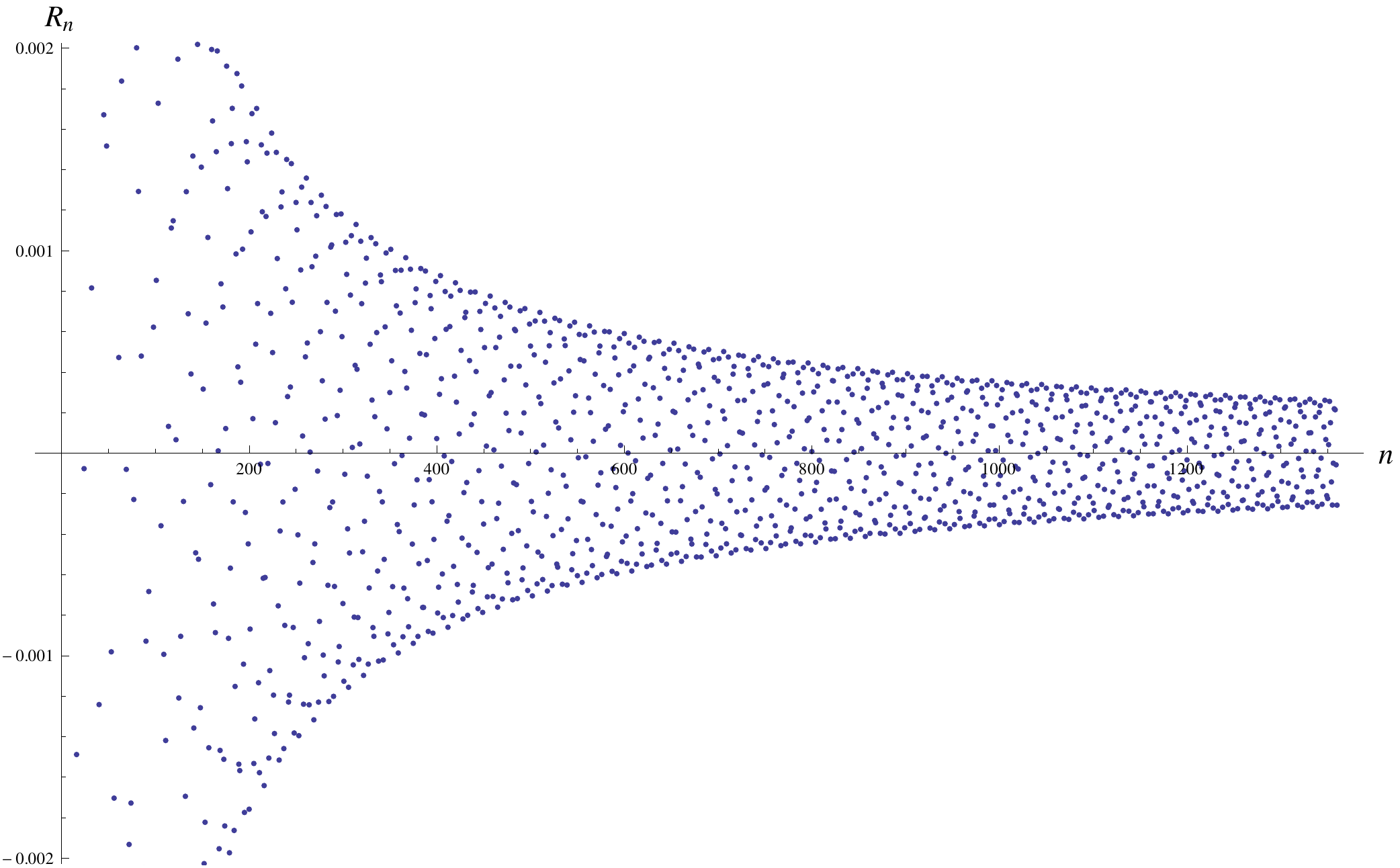}
		\caption{Plot of $R_{n}$ (defined in \eqref{eq:R_n_def}) for $1 \leq n \leq 1360$. \label{fig:stable_asymptotics_error}}
	\end{center}
\end{figure}

\begin{figure}[t!]
	\begin{center}
		 \includegraphics[scale=0.7]{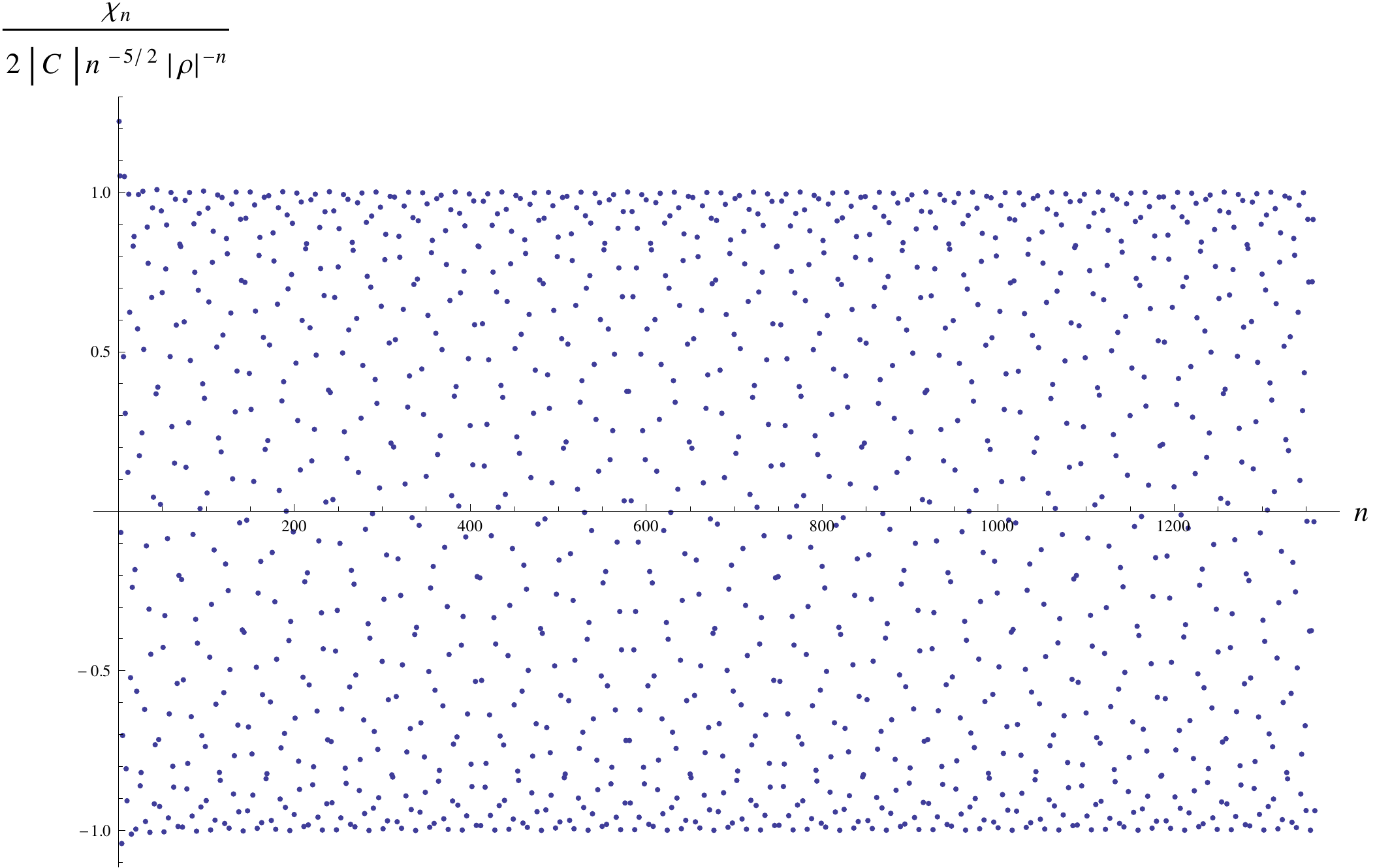}
		\caption{Plot of $\chi_{n}/(2|C| n^{-5/2} |\rho|^{-n})$ for $1 \leq n \leq 1360$  \label{fig:theory_vs_experiment_osc}}
	\end{center}
\end{figure}

\begin{figure}[t!]
	\begin{center}
		 \includegraphics[scale=0.5]{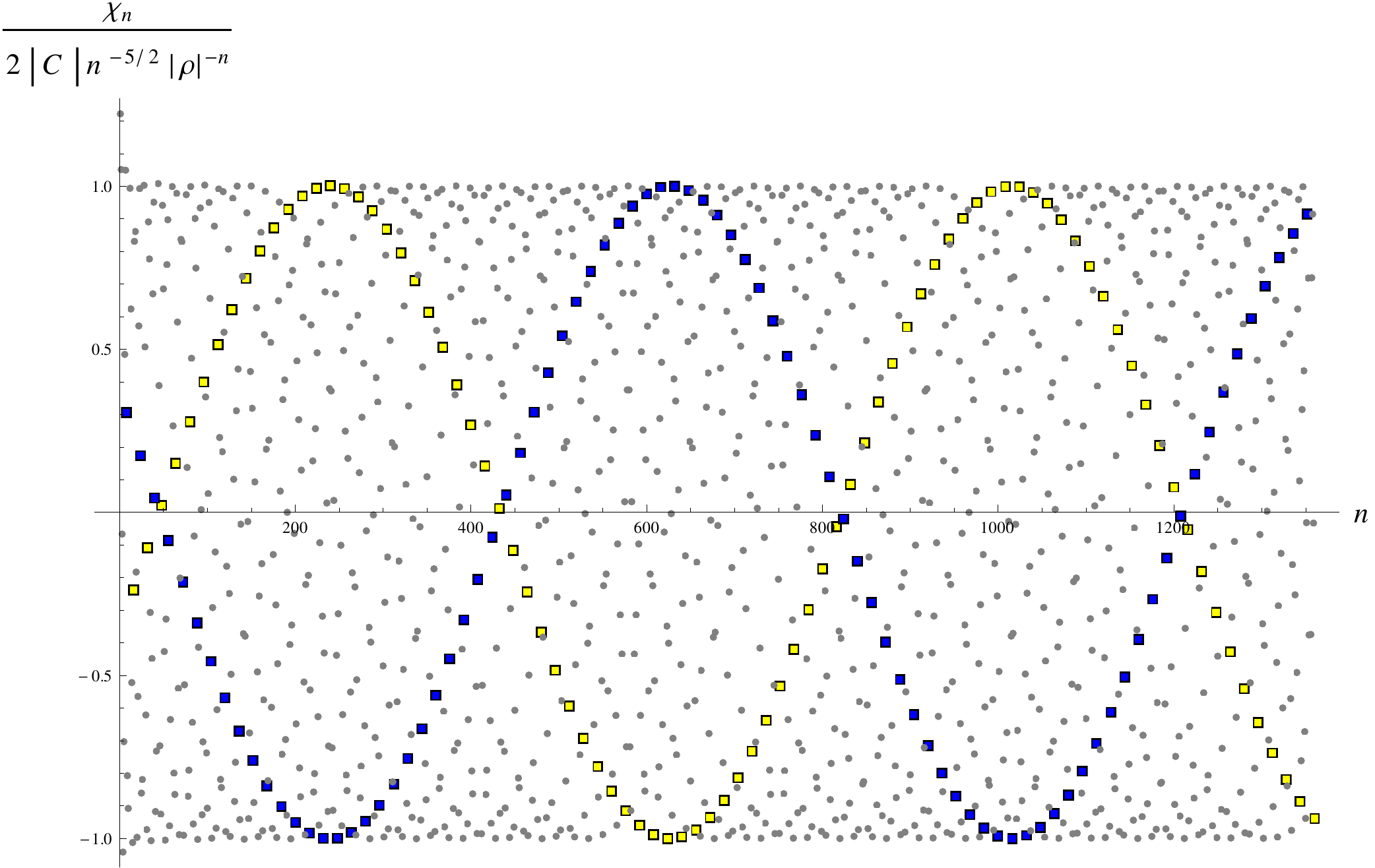} \
		 \vspace{5mm}
		 \includegraphics[scale=0.5]{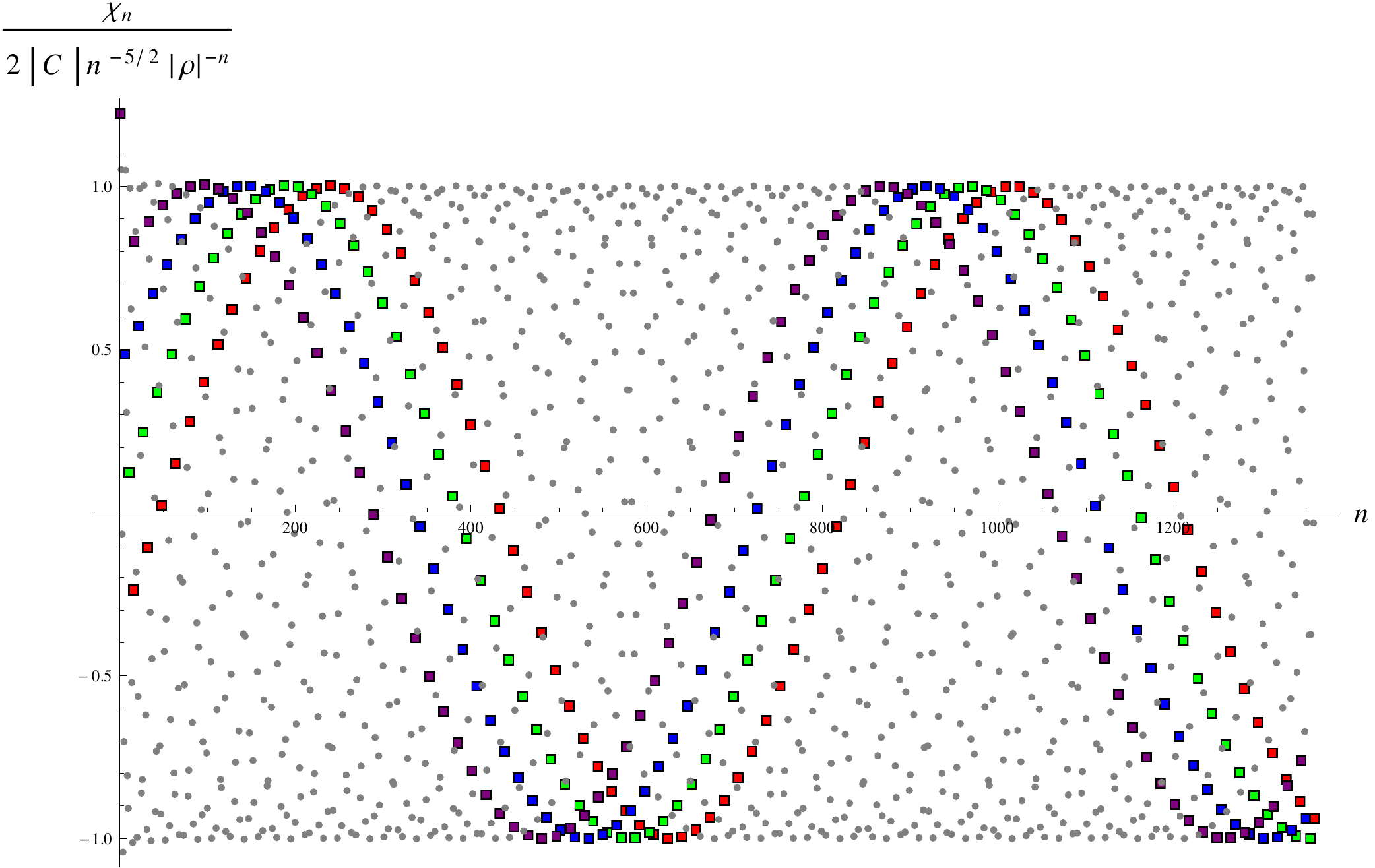}
		\caption{Plot of $\CX_{n} = \chi_{n}/(2|C| n^{-5/2} |\rho|^{-n})$ for $1 \leq n \leq 1360$ with some approximately periodic subsequences highlighted.  \textbf{Top}: Subsequence $(\CX_{n}: n \mod 16 \equiv 0)$ highlighted with yellow squares; subsequence $(\CX_{n}: n \mod 16 \equiv 8)$ highlighted with blue squares.  \textbf{Bottom}: subsequences $(\CX_{n}: n \mod 16 \equiv a)$ highlighted with red ($a = 0$), green ($a = 5$), blue ($a = 10$), and purple ($a = 15$) squares. \label{fig:theory_vs_experiment_osc_highlight}}
	\end{center}
\end{figure}

\subsubsection{Some Numerology} \label{sec:numerology}
	The scattered appearance of Fig.~\ref{fig:theory_vs_experiment_osc} is expected due to the sampling of the cosine function $\cos \left[n \arg(\rho) - \arg(C) \right]$ at integer points $n$:  if $\arg(\rho)$ is an irrational multiple of $2\pi$, then this sampling will never exhibit any exact periodicity in the variable $n$.  Indeed, using the polynomial \eqref{eq:eul_discrim_factor}, a calculation (assisted by computer-algebra software) using resultants shows that the minimal polynomial of $\rho/|\rho|$ is not a cyclotomic polynomial.  As a result, $\rho/|\rho|$ is not a root of unity; hence, $\arg(\rho)$ must be an irrational multiple of $2\pi$.

	Nevertheless, among the chaos there seems to be some regularity: it is possible to recognize patterns in Figs.~\ref{fig:stable_asymptotics_error} and \ref{fig:theory_vs_experiment_osc}.  In particular, a visually striking pattern reveals itself if one defines the sequence
\begin{align*}
	\CX_{n} := \frac{\chi_{n}}{2|C| n^{-5/2} |\rho|^{-n}},
\end{align*}	
(which quickly approaches $\Theta(n)$ defined in \eqref{eq:eul_oscillation}) and studies the sixteen subsequences defined by residue classes of $n$ mod 16:
\begin{align}
	(\CX_{n}: \text{$n$ mod 16} \equiv l)_{n = 0}^{\infty} = \left(\CX_{16k - l} \right)_{k = 0}^{\infty}, \: l = 0 ,\cdots, 15.
	\label{eq:approx_period_subseq}
\end{align}
Each such subsequence appears to be sampled from a sinusoidal function with period $97/2$ in the index $k$.  The reason for this is due to the close proximity of $\arg(\rho^{16})$ to $4\pi / 97$; specifically
\begin{align*}
	\left |\frac{16 \arg(\rho)}{2\pi} - \frac{2}{97} \right| \mod 1  < 2.41 \times 10^{-6}
\end{align*}
which leads to the approximate $97/2$ periodicity in the index $k$ over intervals of size $<10^{6}$.  Plots highlighting some of the subsequences \eqref{eq:approx_period_subseq} (indexed by $n$) are shown in Fig.~\ref{fig:theory_vs_experiment_osc_highlight}.

\section{Algebraic Generating Series and Spectral Networks: General Features} \label{sec:alg_sn}

	The appearance of algebraic equations for BPS indices seems to be a feature of spectral networks: in all degenerate spectral-networks so far encountered, the resulting equations for soliton generating series around joints (e.g. the six-way street equations) and branch points generate a system of algebraic equations that completely determine the corresponding street-factors.  Indeed, in \cite[App. F]{wwc} it was argued that this is a property for ``sufficiently nice" spectral networks: the corresponding street-factors are determined by a collection of algebraic equations.  In this section we will give a more precise statement of algebraicity for spectral networks.  Our key statement is Claim \ref{claim:func_eqn}, which requires a slightly weaker version of the condition on the spectral network compared to the condition\footnote{In \cite[App. F]{wwc}, if $\mathpzc{N}$ is a spectral network subordinate to $\pi: \Sigma \rightarrow C$, the argument required the existence of a compact set $K \subset C$ such that all two-way streets are contained in $K$, and only finitely many streets (one-way or two-way) are contained in $K$.  The condition we state here weakens this to only require finitely many two-way streets; this weakening may be important in the context of $(a,b|m)$-herds whose streets appear to densely fill regions on the cylinder, while still containing only finitely many two way streets.} in \cite[App. F]{wwc}.  We will not prove this claim here, but defer the reader to future work: an honest proof requires a development of careful definitions for spectral networks; in particular, to make sense of the product of soliton generating series---which are supposed to be formal series in non-commutative variables---one must carefully define the appropriate non-commutative ring of formal series as the completion of a filtered ring.  Such a development deserves separate treatment that will lead us too far astray.

In this section, by spectral network we mean something either satisfying a generalization\footnote{Namely, we allow for spectral networks with no singular points and relax the condition on C1 that any compact subset of $C'$ intersects only finitely many segments (allowing for spectral networks with densely defined streets).} of the conditions D1-C4 in \cite[\S 9.1]{sn} or a $\CW$-network.  Recall that, to each street of the spectral network we associate two quantities we call \textit{soliton generating series} (see \cite[\S 2.2.2]{wwc});\footnote{In \cite{wwc} and in the appendices of this paper these generating series are denoted by the symbols $\Upsilon_{p}$ and $\Delta_{p}$ or, in \cite{sn}, as $\tau$ and $\nu$.} they are elements of a non-commutative ring of formal series (which, as mentioned above, takes some care to properly define).  From a polynomial combination of soliton generating series on the street $p$ of a spectral network, we can derive \textit{street factors} $Q(p)$ (see \cite[\S 2.2.2]{wwc})---these latter quantities are elements of a \textit{commutative} ring of formal series, and are directly related to the generating series for BPS indices (see the proof of Cor. \ref{cor:alg_BPS}).

We should emphasize that the soliton generating series and street-factors are not taken as data in the definition of a spectral network, but are derived directly from the data of spectral networks.  Indeed, from the data, we can derive a system of relations that must be obeyed by soliton generating series; let us temporarily dub these as the \textit{vertex relations}.\footnote{Alternatively, as suggested by Greg Moore, these should be called ``non-abelian Kirchoff rules" due to their alternative interpretation as conservation laws for solitons of certain types coming in or out of each vertex.}  We claim that, the vertex relations can be used to define the soliton generating series: interpreting the relations as equations in non-commutative variables, there is a unique solution in a particular non-commutative ring of formal series.

Now let us specialize our attention to sufficiently ``generic" degenerate spectral networks: these are spectral networks that support a single charge (for $\CW$ networks, these are degenerate networks that are off any walls of marginal stability).   To define what this means, recall that the data of any spectral network $\mathpzc{N}$ defines a directed graph $\mathrm{Lift}(\mathpzc{N})$ with embedding into $\Sigma$ (see \S 2.2 and App. C.7.2 of \cite{wwc}).

\begin{definition}
We will say a spectral network $\mathpzc{N}$ \textit{supports a single charge} $\gamma_{c}$ if there exists a non-constant closed path on $\mathrm{Lift}(\mathpzc{N})$ (whose orientation is consistent with each edge) and, moreover, the image of each such closed path under the embedding $\mathrm{Lift}(\mathpzc{N}) \hookrightarrow \Sigma$ has homology class contained in $\mathbb{Z}_{>0} \gamma_{c} \leq H_{1}(\Sigma;\mathbb{Z})$.
\end{definition}

By the discussion above, using the definition of the soliton generating series as the unique solution to the vertex relations (in a particular non-commutative ring of formal series), we state the following.

\begin{claim}
	Suppose that $\mathpzc{N}$ is a spectral network supporting a single charge $\gamma_{c}$; define $\twid{z} := X_{\twid{\gamma}_{c}}$.  Then, for any street $p$ of $\mathpzc{N}$, we can use the vertex relations to define $Q(p)$ as an element of the ring $\formal{\mathbb{Z}}{\twid{z}}$ of formal series; moreover, $Q(p)$ has constant coefficient 1.
\end{claim}

Now we may give a statement of algebraicity in the context of spectral networks: under some mild assumptions, the (a priori) formal series $Q(p)$ are algebraic functions over $\mathbb{Q}$.  The precise statement is the following.

\begin{itclaim} \label{claim:func_eqn}
	Let $\mathpzc{N}$ be a spectral network subordinate to the branched cover $\Sigma \rightarrow C$, such that 
	\begin{enumerate}
		\item $\mathpzc{N}$ supports a single charge $\gamma_{c} \in H_{1}(\Sigma; \mathbb{Z})$;	
	
		\item Its collection of strictly two-way streets 
			\begin{align*}
				\mathrm{degstr}(\mathpzc{N}) := \{p \in \mathrm{str}(\mathpzc{N}): Q(p) \neq 1 \}
			\end{align*}
		 is a finite set.
	\end{enumerate}
	Then, letting $\twid{z} = X_{\twid{\gamma_{c}}}$, for each $p \in \mathrm{degstr}(\mathpzc{N})$ there exists a $\mathcal{F}_{p} \in \left(\mathbb{Q}[\twid{z}] \right)[x]$ (a polynomial with coefficients in $\mathbb{Q}[\twid{z}]$) such that
		\begin{align}
		0 &= \mathcal{F}_{p} \left[Q\left(p\right) \right].
		\label{eq:claim_relations}
	\end{align}

\end{itclaim}

Although we are deferring the proof of this claim until future work, it should be mentioned that the proof is of a constructive nature: the polynomial relations $\mathcal{F}_{p}$ of the claim can be extracted through applications of elimination theory to a system of polynomial equations given produced by ``abelianizing" the vertex relations.  In other words, from a spectral network we can (in principle) algorithmically extract explicit algebraic relations obeyed by the $Q(p)$.

\begin{remark}
	Although no examples are known, in principle it may be possible for a degenerate $\CW$-network to violate the first condition.  However, for the network to correspond to a BPS state of finite mass, the total length of any closed curve supported on the lift of the two-way streets to the spectral cover (defined using the holomorphic differential on the spectral cover) must remain finite.
\end{remark}

The claim induces the following corollary.
\begin{corollary} \label{cor:alg_BPS}
	Let $\gdt$ be the BPS generating series associated to any spectral network satisfying the conditions of Claim \ref{claim:func_eqn},  then there exists a polynomial $\mathcal{R} \in \left(\mathbb{Z}[z]\right)[t]$ such that
\begin{align*}
	0&= \mathcal{R} \left(\gdt \right).
\end{align*}
\end{corollary}
\begin{proof}
	The BPS generating series (c.f \eqref{eq:gen_omega_def}) $\gdt$ is a monomial in the street-factors $Q(p),\, p \in \mathrm{degstr}(\mathpzc{N})$: let
\begin{itemize}
	\item $\mathcal{I}(\Cdot, \Cdot): C_{1}(\Sigma; \mathbb{Z})^{\otimes 2} \rightarrow \mathbb{Z}$ be the skew-symmetric form on $C_{1}(\Sigma; \mathbb{Z})$ defined by linearizing the oriented-intersection pairing of paths on $\Sigma$;
	
	\item $g_{c}^{\vee} \in C_{1}(\Sigma; \mathbb{Z})$ be a \textit{closed} 1-chain such that its corresponding homology class $\gamma_{c}^{\vee} :=[g_{c}^{\vee}] \in H_{1}(\Sigma; \mathbb{Z})$ satisfies
	\begin{align*}
		\langle \gamma_{c}^{\vee}, \gamma_{c} \rangle = 1
	\end{align*}
	under the usual intersection pairing on homology $\langle \Cdot, \Cdot \rangle: H_{1}(\Sigma; \mathbb{Z})^{\otimes 2} \rightarrow \mathbb{Z}$;
	
	\item $\ell(p) \in C_{1}(\Sigma; \mathbb{Z})$ be the lift of $p$ to a 1-chain on $\Sigma$;
\end{itemize}
then, via \eqref{eq:Q-exp}-\eqref{eq:L_omega}, the BPS generating series is given by
\begin{align*}
	\gdt &= \prod_{p \in \mathrm{str}(\mathpzc{N})} Q(p)^{\mathcal{I} \left(g_{c}^{\vee}, \ell(p) \right)},
\end{align*}
where this latter product has finite support by assumption: $Q(p) \neq 1$ if and only if $p$ is two-way.  As any polynomial combination of elements algebraic over some field is again algebraic, then $\gdt$ is algebraic over $\mathbb{Q}(z)$: by successive applications of resultants to the polynomials $\mathcal{F}_{p} \in \left(\mathbb{Z}[z]\right)[x]$, one can produce a polynomial $\mathcal{R} \in \left(\mathbb{Z}[z]\right)[t]$ such that
\begin{align*}
	0&= \mathcal{R} \left(\gdt \right).
\end{align*}
\end{proof}

\begin{example}[Examples]\
	\begin{enumerate}
		\item \textbf{$m$-herds}: Let $m$ be an integer $\geq 1$.  One of the results of \cite{wwc} was that all generating series (the soliton generating series, the street-factors, and the BPS generating series) associated to the $m$-herd can be expressed in terms of powers of a formal series $P$ satisfying the relation
	\begin{align*}
		0 &= P - z P^{(m-1)^2} - 1;
	\end{align*}
	that is, if we define
	\begin{align}
		\mathcal{H}_{k} := p - z p^{k} - 1 \in \left(\mathbb{Z}[z] \right)[p]
		\label{eq:H_k_def}
	\end{align}
	then all street-factors and the BPS generating series are given by powers of the root $P$ of $\mathcal{H}_{(m-1)^2}$ and, hence, must satisfy algebraic relations themselves (which may be computed explicitly by use of resultants).  For example, it was shown in \cite{wwc} that the BPS generating series for an $m$-herd is given in terms of $P$ as 
	\begin{align*}
		\gdt_{(1|m)} = P^{m};
	\end{align*}
one can verify that $\gdt_{(1|m)}$ arises as a root of the polynomial\footnote{By substituting $t = p^m$ into \eqref{eq:m_herd_relation}, one can prove that $\mathcal{H}_{(m-1)^{2}}(p) = 0$ if and only if $\mathcal{F}_{(1|m)}(p^{m}) = 0$ (the ``if" direction requiring a tiny bit more thought than the ``only if" direction).}
	\begin{align}
		\mathcal{F}_{(1|m)} &= -1 + t (1 - z t^{(m-2)})^{m}.
		\label{eq:m_herd_relation}
	\end{align}

	\item \textbf{The $(3,2|3)$-herd}: All street-factors of the $(3,2|3)$-herd can be written as Laurent-polynomials in the quantities $M,\, V,\, W$ satisfying the system of algebraic relations \eqref{eq:func_eqs}.  Using elimination theory, from this system of relations, one can deduce univariate polynomials $\mathcal{F}_{V},\, \mathcal{F}_{W},\, \mathcal{F}_{M}$ with coefficients in $\mathbb{Z}[z]$ such that
	\begin{align*}
		0&= \mathcal{F}_{M}(V)\\
		0&= \mathcal{F}_{V}(W)\\
		0&= \mathcal{F}_{W}(M);
	\end{align*}	
	moreover, the roots of any one of $\mathcal{F}_{M},\, \mathcal{F}_{V},\, \mathcal{F}_{W}$ are in one-to-one correspondence tuples of algebraic functions $(m,v,w)$  obeying 
\begin{equation}
	\begin{aligned}
		0 &= 1 + z m^{4} \left\{ (1 + v) (1 + v - w)^2 [v^2(1+w) - 1]^{3} \right\} - m,\\
		0 &= (-1 + v) (1 + v)^2 + (1 + v^3) W - v (m + v) w^2,\\
		0 &=  v \left(v^2 -1 \right) - \left[m (v + 1) + v(v-2) - 1 \right] w,
	\end{aligned}
	\label{eq:mvw_relations}
\end{equation}	 
 Now, \eqref{eq:mvw_relations} has forty-two tuples $(m,v,w) \in \conj{\mathbb{Q}(z)}^{3}$ of solutions, three of which are constant solutions (elements of $\conj{\mathbb{Q}}^{3}$).  As already mentioned in \S \ref{sec:asymp_DT_3_2}, elimination theory provides us with the algebraic relation \eqref{eq:3_2_dt_relation} involving only $z$ and the BPS generating series $\gdt_{3/2} = M V W$.  In other words, $\gdt_{3/2}$ is a root of the degree 39 polynomial $\mathcal{F}_{(3/2|3)} \in \left(\mathbb{Z}[z] \right)[t]$ listed below:
 \begin{equation}
 	\begin{aligned}
	 	\mathcal{F}_{(3/2|3)} &:= -(1 + z) + t(4 - 5z) +  t^2(-6 + z) + t^3(4 + 21z) +\\
		&t^4(-1 - 34z) - t^5(7z) + t^6(76z + z^2) + t^7(-64z - 13z^2) +\\
		&t^8(6z - 114z^2) + t^9(7z - 80z^2) + t^{10}(6z^2) +\\
	 	&t^{11}(119z^2) + t^{12}(53z^2 + z^3) + t^{13}(-55z^2 + 44z^3) +\\
	 	&t^{14}(-21z^2 - 38z^3) + t^{15}(77z^3) - t^{16}(382z^3) +\\
	 	&t^{17}(270z^3) + t^{18}(80z^3 - z^4) + t^{19}(35z^3 + 7z^4) +\\
	 	&t^{20}(39z^4) - t^{21}(367z^4) - t^{22}(173z^4) -\\
		&t^{23}(30z^4) - t^{24}(35z^4) +  t^{25}(3z^5) - t^{26}(17z^5) -\\
		&t^{27}(77z^5) - t^{28}(14z^5) + t^{29}(21z^5) - t^{32}(3z^6) +\\
		&t^{33}(9z^6) - t^{34}(7z^6) + t^{39}z^7.
	\end{aligned}
	\label{eq:3_2_3_poly}
 \end{equation}
	This polynomial enjoys the following properties:
	\begin{itemize}
		\item it is irreducible as a polynomial in $\mathbb{C}[z,t]$;\\
		
		\item each of its roots is in one-to-one correspondence with a triple $(m,v,w)$ of non-constant solutions to \eqref{eq:mvw_relations}.
	\end{itemize}
	As a consequence of the latter property and the discussion at the end of \S \ref{sec:2_3_herds}, the BPS generating series $\gdt_{3/2} = M V W$ is the unique solution given by a formal power series with constant coefficient 1.\footnote{An alternative derivation of this statement is provided in Example \ref{ex:3_2_curve} at the end of \S \ref{sec:alg_curves}.}
	\end{enumerate}
\end{example}

In \S \ref{sec:alg_asymp} we go about exploring the general properties of the algebraic curves associated to these algebraic equations.  In the process, we show that we can always find a factor of $\mathcal{R} \in \left(\mathbb{Z}[z] \right)[y]$ that is irreducible over $\mathbb{C}(z)$ and, furthermore, still has coefficients in $\mathbb{Z}[z]$.

\section{Algebraic Generating Series and Asymptotics} \label{sec:alg_asymp}
%
%

\begin{definition}[Notation]\
	\begin{enumerate}
		\item $\conj{\mathscr{K}}$ will denote the algebraic closure of a field $\mathscr{K}$.
		\item Let $R$ be a ring, then $\mathsf{taut}_{R}: \left(R[z]\right)[t] \rightarrow R[z,t]$ will denote the tautological map; it is an isomorphism of $R[z]$-modules.  Elements of $\left(R[z] \right)[t]$ will be denoted by calligraphic letters and their images under $\mathsf{taut}_{R}$ will be denoted by their respective Roman-typeface analogues.
		\item In the following two sections: $T$ will denote a general algebraic function, free from a particular context (as opposed to the BPS generating series $\gdt$ that is indicated with sans-serif font). The symbol $z$ will denote a formal variable, which can be interpreted free from any particular context; however, when we specialize to examples coming from $(a,b|m)$-herds, $z$ will specifically denote the variable defined in \eqref{eq:z_var_def} (and explained in Appendix \ref{app:signs}).  We will occasionally abuse notation and write statements like ``$z \in \mathbb{C}$" when we mean the evaluation of the variable $z$ at some value in $\mathbb{C}$.
	\end{enumerate}
\end{definition}

\subsection{Algebraic Curves} \label{sec:alg_curves}

Suppose we wish to study solutions $T$ to the algebraic functional equation
\begin{align}
	\mathcal{R}(T) &= 0
	\label{eq:gen_func_eqn}
\end{align}
for some $\mathcal{R} \in \left(\mathbb{Z}[z] \right)[t]$, i.e. a polynomial with coefficients in $\mathbb{Z}[z]$.  Equivalently, we wish to study roots of $\mathcal{R}$, which splits in $\conj{\mathbb{Q}(z)}[t]$: if $c = \deg_{\mathbb{Z}[z]}\mathcal{R}$, then
\begin{align*}
	\mathcal{R} &= A \left(t - T_{1} \right) \cdots \left(t - T_{c} \right)
\end{align*}
for some $A \in \conj{\mathbb{Q}(z)}$ and $T_{l} \in  \conj{\mathbb{Q}(z)},\, l = 1, \cdots, c$. Being geometrically minded, we will study the collection of roots $\{T_{l}\}_{l = 1}^{c}$ by analysing complex affine/projective (typically singular) algebraic curve(s) associated to $\mathcal{R}$; when we pass to the analytic world, the roots can be interpreted as local sections of a branched cover defined by this curve.

Before proceeding further, we make an observation that will prove useful for analysing the behaviour of $T$ at ``poles" (points on the $z$-plane where $T$ becomes infinite).\footnote{We place ``poles" in quotes as the sort of singularities encountered may also manifest themselves as branch points: e.g. the point $z=0$ of the function $z^{-1/2}$.}

\begin{numrmk} \label{rmk:inversion}
	The space of algebraic functions $\conj{\mathbb{Q}(z)}$ is a field; in particular, if $T \not \equiv 0$ is algebraic, then $1/T$ must be algebraic.  In fact, if $T$ is a root of \eqref{eq:gen_func_eqn}, then 
	\begin{align*}
		0&= \widehat{\mathcal{R}} \left(\frac{1}{T} \right)
	\end{align*}
	where $\widehat{\mathcal{R}} \in \mathbb{Z}[z,y]$ is defined by
	\begin{align*}
		\widehat{\mathcal{R}}(y) &= y^{e} \mathcal{R} \left(\frac{1}{y} \right).
	\end{align*}
	 Note that the operation $\widehat{\left( \cdot \right)}: \mathcal{R} \mapsto \widehat{\mathcal{R}}$ is an involution and commutes with factorizations; in particular $\widehat{\mathcal{R}}$ is irreducible (over some chosen field) iff $\mathcal{R}$ is irreducible.
\end{numrmk}

  Let $r = \mathsf{taut}_{\mathbb{Z}} \mathcal{R} \in \mathbb{Z}[z,t]$ be the polynomial in two variables defined by $\mathcal{R}$.  The zero locus of $r$ defines an affine curve $V_{r} \subset \mathbb{C}^{2}$; to simplify the following discussion, we focus on any given irreducible component of this curve: suppose $r$ factors as a polynomial in $\mathbb{C}[z,t]$ as
  \begin{align}
  	r &= f_{1}^{m_{1}} \cdots f_{k}^{m_{n}}
  	\label{eq:f_factorization}
  \end{align}
  for irreducible $f_{l} \in \mathbb{C}[z,t]$ and multiplicities $m_{1}, \cdots, m_{n} \in \mathbb{Z}_{>0}$, then each $V_{f_{i}},\, i = 1,\cdots k$ is an irreducible complex algebraic curve.

\begin{numrmk} \label{rmk:Q-bar_coeffs}
  	As $r \in \mathbb{Q}[z,t]$, then the irreducible factors $f_{l} \in \mathbb{C}[z,t]$ in \eqref{eq:f_factorization} actually have $\conj{\mathbb{Q}}$ coefficients: $f_{l} \in \conj{\mathbb{Q}}[z,t]$.\footnote{Indeed, the factorization \eqref{eq:f_factorization} induces a factorization $\mathcal{R} = \mathcal{F}_{1} \cdots \mathcal{F}_{n}$ where each factor $\mathcal{F}_{l} \in \left(\mathbb{C}[z] \right)[t]$ is the image of $f_{l}$ under the obvious tautological map.  The algebraic closure $\conj{\mathbb{Q}(z)}$ contains the splitting field of $\mathcal{R}$; so $\mathcal{F}_{l} \in \conj{\mathbb{Q}(z)}[t]$ for all $l$.  On the other hand, $f_{l} \in \mathbb{C}[z,t]$ means that $\mathcal{F}_{l}$ must have coefficients in $\conj{\mathbb{Q}(z)} \cap \mathbb{C}[z] =\conj{\mathbb{Q}}[z]$.  The equality in the previous sentence (which completes the proof), follows by showing that any element of $\conj{\mathbb{Q}(z)}$ can have zeros only at algebraic numbers.}
  \end{numrmk}
Let $f$ be one of the irreducible factors $f_{l}$.  Regarding our choice of $f$ in practice, note that the roots of $\mathcal{F}$ are a subset of the roots of $\mathcal{R}$; so if our interest lies in studying a particular root $T_{*}$ of $\mathcal{R}$, then we choose $f$ to be an irreducible factor such that 
\begin{align*}
	\mathcal{F} := \mathsf{taut}_{\conj{\mathbb{Q}}}^{-1} f \in \left(\conj{\mathbb{Q}}[z] \right)[t]
\end{align*}
has $T_{*}$ as a root.  We make one further remark that will come into use later.
\begin{remark}
 $\mathcal{F}$ is irreducible in $\left(\mathbb{C}(z) \right)[t]$ by virtue of the irreducibility of $f$ in $\mathbb{C}[z,t]$ in combination with Gauss' Lemma.  Hence, $\mathcal{F}$ has a root in rational functions $\mathbb{C}(z)$ if and only if $\deg_{\conj{\mathbb{Q}}[z]} \mathcal{F} = 1$. 

 \end{remark}

  Our ulterior motive behind these statements are, of course, that we wish to choose $f$ such that $\mathcal{F}$ has a root $T \in \formal{\mathbb{Z}}{z} \cap \conj{\mathbb{Q}(z)}$ that enjoys an interpretation either as a street-factor, or the generating series for a collection of BPS indices. If we choose $f$ in such a way, we can use the following gem of a lemma (attributed to Eisenstein \cite{fs:an_comb}) that allows us---without loss of generality---to assume that $f$ has coefficients in $\mathbb{Z}$.

\begin{lemma}[Eisenstein] \label{lem:Eisenstein}
	Let $T \in \formal{\mathbb{Q}}{z} \cap \conj{\mathbb{C}(z)}$ be non-zero (i.e. an algebraic function over $\mathbb{C}$ that can be represented as a formal series with rational coefficients); suppose $f \in \mathbb{C}[z,t]$ is a non-zero polynomial such that $f(z,T(z)) \equiv 0$, then:
	\begin{enumerate}
		\item there exists a non-zero $g \in \mathbb{Z}[z,t]$ such that $g(z,T(z)) = 0$ and the degree of $g$ as a polynomial in $t$ (resp. $z$) is less than or equal to the degree of $f$ as a polynomial in $t$ (resp. $z$);
	
		\item if $f$ is irreducible in $\mathbb{C}[z,t]$, then there exists a constant $c \in \mathbb{C}$ such that $f = c g$.
	\end{enumerate}
\end{lemma}
\begin{proof}
	Assuming the first statement, the second statement is immediate from the first part along with properties of minimal polynomials.  To prove the first statement, we fill out the details of the proof in \cite[VII.37]{fs:an_comb}	. 
	 Let $A \subset \mathbb{C}$ be the set of coefficients of $f$ and $V = \mathbb{Q} A$ the (finite-dimensional) $\mathbb{Q}$-vector space generated by $A$.  Choose a basis $B \subset \mathbb{C}$ (of $\mathbb{Q}$-linearly independent elements) for $V$, then $f$ can be written as
\begin{align*}
	f &= \sum_{b \in B} b f_{b},
\end{align*}
where $f_{b} \in \mathbb{Q}[z,t]$ for every $b \in B$ and at least one $f_{b}$ is non-zero.  If there is only one non-zero $f_{b}$ then we are done; so assume that there are at least two non-zero $f_{b}$.  Because $f(z,T(z)) = 0$, we have
\begin{align}
	0 &= \sum_{b \in B} b f_{b}(z,T(z));
	\label{eq:Q_relation}
\end{align}
the right hand side is a formal power series in $z$.  Now, if none of the $f_{b}(z,T(z))$ are identically zero, then there exists a $k$ such that, extracting the coefficient of $z^{k}$ on the right hand side of \eqref{eq:Q_relation}, produces a $\mathbb{Q}$-linear relation between at least two elements of $B$---a contradiction.
\end{proof}

In summary, the discussion above shows us that if we only care about roots of \eqref{eq:gen_func_eqn} that arise as generating series for a combinatorial/ ``counting" problem (e.g. generating series that count solitons/BPS states, or closed curves on a spectral network), then we can always reduce to another algebraic equation with coefficients in $\mathbb{Z}$ which, moreover is irreducible over $\mathbb{C}[z,t]$ (a.k.a \textit{absolutely irreducible}).  Hence, we can study the geometry of the (complex) curve associated to the ``counting" problem root, while maintaining $\mathbb{Z}$-coefficients.

\begin{remark}
	Because the results of the following sections are quite general and do not rely on $\mathbb{Z}$ coefficients, we will work in complete generality and assume---via Remark \ref{rmk:Q-bar_coeffs}---that we have chosen $f \in \conj{\mathbb{Q}}[z,t]$, irreducible in $\mathbb{C}[z,t]$.
\end{remark}

Now, the curve $V_{f}$ is equipped with two maps to copies of $\mathbb{C}$ via restrictions of the standard coordinate projections of $\mathbb{C}^{2}$ to $\mathbb{C}$:

\begin{equation}
	\begin{tikzpicture}[baseline=(current  bounding  box.center)]
	\node at (0,0) (V) {$V_{f} := \left\{(z,t) \in \mathbb{C}^{2}: f(z,t) = 0 \right\} \subset \mathbb{C}^{2}$};
	
	\node at (-2, -2) (zspace) {$\mathbb{C}$};
	
	\node at (2, -2) (Tspace) {$\mathbb{C}$};
	
	
	\draw[-latex] (V) -- (zspace) node[sloped,above,midway] {\tiny $\text{pr}_{1}|_{V_{f}} = \phi$} node[sloped,below,midway]{\tiny $z \mapsfrom (z,t)$} ;
	
	\draw[-latex] (V) -- (Tspace) node[sloped,above,midway] {\tiny $\tau = \text{pr}_{2}|_{V_{f}}$} node[sloped,below,midway]{\tiny $(z,t) \mapsto t$};
	\end{tikzpicture}
	\label{eq:affine_cover}
\end{equation}
Letting $e := \deg_{\conj{\mathbb{Q}}[z]} \mathcal{F}$, then $\phi$ is a degree $e$ (algebraic) branched cover: there exists a Zariski closed set (the branching locus) $\mathpzc{B}$ such that, defining $V_{f}' := V_{f} \backslash \mathpzc{B}$, the restriction $\phi|_{V_{f}'}: V_{f}'\rightarrow \mathbb{C} \backslash \mathpzc{B}$ is a degree $e$ cover.

 For simplicity, first suppose that $\mathcal{F}$ is monic; then the branching locus (coinciding with the critical values of $\phi$) is given precisely by the values of $z_{*}$ where the polynomial $f|_{z=z_{*}} \in \conj{\mathbb{Q}}[t]$ has a repeated root: hence $\mathpzc{B}$ is just the zero locus $\Delta_{0}$ of the discriminant polynomial $\text{Disc}(\mathcal{F}) \in \conj{\mathbb{Q}}[z]$:
\begin{align}
	\Delta_{0} := \{z \in \mathbb{C}: \text{Disc}(\mathcal{F}) = 0\} \subset \conj{\mathbb{Q}};
	\label{eq:Delta_def}	
\end{align}
the corresponding set of critical points of $\phi$ is given by the \textit{ramification locus}:
\begin{align*}
	\mathpzc{R}_{\:0} := \left\{(z_{*},t_{*}) \in V_{f}: \left. \frac{\partial f}{\partial t} \right|_{(z_{*},t_{*})} = 0 \right \}.
\end{align*}
We can define the ramification index $\nu_{p} \in \mathbb{Z}_{>0}$ of a point $p = (z_{*},t_{*}) \in V_{f}$ directly in terms of derivatives of $f$:
\begin{align}
	\nu_{p} := \min \left \{k \in \mathbb{Z}_{>0}: \left. \frac{\partial^{k} f}{\partial t^{k}} \right|_{(z_{*},t_{*})} \neq 0 \right \} + 1;
	\label{eq:ram_index}
\end{align}
geometrically $\nu_{p}$ is the number of sheets of the cover $V_{f} \rightarrow \mathbb{C}$ that collide at the point $p$.  Of course, $\nu_{p} > 1$ if and only if $p \in \mathpzc{R}_{\:0}$.

Now let us consider the more general situation---the situation encountered in our examples---where $\mathcal{F}$ is not monic.  The resulting story is not the typical ``textbook" picture for affine algebraic branched covers; in particular, the curve may have infinite branches at finite values of $z$.  To see this, note that $f$ can be written in the form
\begin{align*}
	f(z,t) = a_{e}(z) t^{e} + a_{e-1}(z)t^{e-1} + \cdots a_{0}(z),
\end{align*}
where $a_{l} \in \conj{\mathbb{Q}}[z];\, l = 0,\,\cdots, \, e$.  Then, for generic $z_{*} \in \mathbb{C}$, the degree of the polynomial $f|_{z = z_{*}} \in \conj{\mathbb{Q}}[t]$ is just $e$; however, this degree drops precisely for $z_{*}$ in the set of isolated points 
\begin{align}
	\mathpzc{P} := \{z \in \mathbb{C} : a_{e}(z) = 0\} \subset \conj{\mathbb{Q}}.
	\label{eq:explode_def}
\end{align}

From the geometric point of view, branches of $\phi_{f}: 	V_{f} \rightarrow \mathbb{C}$ ``go off to infinity" as we approach points in $\mathpzc{P}$; if the degree of $f$ at $z_{*} \in \mathpzc{P}$ is $r < e$, then precisely $e - r$ branches diverge in this manner.  Thus, the branching locus\footnote{The name ``branching locus" may be misleading in this context: there may be no monodromy around a particular point in $\mathpzc{P}$; this problem will be eliminated from from the projective perspective.} for the affine curve associated to a general irreducible $f$ is given by
\begin{align*}
	\mathpzc{B} = \Delta_{0} \cup \mathpzc{P}.
\end{align*}

Moreover, from the topological perspective (using the analytic topology), it still makes sense to talk about ramification around a chosen point in $\mathpzc{P}$ by analysing the monodromy corresponding to lifts of small loops encircling the point.  One way to see the ramification algebraically is to eliminate the whole infinite branch problem by passing to the projectivized curve, as we will do shortly; however, we can also see any ramification from the perspective of affine curves.  Indeed, let $\widehat{\mathcal{F}}$ be defined according to the operation in Rmk. \ref{rmk:inversion}; then the zero locus of $\widehat{f} = \mathsf{taut}_{\conj{\mathbb{Q}}}^{-1} \widehat{\mathcal{F}}$ defines an affine curve $V_{\widehat{f}}$, birationally equivalent to $V_{f}$.  Note that,
\begin{align*}
	\widehat{f}(z,y) = a_{0}(z) y^{e} + a_{1}(z)y^{e-1} + \cdots a_{e}(z);
\end{align*}
so if the $t$-degree of $f$ at $z_{*}$ is $r \leq e$, then $a_{l}(z_{*}) = 0$ for $l \geq r + 1$; thus,
\begin{align*}
	\widehat{f}(z_{*},y) &= y^{e-r} \left[a_{0}(z_{*})y^{r} + \cdots + a_{r}(z_{*}) \right].
\end{align*}
Hence, there is a ramification point of the ``$z$-projection" 
\begin{align*}
	\begin{array}{lcl}
	\phi_{\widehat{f}}\,: & V_{\widehat{f}} &\longrightarrow  \mathbb{C}\\
		& (z,y) &\longmapsto z
	\end{array}
\end{align*}
at the point $p = (z_{*},0) \in V_{\widehat{f}} \subset \mathbb{C}^{2}$ with index $\nu_{p} = e-r+1$. 

\begin{remark}
It is a useful fact that the discriminant locus $\Delta_{0}$ defined in \eqref{eq:Delta_def} detects the presence of ramification of $\phi_{f}$ at points in $\mathpzc{P}$; however, strictly speaking, there is no corresponding ramification point to speak of on the affine curve $V_{f}$ (although there will be on its projectivization).
\end{remark}

 Similarly, the branched cover $\phi_{\widehat{f}}$ also has infinite (possibly ramified) branches over any point $z_{*} \in \mathbb{C}$ such that $(z_{*},0) \in V_{f}$, i.e. for any $z_{*}$ in the set
\begin{align}
	\mathpzc{Z} := \{z \in \mathbb{C} : a_{0}(z) = 0\} \subset \conj{\mathbb{Q}}.
	\label{eq:zed_def}
\end{align}

\begin{remark}
The sets $\mathpzc{R},\, \Delta,\, \mathpzc{P},$ and $\mathpzc{Z}$ will play an important role in \S \ref{sec:asymptotics}, during the analysis of asymptotics of sequences (e.g. of DT invariants/Euler characteristics) encoded in algebraic generating series.
\end{remark}

Returning to the diagram \eqref{eq:affine_cover}, the map $\tau \in \mathbb{C}[V_{f}] = \mathbb{C}[z,t]/(f) \subset \mathbb{C}(V_{f}) = \left(\mathbb{C}(z) \right)[t]/(\mathcal{F})$ can be thought of as a root of $\mathcal{F}$ in the field $\left(\mathbb{C}(z) \right)[t]/(\mathcal{F})$ (where $\tau$ is just the class defined by $t$).  In other words, there is a solution of \eqref{eq:gen_func_eqn} that is a regular algebraic function on a branched cover of the ``$z$-plane" $\mathbb{C}$; the attentive reader will realize nothing deep in this statement: from the analytic point of view $\tau$ is the analytic continuation of any root of $\mathcal{F}$ to a function on $V_{f}$.

To understand the behaviour of $\phi$ and $\tau$ near various sorts of ``infinity", it is convenient to consider the projectivized curve $\conj{V_{f}}$ formed from the zero locus of the homogenization $f_{h} \in \conj{\mathbb{Q}}[z,t,w]$ of $f \in \conj{\mathbb{Q}}[z,t]$, defined below.

\begin{center}
	\begin{tikzpicture}[baseline=(current  bounding  box.center)]
	\node at (0,0) (V) {$\conj{V_{f}} = \left\{[z:t:w] \in \mathbb{P}^{2}: f_{h}(z,t,w) = 0 \right\}$};
	
	\node at (-2.5, -2.5) (zspace) {$\mathbb{P}^1$};
	
	\node at (2.5, -2.5) (Tspace) {$\mathbb{P}^1$};
	
	
	\draw[-latex] (V) -- (zspace) node[sloped,above,midway] {\tiny $\conj{\phi}$} node[sloped,below,midway]{\tiny $[z:w] \mapsfrom [z:t:w]$} ;
	
	\draw[-latex] (V) -- (Tspace) node[sloped,above,midway] {\tiny $\conj{\tau}$} node[sloped,below,midway]{\tiny $[z:t:w] \mapsto [t:w]$};
	\end{tikzpicture}
\end{center}
Restricting $\conj{V_{f}}$ to the affine coordinate patch $U_{w \neq 0} = \{[z:t:w] \in \mathbb{P}^{2}: w \neq 0\}$, we recover  \eqref{eq:affine_cover}.   Indeed, the rational function $\conj{\tau} \in \mathbb{C}(\conj{V_{f}}) \cong \mathbb{C}(V_{f})$ is an extension of $\tau \in \mathbb{C}[V_{\mathcal{F}}]$ (i.e. $\conj{\tau}|_{U_{w \neq 0}} = \tau$), and the the map $\conj{\phi}$ is a degree $e$ branched cover of $\mathbb{P}^{1}$ that extends $\phi$.

\begin{remark}[Technical Note]
Of course, $\conj{\phi}$ is everywhere defined only if $[0:1:0] \notin \conj{V_{f}}$, and $\conj{\tau}$ is everywhere defined only if $[1:0:0] \notin \conj{V_{f}}$.  Annoyingly, our examples include such points, but luckily the point $[1:0:0]$ will not cause any philosophical or practical qualms: we just eliminate it from the domain of $\conj{\tau}$ and consider $\conj{\tau}$ as a birational map $\conj{V_{f}} \DashedArrow[->, densely dashed] \mathbb{P}^{1}$.  On the other hand, the point $[0:1:0]$ is a technical nuisance as it is convenient to have $\conj{\phi}$ defined everywhere; however, it can be removed by application of a projective transformation (an element of $\mathrm{PGL}_{2}(\mathbb{C})$) on $\mathbb{P}^{2}$ in order to change the embedding $\conj{V}_{f} \rightarrow \mathbb{P}^{2}$ before performing the coordinate projections that define $\conj{\phi}$ and $\conj{\tau}$.  Because we will eventually wish to write down expressions for sections of the cover $\conj{\phi}$ in terms of the coordinate $z$, the transformation should not mix $t$ and $z$.  In examples, applying one of the two $\mathrm{PGL}_{2}(\mathbb{Z})$ transformations:
\begin{align*}
	[z:t:w] &\mapsto [z':t':w']=[z+w:t:w]
\end{align*}
or
\begin{align*}
	[z:t:w] &\mapsto [z':t':w']=[z:w:t]
\end{align*}
will eliminate the annoyance points.  From the affine perspective (restricting to the coordinate patch where $w \neq 0$), the first map is just the shift $z \mapsto z + 1$ and the second map is the birational map 
\begin{align*}
	V_{f} &\DashedArrow[->, densely dashed] V_{\widehat{f}}\\
	 (z,t) &\mapsto (z,1/t).
\end{align*}
 Without further discussion of this issue, we will proceed as if this coordinate transformation has been performed (if necessary); and assume that $z,t,$ and $w$ represent the transformed coordinates.
\end{remark}

For later convenience, we note that $\conj{\phi}$ is ramified over the points
\begin{align*}
	\mathpzc{R} &:= \left \{[z:t:w] \in  \conj{V_{f}}: \left. \frac{\partial{f_{h}}}{\partial t} \right|_{(z,t,w)} = 0 \right \};
\end{align*}

Moreover, the image $\conj{\phi}(\mathpzc{R})$ (also referred to as \text{branch} points) is encoded in the zero locus of the discriminant $\text{Disc}\left(\mathcal{F}_{h} \right) \in \conj{\mathbb{Q}}[z,w]$ of $\mathcal{F}_{h} := \mathrm{taut}_{\mathbb{Q}[w]} f_{h} \in \left(\conj{\mathbb{Q}}[z,w] \right)[t]$:
\begin{align*}
	\Delta &:= \conj{\phi}(\mathpzc{R})\\
	&= \left \{[z:w] \in \mathbb{P}^{1}: \text{Disc}(\mathcal{F}_{h})(z,w) = 0 \right \}.
\end{align*}
Restricting to the standard coordinate patch $U_{w \neq 0}$ (and implicitly using the isomorphism $\mathbb{C}^{2} \rightarrow U_{w \neq 0}$), then $\mathpzc{R}$ and $\Delta$ restrict to $\mathpzc{R}_{\:0}$ and $\Delta_{0}$ respectively.

\begin{example}[Examples]\
	\begin{enumerate}
		\labitem{1}{ex:m_herd_curve} \textbf{$m$-herds}: Recall that all generating series (e.g. soliton generating series, street factors, and the BPS generating series) associated to an $m$-herd are given by powers of a root $P$ of the polynomial $\mathcal{H}_{(m-1)^2} \in \left(\mathbb{Z}[z]\right)[p]$, where
	\begin{align}
		\mathcal{H}_{k} := p - z p^{k} - 1 \in \left(\mathbb{Z}[z] \right)[p].
		\tag{\ref{eq:H_k_def}}
	\end{align}
	The discriminant of this polynomial is given as \cite[\S 2.7, pg. 41]{samuel:alg_num}
	\begin{align*}
		\mathrm{Disc}(\mathcal{H}_{k}) &= (-1)^{k(k-1)/2} z^{k-2} \left(-(k-1)^{k-1}+k^k z\right).
	\end{align*}
	which has two roots, one at $z = 0$, and the other at
	\begin{align*}
		z_{*} &= (k-1)^{k-1} k^{-k}.
	\end{align*}
		The root of the discriminant at $z = 0$ corresponds to ramification point of ramification degree $k-1$ (present on the projective curve, but not the affine curve): the degree of \eqref{eq:H_k_def} drops at the point $z = 0 \in \mathpzc{P}$.  To study ramification around this point, we can either pass to the projective curve or study the birationally equivalent curve defined by
		\begin{align*}
			\widehat{\mathcal{H}_{k}} &= y^{k-1} - y^{k} - z.
		\end{align*}

		\labitem{2}{ex:3_2_curve} $(3,2|3)\textbf{-herd}$: Using computer algebra software (e.g. \textit{Mathematica} or \textit{Magma}), one can check that the polynomial $f_{(3/2|3)} = \mathsf{taut}_{\mathbb{Z}}^{-1} \mathcal{F}_{(3/2|3)} \in \mathbb{Z}[z,t]$ is irreducible; the associated affine curve $V_{f}$ is a degree 39 branched cover of ``$z$-plane" $\mathbb{C}$.

		  The discriminant $\mathrm{Disc}(\mathcal{F}_{(3/2|3)}) \in \mathbb{Z}[z]$ is a degree 306 polynomial that factorizes as a product
	\begin{align}
		\mathrm{Disc}(\mathcal{F}_{(3/2|3)}) = -387420489 \left[d_{1}(z) \right]^{256} \left[d_{2}(z) \right] \left[d_{3}(z) \right]^2;
		\label{eq:Disc_F_3_2}
	\end{align}		  
	 where the $d_{i};\, i = 1,2,3;$ are all distinct irreducible polynomials; specifically $d_{1} = z, \, d_{2}$ is a degree 10 polynomial,\footnote{The expression for $d_{2}$ is explicitly shown below in \eqref{eq:d_2}; the root of least magnitude of this polynomial will play a role in the asymptotics of the associated BPS indices.} and $d_{3}$ is a polynomial of order $20$.  Thus, there are 31 distinct points in $\Delta_{0}$

		  The branch point $z = 0$ is of particular interest as this is the point around which we expect to find a series solution for the BPS generating series.  The degree of $f$ drops at the point $z = 0$ where $35$ roots become infinite; there are two associated ramification points above $z = 0$: one at infinity (present on the projective curve, or after performing the birational transformation $t \mapsto 1/t = y$), and one at the point $(z_{*}, t_{*}) = (0,1)$. 
		  
		  The point $(z_{*},t_{*}) = (0,1)$ is a singular point where four branches\footnote{Two of these branches can be seen in the depiction of the real points of the curve in Fig.~\ref{fig:(3_2|3)_curve}.} of the curve $V_{f}$ meet transversely; however, only one of these four branches extends as an analytic function in the variable $z$---it is this unique analytic extension that represents the BPS generating series.  Monodromy around the point $z = 0$ induces a permutation that is the product of a 3-cycle that permutes the three roots that do not analytically extend, a 35-cycle that permutes the roots that become infinite, and a 1-cycle that preserves the BPS generating series root.
	\end{enumerate}
\end{example}

\begin{figure}[t!]
	\begin{center}
		 \includegraphics[scale=0.5]{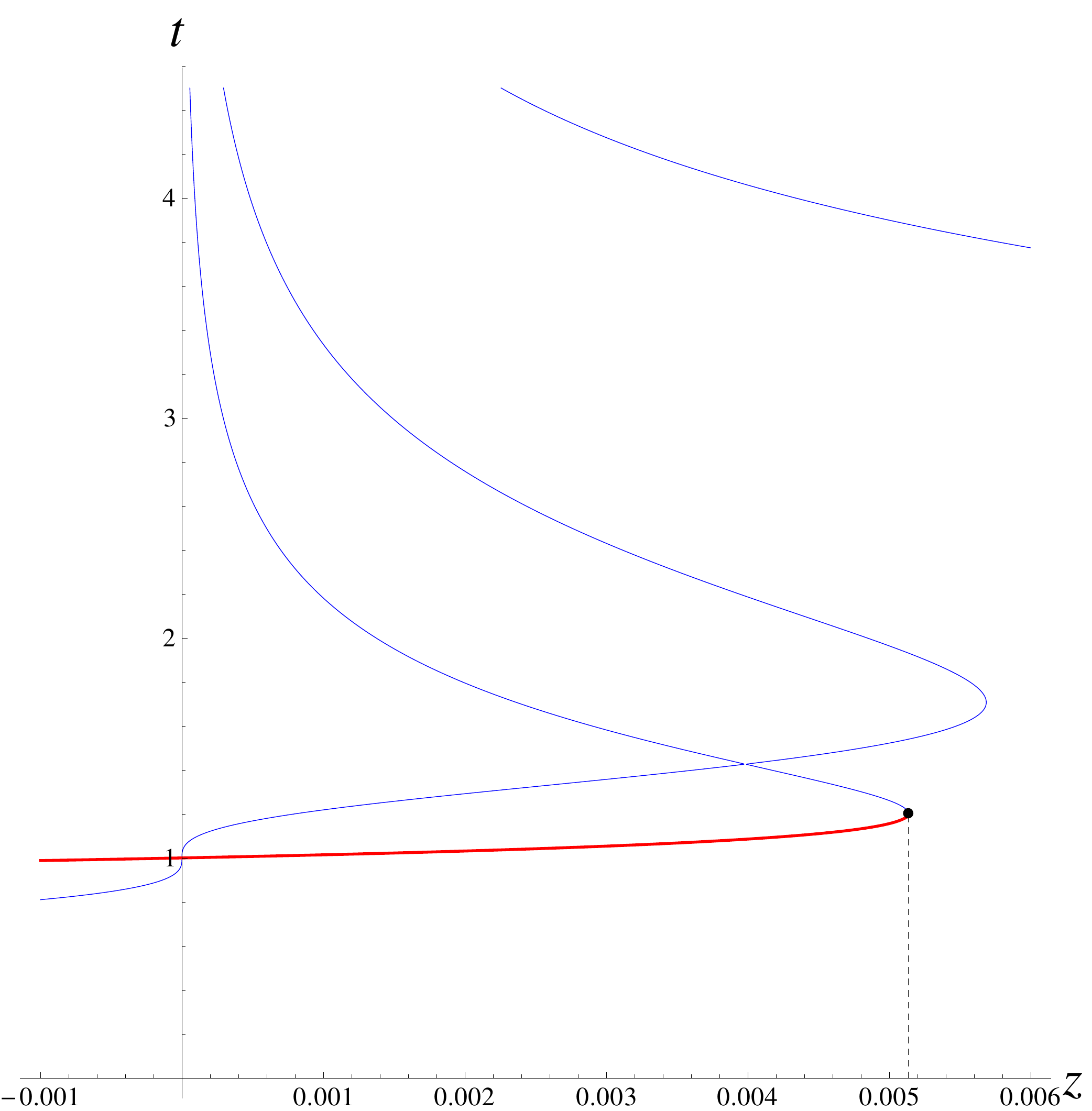}
		\caption{Real points of the curve defined by $\mathcal{F}_{(3,2|3)}$.  The portion shaded in red is the graph of the BPS generating series solution on its maximal domain of analyticity. \label{fig:(3_2|3)_curve}}
	\end{center}
\end{figure}

\subsubsection{Analytic Viewpoint} \label{sec:analytic_view}
The above discussion gives us a solution $\tau$ of \eqref{eq:gen_func_eqn} that is a regular function on a branched cover of the $z$-plane $\mathbb{C}$.  However, one of our goals was to seek out the roots $T_{l} \in \conj{\mathbb{Q}(z)}$ of \eqref{eq:gen_func_eqn} that are functions of the coordinate $z$; geometrically this is equivalent to seeking (local) sections of the branched cover $\phi: V_{f} \rightarrow \mathbb{C}$.  When $e > 1$, such sections cannot be birational maps (rational functions of $z$) as their existence would violate irreducibility of $\mathcal{F}$ over $\mathbb{C}(z)$. However, passing over into the holomorphic world, on the complement of the finite set $\Delta \cup \mathpzc{P}$, local holomorphic sections are guaranteed by the analytic implicit function theorem.  In fact, all of the above could have been said in analytic language.

\begin{definition}[Notation]
	Let $X$ be an analytic space, then $\mathcal{O}_{X}$ will denote the sheaf of analytic functions on $X$ and $\mathcal{M}_{X}$ will denote the sheaf of meromorphic functions on $X$.
\end{definition}

Passing to the world of analytic spaces:
\begin{itemize}
	\item  The projective variety $\conj{V_{f}}$ is re-interpreted as an analytic space; we will denote this corresponding analytic space via $\mathcal{V}_{f} \subset \mathbb{P}^{2}$.
	
	\item $\conj{\tau} \in \mathbb{C}(\conj{V_{f}}) \hookrightarrow \mathcal{M}_{\mathcal{V}_{f}}(\mathcal{V}_{f})$ is a meromorphic function on $\mathcal{V}_{f}$; it restricts to the analytic function $\tau \in \mathbb{C}[V_{f}] \hookrightarrow \mathcal{O}_{\mathcal{V}_{f}}(\mathcal{V}_{f} \cap U_{w \neq 0})$ where $U_{w \neq 0} := \{[z:t:w] \in \mathbb{P}^{2}: w \neq 0\}$.

	\item The (algebraic) branched cover $\conj{\phi}$ is re-interpreted as a branched cover of analytic spaces $\varphi: \mathcal{V}_{f} \rightarrow \mathbb{P}^{1}$.
\end{itemize}

  If the domain of $\varphi$ is restricted to $\mathcal{V}_{f}' := \mathcal{V}_{f} \backslash \varphi^{-1}(\Delta)$, we have a holomorphic degree $e$ cover of Riemann surfaces:
\begin{align*}
	\varphi_{\mathrm{cov}} := \varphi|_{\mathcal{V}_{f}'}: \mathcal{V}_{f}' \rightarrow \mathbb{P}^{1} \backslash \Delta.
\end{align*}
Then for any simply-connected open set $U \subset \mathbb{P}^{1} \backslash \Delta$, the pre-image $\varphi_{\mathrm{cov}}^{-1}(U)$ is a disjoint union of $e$ open sets; in fact, for any component $V$ of $\varphi_{\mathrm{cov}}^{-1}(U)$ we can use the analytic implicit function theorem, along with properties of analytic continuation, to show that there is a unique holomorphic section $s_{V}:U \rightarrow V \subset \mathcal{V}_{f}'$; by precomposing this section with the map $\conj{\tau}$ we get a \textit{meromorphic} function (a holomorphic map to $\mathbb{P}^{1}$) on $U$: 
\begin{align*}
	T_{V} := \conj{\tau} \circ s_{V}: U \rightarrow \mathbb{P}^{1}.
\end{align*}
Moreover, this meromorphic function\footnote{It is not hard to see that the poles of any such $T_{V}$ are precisely at points in $\mathpzc{P}$ with ramification index $1$, i.e. points in $\mathpzc{P} \cap \left(\mathbb{P}^{1} \backslash \Delta \right)$.}  is a root of $\mathcal{F}$ in the field $\mathcal{M}_{\mathbb{P}^{1} \backslash \Delta}(U)$; in fact, $\mathcal{F}$ splits over $\mathcal{M}_{\mathbb{P}^{1} \backslash \Delta}(U)$ as there are $e$ distinct roots corresponding to the $e$ components of $\varphi_{\mathrm{cov}}^{-1}(U)$.  

\subsection{Asymptotics and Algebraicity} \label{sec:asymptotics}
Let $\ggen$ be a generating series for the sequence of rational numbers $(\beta_{n})_{n = 1}^{\infty}$, i.e.:
\begin{align}
	\ggen &:= \prod_{n = 1}^{\infty} (1 - \left(\mathtt{s} z \right)^{n} )^{n \beta_{n}} \in \formal{\mathbb{Q}}{z},
	\label{eq:general_gen}
\end{align} 
where $\mathtt{s} \in \{+1, -1 \}$ is some fixed sign depending on the context under consideration; of course, we have two cases in mind.
\begin{enumerate}
	\item $\ggen = \gdt$ is the generating series for BPS indices/DT invariants: $\beta_{n} = \Omega(n \gamma_{c})$ (where $\gamma_{c}$ is some primitive charge) $\gamma_c$ and $\mathtt{s} z = X_{\twid{\gamma_{c}}} =: \twid{z}$.  The particular value of $\mathtt{s} \in \{ \pm 1\}$ depends on the definition of $z$ suited to the problem at hand; however, for $(a,b|m)$-herds/the $m$-Kronecker quiver we choose $\mathtt{s} = (-1)^{mab - a^2 - b^2}$ according to the definition of $z$ given in \eqref{eq:zed_def} and Appendix \ref{app:signs}.
	
	\item $\ggen = \geul$ is the generating series for the Euler characteristics of stable moduli for the Kronecker $m$-quiver: $\beta_{n} = -\chi(\CM_{\st}^{m}(an,bn))$ where $(a,b) \in \mathbb{Z}_{>0}^{2}$ is a pair of coprime integers, $\CM_{\st}^{m}(an,bn)$ is the moduli space of stable representations of the $m$-Kronecker quiver with dimension vector $(an,bn)$ (with respect to the non-trivial stability condition), and $\chi$ denotes the Euler characteristic.  In this case the choice $\mathtt{s} = 1$ is forced upon us via Reineke's functional equation \eqref{eq:reineke_func} if we wish to use the same variable $z$ used in the DT invariant generating series $\gdt$ (for which we choose the signs $\mathtt{s} = (-1)^{mab - a^2 - b^2}$).
\end{enumerate}

If $\ggen$ is an algebraic function, then the discussion in \S \ref{sec:alg_curves} presents a geometric picture that aids in the classification of the $n \rightarrow \infty$ behaviour of the $\beta_{n}$ via the study of local sections of the branched cover $\varphi$.  In particular holomorphic techniques allow us to extract asymptotics of the $\beta_{n}$ through a study of the possible zeros, poles, and branch point singularities of the local section corresponding to the $\ggen$; the fact that $\ggen$ is, moreover, algebraic places further restrictions on the asymptotic behaviour.  

In the following section we will derive the most general possible form for the asymptotics of the $\beta_{n}$ relying only on the assumption that $\ggen$ is an algebraic function over $\mathbb{Q}$.  As a warmup, we will first determine the asymptotics of the coefficients $\{t_{n}\}_{n = 0}^{\infty}$ in the series expansion
\begin{align*}
	T = \sum_{n = 0}^{\infty} t_{n} (z - z_{0})^{n}
\end{align*}
of an arbitrary algebraic function $T$, holomorphic around the point $z_{0}$ (for simplicity we will eventually take $z_{0} = 0$);\footnote{When this algebraic function is the BPS generating series, these coefficients can be interpreted as counts of halo-bound state BPS particles \cite[\S 3.4]{gmn:framed}.} the asymptotics of the $\beta_{n}$ follow a closely related story. The techniques used in this section are drawn heavily from the rather powerful book \cite{fs:an_comb} by Flajolet and Sedgewick.  Before beginning, we introduce some terminology.

\begin{definition}[Terminology]
In the subsequent discussion: a \textit{singularity} is one of the following:
	\begin{enumerate}	
	 \labitem{(P)}{list:poles} A pole of the function (where the function can be interpreted as a meromorphic function), e.g. the point $0$ in $h(\zeta) = \zeta^{-k}$ for some $k \in \mathbb{Z}$;
	 
	 \labitem{(B${}_{>0}$)}{list:fin_branch} A ``finite" branch point: the singularity is due to a failure of coordinates and the function can be continued as a holomorphic function on some finite-degree branched cover of $U$, e.g. the point $\zeta = 0$ in $\zeta^{\alpha}$ for some $\alpha \in \mathbb{Q}_{>0} \backslash \mathbb{Z}_{>0}$;
	 
	 \labitem{(B${}_{<0}$)}{list:inf_branch} An ``infinite" branch point: the singularity is due to a failure of coordinates and the function can be continued as a meromorphic function on some finite-degree cover of $U$, e.g. the point $\zeta = 0$ in $h(\zeta) = \zeta^{-\alpha}$ for some $\alpha \in \mathbb{Q}_{>0} \backslash \mathbb{Z}_{>0}$;
	 
	 \labitem{(B${}_{\mathrm{log}}$)}{list:log_branch} A logarithmic branch point: the singularity is due to a failure of coordinates and the function can be continued as a meromorphic function on some infinite-degree branched cover of $U$, e.g. the point $\zeta = 0$ in $h(\zeta) = \log(\zeta)$.
	 \end{enumerate}
	 Singularities of type \ref{list:poles}, \ref{list:fin_branch}, or \ref{list:inf_branch} are the only possibilities for an algebraic function; a singularity of type \ref{list:log_branch} can occur as the logarithm of a meromorphic function (specifically $\log(h)$ for some meromorphic function $h$ will have a logarithmic singularity at any zero or pole of $h$).    Furthermore, taking the logarithm of an algebraic function converts all zeros, singularities of type \ref{list:poles}, and singularities of type \ref{list:inf_branch} into logarithmic branch points.
\end{definition}

\subsubsection{Asymptotics of Algebraic Series Coefficients}



First, we begin by noting that local sections of the cover $\varphi_{\mathrm{cov}}: \mathcal{V}_{f}' \rightarrow \mathbb{P}^{1} \backslash \Delta$, defined only on open sets $U \subset \mathbb{P}^{1} \backslash \Delta$, are guaranteed to be analytically continued to regions including points in $\Delta$ as long as the image of the section does not hit a ramification point.  Thus, if we wish to study maximal analytic extensions of sections---a necessary study in order to understand type of singularities local sections may encounter--- we should include non-ramification points of $\mathcal{V}_{f}$ in our study as well.  Indeed, on the complement of ramification points, $\varphi$ defines a holomorphic map of Riemann surfaces (non-singular analytic curves)
\begin{align*}
	\varphi_{\mathrm{sub}} := \varphi|_{\mathcal{V}_{f} \backslash \mathpzc{R}} : \mathcal{V}_{f} \backslash  \mathpzc{R} \rightarrow \mathbb{P}^{1}.
\end{align*}
Even though this is generally not a covering space map, it remains a submersion; hence, we can still apply the analytic inverse function theorem to speak of sections.

\begin{lemma}
	 Denote the disk of radius $r$, centred at $z_{0}$, via
\begin{align*}
	D_{z_{0}}(r) := \{z \in \mathbb{C}: |z - z_{0}| < r \};
\end{align*}
 Let $s_{0} \in \varphi_{\mathrm{sub}}^{-1}(z_{0})$ for some $z_{0} \in \mathbb{C} \cap \varphi_{\mathrm{sub}}(\mathcal{V}_{f})$; if the component of $\varphi^{-1}(D_{z_{0}}(r))$ containing $s_{0}$ does not contain any point in $\mathpzc{R}$, then there is a unique holomorphic section $s: D_{z_{0}}(r) \rightarrow \mathcal{V}_{f}$ of $\varphi_{\mathrm{sub}}$ such that $s(z_{0}) = s_{0}$.
\end{lemma}
\begin{proof}
	For sufficiently small $\epsilon > 0$, the analytic implicit function theorem guarantees the existence of a unique holomorphic section $s: D_{z_{0}}(\epsilon) \rightarrow \mathcal{V}_{f}$ such that $s(z_{0}) = s_{0}$; we can analytically continue $s$ to any disk $D_{z_{0}}(r)$ such that the component of $\varphi^{-1}(D_{z_{0}}(r))$ containing $s_{0}$ does not contain any ramification points.
\end{proof}

Our true interests lie in finding roots of the polynomial $\mathcal{F}$, thought of as functions of the coordinate $z$; as described via the discussion in \S \ref{sec:analytic_view}, these are given by pushing holomorphic sections of $\mathcal{V}_{f} \rightarrow \mathbb{P}^{1}$ forward via $\conj{\tau}$.  Doing so, we obtain a meromorphic function, i.e. a holomorphic map $T := \conj{\tau} \circ s: D_{z_{0}}(r) \rightarrow \mathbb{P}^{1}$; the poles of this function correspond precisely to the location of the poles of $\conj{\tau}$
\begin{align*}
	\widehat{\mathpzc{P}} := \mathpzc{\conj{\tau}}^{-1}([1:0]) \subset \mathcal{V}_{f}
\end{align*}
(note that $\varphi(\widehat{\mathpzc{P}}) = \mathpzc{P}$).  Hence, we have the following corollary.

\begin{corollary} \label{cor:an_extension}
 Let $s$ be defined as in the previous lemma; if the component of $\varphi^{-1}(D_{z_{0}}(r))$ containing $s_{0}$ does not contain any point in $\mathpzc{R} \cup \widehat{\mathpzc{P}}$, then $T := \conj{\tau} \circ s: D_{z_{0}}(r) \rightarrow \mathbb{C}$ is a holomorphic function.
\end{corollary}

With this in mind, suppose $T$ is a root of $\mathcal{F}$ that is holomorphic on a disk centred at $z_{0} \in \mathbb{C}$.  Then, to speak of the \textit{maximal} analytic continuation to such a disk, we define the following.

\begin{definition}
	Let $h$ be a holomorphic function on a disk centered at $z_{0} \in \mathbb{C}$, then 
	\begin{align*}
		\mathtt{R}_{h} := \sup \left \{r \in  \mathbb{R}_{>0} \cup \{\infty\}: \text{$h$ can be analytically continued to $D_{z_{0}}(r)$} \right \}.
	\end{align*}
\end{definition}
On $D_{z_{0}}(\mathtt{R}_{T})$, we may write $T$ as a convergent series
\begin{align}
	T(z) &= \tau(s_{0}) + \sum_{n = 1}^{\infty} t_{n} (z - z_{0})^{n},\quad z \in D_{z_{0}}(\mathtt{R}_{T}).
	\label{eq:max_section}
\end{align}
Assume that $\mathtt{R}_{T} < \infty$; then, heuristically speaking, we can predict the asymptotic behaviour of the coefficients $t_{n}$ in the large $n$ limit by noticing that the failure of the series representation of $T$ to converge on a disk of radius larger than $\mathtt{R}_{T}$ is due to the growth of $\left |t_{n} \right |$ as fast as $\left |\mathtt{R}_{T} \right |^{-n}$;  in fact, we can even derive the precise asymptotics of $t_{n}$ through a closer analysis of its behaviour near the singular points of $T$ (a collection of points in in $\mathpzc{R} \cup \widehat{\mathpzc{P}}$) that obstruct analytic continuation to a larger disk.

\begin{remark}
	If $\mathtt{R}_{T} < \infty$, then via Cor. \ref{cor:an_extension} the failure to holomorphically extend $T$ to a larger disk is due to a (non-empty) subset of points in $\varphi^{-1}\left[\partial \conj{D_{z_{0}}(\mathtt{R}_{T})} \right] \cap \left(\mathpzc{R} \cup \widehat{\mathpzc{P}} \right)$.
\end{remark}

\begin{definition}
	Let $h$ be a holomorphic function on some disk centred at $z_{0}$, then the set of \textit{dominant singularities} of $h$ (relative to $z_{0}$) is
	\begin{align*}
		\mathsf{sing}_{z_{0}}(h) &:= \left \{
		z \in \partial \conj{D_{z_{0}}(\mathtt{R}_{h})}: 
		\begin{array}{c}	
			\text{$h$ does not extend as a holomorphic function}\\
			 \text{to any open set containing $z$} 
		\end{array}		 
		\right \}.
	\end{align*}
\end{definition}

For the algebraic function $T$, we can lift each dominant singularity uniquely to a point on the curve $\mathcal{V}_{f}$.  Denote this set of lifts via $\mathsf{Sing}_{z_{0}}(T)$.



\begin{remark}\
	\begin{enumerate}
		\item By the note above,
			\begin{align*}
				\mathsf{sing}_{z_{0}}(T) \subseteq \Delta \cup \mathpzc{P}
			\end{align*}
			Moreover, under the map $\iota: \mathbb{C}^{2} \cup \{\infty \} \hookrightarrow \mathbb{P}^2$ which sends $(z,t) \mapsto [z:t:1]$ for $t \neq \infty$ and $(z,\infty) \mapsto [z:1:0]$, we have
			\begin{align*}		
				 \iota \left(\mathsf{Sing}_{z_{0}}(T) \right) \subseteq \mathpzc{R} \cup \widehat{\mathpzc{P}} \subset \mathcal{V}_{f}\\
			\end{align*}
			In particular, $\mathsf{sing}_{z_{0}}(T)$ (and $\mathsf{Sing}_{z_{0}}(T) $) are finite sets (of the same order) as the set $\mathpzc{R} \cup \widehat{\mathpzc{P}}$ is finite.
		
	
		\item It is helpful to keep in mind that $T(z)$ becomes infinite as $z$ approaches dominant singularities in $\mathpzc{P}$, while it remains finite as $z$ approaches dominant singularities in $\mathpzc{R} \backslash \left(\mathpzc{R} \cap \mathpzc{P} \right)$.
	\end{enumerate}
\end{remark}

Unfortunately, when $\mathcal{V}_{f}$ has singular points (in particular singular points projecting, via $\varphi$ to the affine patch containing 0), (smooth) complex geometric techniques are insufficient for finding $\mathtt{R}_{T}$ and the location of dominant singularities.  Indeed, let $\mathcal{W} \subset \mathcal{V}_{f}$ be the component of $\varphi^{-1}(D_{z_{0}}(r))$ containing the point $s_{0}$; assume that we have chosen $r$ large enough such that the there is a point $p \in \mathpzc{R} \cap \partial \conj{\mathcal{W}}$.  If $p$ is a smooth point of $\mathcal{V}_{f}$, and $\nu_{p} \geq 1$ is its ramification-index, then the monodromy associated to small loops containing the point $p$ induces a cyclic permutation of order $\nu_{p}$ on the components of
	\begin{align*}
		\varphi^{-1}(D_{\varphi(p)}(\epsilon)) = \bigsqcup_{i=1}^{\nu_{p}} L_{i},
	\end{align*}
where $\epsilon > 0$ is sufficiently small---i.e. a cyclic permutation on the set of all $\varphi$-preimages of small disks surrounding $\varphi(p)$.  The result is that $p$ is a branch-point singularity and we cannot extend $s$ holomorphically to a larger disk.  On the other hand, if $p$ is a \textit{singular} point,  then it may happen that one of the pre-images $L = L_{j}$, for some $j \in \{1, \cdots, \nu_{p} \}$, is fixed by the monodromy associated to a small loop around $p$; if $L \cap \mathcal{W} \neq \emptyset$, then we can analytically continue $s$ to $\mathcal{W} \cup L$---i.e. analytically continue to a disk larger than the na\"{i}ve one whose radius is constrained by $p$.  Furthermore, as long as $p$ is not an element of $\mathpzc{P}$, then $T = \conj{\tau} \circ s$ can also be analytically continued.  In order to detect if we can perform such continuations beyond the na\"{i}ve disks constrained by singular points of $V_{f}$, it is best to take a more algebraic viewpoint and use the Puiseux expansion of the root $T$.  The key idea is that the behaviour of $T$ near a singularity is encoded in its Puiseux expansion around the singularity.  To recall this expansion, we paraphrase the statement of Thm. VII.7 of \cite[\S VII.7]{fs:an_comb}.

\begin{theorem}[(Newton-)Puiseux Expansion, C.f. Theorem VII.7 of \cite{fs:an_comb}] \label{thm:Puiseux}
	Let $T \in \conj{\mathbb{Q}(z)}$ be an algebraic function; let $z_{0} \in \mathbb{C}$, then there exists an expansion of the form
	\begin{align}
		T(z) &= \sum_{l \geq l_{0}} c_{l} (z - z_{0})^{l/\kappa_{z_{0}}}
		\label{eq:T_puiseux}
	\end{align}
	where $l_{0} \in \mathbb{Z},$ $\kappa_{z_{0}} \geq 1$ is an integer,\footnote{Note that $\kappa_{z_{0}}$ depends on the choice of root $T$, not just the point $z_{0} \in \mathbb{C}$; so the truly pedantic should adopt a notation that suggests this---however, there should be no confusion due to our notation for the following discussion.} and $\{c_{l} \}_{l = l_{0}}^{\infty} \subset \mathbb{C}$.  Moreover, for sufficiently small $r > 0$, then for any $\theta, \vartheta \in \mathbb{R}/(2\pi \mathbb{R})$, this expansion gives a well-defined analytic function on a neighbourhood of the form (see Fig.~\ref{fig:wedge_region})
	\begin{equation}
		\begin{aligned}
			W_{z_{0}}(r, \theta, \vartheta) &:= \{z \in \mathbb{C} \backslash \{z_{0}\} : \text{$|z - z_{0}| < r$ and }\mathrm{Arg}(z - z_{0}) \notin [ \vartheta + \theta,  \vartheta - \theta ] \}\\
			&= e^{i \vartheta} \{z \in \mathbb{C} \backslash \{0\}: \text{$|z| < r$ and }\mathrm{Arg}(z) \notin [ \theta,  - \theta ] \} + z_{0}
		\end{aligned}
			\label{eq:W_rtheta_def}
	\end{equation}	
	(i.e. an indented disk given by a disk neighbourhood of the point $z_{0}$, minus a closed wedge spanning an angle $2 \theta$ and bisector the line at angle $\vartheta$).
\end{theorem}

\begin{figure}[t!]
	\begin{center}
		 \includegraphics[scale=1.2]{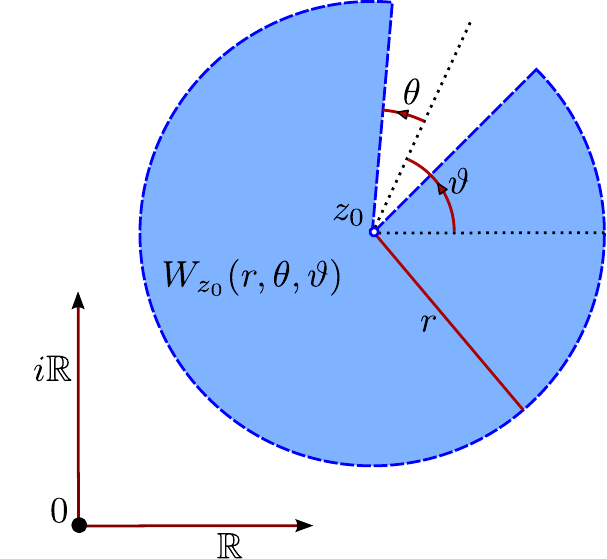}
		\caption{Wedge Region $W_{z_{0}}(r,\theta,\vartheta)$ defined in \eqref{eq:W_rtheta_def} \label{fig:wedge_region}}
	\end{center}
\end{figure}

The Puiseux expansion for $T$ is a Laurent series in the fractional power $(z - z_{0})^{1/\kappa_{z_{0}}}$; when it is not possible to choose $\kappa_{z_{0}} = 1$, then such an expansion cannot define a meromorphic function on any disk containing $z_{0}$ (but does on a finite branched cover of such a disk) due to the existence of monodromy around small loops that encircle $z_{0}$.  However, it is holomorphic after removing a small wedge from the disk containing $z_{0}$ and choosing one of the $\kappa_{z}$ branches of $(z - z_{0})^{1/\kappa_{z_{0}}}$.  Each choice of branch corresponds to a distinct choice of root around $z_{0}$ and the monodromy around $z_{0}$ induces a cyclic permutation of these $\kappa_{z}$ roots.  When it is possible to choose $\kappa_{z_{0}} = 1$, then there is no monodromy around $z_{0}$ (i.e. no branch point singularity) and the expansion \eqref{eq:T_puiseux} defines a meromorphic function on a sufficiently small disk $D_{z_{0}}(r) \supset W_{z_{0}}(r, \theta, \vartheta)$; furthermore, if $l_{0} \geq 0$, then the series defines a holomorphic function on $D_{z_{0}}(r)$.

In order to eliminate the various choice ambiguities that arise when fractional powers arise in the Puiseux expansion, we fix a convention.

\begin{definition}[Convention]
	When considering expressions of the form
	\begin{align*}
		\left(1 - \frac{z}{\rho} \right)^{\gamma}
	\end{align*}
	for $\gamma \in \mathbb{Q} \backslash \mathbb{Z}$, we will always choose the unique branch that is defined for $z$ in the domain
	\begin{align*}
		\mathbb{C} \backslash \left \{z: 1 - \frac{z}{\rho} \in \mathbb{R}_{<0} \right \},
	\end{align*}	
	and is \textit{positive} when $\rho^{-1} z \in [0,1)$.
\end{definition}

With this choice of convention, expanding $T$ around a point $\rho \in \mathsf{sing}_{z_{0}}(T)$, we have
\begin{align}
	T(z) &= \sum_{l \geq l_{0}} c_{l} \left(1 - \frac{z}{\rho} \right)^{l/\kappa_{\rho}}
	\label{eq:T_sing_puiseux}
\end{align}
where the collection of constants $\{c_{l}\}_{l = l_{0}}^{\infty} \subset \mathbb{C}$ are unambiguously determined when fixing the convention above, and $T$ converges on any ``wedge-shaped" region $W_{\rho}(r, \theta, \vartheta)$ (c.f. \eqref{eq:W_rtheta_def}) for sufficiently small $r$.

\begin{remark}
	To calculate $l_{0},\, \kappa_{\rho}$, and the coefficients $c_{l}$, one can use a suitable Newton-polygon associated to either the polynomial $f \in \mathbb{Z}[z,t]$ (if $\rho \in \mathpzc{P}$), or the polynomial $\widehat{f} \in \mathbb{Z}[z,y]$ (if $\rho \in \mathpzc{P}$) (see e.g. \cite[Ch. 7.2]{kirwan:alg_curves}).  Indeed, assume $\rho \notin \mathpzc{P}$ and let $p = (\rho, \alpha) = (\rho, T(\rho)) \in \mathsf{Sing}_{z_{0}}(T)$, then we can expand $f$ in a series about the point $p = (\rho, \alpha) \in \mathsf{Sing}_{z_{0}}(T) \hookrightarrow \mathbb{C}^{2}$
	\begin{align*}
		f(z,t) = \sum_{n,m = 0}^{\infty} f_{n,m}(p) (z - \rho)^{m} (t - \alpha)^{n} 
	\end{align*}
	where
	\begin{align}
		f_{n,m}(p) &:= \frac{1}{n!m!} \left. \left(\frac{\partial^{n+m} f}{\partial z^{n} \partial t^{m}} \right) \right|_{(z,t) = (\rho, \alpha)} \in \conj{\mathbb{Q}};
		\label{eq:f_coeffs}
	\end{align}
	then the possibilities for the values $\kappa_{\rho}$ are given by the inverse slopes of the leftmost convex envelope of the set $\{(n,m) \in \mathbb{Z}: f_{n,m}(p) \neq 0\}$; for each choice of $\kappa_{\rho}$, upon substituting in the expression 
	\begin{align}
		T = \sum_{l \geq l_{0}} c_{l} \left(1 - \frac{z}{\rho} \right)^{l/\kappa_{\rho}}
		\label{eq:sing_puiseux}
	\end{align}	
 into $f$, the condition that all coefficients of a particular order must vanish determines $l_{0}$ and imposes polynomial conditions on the $c_{l}$ that can be solved order-by-order in $k$.  (Note that, because these polynomials have coefficients in $\conj{\mathbb{Q}}$, then $c_{l} \in \conj{\mathbb{Q}}$.)  When $\rho \in \mathpzc{P}$ one can repeat the same procedure to determine the Newton-Puiseux expansion of $T^{-1}$ (and, hence, $T$) by expanding $\widehat{f} \in \conj{\mathbb{Q}}[z,y]$ around the point $p = (\rho,0) \in \mathsf{Sing}_{z_{0}}(T)$.
\end{remark}

It is important to observe that the expansion \eqref{eq:sing_puiseux} must contain at least one non-vanishing $c_{l}$ associated to an $l$ such that $l/\kappa_{\rho}$ is not a non-negative integer---otherwise we can analytically continue $T$ to an open set containing $\rho \in \mathrm{sing}_{z_{0}}(T)$, which contradicts the definition of $\mathsf{sing}_{z_{0}}(T)$. 

\begin{definition}
	Let $T$ be an algebraic function and \eqref{eq:T_sing_puiseux} an expansion of $T$ around a point $\rho \in \mathsf{sing}_{z_{0}}(T)$; then define
\begin{equation}
	\begin{aligned}
		\ell_{\rho}(T) &:=  \min \left\{l: \text{$c_{l} \neq 0$ and $\frac{l}{\kappa_{\rho}} \notin \mathbb{Z}_{\geq 0}$} \right \},\\
		C_{\rho}(T) &:= c_{\ell(\rho)},\\
		\sigma_{\rho}(T) &:= \frac{\ell_{\rho}}{\kappa_{\rho}};
	\end{aligned}
	\label{eq:sigma_rho_def}
\end{equation}
	In practice we will be working with respect to a fixed function $T$; so we will simplify our notation by suppressing appearances of $T$, e.g. just writing $C_{\rho}$ and $\sigma_{\rho}$.
\end{definition}

Using the definition above, as $z \rightarrow \rho$
\begin{align}
	T(z) = E_{\rho}(z) + C_{\rho} \left(1- \frac{z}{\rho} \right)^{\sigma_{\rho}}  + c_{0} + \mathcal{O} \left[\left(1- \frac{z}{\rho} \right)^{\sigma_{\rho} + 1/\kappa_{\rho}} \right],
	\label{eq:T_around_rho}
\end{align}
where
\begin{align*}
	E_{\rho}(z) &= \sum_{\{k:\text{$k \neq 0$ and $k < \sigma_{\rho}$} \}} c_{k} (z - \rho)^{k}
\end{align*}
is a (possibly identically zero) polynomial function of $z$ such that $E_{\rho}(z) \rightarrow 0$ as $z \rightarrow \rho$.

\begin{numrmk} \label{rmk:smooth_dom}
In the special case that the point $p \in \mathsf{Sing}_{z_{0}}(T)$ is a \textit{smooth} point of the curve $\mathcal{V}_{f}$ (i.e. $f_{1,0} := f_{1,0}(p) \neq 0$), then $\kappa_{\rho}$ is equal to the ramification index $\nu:= \nu_{p}$ (defined by equation \eqref{eq:ram_index}) at the point $p$ and
\begin{align*}
	\ell_{\rho} &= 
	\left\{
	\begin{array}{ll}	
		+1, & \text{if $\rho \notin \mathpzc{P}$}\\
		-1, & \text{if $\rho \in \mathpzc{P}$}
	\end{array}	\right. ; \\
	C_{\rho} &= 
	\left\{
	\begin{array}{lr}	
		\omega_{\nu} \left( \rho \frac{f_{1,0}}{f_{0,\nu}} \right)^{1/\nu} , & \text{if $\rho \notin \mathpzc{P}$}\\
		\omega_{\nu} \left( \rho \frac{\widehat{f}_{1,0}}{\widehat{f}_{0,\nu}} \right)^{-1/\nu} , & \text{if $\rho \in \mathpzc{P}$}
	\end{array}
	\right. ;\\
	\sigma_{\rho} &=
	\left\{
	\begin{array}{ll}	
		+\frac{1}{\nu_{\rho}} , & \text{if $\rho \notin \mathpzc{P}$}\\
		-\frac{1}{\nu_{\rho}} , & \text{if $\rho \in \mathpzc{P}$}
	\end{array}	\right. .
\end{align*}
where $\omega_{\nu}$ is a suitable choice of $\nu^{\mathrm{th}}$ root of unity.  However, if $p$ is a singular point of $\mathcal{V}_{f}$, then (the smallest choice for) $\kappa_{\rho}$ in the expansion for $T$ is $\leq \nu_{p}$; its precise value depends on the choice of section $T$, not just the point $p$ (different sections passing through the same point $p$ may have different smallest values for $\kappa_{\rho}$).
\end{numrmk}

Now, one can show that these local expansions completely determine the asymptotics of the coefficients $t_{n}$ in \eqref{eq:max_section}.  To see why this is possible, we make the following remark.

\begin{remark}
As we are ultimately interested in expansions about $z_{0} = 0$, it is in our best interests to simplify our notation by imposing $z_{0} = 0$ in the rest of discussion.  The formulae for non-zero $z_{0}$ can be derived by a simple translation of coordinates.
\end{remark}

\begin{numrmk} \label{rmk:indented_disk}
We have $W_{\rho}(r,\theta,\vartheta) \cap D_{0}(\mathtt{R}_{T}) \neq \emptyset$ so it makes sense to compare the expansions $\eqref{eq:T_puiseux}$ around any singular point $\rho$ with the expansion $\eqref{eq:max_section}$ around $0$; furthermore, $T$ can be analytically continued to the region 
\begin{align*}
	D_{0}(\mathtt{R}_{T}) \cup \bigcup_{\rho \in \mathsf{sing}_{0}(T)} W_{\rho}(r, \theta, \arg(\rho)),
\end{align*}
for some choice of sufficiently small $r>0$.

\begin{align*}
	\Delta_{0}(r_{\mathrm{in}}, r_{\mathrm{out}} \vartheta) &= \{z \in D_{0}(r): \arg(z) \notin [-\vartheta,\vartheta] \}
\end{align*}

\end{numrmk}

In the case that there is only one dominant singularity $\rho \in \mathsf{sing}_{0}(T)$, one can pass from the Puiseux expansion around $\rho$ to the expansion around $0$ by using the series
\begin{align*}
	\left(1 - \frac{z}{\rho} \right)^{\gamma} &= \sum_{n = 0}^{\infty} \binom{\gamma}{n} (-1)^{n} \rho^{-n} z^{n},
\end{align*}
which converges on $D_{0}(\rho)$.  using Stirling's asymptotics, to leading order in $n$ we have
\begin{align*}
	[z^{n}] \left(1 - \frac{z}{r} \right)^{\gamma}  = \left(\frac{n^{-1 - 1/\gamma}}{\Gamma(-\gamma)} \right)  \rho^{-n}  + \mathcal{O}(n^{-2 - 1/\gamma} \rho^{-n}).
\end{align*}
In the case of a single dominant singularity, it happens that the $n \rightarrow \infty$ asymptotics of the coefficients $t_{n}$ can be extracted by applying the above expansion to the expression \eqref{eq:T_around_rho}, ignoring the big-$\mathcal{O}$-terms of \eqref{eq:T_around_rho}; this yields the $n \rightarrow \infty$ asymptotics
\begin{align*}
	t_{n} \sim  \left(\frac{C_{\rho}}{\Gamma(-\sigma_{\rho})} \right)  n^{-1 - \sigma_{\rho}} \rho^{-n}.
\end{align*}
In the case that there are multiple dominant singularities, the correct asymptotics are given by summing up the individual contributions from each dominant singularity.  The following theorem expresses this more general situation, including the subleading asymptotics.  The full proof, which is mainly an application of Cauchy's integral formula, can be found in Theorem VI.5 of \cite{fs:an_comb}.

\begin{theorem} \label{thm:transfer}
	If $h$ is an analytic function on $D_{0}(\mathtt{r})$ such that
	\begin{itemize}
		\item $h$ has a finite number of dominant singularities $\{\rho_{1}, \cdots, \rho_{r} \} \subset \partial \conj{D_{0}(\mathtt{r})}$;
			
		\item $h(z)$ is analytic on a region of the form specified in Rmk. \ref{rmk:indented_disk};
			
		\item there exist functions
		\begin{align*}
			\Phi_{1},\, \cdots,\, \Phi_{r},\, \mathpzc{E} \in \left \{ (1-z)^{a} \log(1-z)^{b}: a \in \mathbb{C},\, b \in \mathbb{Z} \right \}
		\end{align*} 
	such that, for each $i = 1, \cdots, r$, 
		\begin{align*}
			h(z) = \Phi_{i} \left(\frac{z}{\rho_{i}} \right) + \mathcal{O} \left[ \mathpzc{E} \left(\frac{z}{\rho_{i}} \right) \right]
		\end{align*}
		as $z \rightarrow \rho_{i}$;
	\end{itemize}	
	 then 
	\begin{align*}
		[z^{n}] h(z) &= \sum_{i=1}^{r} \rho_{i}^{-n} \left( [z^n] \Phi_{i} \right) + \mathcal{O} \left[ \mathtt{r}^{-1} \left([z^{n}] \mathpzc{E} \right) \right].
	\end{align*}
\end{theorem}

The following is a corollary of the above theorem in the context of the algebraic functions under study.\footnote{For a more refined statement of the structure of the subleading asymptotics for algebraic functions, see \cite[Thm. VII.8.]{fs:an_comb}.}

\begin{theorem}	
	Let
	\begin{equation}
		\begin{aligned}
			\sigma_{*} &:= \min \{\sigma_{\rho}: \rho \in \mathrm{sing}_{0}(T) \},\\
			\sigma_{\mathrm{sub}} &:= \min \left \{\sigma_{\rho} + \frac{1}{\kappa_{\rho}}: \text{$\rho \in \mathrm{sing}_{0}(T)$ and $\sigma_{\rho} + \frac{1}{\kappa_{\rho}} \notin \mathbb{Z}_{>0}$} \right \},\\
			\mathsf{leadsing}_{0}(T) &:= \{\rho \in \mathsf{sing}_{0}(T): \sigma_{\rho} = \sigma_{*} \}.
		\end{aligned}
		\label{eq:leading_exps}
	\end{equation}
	Then	
	\begin{align}
		t_{n} = \left( \frac{n^{-1 - \sigma_{*}}}{\Gamma(-\sigma_{*})}  \right)   \sum_{\rho \in \mathsf{leadsing}_{0}(T)} C_{\rho} \rho^{-n} + \mathcal{O} \left( n^{-1-\sigma_{\mathrm{sub}}} \mathtt{R}_{T}^{-n} \right).
		\label{eq:t_n_asymp}
	\end{align}
\end{theorem}
\begin{proof}
	From the discussion above, for any $\rho \in \mathsf{sing}_{0}(T)$, as $z \rightarrow \rho$ we have
	\begin{align*}
		T(z) = C_{\rho} \left(1 - \frac{z}{\rho} \right)^{\sigma_{\rho}} + \mathrm{Poly}_{\rho}(z) + \mathcal{O} \left[ \left(1-\frac{z}{\rho} \right)^{\sigma_{\mathrm{sub}}} \right],
	\end{align*}
	where $\mathrm{Poly}_{\rho}(z)$ is some polynomial in $z$ of degree $<\sigma_{\mathrm{sub}}$.  Then from Thm. \ref{thm:transfer},
	\begin{align*}
		t_{n} = \sum_{\rho \in \mathsf{sing}_{0}(T)} \left(\frac{C_{\rho}}{\Gamma(-\sigma_{\rho})} \right)  n^{-1 - \sigma_{\rho}} \rho^{-n} + \mathcal{O} \left( n^{-1 - \sigma_{\mathrm{sub}}} \mathtt{R}_{T}^{-n} \right).
	\end{align*}
	Note that this expression may have some redundancies: some terms in the summation may be in the same big-$\mathcal{O}$ class as the unspecified subleading terms; absorbing such redundant terms into the unspecified big-$\mathcal{O}$ terms, we arrive at \eqref{eq:t_n_asymp}.
\end{proof}

\begin{numrmk} \label{remark:osc}
	Assume that we can choose $z_{0} = 0$, and $f$ has real coefficients.  Then the fact that $f$ has real coefficients ensures that all non-real dominant singularities of $T$ come in pairs with their complex conjugates.  Moreover, if $T$ is a root such that $T(x) \in \mathbb{R}$ for all $x \in \mathbb{R} \cap D_{0}(\mathtt{R}_{T})$, then $T(\conj{z}) = \conj{T(z)}$ for every $z \in D_{0}(\mathtt{R}_{T})$.  As a corollary of this latter fact, the fact that the Puiseux expansions around $z= \rho$ and $z = \conj{\rho}$ are holomorphic on a common domain, and the fact that $\conj{(1- \rho^{-1} z)^{\gamma}} = (1- \conj{\rho}^{-1} \conj{z})^{\gamma}$, it follows that $C_{\conj{\rho}} = \conj{C_{\rho}}$; hence, all (non-real) terms in \eqref{eq:t_n_asymp} come in complex-conjugate pairs and we have,
	\begin{align*}
		t_{n} = \left(\frac{1}{\Gamma(-\sigma_{*})}  \right)  \mathrm{Osc}(n)   n^{-1 - \sigma_{*}}\mathtt{R}_{T}^{-n}  + \mathcal{O} \left( n^{-1-\sigma_{\mathrm{sub}}} \mathtt{R}_{T}^{-n} \right)
	\end{align*}
	where
	\begin{align*}
		\mathrm{Osc}(n) &:=  \sum_{\rho \in \mathsf{leadsing}_{0}(T)} |C_{\rho}| \cos \left[n \arg(\rho) - \arg \left(C_{\rho} \right) \right].
	\end{align*}
\end{numrmk}

\subsubsection{Asymptotics of Euler-Product Exponents from Algebraicity} \label{sec:generating_asymptotics}
Let $\ggen \in \formal{\mathbb{Q}}{z}$ be a generating series for a sequence of rational numbers.  If $\ggen$ is algebraic, then it must be the series representation (Puiseux-expansion) of a holomorphic function $T$ around $z=0$ such that $T(0) = 1$.  Alternatively, if $T$ is an algebraic function that:
\begin{enumerate}
	\item is holomorphic on a disk containing $z=0$,
	
	\item satisfies $T(0) = 1$,
\end{enumerate}
then $T$ can has a series representation $T(z) = 1 + \sum_{n = 1}^{\infty} t_{n} z^{n} \in \formal{\conj{\mathbb{Q}}}{z}$.  Note that any such series admits an Euler-product factorization: for choice of a fixed $\mathtt{s} \in \{+1,-1\}$, we may write
\begin{align*}
	T(z) = \prod_{n = 1}^{\infty} \left(1 - \left(\mathtt{s} z \right)^{n} \right)^{n \beta_{n}}
\end{align*}
where the sequence of algebraic numbers $(\beta_{n})_{n = 1}^{\infty}$ is defined by \eqref{eq:beta_from_ggen}.  In other words, $T$ can be thought of as a generating series for a sequence of algebraic numbers.  

Of course, our interests lies in the case where the $(\beta_{n})$ are rational (or better yet, integers).  However, in the spirit of generality of the results in this section, we do not impose any rationality condition and allow $\ggen \in \formal{\conj{\mathbb{Q}}}{z}$ to denote a generating series of \textit{algebraic} numbers.

The following lemma hints that, if we wish to study asymptotics of the $\beta_{n}$, then we should really be studying the coefficient asymptotics of $\log(\ggen)$ expanded around $z = 0$.

\begin{lemma} \label{lem:omega_to_log}
	Let $\ggen \in \formal{\conj{\mathbb{Q}}}{z}$ be the generating series for $\left( \beta_{n} \right)_{n = 1}^{\infty}$; suppose $\ggen$ is algebraic over $\mathbb{Q}(z)$ and define
	\begin{align*}
		\mathtt{R} = \min \{ \mathtt{R}_{\ggen}, \mathtt{R}_{1/\ggen} \}.
	\end{align*} 
	\begin{enumerate}	
			\labitem{(A)}{case:growth} If $\mathtt{R} \leq 1$, then for any $0 < r < \mathtt{R}$
		\begin{align}
			\beta_{n} = -\frac{1}{n}  \left([z^{n}] \log \left(\ggen \right) \right) + \mathcal{O} \left( d(n) r^{-n/2} \right),
			\label{eq:Omega_Log_asymp}
		\end{align}
		where $d(n)$ is the number of divisors of $n$ that are less than $n$. 
	
		\labitem{(B)}{case:shrink} If $\mathtt{R} > 1$, then $\beta_{n} \in \mathcal{O} \left( n^{-2} \right)$.
	\end{enumerate}
\end{lemma}
\begin{proof}
	From \eqref{eq:omega_invert} we have
	\begin{align*}
		\beta_{n} &= -\frac{1}{n^2} \sum_{k|n}  k \mu \left(\frac{n}{k} \right ) \mathtt{s}^{k}  \left([z^{k}] \log \left(\ggen \right) \right)\\
		&=-\frac{\mathtt{s}^{n}}{n}  \left([z^{n}] \log \left(\ggen \right) \right) - \underbrace{\frac{1}{n^2} \sum_{\substack{k|n\\ k<n}} k \mu \left(\frac{n}{k} \right) \mathtt{s}^{k} \left([z^{k}] \log \left(\ggen \right) \right)}_{R(n)}
	\end{align*}
	Next, note that $\ggen$ is defines a	holomorphic function on $D_{0}(\mathtt{R}_{\ggen})$ with $\ggen(0) = 1$; hence, the composite function $\log(\ggen)$ is also holomorphic on any disk $D_{0}(r) \subset D_{0}(\mathtt{R}_{\ggen})$ such that $D_{0}(r)$ contains no zeros of $\ggen$, i.e. any sub-disk where both $\ggen$ and $1/\ggen$ are holomorphic.  The maximal such sub-disk has radius $\mathtt{R} = \min\{\mathtt{R}_{\ggen}, \mathtt{R}_{1/\ggen} \}$.  Now, it is a corollary of Cauchy's integral formula that if $h: D_{z_{0}}(r) \rightarrow \mathbb{C}$ is a holomorphic function on $D_{z_{0}}(r)$, then $[z^{n}] h \leq r^{-n} \left(\sup_{r' \in D_{z_{0}}(r')} h \right)$ for any $0 < r' < r$.  Applying this to our situation: for any $0 < r < \mathtt{R}$,
	\begin{align}
		\left|[z^{k}] \log \left(\ggen \right) \right| \leq C r^{-k}
		\label{eq:log_coeff_bound}
	\end{align}
	where $0<C_{r}< \infty$ is given by the supremum of $|\log(\ggen)|$ over the circle of radius $r$.  
	\begin{enumerate}
	
	\item \textbf{Case} \ref{case:growth}.  First note that
	\begin{align*}
		|R(n)| \leq \frac{1}{n} \sum_{\substack{k|n\\ k<n}} \left|[z^{k}] \log \left(\ggen \right) \right|.
	\end{align*}	
	Using the fact that $r < \mathtt{R} < 1$, the bound \eqref{eq:log_coeff_bound}, and the fact that the next largest divisor of $n$ (other than $n$ itself) is $\geq n/2$, then
	\begin{align*}
		|R(n)| \leq C_{r} d(n) r^{-n/2},
	\end{align*}
	where $d(n)$ is the number of divisors of $n$ that are $< n$.  This verifies \eqref{eq:Omega_Log_asymp}.

	\item \textbf{Case} \ref{case:shrink}. Choose $r$ such that $1 < r < \mathtt{R}$.  In this case
	\begin{align*}
		|R(n)| &\leq \frac{1}{n^2} \sum_{k = 1}^{n} k \mathtt{R}^{-k}\\
		 	&\leq \frac{1}{(r-1)^2} \left[\frac{r}{n^2}-\frac{r^{1-n}}{n^2}-\frac{r^{1-n}}{n}+\frac{r^{-n}}{n} \right]
	\end{align*}
	so $R(n) \in \mathcal{O}(n^{-2})$ as $n \rightarrow \infty$.  Furthermore, by \eqref{eq:log_coeff_bound} $[z^{n}] \log \left(\ggen \right) \in \mathcal{O}(r^{-n})$ as $n \rightarrow \infty$; hence, $\beta(n) \in \mathcal{O}(n^{-2})$.
	\end{enumerate}
\end{proof}

\begin{numrmk} \label{rmk:divisor_bound}
 One can use the rather crude estimate $d(n) \leq n$ to rewrite \eqref{eq:Omega_Log_asymp} as
	\begin{align*}
		\beta_{n} = -\frac{\mathtt{s}^{n}}{n}  \left([z^{n}] \log \left(\ggen \right) \right) + \mathcal{O} \left(n r^{-n/2} \right)
	\end{align*}	
	 however, there are various improvements on this bound.  As an example of one improvement (see \cite[\S 13.10]{apostol}): for any $\epsilon > 0$, there exists a $C_{\epsilon}$ such that
	\begin{align*}
		d(n) \leq C_{\epsilon} n^{\epsilon}
	\end{align*}
	for all $n \geq 1$; with this estimate, the subleading terms of \eqref{eq:Omega_Log_asymp} are in $\mathcal{O}(n^{\epsilon} r^{-n/2})$.
\end{numrmk}	

We are most interested in the case that $\ggen \in \formal{\mathbb{Z}}{z}$; indeed, if $\left(\beta_{n}\right)_{n = 1}^{\infty} \subset \mathbb{Z}$---which must be the case for BPS indices\footnote{A BPS index can be expressed as the trace of an integer-valued operator over a finite dimensional vector space: see \eqref{eq:BPS_superdim}).} or Euler characteristics---it follows that $\ggen \in \formal{\mathbb{Z}}{z}$.

\begin{proposition} \label{prop:growth}
	If $\ggen \in \formal{\mathbb{Z}}{z}$, then $\mathtt{R} \leq 1$.
\end{proposition}
\begin{proof}
	Assume that $\mathtt{R} > 1$; then if $\ggen$ has integer coefficients, it must have only finitely many as \eqref{eq:t_n_asymp} guarantees that $|g_{n}|$ becomes arbitrarily small as $n \rightarrow \infty$.  Hence, $\ggen$ must be a polynomial with integer coefficients and
\begin{align*}
	\mathtt{R} = \min \{|\omega|: \text{$\omega \in \mathbb{C}$ and $\ggen(\omega) = 0$}\}.
\end{align*}	
	Factorizing $\ggen$ over $\mathbb{C}$ we have $\ggen = K (z - \omega_{1}) \cdots (z - \omega_{n})$ for some collection of roots $\omega_{i} \in \conj{\mathbb{Q}},\, i = 1, \cdots, n$.  Integrality of the coefficients of $\ggen$ ensures that $K \in \mathbb{Z}$, but the product of the roots of $T$ must satisfy
	\begin{align*}
		\omega_{1} \cdots \omega_{n} = (-1)^{n} \frac{\ggen(0)}{K} = (-1)^{n} \frac{1}{K};
	\end{align*}
	so at least one root must have magnitude $\leq 1$, a contradiction.
\end{proof}

Hence, we will restrict our attention to the situation where $\mathtt{R} \leq 1$. For the interested reader, in \S \ref{sec:asymptotics_rational} we will revisit the case that $\mathtt{R} > 1$ (which can only occur for generating series with some $\beta_{n}$ non-integral).

Our next step is to mimic the singularity analysis of the previous section with the function $\log(\ggen)$; although $\log(\ggen)$ is no longer an algebraic function,\footnote{Let $T$ be an algebraic function such that $T \not \equiv 0$ and $T \not \equiv 1$, then it follows as a corollary of the Lindemann-Weierstrass theorem that $\log(T)$ is transcendental over $\mathbb{Q}(z)$.  Indeed, suppose $T \not \equiv 0,1$ is an algebraic function---i.e. $T$ satisfies $r(z,T(z)) = 0$ for some $r \in \mathbb{Q}[z,t]$---then evaluating $z$ at an algebraic number $z_{*} \in \conj{\mathbb{Q}}$ we must have $T(z_{*}) \in \conj{\mathbb{Q}}$.  Furthermore, we can choose $z_{*}$ such that $T(z_{*}) \neq 0,1$.  Now, by Lindemann-Weierstrass $\log(T(z_{*})) \notin \conj{\mathbb{Q}}$; hence, $\log(T)$ cannot be an algebraic function.} its singularities are closely related to the singularities (and zeros) of $\ggen$.

\begin{remark}
	Denote the set of dominant singularities of $\log(\ggen)$, defined with respect to its analytic continuation to disks centred about $z=0$, by $\mathsf{sing}_{0}(\log(\ggen))$.  This set of singularities can be described in terms of the singularities and zeros of $\ggen$:
	\begin{enumerate}
		\item $\conj{D_{0}(\mathtt{R}_{\ggen})} \cap \mathpzc{Z} = \emptyset$, i.e. there are no zeros of $\ggen$ in the closed-disk $\conj{D_{0}(\mathtt{R}_{\ggen})}$.  Then $\mathtt{R} = \mathtt{R}_{\ggen} = \mathtt{R}_{1/\ggen}$ and we have
		\begin{align*}
			\mathsf{sing}_{0}(\log(\ggen)) = \mathsf{sing}_{0}(\ggen) = \mathsf{sing}_{0}(1/\ggen).
		\end{align*}
		
		\item $\conj{D_{0}(\mathtt{R}_{\ggen})} \cap \mathpzc{Z} \neq \emptyset$, i.e. there is a zero of $\ggen$ in the closed disk $\conj{D_{0}(\mathtt{R}_{\ggen})}$.  We separate this case into two further subcases:
			\begin{enumerate}
				\item There are no zeros of $\ggen$ in the interior $D_{0}(\mathtt{R}_{\ggen})$ of $\conj{D_{0}(\mathtt{R}_{\ggen})}$: then $\mathtt{R} = \mathtt{R}_{\ggen} = \mathtt{R}_{1/\ggen}$ and $\mathsf{sing}_{0}(\log(\ggen))$ is given by adjoining the zeros on the boundary to the set  $\mathrm{sing}_{0}(\log(\ggen))$.  More pedantically
				\begin{align*}
					\mathsf{sing}_{0}(\log(\ggen)) = \mathsf{sing}_{0}(\ggen) \cup \mathsf{sing}_{0}(1/\ggen).
				\end{align*}
				
				\item There is a zero of $\ggen$ in the interior of $\conj{D_{0}(\mathtt{R}_{\ggen})}$: then $\mathtt{R} = \mathtt{R}_{1/\ggen} < \mathtt{R}_{\ggen}$ and
				\begin{align*}
					\mathsf{sing}_{0}(\log(\ggen)) = \mathsf{sing}_{0}(1/\ggen).
				\end{align*}
			\end{enumerate}
	\end{enumerate}	
As the proof of Thm. \ref{thm:main} will show, the dominant singularities of $\log(\ggen)$ will be a ``finite" branch point if it lies in the complement of $\mathpzc{Z} \cup \mathpzc{P}$, and it will be a logarithmic branch point if it lies in $\mathpzc{Z} \cup \mathpzc{P}$.
\end{remark}

Because singularities in $\mathpzc{Z} \cup \mathpzc{P}$ (i.e. singularities arising as zeros and poles of $\ggen$) will play a different role than finite branch points, it is helpful to introduce some notation implicitly defined in the following remark.

\begin{numrmk} \label{rmk:m_rho_def}
Let $\rho \in \mathrm{sing}_{0}(\log(\ggen)) \cap \left(\mathpzc{Z} \cup \mathpzc{P} \right)$; then the Puiseux expansion of $\ggen$ around $\rho$ will take the form
\begin{align*}
	\ggen(z) = (1 - \rho^{-1} z)^{m_{\rho}} g \left[ \left(1 - \rho^{-1} z \right)^{1/\kappa_{\rho}} \right]
\end{align*}
where $g$ is a holomorphic function on the unit disk centred about $0$ such that $g(0) \neq 0$ and
	\begin{itemize}
		\item $m_{\rho} > 0$ if $\rho \in \mathpzc{Z}$ (i.e. is a zero of $\ggen$),
		
		\item $m_{\rho} < 0$ if $\rho \in \mathpzc{P}$ (i.e. is a zero of $1/\ggen$).
	\end{itemize}
Note that if $m_{\rho} \in \mathbb{Q} \backslash \mathbb{Z}_{>0}$ then $m_{\rho} = \sigma_{\rho}$; otherwise, if $m_{\rho} \in \mathbb{Z}_{>0}$ (i.e. $\rho$ is an unramified zero of $\ggen$), then $m_{\rho} < \sigma_{\rho}$.
\end{numrmk}

Equipped with Lemma \ref{lem:omega_to_log}, we state the following classification theorem.

\begin{theorem} \label{thm:main}
	Suppose $\ggen \in \formal{\conj{\mathbb{Q}}}{z}$ is an algebraic series generating the sequence $(\beta_{n})_{n = 1}^{\infty} \subset \conj{\mathbb{Q}}$, and such that $\mathtt{R} \leq 1$. 
	\begin{enumerate}
		\item If $\conj{D_{0}(\mathtt{R}_{\ggen})} \cap (\mathpzc{Z} \cup \mathpzc{P}) = \emptyset$, i.e. $\ggen$ has no zeros or poles on $D_{0}(\mathtt{R}_{\ggen})$, then $\mathtt{R} = \mathtt{R}_{\ggen}$, $\sigma_{*} > 0$, and 
		\begin{align}
			\beta_{n} &= \left(\frac{\mathtt{s}^{n}}{\Gamma(-\sigma_{*})} \right) n^{-2 - \sigma_{*}} \sum_{\rho \in \mathsf{leadsing}_{0}(\ggen)} \left( \ggen(\rho) C_{\rho} \right) \rho^{-n} + \mathcal{O} \left(n^{-2 - \sigma_{\mathrm{sub}}} \mathtt{R}^{-n} \right),
			\label{eq:beta_asymp_nozeropole}
		\end{align}
		as $n \rightarrow \infty$; where $\sigma_{*},\, \sigma_{\mathrm{sub}} > \sigma_{*},$ and $\mathsf{leadsing}_{0}(\ggen)$ are defined in \eqref{eq:leading_exps}.  Furthermore, if $\ggen$ is a series with real-coefficients,
		\begin{align}
			\beta_{n} &= \left(\frac{\mathtt{s}^{n}}{\Gamma(-\sigma_{*})} \right) \mathrm{Osc}(n) n^{-2 - \sigma_{*}} \mathtt{R}^{-n} + \mathcal{O} \left(n^{-2 - \sigma_{\mathrm{sub}}} \mathtt{R}^{-n} \right)
			\label{eq:beta_asymp_nozeropole_cosine}
		\end{align}
		as $n \rightarrow \infty$; where
		\begin{align*}
			\mathrm{Osc}(n) :=  \sum_{\rho \in \mathsf{leadsing}_{0}(\ggen)} \left| \ggen(\rho) C_{\rho} \right| \cos \left[n \arg(\rho) - \arg \left( \ggen(\rho) C_{\rho} \right) \right].
		\end{align*}
		
		\item If $\conj{D_{0}(\mathtt{R}_{\ggen})} \cap (\mathpzc{Z} \cup \mathpzc{P}) \neq \emptyset$, i.e. $\ggen$ has a zero or a pole on $\conj{D_{0}(\mathtt{R}_{\ggen})}$, then
		\begin{align}
			\beta_{n} &= \mathtt{s}^{n} n^{-2} \left( \sum_{\rho \in \mathrm{sing}_{0}(\log(\ggen)) \cap (\mathpzc{Z} \cup \mathpzc{P})} m_{\rho} \rho^{-n} \right) + \mathcal{O}(n^{-2-1/\kappa_{*}} \mathtt{R}^{-n});
			\label{eq:beta_asymp_zeropole}
		\end{align}	
		where $m_{\rho}$ is defined in Rmk. \ref{rmk:m_rho_def} and
		\begin{align*}
			\kappa_{*} := \max \left \{\kappa_{\rho}: \rho \in \mathpzc{Z} \cup \mathpzc{P} \right\} \geq 1.
		\end{align*}
		Furthermore, if $\ggen$ is a series with real-coefficients,
		\begin{align}
			\beta_{n} = \mathtt{s}^{n} \mathrm{Osc}(n) n^{-2} \mathtt{R}^{-n} + \mathcal{O}(n^{-2 - 1/\kappa_{*}} \mathtt{R}^{-n})
			\label{eq:beta_asymp_zeropole_cosine}
		\end{align}
		as $n \rightarrow \infty$; where
		\begin{align*}
			\mathrm{Osc}(n) &:=  \sum_{\rho \in \mathrm{sing}_{0}(\log(\ggen)) \cap (\mathpzc{Z} \cup \mathpzc{P})} m_{\rho}  \cos \left[n \arg(\rho) \right].
		\end{align*}
	\end{enumerate}
\end{theorem}
\begin{proof}
	To prove the theorem, we first analyze the leading order asymptotics of $\log(\ggen)$ and then apply Lem. \ref{lem:omega_to_log}.  To begin, we divide the dominant singularities of $\log(\ggen)$ into two cases:
	\begin{enumerate}
		\labitem{(i)}{list:log} $\rho \in \mathrm{sing}_{0}(\log(\ggen)) \cap \left(\mathpzc{Z} \cup \mathpzc{P} \right)$, i.e. $\rho$ is a zero of $\ggen$ or $1/\ggen$;

		\labitem{(ii)}{list:finite} $\rho \notin \mathrm{sing}_{0}(\log(\ggen)) \cap \left(\mathpzc{Z} \cup \mathpzc{P} \right)$, i.e. $\ggen(\rho) \neq 0$ and $\rho$ is a ``finite" branch point of $\ggen$ in the sense of \ref{list:fin_branch}.
	\end{enumerate}

	For simplicity of notation, throughout we define
	\begin{align*}
		\zeta := (1 - \rho^{-1} z)
	\end{align*}
	where the value of $\rho$ under consideration will be clear from context.
	
	\subsubsection*{Case \ref{list:log}}
	Assume that $\rho \in \mathrm{sing}_{0}(\log(\ggen))$ is a zero of $\ggen$ or $1/\ggen$, then $\exists m_{\rho} \in \mathbb{Q}_{\neq 0}$ such that the Puisuex expansion of $\ggen$ around $\rho$ can be written as
	\begin{align*}
		\ggen(z) &= \zeta^{m_{\rho}} g \left[ \zeta^{1/\kappa_{\rho}} \right]
	\end{align*}
	where $g$ is a holomorphic function on the unit disk centred about $0$ such that $g(0) \neq 0$, $m_{\rho}$ is a positive rational number, and $\kappa_{\rho} \in \mathbb{Z}_{ \geq 1}$.  Hence,
	\begin{align*}
		\log(\ggen) &= \log \left( \zeta ^{m_{\rho}} \right) +  \log \left\{ g \left[\zeta^{1/\kappa_{\rho}} \right]  \right\}\\
		&=  m_{\rho} \log(\zeta)  + \log \left\{g(0) + \mathcal{O} \left[\zeta^{1/\kappa_{\rho}} \right] \right\}\\
		&= m_{\rho} \log \left(\zeta \right) + \log(-g(0)) + \mathcal{O} \left[\zeta^{1/\kappa_{\rho}} \right].
	\end{align*}

	\subsubsection*{Case \ref{list:finite}}
	Next we move on to case \ref{list:finite}.  Let $\alpha := \ggen(\rho) \neq 0$ so that $(\rho, \alpha) \in \mathrm{Sing}_{0}(\ggen) \subset \mathpzc{R}_{\:0}$.  Then, on the appropriate open region, we may write
	\begin{align*}
		\ggen(z) = \zeta^{\sigma_{\rho}} g \left[\zeta^{1/\kappa_{\rho}} \right] + \alpha + E_{\rho}(z).
	\end{align*}
		where $g$ is a holomorphic function on the unit disk such that $g(0) = C_{\rho} \neq 0$, $\sigma_{\rho}$ is a \textit{positive} rational number, and $E_{\rho}(z)$ is a polynomial in $z$ of degree $< \sigma_{\rho}$.  Taking the logarithm of both sides, we have
	\begin{align*}
		\log \left[\ggen(z) \right] &= \log(\alpha) + \sum_{m = 1}^{\infty} \frac{(-1)^{m + 1}}{m} \alpha^{-m} \left[\zeta^{\sigma_{\rho}} g \left(\zeta^{1/\kappa_{\rho}} \right) +  E_{\rho}(z)  \right]^{m};
	\end{align*}
	Hence, as $z \rightarrow \rho$,
	\begin{align*}
	\log \left[\ggen(z) \right] &= -\alpha^{-1} g(0) \zeta^{\sigma_{\rho}} + \mathrm{Poly}_{\rho}(z) + \mathcal{O} \left[\zeta^{\sigma_{\rho} + 1/\kappa_{\rho}} \right].
	\end{align*}
	for some polynomial $\mathrm{Poly}_{\rho}(z)$ of degree $\leq \sigma_{\rho}$.

Our analyses of cases \ref{list:log} and \ref{list:finite} show that, as $z \rightarrow \rho$,
	\begin{align*}
		\log(\ggen)  &= \left\{
		\begin{array}{lr}
			m_{\rho} \log \left(1 - \rho^{-1}z \right), & \text{if $\rho \in \mathpzc{Z} \cup \mathpzc{P}$}\\
			-\alpha^{-1} C_{\rho} \left(1 - \rho^{-1}z \right)^{\sigma_{\rho}}, & \text{if $\rho \notin \mathpzc{Z} \cup \mathpzc{P}$}
		\end{array}
		\right. + \mathrm{Poly}_{\rho}'(z) + \mathcal{O} \left[\left(1 - \rho^{-1}z \right)^{D} \right]
	\end{align*}
	where
	\begin{itemize}	
		\item $D = \min \left[ \{ \sigma_{\rho} + \frac{1}{\kappa_{\rho}}: \text{$\rho \notin \mathpzc{Z} \cup \mathpzc{P}$ and $\sigma_{\rho} + \frac{1}{\kappa_{\rho}} \notin \mathbb{Z}_{>0}$} \} \cup \{\frac{1}{\kappa_{\rho}}: \rho \in \mathpzc{Z} \cup \mathpzc{P} \} \right]$
		
		\item $\mathrm{Poly}_{\rho}'(z)$ is a polynomial of degree $<D$;
		
	\end{itemize}		
		 The asymptotics of $[z^n] \log(\ggen)$ follow after applying Thm.~\ref{thm:transfer} to the above discussion, eliminating redundant terms that are absorbed into subleading asymptotics.  Equations \eqref{eq:beta_asymp_nozeropole} and \eqref{eq:beta_asymp_zeropole} follow by application of Lemma \ref{lem:omega_to_log}, which transfers the asymptotics of $[z^{n}] \log(\ggen)$ to the asymptotics of $\beta_{n}$.  A few more words are in order for this last statement: it remains to show that the subleading asymptotics of \eqref{eq:Omega_Log_asymp}, which live in $\mathcal{O}(d(n) r^{-n/2})$, can be safely absorbed into the subleading asymptotics of $n^{-1} \left([z^{n}] \log(\ggen) \right)$, which are in $\mathcal{O}(n^{\beta} \mathtt{R}^{-n})$ for some $\beta \in \mathbb{R}$.  There are two cases at hand.
		 \begin{enumerate}
		 	\labitem{(A)}{list:lessthan1}  \textbf{When $\mathtt{R} < 1$}. We may choose an $r$ in \eqref{eq:Omega_Log_asymp} such that $0 < r < \mathtt{R}$ and $r^{1/2} > \mathtt{R}$ (e.g. $r = \mathtt{R}^{3/2}$); the choice of such an $r$ ensures that $\mathcal{O}(n^{\alpha} r^{-n/2}) \subset \mathcal{O}(n^{\beta} \mathtt{R}^{-n})$ for any $\alpha,\, \beta \in \mathbb{R}$.  So, by Rmk. \ref{rmk:divisor_bound}, we have $\mathcal{O}(d(n) r^{-n/2}) \subset \mathcal{O}(n^{\beta} \mathtt{R}^{-n})$ for any $\beta \in \mathbb{R}$ and we are done.

			\labitem{(B)}{list:equalto1} \textbf{When $\mathtt{R} = 1$}.  The situation is slightly more subtle: the subleading asymptotics of $[z^{n}] \log(\ggen)$ are decreasing (being in $\mathcal{O}(n^{-\beta})$ for some $\beta >0$) while $d(n) r^{-n/2}$ grows exponentially for any $r < \mathtt{R} = 1$; hence, Lem. \ref{lem:omega_to_log} appears useless at first sight.  However, we may salvage the situation by utilizing the following trick that takes us back to the $\mathtt{R} < 1$ situation.  First, let $a > 1$ be any integer and define the series
	\begin{align*}
		\twid{\ggen} &= 1 + \sum_{n = 1}^{\infty} g_{n} a^{n} z^{n};
	\end{align*}
	as $a$ is an integer, then $\twid{\ggen} \in \formal{\conj{\mathbb{Q}}}{z}$.  Moreover, $\twid{\ggen}$ is an algebraic series over $\mathbb{Q}$ (if $f(z, \ggen) = 0$ for some $f \in \mathbb{Q}[z,t]$ then $f(az,\twid{\ggen}) = 0$).  Now, because $\twid{\ggen}$ has constant coefficient 1, it admits an Euler product expansion that defines a sequence of algebraic numbers $(\twid{\beta}_{n})_{n}$ (defined via \eqref{eq:beta_from_ggen}):
	\begin{align*}
		\twid{\ggen} &= \prod_{n = 1}^{\infty} \left( 1 - \left(\mathtt{s} z \right)^{k} \right)^{n \twid{\beta}_{n}}.
	\end{align*}
	Next, observe that $\rho \in \mathrm{sing}_{0}(\log(\mathsf{\ggen}))$ if and only if $a^{-1} \rho \in \mathrm{sing}_{0}(\log(\twid{\mathsf{G}}))$; in particular, $\twid{\mathtt{R}} = \min (\mathtt{R}_{\twid{\ggen}}, \mathtt{R}_{\twid{\ggen}^{-1}}) = a^{-1} \mathtt{R} < 1$.  Hence, by the argument in \ref{list:lessthan1}, Thm. \ref{thm:main} holds with $\ggen,\, \beta,\,$ and $\mathtt{R}$ replaced with their twiddled-counterparts $\twid{\ggen},\, \twid{\beta}_{n}$, and $\twid{\mathtt{R}}$.  To reduce to the asymptotics for $\beta_{n}$, we observe that $[z^{n}] \log \left(\twid{G} \right) = a^{n} [z^{n}] \log \left(G \right)$; so from \eqref{eq:beta_from_ggen}, we have
\begin{align*}	
	\twid{\beta}_{n} &= a^{n} \beta_{n}.
\end{align*}
	Moreover, from our observations, when passing back to all un-twiddled quantities, an overall multiplicative factor of $a^{n}$ may be extracted from the right hand side of the twiddled analogues of \eqref{eq:beta_asymp_nozeropole} or \eqref{eq:beta_asymp_zeropole}.  Cancelling factors of $a^{n}$, the un-twiddled \eqref{eq:beta_asymp_nozeropole} and \eqref{eq:beta_asymp_zeropole} remain valid.
\end{enumerate}
	 Lastly, if $\ggen$ has coefficients in $\mathbb{R}$, we may rewrite \eqref{eq:beta_asymp_nozeropole} (respectively \eqref{eq:beta_asymp_zeropole}) as \eqref{eq:beta_asymp_nozeropole_cosine} (respectively \eqref{eq:beta_asymp_zeropole_cosine}) by using the reality of the coefficients of $f$ (c.f. Remark \ref{remark:osc}).
\end{proof}

\begin{remark}
	As shown in \S \ref{sec:asymptotics_rational}, if we restrict to $\ggen \in \formal{\mathbb{Q}}{z}$, we may drop the condition that $\mathtt{R} \leq 1$ from the statement of Thm. \ref{thm:main}.
\end{remark}

The following explicit expression follows immediately when combining 	Thm.~\ref{thm:main} with Prop. \ref{prop:growth} and Rmk. \ref{rmk:smooth_dom}.

\begin{corollary} \label{cor:multi}
	Let 
	\begin{enumerate}
		\item $\ggen \in \formal{\conj{\mathbb{Q}}}{z}$ be a generating series for $(\beta_{n})_{n = 1}^{\infty} \subset \conj{\mathbb{Q}}$, algebraic over $\mathbb{Q}(z)$;
		
		\item  $f \in \conj{\mathbb{Q}}[z,t]$ be the (unique up to a constant) absolutely irreducible polynomial such that $f(z,\ggen) = 0$;
		
		\item $\mathrm{sing}_{0}(\log(\ggen)) = \{\rho_{1}, \cdots, \rho_{k}\}$ be the set of dominant singularities of $\log(\ggen)$;
	\end{enumerate}
	 If $\mathtt{R} \leq 1$, $p_{i} := (\rho_{i}, \ggen(\rho_{i}))$ is a smooth ramification point of $\phi: V_{f} \rightarrow \mathbb{C}$ for all $i \in \{1, \cdots, k\}$, and the ramification indices of each $p_{i}$ are all equal: $\nu := \nu_{p_{1}} = \nu_{p_{2}} = \cdots = \nu_{p_{n}}$, then as $n \rightarrow \infty$
	\begin{align}
		\beta_{n} = - \left[\frac{\mathtt{s}^{n} n^{-2 - 1/\nu} }{\Gamma \left(-\frac{1}{\nu} \right)} \right]  \sum_{i = 1}^{k} \frac{C_{\rho_{i}}}{\ggen(\rho_{i})} \rho_{i}^{-n} + \mathcal{O} \left(n^{-2 - \sigma_{\mathrm{sub}} } \mathtt{R}^{-n} \right)
		\label{eq:multi_dominant}
	\end{align}
	where 
		\begin{align*}
		\sigma_{\mathrm{sub}} &= 
		\left\{
		\begin{array}{ll}
			3/2, & \text{if $\nu = 2$}\\
			2/\nu, & \text{if $\nu > 2$}
		\end{array}
		\right.
	\end{align*}
	and the constants $C_{\rho_{i}}$ are defined in \eqref{eq:sigma_rho_def} (via the Newton-Puiseux expansion of $\ggen$ around $\rho_{i}$): explicitly, via Rmk. \ref{rmk:smooth_dom}, there exist a collection of $\nu$th roots of unity $\{\omega_{\nu,i}\}_{i = 1}^{k}$ (satisfying the condition that $\omega_{i} = \conj{\omega_{j}}$ if $\rho_{i}$ and $\rho_{j}$ are conjugate) such that
	\begin{align*}
		C_{\rho_{i}} = \omega_{\nu,i} \left[ \rho_{i} \left( \frac{f_{1,0}(p_{i})}{f_{0,\nu}(p_{i})} \right) \right]^{1/\nu} 
	\end{align*}
	where $f_{n,m}(p_{i})$ is defined via \eqref{eq:f_coeffs}.  Moreover, if $\ggen \in \formal{\mathbb{Q}}{z}$, then the quantity in the radical is a rational number and $\omega_{\nu,i} \in \{\pm 1\}$.
\end{corollary}

\begin{remark}
Notice that we have been careful to not use the symbol ``$\sim$"; the reason for this is the following.  If there are multiple dominant singularities, at least one (along with its complex conjugate) of which is complex, then $\mathrm{Osc}(n)$ is a non-constant ``oscillatory" function.  The result is that, as $n \rightarrow \infty$
\begin{align*}
	\frac{\beta_{n}}{\mathtt{s}^{n} n^{-l} \mathrm{Osc}(n) \mathtt{R}^{-n}} = 1 + \mathcal{O} \left(n^{-s} \mathrm{Osc}(n)^{-1}  \right);
\end{align*}
where $l > 0$ and $s > l$ are rational numbers, depending on the specific example at hand, that are specified by Thm.~\ref{thm:main}.  The big-$\mathcal{O}$-terms, however, are not guaranteed to vanish (or even have a well-defined limit) as $n \rightarrow \infty$. For example, in the case of two complex dominant singularities, we would have
\begin{align*}
	\mathrm{Osc}(n) &= \cos \left[ n \arg(\rho) - \phi \right]
\end{align*}
for some $\phi \in [0,2\pi)$.  When $\arg(\rho)$ is an irrational multiple of $\pi$, then the set $\{\mathrm{Osc}(n)\}_{n = 1}^{\infty}$ densely fills the interval $[-1,1]$.  In particular the set $\{\mathrm{Osc}(n)\}_{n = 1}^{\infty}$ contains infinitely many elements of magnitude $< \epsilon$ for any fixed $\epsilon$; so for any fixed $M > 0$, then $\{\mathrm{Osc}(n)^{-1}\}_{n = 1}^{\infty}$ contains infinitely many elements of magnitude $>M$---it follows that $n^{-s} \mathrm{Osc}(n)^{-1}$ has no well-defined limit as $n \rightarrow \infty$.

However, returning to the general scenario, note that, because $\mathrm{Osc}(n)$ is a bounded function we always have that
\begin{align*}
	\limsup_{n \rightarrow \infty} \frac{\beta_{n}}{\mathtt{s}^{n} n^{-l} \mathtt{R}^{-n}} &= \limsup_{n \rightarrow \infty} \mathrm{Osc}(n) = C_{1}\\
	\liminf_{n \rightarrow \infty} \frac{\beta_{n}}{\mathtt{s}^{n} n^{-l} \mathtt{R}^{-n}} &=  \liminf_{n \rightarrow \infty} \mathrm{Osc}(n) = C_{2}
\end{align*}
for some constants $C_{1},\, C_{2} \in \mathbb{R}$ (at least one of which is non-vanishing). 

 In the simple case of a single dominant singularity then $\mathrm{Osc}(n) \equiv C$ for some non-zero $C \in \mathbb{R}$; hence, we are justified in writing
\begin{align*}
	\beta_{n} \sim \mathtt{s}^{n} C n^{-l} \mathtt{R}^{-n}.
\end{align*}
\end{remark}

\begin{corollary} \label{cor:single}
Let $\ggen$ and $f$ be defined as in Corollary \ref{cor:multi}. If $\log(\ggen)$ has a single dominant singularity $\rho$ such that $\rho \notin \mathpzc{Z} \cup \mathpzc{P}$ and $(\rho, \ggen(\rho))$ is a smooth ramification point of $\phi: V_{f} \rightarrow \mathbb{C}$ with ramification index $\nu$, then as $n \rightarrow \infty$
\begin{align}
	\beta_{n} \sim \omega_{v} \left(\rho  \frac{f_{1,0}(p)}{f_{0,\nu}(p)} \right)^{1/\nu} \left[\frac{\mathtt{s}^{n}}{\ggen(\rho) \Gamma \left(-\frac{1}{\nu} \right)} \right]n^{-2 - 1/\nu}  \rho^{-n}
	\label{eq:single_dominant}
\end{align}
where $f_{n,m}(p)$ is defined via \eqref{eq:f_coeffs}, $\omega_{v}$ is a $\nu$th root of unity and the big-$\mathcal{O}$ class of the subleading asymptotics is provided by the previous corollary.
\end{corollary}

%

\begin{example}[Examples]\
	\begin{enumerate}
		\labitem{1}{ex:m_herd_asy} \textbf{The $m$-herd}. One can study the $m$-herd BPS asymptotics by studying the BPS generating series $\gdt$, given as a root of
		\begin{align*}
			\mathcal{F}_{(3,2|3)} = -1 + t (t - z t^{(m-2)})^{m} \in \left(\mathbb{Z}[z] \right)[t];
		\end{align*}
		however, as with our previous $m$-herd examples, we will proceed indirectly with the slightly simpler polynomial
		\begin{align*}
			\mathcal{H}_{k} = p - z p^{k} - 1 \in \left(\mathbb{Z}[z] \right)[p]
		\end{align*}
		Letting $P$ be the unique root of this equation with $P(0) = 1$, we define the sequence $(c_{n})_{n=1}^{\infty}$ via the Euler-product factorization
		\begin{align*}
			P = \prod_{n = 1}^{\infty} \left(1  - (\mathtt{s} z)^{n} \right)^{n c_{n}};
		\end{align*}
		noting that, when $k = (m-1)^2$ and choosing $\mathtt{s} = (-1)^{m}$, the $m$-herd BPS generating series is given by 
		\begin{align*}
			\gdt= P^{m}
		\end{align*}		
		hence,
		\begin{align}
			\Omega(n \gamma_{c}) = m c_{n}.
			\label{eq:c_to_Omega}
		\end{align}
		Letting $k := (m-1)^{2}$, there is a single dominant singularity located at 
		\begin{align*}
			\rho = k^{-k} (k-1)^{k-1}
		\end{align*}		
		with
		\begin{align*}
			P(\rho) = \frac{k}{k-1} = \frac{(m-1)^{2}}{m(m-2)}.
		\end{align*}
		The corresponding ramification point $p = (\rho, P(\rho))$ is a smooth point of ramification index $\nu_{p} = 2$.  Working through the details of \eqref{eq:single_dominant}, we have
		\begin{align*}
			c_n \sim - \frac{\mathtt{s}^{n}}{k} \sqrt{\frac{k}{2 \pi (k-1)}} n^{-5/2} \left(k^{k} (k-1)^{-(k-1)} \right)^{n}.
		\end{align*}
		Setting $k = (m-1)^2$, $\mathtt{s} = (-1)^{m}$ and using \eqref{eq:c_to_Omega}, we have
		\begin{align*}
			\Omega(n \gamma_{c}) \sim (-1)^{m n + 1} \left(\frac{1}{m-1}  \sqrt{\frac{m}{2 \pi (m-2)}} \right) n^{-5/2} \left[(m-1)^{(m-1)} m(m-2)^{-m(m-2)} \right]^{n};
		\end{align*}
		which agrees with the asymptotics of \cite[Prop. 3.4]{wwc}, proved via a procedure using the explicit series expansion of $P$ around $z =0$.
			
		\labitem{2}{ex:3_2_asy} \textbf{The $(3,2|3)$-herd}.  The BPS generating series $\gdt$ is given by the root of a 39th degree (absolutely irreducible) polynomial (see \eqref{eq:3_2_dt_relation})---in particular, the root defined by the unique \textit{analytic} section of the associated algebraic curve, passing through the point $(0,1) \in\mathbb{C}^{2}$ (the point $(0,1)$ is a singular point of the algebraic curve corresponding to the intersection of two branches; only one branch defines an analytic section at $z=0$).  There is a single dominant singularity of $\gdt$ at $\rho \approx 0.005134$ that arises as the root of smallest magnitude of the degree 10 polynomial $d_{2}(z)$ (a factor of the discriminant polynomial of $\mathcal{F}_{(3/2|3)}$, c.f. \eqref{eq:Disc_F_3_2}); explicitly,
			\begin{equation}
				\begin{aligned}
				d_{2}(z) &= 7453227051205047621210969560803493368439376354217529296875\\
				& - 2322891406807452663970230316950021070808611526489257812500000 z  \\
				& + 153110106665001377524256387694240634085046455866589161785600000 z^2 \\
&- 7817871138183350523859861706656575827739367498904849637993525248 z^3  \\
&+ 2216428789463767802947118021038105095664745743271952551598777966592 z^4 \\ 
&- 12918183321713078299888780094691016787906213863596568047251860094976 z^5  \\
&+ 54012861575241903106685494870153854374366753649646197438214701056000 z^6  \\
&+ 2067823170060188178302247072833500754014100750420991514751809880064 z^7\\
&- 73168130623725009119914340392401973195062456631524680181743616 z^8 \\
&- 692626494138074646657585699931857043089076077167902720 z^9  \\
&+ 106627982583039156618936454468596279550148608 z^{10}.
				\end{aligned}
			\label{eq:d_2}
			\end{equation}
			It can be checked that $(\rho,\gdt(\rho))$ (where $\gdt(\rho) \approx 1.203051$) is a smooth point with ramification index $2$ (and one can see both associated branches within the real points of the curve: see Fig.~\ref{fig:(3_2|3)_curve}). Applying \eqref{eq:single_dominant}, we arrive at the asymptotics stated in \eqref{eq:3_2_herd_asymptotics}, which we restate here for convenience:
			\begin{align*}
				\Omega \left(n \gamma_{c} \right) \sim (-1)^{n+1} C  n^{-5/2} \rho^{-n},
			\end{align*}
		where $C \approx 0.075084$.
		
		\item \textbf{$3$-Kronecker Quiver Euler characteristics}: The generating series $\geul = \geul_{3/2}$ is a root of a 9th degree (absolutely irreducible) polynomial (c.f. \eqref{eq:3_2_eul_relation}) $\mathcal{E} \in \left(\mathbb{Z}[z]\right)[e]$.  It has two dominant singularities at two complex conjugate roots of the polynomial \eqref{eq:eul_discrim_factor}, a factor of the discriminant polynomial $\mathrm{Disc}_{\mathbb{Z}[e]}(\mathcal{E}) \in \mathbb{Z}[z]$.  Explicitly these dominant singularities are given by
		\begin{align*}
			\rho \approx 0.0151352 + 0.0373931 i
		\end{align*}
		and its complex conjugate.  The corresponding points $(\rho, \geul(\rho))$ and $(\conj{\rho}, \geul(\conj{\rho}) = \conj{\geul(\rho)})$, where
\begin{align*}
		\geul(\rho) \approx 0.6590388 + 0.7452078 i
\end{align*}		
		are non-singular points of ramification index $2$.  Applying \eqref{eq:multi_dominant}, we arrive at the asymptotics stated in \eqref{eq:chi_3_2_asymps_cx}.
	\end{enumerate}
\end{example}

The physical consequences of the results in this section, when combined with Claim \ref{claim:func_eqn},  can be summarized in a single statement.

\begin{phcorollary}
	Let $\gamma_{c}$ be a primitive charge of a BPS state in a theory of class $S[A_{K-1}]$ such that:
	\begin{itemize}
		\item the state occurs at a point on the Coulomb branch that is off of any walls of marginal stability,
		\item and it is represented by a spectral network with finitely many two-way streets;
	\end{itemize}
	then the generating series for the BPS indices $\{\Omega(n \gamma_{c}) \}_{n = 1}^{\infty} \subset \mathbb{Z}$ is algebraic.  Furthermore, $\Omega(n \gamma_{c})$ grows asymptotically as
	\begin{align*}
		\Omega(n \gamma_{c}) = \mathrm{Osc}(n) n^{-2 - \alpha} \mathtt{R}^{-n} + \mathcal{O} \left(n^{-2 - \alpha - \epsilon} \mathtt{R}^{-n} \right)
	\end{align*}
	as $n \rightarrow \infty$; where $\mathtt{R} \in \conj{\mathbb{Q}} \cap (0,1]$, $\alpha \in \mathbb{Q}_{\geq 0},\, \epsilon \in \mathbb{Q}_{>0}$, and $\mathrm{Osc}(n)$ is a bounded ``oscillatory" function.  Furthermore, because $\left(\Omega(n \gamma_{c}) \right)_{n} \subset \mathbb{Z}$ then either there are finitely many BPS indices ($\mathtt{R}=1$), or there are infinitely many ($\mathtt{R} < 1$) with asymptotic growth of the form given above.
\end{phcorollary}

\subsubsection{Algebraic Generating Series with Rational Coefficients} \label{sec:asymptotics_rational}
Suppose we were interested in studying the asymptotics of a series of rational numbers $(\beta_{n})_{n = 1}^{\infty}$ with algebraic generating series $\ggen \in \formal{\mathbb{Q}}{z}$ that does not necessarily lie in $\formal{\mathbb{Z}}{z}$.  Such a scenario cannot occur for the generating series of $m$-Kronecker Euler characteristics or BPS index generating series: in both scenarios $(\beta_{n})_{n = 1}^{\infty} \subset \mathbb{Z}$; however, this case may be of interest in the study of DT invariants that do not necessarily satisfy an integrality condition (and, hence, cannot arise as BPS indices).  If the radius of convergence $\mathtt{R}$ of $\log(\ggen)$ is less than or equal to 1, then we can simply apply Thm. \ref{thm:main}; however, it is possible that $\mathtt{R} > 1$ (a situation that cannot occur for series in $\formal{\mathbb{Z}}{z}$ by Prop. \ref{prop:growth}).  In this latter situation, Lemma \ref{lem:omega_to_log} indicates that the associated $\beta_{n}$ must shrink: $\beta_{n} \in \mathcal{O}(n^{-2})$ as $n \rightarrow \infty$.  However, as we will see, this is a rather crude estimate and, in fact, the formulae of Thm. \ref{thm:main} continue to hold for $\ggen \in \formal{\mathbb{Q}}{z}$.  The reason for this magic is due to the following Lemma.

\begin{lemma}[Eisenstein, Heine] \label{lem:Eisenstein_Heine}
	Let $\sum_{n = 0}^{\infty} t_{n} z^{n} \in \formal{\mathbb{Q}}{z}$ be algebraic over $\mathbb{Q}(z)$, then there exists an integer $a$ (an Eisenstein constant) such that $a^{n} t_{n},\, n \geq 1$ are all integers.
\end{lemma}
\begin{proof}
	See \cite[pg. 327]{dienes:taylor_series} which attributes the statement to Eisenstein (1852) and the proof to Heine (1854); part of the proof is identical to the proof of Lem. \ref{lem:Eisenstein}.
\end{proof}
	
Now let $a$ denote an Eisenstein constant of $\ggen = 1 + \sum_{n = 1}^{\infty} g_{n} z^{n}$, and define the formal series
\begin{align*}
	\twid{\ggen} &= 1 + \sum_{n = 1}^{\infty} g_{n} a^{n} z^{n}
\end{align*}
by definition of $a$, then $\twid{\ggen} \in \formal{\mathbb{Z}}{z}$.  Moreover, $\twid{\ggen}$ is algebraic (if $f(z, \ggen) = 0$ for some $f \in \mathbb{Z}[z,t]$ then $f(az,\twid{\ggen}) = 0$).  Now, because $\twid{\ggen}$ has constant coefficient 1, it admits an Euler-product expansion that defines a sequence of rational numbers $\left(\twid{\beta}_{n} \right)_{n = 1}^{\infty}$ (defined via \eqref{eq:beta_from_ggen})
\begin{align*}
	\twid{\ggen} &= \prod_{n = 1}^{\infty} \left( 1 - \left(\mathtt{s} z \right)^{k} \right)^{n \twid{\beta}_{n}}.
\end{align*}
Note that we can apply Thm. \ref{thm:main} to extract the $n \rightarrow \infty$ asymptotics of $\twid{\beta}_{n}$: because $\twid{\ggen} \in \formal{\mathbb{Z}}{z}$, by Prop. \ref{prop:growth}
\begin{align*}
	\twid{\mathtt{R}} := \min \left \{\mathtt{R}_{\twid{\ggen}} , \mathtt{R}_{1/\twid{\ggen}} \right \} \leq 1.
\end{align*}

Now, because $\twid{\ggen}$ is just given by the composition of $\ggen$ with a rescaling of the variable $z$, we can translate the asymptotic results which express the asymptotics of $\twid{\beta}_{n}$ in terms of asymptotics of the $\beta_{n}$; indeed, it is easy to see that:\footnote{As an aside: note that because $\twid{\mathtt{R}} \leq 1$ then we have the bound $\mathtt{R} \leq |a|$; hence, for any algebraic series $\ggen$ with constant coefficient 1 and Eisenstein constant $n \in \mathbb{Z}_{>0}$, then $\log(\ggen)$ must have radius of convergence $\leq n$.}
\begin{align}
	\mathrm{sing}_{0}(\twid{\ggen}) = \{a^{-1} \rho: \rho \in \mathrm{sing}_{0}(\ggen)\};
	\label{eq:sing_twid_to_untwid}
\end{align}
in particular,
\begin{align}
	\twid{\mathtt{R}} = |a|^{-1} \mathtt{R};
	\label{eq:R_twid_to_untwid}
\end{align}
and from \ref{eq:beta_from_ggen} along with the fact that $[z^{n}] \log \left(\twid{G} \right) = a^{n} [z^{n}] \log \left(G \right)$, it is immediate that
\begin{align}
	\twid{\beta}_{n} &= a^{n} \beta_{n}.
	\label{eq:beta_twid_to_untwid}
\end{align}

\begin{corollary}
	Theorem \ref{thm:main} holds for any algebraic $\ggen \in \formal{\mathbb{Q}}{z}$ (without the condition $\mathtt{R} \leq 1$); moreover, Corollaries \ref{cor:multi} and \ref{cor:single} hold for any algebraic $\ggen \in \formal{\mathbb{Q}}{z}$.
\end{corollary}
\begin{proof}
	As mentioned, Thm. \ref{thm:main} holds with all  $\ggen,\, \beta,\,$ and $\mathtt{R}$ replaced with their twiddled-counterparts $\twid{\ggen},\, \twid{\beta},\,$ and $\twid{\mathtt{R}}$.  With \eqref{eq:sing_twid_to_untwid} - \eqref{eq:beta_twid_to_untwid}, one can check that the expressions \eqref{eq:beta_asymp_nozeropole} and \eqref{eq:beta_asymp_zeropole} remain valid when passing back to the un-twiddled quantities.
\end{proof}

\appendix

\section{Some Generalities and Recollections on Quiver Representations} \label{app:quiv_rep}


\subsection{Slope-stability and (Semi)-Stable Moduli} \label{app:slope_stab}
We recall the basics of quiver representations and associated moduli spaces. Most of this material is drawn from \cite{reineke:quiv_rep} and \cite{king:quiv_rep}.  Let $Q = (Q_{0}, Q_{1})$ be a quiver specified by a finite set of vertices $Q_{0}$ and a finite set of arrows $Q_{1}$. In the following sections we will progressively specialize to the cases where $Q$ is acyclic (possesses no oriented cycles) and where $Q$ is the Kronecker $m$-quiver.

\begin{definition}\
	\begin{enumerate}
 \item A \textit{representation} $V = \left( (V_{i})_{i \in Q_{0}}, (V_{\alpha})_{\alpha \in Q_{1}} \right)$ of $Q$ is a collection of finite-dimensional complex vector spaces $V_{i},\, i \in Q_{0}$ indexed by vertices, and a collection of $\mathbb{C}$-linear homomorphisms $V_{\alpha}: V_{i} \rightarrow V_{j}$ indexed by arrows $\alpha: i \rightarrow j \in Q_{1}$.  
 
 \item A \textit{subrepresentation} $W$ of $V$ is a representation $\left( (W_{i})_{i \in Q_{0}}, (W_{\alpha})_{\alpha \in Q_{1}} \right)$ such that $W_{i} \leq V_{i}$ for each $i \in Q_{0}$ and $W_{\alpha}$ is given by the respective restrictions of the $V_{\alpha}$ for each $\alpha \in Q_{1}$.
 
 \item A morphism $L: V \rightarrow W$ between representations of $Q$ is a collection of maps $\left(L_{i}: V_{i} \rightarrow W_{i}\right)_{i \in Q_{0}}$ satisfying the obvious commutative diagrams: $L_{j} V_{\alpha} = W_{\alpha} L_{i}$ for each arrow $\alpha: i \rightarrow j$.
 \end{enumerate}
\end{definition}

\begin{remark}[Remark/Definition]
	Representations of $Q$, as defined above, form an abelian category $\mathsf{Rep}_{\mathbb{C}}(Q)$.
\end{remark}

Define $\Lambda := \mathbb{Z} Q_{0}$ as the free abelian group generated by $Q_{0}$ and $\Lambda^{+} := \mathbb{Z}_{\geq 0} Q_{0}$ the set of possible ``dimension vectors".  For each representation $V$ we will denote its dimension vector via $\underline{\dim}(W) = \sum_{i \in Q_{0}} \dim_{\mathbb{C}}(V_{i}) i \in \Lambda^{+}$.  Now, fix a functional $\Theta \in \text{Hom}_{\mathbb{Z}}(\Lambda, \mathbb{Z})$ (specified by assigning each vertex an integer weight), then we define its corresponding \textit{slope function} $\mu_{\Theta}: \Lambda^{+} \rightarrow \mathbb{Q}$ as
\begin{align}
	\mu_{\Theta}(d) = \frac{\Theta(d)}{\dim(d)}
	\label{eq:quiv_slope}
\end{align}
where $\dim : \Lambda^{+} \rightarrow \mathbb{Z}$ is given by $\dim (d) := \sum_{i \in Q_{0}} d_{i}$ for $d = \sum_{i \in Q_{0}} d_{i} i$.

\begin{definition}\
	\begin{enumerate}
	\item A representation $V$ of $Q$ is $\Theta$-\textit{semistable} if for every non-zero subrepresentation $W$ of $V$ we have $\mu_{\Theta}(\underline{\dim}(W)) \leq \mu_{\Theta}(\underline{\dim}(V))$.
	\item 
	A representation $V$ of $Q$ is $\Theta$-\textit{stable} if for every non-zero proper subrepresentation $W$ of $V$ we have $\mu_{\Theta}(\underline{\dim}(W)) < \mu_{\Theta}(\underline{\dim}(V))$.
	
	\item A representation $V$ of is $\Theta$-\textit{polystable} if it can be decomposed as a direct sum of $\Theta$-stable representations $(V_{i})_{i \in I}$ such that $\mu_{\Theta}(\underline{\dim}(V_{i})) = \mu_{\Theta}(\underline{\dim}(V_{j}))$ for any $i,j \in I$.
	\end{enumerate}
\end{definition}	

Because any map between $\Theta$-semistable representations has kernel, cokernel and image another $\Theta$-stable representation, we have the following.\cite{king:quiv_rep}
\begin{remark}
	The full subcategory of $\Theta$-semistable representations in $\mathsf{Rep}_{\mathbb{C}}(Q)$ forms an abelian subcategory $\mathsf{Rep}^{\Theta}_{\mathbb{C}}(Q)$; the simple (semisimple) objects of this category are the $\Theta$-stable (polystable) representations.
\end{remark}

Any object $V$ in  $\mathsf{Rep}^{\Theta}_{\mathbb{C}}(Q)$ admits a composition series\footnote{Arising as a Jordan-H\"{o}lder filtration when thought of as a module of the quiver path algebra.}, i.e. there exists a sequence of subobjects 
\begin{align*}
	0 \subset V_{1} \subset V_{2} \subset \cdots \subset V_{n} = V
\end{align*}
such that $V_{i}/V_{i-1}$ is a simple object (stable representation) for $2 \leq i \leq n$; a composition series is unique up to permutation and isomorphism of composition factors.  Two objects are said to be \textit{S-equivalent} if they have equivalent composition series (that is same composition length and composition factors up to permutation and isomorphism); as its name suggests, the notion of $S$-equivalence defines an equivalence relation.

Fixing a dimension vector $d \in \mathbb{Z}_{>0}Q_{0}$, one can then ask about the set of $S$-equivalence classes of semistable representations.   By a theorem of King \cite{king:quiv_rep}, using the stability condition $\Theta$, this set can be given the structure of:
\begin{enumerate}
	\item A (possibly singular) complex variety $\CM_{\sst}^{Q}(d;\Theta)$ by taking a projective GIT quotient of the affine space $\mathfrak{R}_{Q} := \bigoplus_{i \rightarrow j \in Q_{1}} \mathrm{Hom} \left(\mathbb{C}^{d_{i}}, \mathbb{C}^{d_{j}} \right)$: indeed, $\mathfrak{R}_{Q}$ is equipped with an action of $G = \mathbb{P} \left(\prod_{i \in Q_{0}} \GL_{d_{i}} \mathbb{C} \right)$ induced by conjugation; for any character $\phi: G \rightarrow \mathbb{C}^{\times}$, let $\mathbb{C}[\mathfrak{R}_{Q}]^{G,\phi}$ denote the ring of functions on $\mathfrak{R}_{Q}$ that satisfy $f(g \cdot x) = \phi(x) f(g)$ for any $g \in G$ and $x \in \mathfrak{R}_{Q}$; then, we define
	\begin{align*}
		\CM_{\sst}^{Q}(d;\Theta) := \mathrm{Proj} \left( \bigoplus_{n \geq 0} \mathbb{C}[\mathfrak{R}_{Q}]^{G,\chi_{\Theta}^{n}} \right)
	\end{align*}
	where $\chi_{\Theta}$ is the character defined by the descent of the character 
	\begin{align*}
		\prod_{i \in Q_{0}} \GL_{d_{i}} \mathbb{C} &\longrightarrow \mathbb{C}^{\times}\\
		(g_{i})_{i \in Q_{0}}	 &\longmapsto \prod_{i \in Q_{0}} \det(g_{i})^{\Theta(d) - \dim(d) \Theta(i)}
	\end{align*}		
	to $G$.

	\item A smooth K\"{a}hler orbifold by taking a K\"{a}hler quotient.
\end{enumerate}

When $Q$ is acyclic, the GIT quotient has a classical algebro-geometric interpretation.

\begin{numrmk} \label{rmk:projective}
	If $Q$ is acyclic, then $\mathbb{C}[\mathfrak{R}_{Q}]^{G} = \mathbb{C}$; so it follows that $\CM_{\sst}^{Q}(d;\Theta)$ is a (possibly singular) complex projective variety.
\end{numrmk}

\begin{numrmk} \label{rmk:GIT_vs_Kahler}
We do not describe the K\"{a}hler construction explicitly; however, in the case that $Q$ is an acyclic quiver (so that the GIT quotient is a projective variety); the analytificiation of the GIT construction is biholomorphic to the K\"{a}hler construction away from any singular points of the K\"{a}hler quotient; moreover, both spaces are always homeomorphic as topological spaces.
\end{numrmk}

\begin{remark}
	Clearly every S-equivalence class contains a polystable representative, unique up to isomorphism. Hence, the space $\CM_{\sst}^{Q}(d;\Theta)$ parametrizes polystable representations of dimension vector $d$ up to isomorphism.
\end{remark}

The set of \textit{stable} representations up to isomorphism has the structure of a \textit{smooth} open subvariety $\CM_{\st}^{Q}(d;\Theta)$ of $\CM^{Q}_{\sst}(d; \Theta)$.  These spaces are indeed ``moduli spaces" as they are solutions to a moduli problem given by a functor from $\mathbb{C}$-schemes to \textsf{Set}.

\begin{numrmk} \label{rmk:moduli}
$\CM_{\sst}^{Q}(d;\Theta)$ is a coarse moduli space for families of $\Theta$-semistable modules (parametrized by $\mathbb{C}$-schemes) of dimension vector $d$ (up to $S$-equivalence). When $d$ is a primitive dimension vector then $\CM_{\st}^{Q}(d;\Theta) = \CM^{Q}_{\sst}(d; \Theta)$ and, moreover, $\CM_{\st}^{Q}(d;\Theta)$ is equipped with a universal family making it into a fine moduli space.\footnote{Some readers would correctly argue that the failure to always be a fine moduli space for non-primitive dimension vectors is an indication we should be working with moduli stacks; indeed,one always has a fine moduli (Artin) stack (even for non-primitive dimension vectors).  However, we do not need to explicitly use the language of stacks in our discussion.}
\end{numrmk}

\begin{definition}
We define two equivalence relations on stability conditions in the following manner:
	\begin{enumerate}
		\item Given a fixed dimension vector $d$, two stability conditions $\Theta_{1}$ and $\Theta_{2}$ are \textit{$d$-equivalent} if a representation of dimension vector $d$ is $\Theta_{1}$-(semi)stable if and only if it is $\Theta_{2}$-(semi)stable.
		
		\item Two stability conditions are \textit{stability-equivalent} if they are $d$-equivalent for all dimension vectors $d$.
	\end{enumerate}
\end{definition}

By definition, if $\Theta_{1}$ and $\Theta_{2}$ are $d$-equivalent, then the moduli spaces $\CM_{\sst}^{Q}(d;\Theta_{1})$ and $\CM_{\sst}^{Q}(d;\Theta_{2})$ are the same as subsets of ($S$-equivalence classes of) representations.  But Rmk.~\ref{rmk:moduli} provides a stronger equivalence: $\CM_{\sst}^{Q}(d;\Theta_{1})$ and $\CM_{\sst}^{Q}(d;\Theta_{2})$ are coarse-moduli spaces for the same moduli problem; hence, they must be isomorphic as varieties.  It follows that, at the level of constructing moduli spaces (associated to a fixed dimension vector $d$0, one need only be concerned with ($d-$)equivalence classes of stability conditions.

\begin{numrmk} \label{rmk:stability_equivs} \
There are two operations on functionals $\Theta \in \text{Hom}_{\mathbb{Z}}(\Lambda, \mathbb{Z})$ that do not change the stability-equivalence classes of (semi)stable representations:
	\begin{itemize}
		\item Rescaling: $\Theta \mapsto n \Theta$ for $n \in \mathbb{Z}_{>0}$
	
		\item Translating by $\dim$: $\Theta \mapsto \Theta + n \dim$ for $n \in \mathbb{Z}$.
	\end{itemize}
\end{numrmk}

We can generalize the notion of stability to allow $\Theta$ to be a real-valued functional, i.e. we can choose $\Theta \in \text{Hom}_{\mathbb{Z}}(\Lambda, \mathbb{R})$.  The corresponding slope function $\mu_{\Theta}$ is then a real-valued map $\Lambda^{+} \rightarrow \mathbb{R}$ and the definition of (semi)stability remains the same with respect to this slope function.  The space $\CM_{\sst}^{Q}(d;\Theta)$ can then be defined as a solution to the moduli problem of Rmk~\ref{rmk:moduli}; if such a solution exists, then $\CM_{\sst}^{Q}(d;\Theta)$ is unique up to unique isomorphism.  For existence, we may try to construct $\CM_{\sst}^{Q}(d;\Theta)$ via a GIT quotient; however, this requires $\mathbb{Z}$-valued $\Theta$.  For the $m$-Kronecker quiver, one need not worry about this condition: section \ref{app:stab_m_kron} below shows that every real-valued stability functional on the $m$-Kronecker quiver is stability-equivalent to an integer-valued stability functional.  However, for those concerned with generalities, the following shows that we can always replace an $\mathbb{R}$-valued functional with a $d$-equivalent $\mathbb{Z}$-valued functional; hence, we may always use GIT to construct $\CM_{\sst}^{Q}(d,\Theta)$ as a projective variety.

\begin{lemma} \label{lem:integer_replacement}
	Let $d \in \Lambda^{+}$ be a fixed dimension vector and $\Theta \in \mathrm{Hom}_{\mathbb{Z}}(\Lambda, \mathbb{R})$, sufficiently generic in the sense that
	\begin{align*}
		\Theta(k) = \frac{\Theta(d)}{\dim(d)} \dim(k)
	\end{align*}	
	if and only if $d = n k$ for some $n \in \mathbb{Z}_{>0}$.  Then $\Theta$ is $d$-equivalent to a stability functional in $\mathrm{Hom}_{\mathbb{Z}}(\Lambda, \mathbb{Z})$.
\end{lemma}
\begin{proof}
	Fix any $\epsilon > 0$, then we can find $\phi_v \in \mathbb{Q}$ such that
	\begin{align*}
		|\Theta(q) - \phi_{q} | < \epsilon
	\end{align*}
	for all $q \in Q_{0}$.  Define
	\begin{align*}
		\Phi = \sum_{q \in Q_{0}} \phi_{q} q^{*} \in \mathrm{Hom}_{\mathbb{Z}}(\Lambda, \mathbb{Q});
	\end{align*}	
	then,
	\begin{align*}
		\left|\mu_{\Theta}(k) - \mu_{\Phi'}(k) \right| < \epsilon.
	\end{align*}
	for any $k \in \Lambda$.  Now equip $\Lambda^{+}$ with the poset structure defined by: $k \leq d$ if and only if $k_{q} \leq d_{q}$ for all $q \in Q_{0}$.  Clearly, for a fixed $d$, there are only finitely many $k$ such that $k \leq d$; hence, if we choose $\epsilon > 0$ such that
	\begin{align*}
		\epsilon < \frac{1}{2} \min \left\{ \left|\mu_{\Theta}(k)  - \mu_{\Theta}(d) \right|: \text{$k \leq d$ and $d \neq n k$ for any $n \in \mathbb{Z}_{>0}$}\right\},
	\end{align*}	
	then $\Theta$ and $\Phi$ are $d$-equivalent.  Letting $\lambda$ denote the largest denominator of the $\phi_{q},\, q \in Q_{0}$, we have $\lambda \Phi \in \mathrm{Hom}_{\mathbb{Z}}(\Lambda, \mathbb{Z})$ and, moreover, $\lambda \Phi$ is stability-equivalent to $\Phi$. 
\end{proof}

\subsection{Slope-Stability Conditions on the \texorpdfstring{$m$-Kronecker}{m-Kronecker} Quiver} \label{app:stab_m_kron}
Recall the $m$-Kronecker quiver is defined as the quiver $Q=K_{m}$ with two vertices $\{q_{1},q_{2}\}$ and $m \geq 1$ parallel arrows $\{\alpha_{l}: q_{2} \rightarrow q_{1} \}_{l=1}^{m}$ from $q_{2}$ to $q_{1}$.

\begin{center}
\begin{tikzpicture}
\tikzstyle{block} = [rectangle, draw=blue, thick, fill=blue!10,
text width=16em, text centered, rounded corners, minimum height=2em]

\node at (1.2, 0)[circle,draw=blue,very thick] (second) {$q_{2}$};

\node at (-1.2, 0)[circle,draw=blue,very thick] (first) {$q_{1}$};

\node at (0,0.1) {{\Huge \vdots}};

\draw [decorate,decoration={brace, amplitude=6pt, mirror},xshift=0pt,yshift=-1pt,red,thick]
(-1.2,-0.9) -- (1.2,-0.9) node[black, midway,below,yshift = -6pt]{{\tiny $m$ arrows}};
\draw[-latex] (second) to  [bend right = 35] (first);
\draw[-latex] (second)  to  [bend right = 60] (first);
\draw[-latex] (second)  to  [bend right = -35] (first);
\draw[-latex] (second)  to  [bend right = -60] (first);
	\end{tikzpicture}
\end{center}

Via the operations in Remark \ref{rmk:stability_equivs}, any stability functional $\Theta$ on the $m$-Kronecker quiver (real or integer-valued) can be reduced to one of three possibilities: 
\begin{enumerate}
	\item $\Theta_{1} \equiv 0$;
	
	\item $\Theta_{2} = q_{1}^{*}:e q_{1} + d q_{2} \mapsto e$;
	
	\item $\Theta_{3} = q_{2}^{*}: e q_{1} + d q_{2} \mapsto d$;
\end{enumerate}
that induces the same slope-stability condition as $\Theta$.  As described in \cite[\S 5.1]{reineke:quiv_rep}, $\Theta_{3}$ is the only choice that leads to an interesting notion of (semi-)stable representations.  Indeed, consider the representations:

\begin{center}
\begin{tikzpicture}
\node at (-2.2,0)[] {$S_{1} =$};
\tikzstyle{block} = [rectangle, draw=blue, thick, fill=blue!10,
text width=16em, text centered, rounded corners, minimum height=2em]

\node at (1.2, 0)[circle,draw=blue,very thick] (second) {$0$};

\node at (-1.2, 0)[circle,draw=blue,very thick] (first) {$\mathbb{C}$};

\node at (0,0.1) {${\Huge \vdots}$};

\draw[-latex] (second) to  [bend right = 35] (first);
\draw[-latex] (second)  to  [bend right = 60] (first);
\draw[-latex] (second)  to  [bend right = -35] (first);
\draw[-latex] (second)  to  [bend right = -60] (first);
\end{tikzpicture}

\begin{tikzpicture}
\tikzstyle{block} = [rectangle, draw=blue, thick, fill=blue!10,
text width=16em, text centered, rounded corners, minimum height=2em]

\node at (-2.2,0)[] {$S_{2} =$};

\node at (1.2, 0)[circle,draw=blue,very thick] (second) {$\mathbb{C}$};

\node at (-1.2, 0)[circle,draw=blue,very thick] (first) {$0$};

\node at (0,0.1) {${\Huge \vdots}$};

\draw[-latex] (second) to  [bend right = 35] (first);
\draw[-latex] (second)  to  [bend right = 60] (first) ;
\draw[-latex] (second)  to  [bend right = -35] (first);
\draw[-latex] (second)  to  [bend right = -60] (first);
	\end{tikzpicture}
\end{center}
For all three stability conditions, $S_{1}$ and $S_{2}$ are stable representations; for $\Theta_1$ and $\Theta_{2}$, these are the \textit{only} stable representations.  For $\Theta_1$, all representations are semi-stable; for $\Theta_2$, the only semi-stable representations are isomorphic to $S_{1}^{\oplus r}$ or $S_{2}^{\oplus r}$ for some $r>0$.  We will return to the classification of semi-stable representations for $\Theta_{3}$ in \S \ref{app:dense_arc}.

\begin{definition}[Terminology]
	A stability condition on an $m$-Kronecker quiver that is stability-equivalent to $\Theta_{3}$ is a \textit{wild} stability condition.
\end{definition}

\subsection{Stability Conditions on \texorpdfstring{$\mathsf{Rep}_{\mathbb{C}}(Q)$}{Rep_C(Q)}} \label{app:bridgeland_to_slope}
Perhaps closer to the spirit of the notion of stability that appears in the context of physics are stability conditions on the abelian category\footnote{
Even more natural are Bridgeland stability conditions on the triangulated category $D^{b}\mathsf{Rep}_{\mathbb{C}}(Q)$; however, we will only be concerned with the subspace of stability conditions after choosing the $t$-structure whose heart is the abelian subcategory $\mathsf{Rep}_{\mathbb{C}}(Q) \hookrightarrow D^{b}\mathsf{Rep}_{\mathbb{C}}(Q)$ via its obvious embedding into chain complexes supported in the zeroth degree.} $\mathsf{Rep}_{\mathbb{C}}(Q)$.  First, we recall the notion of a stability condition on an abelian category \cite{king:quiv_rep, bridgeland}.

\begin{definition}\
 Let $\mathsf{A}$ be an abelian category and $K_{0}(\mathsf{A})$ its Grothendieck group. A stability function on $\mathsf{A}$ is a group homomorphism $Z: K_{0}(\mathsf{A}) \rightarrow \mathbb{C}$ such that for all $0 \neq E \in \text{object}(\mathsf{A})$,
	\begin{align}
		Z(E) \in \mathbb{H} \cup \mathbb{R}_{<0} = \{r e^{i \theta}:\text{ $r \in \mathbb{R}_{>0}$ and $\theta \in (0,\pi]$}\} \subset \mathbb{C}
		\label{eq:stab_func_def}
	\end{align}
\end{definition} 
The appropriate notion of (semi-)stability is then extracted from the ``phase" $\arg \left[ Z \right]$ thought of as a function $\mathrm{ob}(\mathsf{A}) \rightarrow (0,\pi]$.

\begin{definition}\
 An object $E \in \text{ob}\left(\mathsf{A} \right)$ is $Z$-\textit{semistable} if for every non-zero proper subobject $F$ of $E$ we have 
	\begin{align}	
		\arg \left[ Z(F) \right] \leq \arg \left[ Z(E) \right].
		\label{eq:Z_semistable}
	\end{align}
  $E$ is $Z$-\textit{stable} if strict equality holds in \eqref{eq:Z_semistable} for all proper subobjects $F$.
\end{definition}

Our interest lies in the case $\mathsf{A} = \mathsf{Rep}_{\mathbb{C}}(Q)$.  For simplicity we will take $Q$ to be an \textit{acyclic} quiver, then the map that takes each representation to its dimension vector, $\underline{\dim}: K_{0}(\mathsf{Rep}_{\mathbb{C}}(Q)) \rightarrow  \mathbb{Z} Q_{0} =: \Lambda$, is an isomorphism.  Indeed, for $Q$ acyclic, $K_{0}(\mathsf{Rep}_{\mathbb{C}}(Q))$ is freely generated by the (isomorphism classes of) one-dimensional representations supported at a single vertex, and the zero map associated to each arrow.

We can equivalently state the notion of stability in terms of an associated slope-function $\mu_{Z}: \Lambda^{+} \rightarrow \mathbb{R} \cup \{\infty\}$ defined as
\begin{align}
	\mu_{Z}(d) &= -\frac{\text{Re}\left[Z(d)\right]}{\text{Im}\left[Z(d) \right]}\\
	&= - \cot \left\{\arg \left[ Z(d) \right] \right\}.
	\label{eq:Z_slope}
\end{align}
Then $\mu_{Z}(F) \leq \mu_{Z}(E)$ if and only if $\arg \left[ Z(F) \right] \leq \arg \left[ Z(E) \right]$ (where $\arg[Z]$ is thought of as a function taking values in $(0,\pi]$).

\begin{numrmk} \label{rmk:stab_stdform}
	Note that, if we are given $\Theta \in \mathrm{Hom}(\Lambda, \mathbb{Z})$, then we can define $Z \in \mathrm{Hom}(\Lambda, \mathbb{C})$ via
	\begin{align}
		Z &= - \Theta + i \dim
		\label{eq:Z_normal_form}
	\end{align}
	because $\dim$ is positive, then the image of $Z$ lies in $\mathbb{Z} + \mathbb{Z}_{>0} i \subset \mathbb{H}$; so this is a valid stability condition.  In this case, the function \eqref{eq:Z_slope} is precisely the slope function $\mu_{\Theta}$ defined in \eqref{eq:quiv_slope}.
\end{numrmk}

In the context of the $m$-Kronecker quiver, the following proposition shows that a stability condition on $\mathsf{Rep}_{\mathbb{C}}(K_{m})$ is equivalent to one of the three slope-stability conditions in \S \ref{app:stab_m_kron}.

\begin{proposition} \label{prop:kron_stdform}
	Let $K_{m}$ denote the $m$-Kronecker quiver (as discussed in \S \ref{app:stab_m_kron}) with vertices $q_{1}$ and $q_{2}$ and $m$ arrows from $q_{2}$ to $q_{1}$; $\Lambda := \mathbb{Z} \langle q_{1}, q_{2} \rangle \cong K_{0} \left(\mathsf{Rep}_{\mathbb{C}} (K_{m}) \right)$.  If $Z: \Lambda \rightarrow \mathbb{C}$ is a stability function on $K_{m}$ sufficiently generic in the sense that $Z(q_{1})$ and $Z(q_{2})$ are contained in $\mathbb{H}$, then
	\begin{enumerate}
		\item If $\arg[Z(q_{1})] = \arg[Z(q_{2})]$ then $Z$ induces a stability condition stability-equivalent to the slope-stability condition induced by $\Theta_{1}$;
		
		\item If $\arg[Z(q_{1})] > \arg[Z(q_{2})]$ then $Z$ induces a stability condition stability-equivalent to the slope-stability condition induced by $\Theta_{2}$;
		
		\item If $\arg[Z(q_{1})] < \arg[Z(q_{2})]$ then $Z$ induces a stability condition stability-equivalent to the slope-stability condition induced by $\Theta_{3}$.
	\end{enumerate}
\end{proposition}
\begin{proof}
	The proof of this statement is an exercise in linear algebra and relies on the fact that the Kronecker $m$-quiver only has two vertices.  We begin by showing that $Z$ induces the same notion of stability as a stability condition of the form \eqref{eq:Z_normal_form}.  By taking the real and imaginary parts of $Z$ we get two group homomorphisms from $\Lambda$ to $\mathbb{R}$.  Indeed,
		\begin{align*}
			\mathrm{Im}(Z) = a q_{1}^{*} + b q_{2}^{*} \in \mathrm{Hom}_{\mathbb{Z}}(\Lambda, \mathbb{R})
		\end{align*}
	for some $a,b \in \mathbb{R}_{>0}$. We claim that there is a positive-rescaling of the lattice $\Lambda$, thought of as embedded inside of $\Lambda_{\mathbb{R}} := \Lambda \otimes_{\mathbb{Z}} \mathbb{R}$, such that the induced pullback on $Z$ takes the form \eqref{eq:Z_normal_form}.  Indeed, let $D \in \Aut_{\mathbb{R}} \left[ \Lambda \otimes_{\mathbb{Z}} \mathbb{R} \right]$ given by
	\begin{align*}
		D = a q_{1} \otimes q_{1}^{*} + b q_{2} \otimes q_{2}^{*}
	\end{align*}
	then it follows that
	\begin{align*}
		\mathrm{Im}(Z) =\dim \circ D
	\end{align*}
	Hence, defining
	\begin{align*}
		Z' :=  \mathrm{Re}(Z) \circ D^{-1}  + i \dim \in \mathrm{Hom}_{\mathbb{Z}}(\Lambda,\mathbb{C})
	\end{align*}
	we may write
	\begin{align*}
		Z = Z' \circ D.
	\end{align*}
	Now we claim that $Z'$ induces the same stability condition as $Z$.  To see this, we think of $Z$ and $Z'$ as valued in the real two-dimensional vector space $\mathbb{R} \oplus i \mathbb{R}$; then
	\begin{align*}
		Z' \circ D = \twid{D} Z'
	\end{align*}
	where $\twid{D}$ is a (real) automorphism of the image $I = \mathrm{span}_{\mathbb{R}} \langle Z'(q_{1}), Z'(q_{2}) \rangle \leq \mathbb{R} \oplus i \mathbb{R}$ of $Z'$, defined via
	\begin{align*}
		\twid{D} = a Z'(q_{1}) \otimes Z'(q_{1})^{*} + b Z'(q_{2}) \otimes Z'(q_{2})^{*} \in \mathrm{Hom}_{\mathbb{R}}(I,I).
	\end{align*}
	But $\twid{D}$ is an orientation-preserving map: when $I$ is one-dimensional, $\twid{D}$ is a positive rescaling of $I$; when $I$ is two-dimensional, $\twid{D}$ is a diagonal matrix with positive entries when expressed in terms of the basis $\{Z'(q_{1}), Z'(q_{2)} \}$; hence, $Z = \twid{D} Z'$ and $Z'$ induce the same notion of stability.
	
	Next, by Remark \ref{rmk:stab_stdform}, we have that $Z$ induces the same notion of stability as the slope-stability defined by the linear functional, $\Theta = -\mathrm{Re}(Z) \circ D^{-1}$; the statement of the proposition holds by taking $\Theta$ into one of the standard forms of \S \ref{app:stab_m_kron} via Remark~\ref{rmk:stability_equivs}.
\end{proof}

\subsection{The BPS Quiver} \label{app:BPS_quiv}

Let us put the above technology into the context of four-dimensional $\mathcal{N}=2$ field theories admitting a Coulomb branch $\CB$---a more complete review can be found in \cite{alim_cecotti_BPS_quiv_1,alim_cecotti_BPS_quiv_2}.  First, begin by fixing a point $u \in \CB$; let $\widehat{\Gamma}_{u}$ be the charge lattice at the point $u$, $\langle \cdot, \cdot \rangle_{u}: \widehat{\Gamma}_{u}^{\otimes 2} \rightarrow \mathbb{C}$ the associated anti-symmetric pairing, and $\mathcal{Z}_{u}: \widehat{\Gamma}_{u} \rightarrow \mathbb{C}$ the central charge function.

For notational convenience, first define the subspace of \textit{occupied} BPS charges
\begin{align*}
	\mathtt{Occ} := \{\gamma \in \widehat{\Gamma}_{u}: \text{There exists a BPS state of charge $\gamma$} \}.
\end{align*}

In order to define a BPS quiver, we require the existence of $\vartheta \in \left(\mathbb{R}/\mathbb{Z} \right) \backslash \arg \left[\mathcal{Z}_{u}(\mathtt{Occ}) \right]$ such that:

\begin{enumerate}
	\labitem{(C1)}{list:quivC1} 	There exists a finite set of linearly-independent charges $\mathpzc{b} \subset \widehat{\Gamma}_{u}$ such that $\mathcal{Z}_{u}^{-1} \left(e^{i \vartheta} \mathbb{H} \right) \subset \mathbb{Z}_{\geq 0} \mathpzc{b}$, i.e. any charge $\gamma \in \mathtt{Occ}$ such that $Z(\gamma) \in e^{i \vartheta} \mathbb{H}$ is a positive linear combination of elements of $\mathpzc{b}$;
		
	\labitem{(C2)}{list:quivC2} Each $\gamma \in \mathpzc{b}$ is the charge of a hypermultiplet.
\end{enumerate}

Condition \ref{list:quivC1} is a purely linear/convex algebraic question about the set $\mathtt{Occ}$.  It is instructive to construct a hypothetical counterexample of a set $\mathtt{Occ} \subset \mathbb{Z}^{2}$ and a map $Z: \mathbb{Z}^{2} \rightarrow \mathbb{C}$ such that the $\mathpzc{b}$ of condition \ref{list:quivC1} does not exist for any choice of $\vartheta \in S^{1} \backslash \arg \left[\mathcal{Z}_{u}(\mathtt{Occ}) \right]$.

\begin{numrmk} \label{rmk:unique_basis}
	If there exists a $\vartheta$ satisfying $\ref{list:quivC1}$, then the corresponding $\mathpzc{b}$ is unique.  This follows from the fact that any invertible matrix with non-negative integer entries whose inverse also has non-negative entries must necessarily be a permutation matrix.
\end{numrmk}
Thus, if such a $\mathpzc{b}$ exists for a given element of $\left(\mathbb{R}/(2\pi \mathbb{Z}) \right) \backslash \arg \left[\mathcal{Z}_{u}(\mathtt{Occ}) \right]$, one knows that it is unique and need only check that \ref{list:quivC2} holds.

Assuming $\vartheta$ satisfies \ref{list:quivC1} and \ref{list:quivC2}, then the \textit{BPS quiver} $Q = (Q_{0}, Q_{1})$ at the point $u \in \CB$ (and associated to $\vartheta$) is the quiver with:
\begin{enumerate}
	\item vertices $Q_{0} = \mathpzc{b}$;
	
	\item $\langle \gamma_{1}, \gamma_{2} \rangle$ arrows from $\gamma_{2} \in \mathpzc{b}$ to $\gamma_{1} \in \mathpzc{b}$ if $\langle \gamma_{1}, \gamma_{2} \rangle > 0$.
\end{enumerate}.
Moreover, $Q$ is equipped with a stability condition determined via the central charge function:
	\begin{enumerate}
	\item The lattice $\Lambda = \mathbb{Z}Q_{0}$, equipped with the canonical basis spanned by the vertices, is identified with the charge (sub)lattice $\mathbb{Z} \mathpzc{b} \leq \widehat{\Gamma}_{u}$, equipped with the basis $\mathpzc{b}$;

	\item The stability function $Z: \Lambda \rightarrow \mathbb{C}$ defined in \eqref{eq:stab_func_def} is the central charge function $e^{-i \vartheta} \mathcal{Z}_{u}: \Gamma \cong \Lambda \rightarrow \mathbb{C}$.  
\end{enumerate}

A few remarks are in order.

\begin{remark}[Remarks]\
	\begin{enumerate}
	\item Note that (with Rmk.~\ref{rmk:unique_basis} in mind), the quiver $Q$ only depends on a choice of $\vartheta \in S^{1} \backslash \arg \left[\mathcal{Z}_{u}(\mathtt{Occ}) \right]$ satisfying \ref{list:quivC1} and \ref{list:quivC2}.  If there are several such $\vartheta$, then the corresponding quivers are not necessarily the same; however, it might be possible to relate them by a finite sequence of mutations \cite{alim_cecotti_BPS_quiv_2}.
	
	\item The existence and construction of a quiver $Q$ (when it exists) depends on a priori knowledge about properties of the full set $\mathtt{Occ}$ of BPS states.  However, it is not necessary to understand the full content of the 1-particle BPS subspaces associated to every element of $\mathtt{Occ}$: one needs only check \ref{list:quivC1} and \ref{list:quivC2}.  Once these are verified (or assumed via conjecture), then all information about the BPS spectrum for non-basis elements can be computed by looking at the corresponding quiver moduli spaces (see \S \ref{app:BPS_index_quiver}).  In particular, all of the BPS indices are given by computing DT invariants for the associated quiver. 
	\end{enumerate}
	
\end{remark}

\begin{definition}[Terminology]
	Let $\CB$ be the Coulomb branch for a four-dimensional $\mathcal{N}=2$ theory.  We will say a point $u \in \CB$ is an \textit{$m$-wild point} ($m \geq 3$), if 
	\begin{enumerate}	
		\item there exists a BPS quiver $Q$, associated to some $\vartheta \in \mathbb{R}/(2 \pi \mathbb{Z})$, such that $Q$ contains the $m$-Kronecker quiver as a full subquiver $K$.

		\item the central charge function $\mathcal{Z}_{u}: \Gamma \rightarrow \mathbb{C}$ determines a wild stability condition on $K$: suppose that the nodes of $K$ are labelled by $\gamma_{1}$ and $\gamma_{2}$ with $m$ arrows from $\gamma_{2}$ to $\gamma_{1}$, then $\mathcal{Z}$ equips $K$ with a stability condition stability-equivalent to $\Theta_{3}$ if and only if
	\begin{align}
		 \arg \left(e^{-i \vartheta} \mathcal{Z}_{\gamma_{1}} \right) <   \arg  \left(e^{-i \vartheta} \mathcal{Z}_{\gamma_{2}} \right)
		 \label{eq:cond_wild_centralcharge}
	\end{align}
	as elements of $(0,\pi)$. 
	\end{enumerate}
\end{definition}

\begin{remark}
\eqref{eq:cond_wild_centralcharge} is satisfied if
 	\begin{align*}
		 \arg \left( e^{-i \theta} \mathcal{Z}_{\gamma_{1}} \right) <  \arg  \left(e^{-i \theta} \mathcal{Z}_{\gamma_{2}} \right)
	\end{align*}
	for \textit{any} $\theta \in \mathbb{R}/\mathbb{Z}$ such that $e^{-i \theta} \mathcal{Z}_{\gamma_{1}}$ and $e^{-i \theta} \mathcal{Z}_{\gamma_{2}}$ lie in $\mathbb{H}$.  Thus, one can check if a point is $m$-wild only knowing the existence of an $m$-Kronecker subquiver of a BPS quiver $Q$, without knowing the precise (possibly family of) $\vartheta$ required to produce $Q$.
\end{remark}

\subsection{BPS States and Quiver Moduli} \label{app:BPS_index_quiver}
In order to define the BPS index, we first fix a rest-frame (a spacelike slice) and then focus our attention on a subspace of the 1-particle BPS states at rest:	$\mathcal{H}^{\mathrm{rest}}_{\mathrm{BPS}}(u)$; this space is naturally graded by the charge lattice at $u$:
\begin{align*}
	\mathcal{H}^{\mathrm{rest}}_{\mathrm{BPS}}(u) = \bigoplus_{\gamma \in \widehat{\Gamma}_{u}} \mathcal{H}^{\mathrm{rest}}_{\mathrm{BPS}}(\gamma;u).
\end{align*}
Furthermore, each direct summand is a finite dimensional representation of the massive little-algebra, the even part of which contains the spatial rotation subalgebra $\mathfrak{r} \cong \mathfrak{so}(3)$.  It is a fact (see \cite[\S 4.2.3]{moore:felix} and \cite[\S II]{wess-bagger}) that there is always a factorization as little-algebra representations:
\begin{align*}
	\mathcal{H}^{\mathrm{rest}}_{\mathrm{BPS}}(\gamma;u) \cong \rho_{\mathrm{hh}} \otimes \mathfrak{h}(\gamma;u)
\end{align*}
where $\rho_{\mathrm{hh}}$ is the half-hypermultiplet representation.  Using this factorization, the BPS index $\Omega(\gamma;u)$ can be written as
\begin{align}
	\Omega(\gamma;u) &= \Tr_{\mathfrak{h}(\gamma;u)} (-1)^{2J_{3}}
	\label{eq:BPS_superdim}
\end{align}
where, $J_{3}$ is a generator of a Cartan subalgebra of $\mathfrak{r}$.

Now, assume that $Q$ is an acyclic (full) subquiver of a BPS quiver at the point $u \in \CB$.  In order to define the moduli space(s) of (semi)stable representations with respect to the stability condition $Z: \mathbb{Z}Q_{0} \rightarrow \mathbb{C}$ (determined from the central charge---see \S \ref{app:BPS_quiv}), we replace $Z$ with a slope-stability condition.  If $Q$ is the Kronecker quiver, then Prop.~\ref{prop:kron_stdform} automatically provides us with a slope-stability condition that provides an equivalent notion of stability as $Z$.  For more general acyclic quivers this can be done in the following manner: 
fix $\gamma \in \mathbb{Z}_{>0} Q_{0}$ and define $\theta_{\gamma,Z} \in \mathrm{Hom}_{\mathbb{R}} \left(\mathbb{Z} Q_{0}, \mathbb{R} \right)$ by
\begin{align*}
	\theta_{\gamma,Z}: k \mapsto \mathrm{Im}\left( \frac{Z(k)}{Z(\gamma)} \right);
\end{align*}
then one can verify that a representation $V$ of dimension $\underline{\dim}(V) = \gamma$ is $\theta_{\gamma,Z}$-(semi)stable if and only if it is $Z$-(semi)stable.  Hence, we may define  
\begin{align*}
	\begin{array}{lcc}
	\CM_{\sst}^{Q}(\gamma;Z) & := &\CM_{\sst}^{Q}(\gamma;\theta_{\gamma,Z}) \\
	{} & {} & \rotatebox{90}{$\subseteq$}\\
	\CM_{\st}^{Q}(\gamma;Z) & := & \CM_{\st}^{Q}(\gamma;\theta_{\gamma,Z})
	\end{array}
\end{align*}
where $\CM_{\sst}^{Q}(\gamma;\theta_{\gamma,Z})$ is constructed via GIT quotient.

Now BPS indices are conjectured to be computed via quiver DT invariants.
\begin{enumerate}
	\item Generalized DT invariants in the sense of Kontesevich and Soibelman: A specialization of the motivic DT invariants of $D^{b} \mathsf{Rep}_{\mathbb{C}}(Q)$ equipped with the Bridgeland stability condition defined by the heart $\mathsf{Rep}_{\mathbb{C}}(Q) \hookrightarrow \mathbb{C}$ and $Z: K_{0}(\mathsf{Rep}_{\mathbb{C}}(Q)) \cong \mathbb{Z}Q_{0} \rightarrow \mathbb{C}$.
	
	\item DT invariants in the sense of Joyce-Song-Behrend: As weighted Euler characteristics of a suitable moduli space/stack of semistable representations.
\end{enumerate}

  Via Rmk.~\ref{rmk:projective}, $\CM_{\sst}^{Q}(\gamma;Z)$ is a projective variety and so its analytification is equipped with the pullback of the Fubini-Study metric given a choice of embedding into a projective space.

If $\gamma \in \mathbb{Z}_{\geq 0} Q_{0} \leq \widehat{\Gamma}_{u}$ is a primitive dimension vector, then $\CM_{\st}^{Q}(\gamma;Z) = \CM_{\sst}^{Q}(\gamma;Z)$ is a smooth projective variety; hence, as an analytic space $\CM_{\st}^{Q}(\gamma;Z)$ is a K\"{a}hler manifold.  Now, as shown in \cite{denef:qqhh}, the space $\mathfrak{h}(\gamma;u)$ is precisely the space of BPS states for supersymmetric quantum mechanics of a point particle travelling on the configuration space $\CM_{\st}^{Q}(\gamma;Z)$ (c.f. \cite[\S 10.4.3]{mirror_sym_1}).  In this situation, the $\mathfrak{r} \cong \mathfrak{su}(2)$-representation $\mathfrak{h}(\gamma;u)$ can be identified with the ring of harmonic forms on $\CM_{\st}^{Q}(\gamma;Z)$, equipped with the Lefschetz $\mathfrak{su}(2)$-action \cite[\S 4.3]{denef:qqhh}---i.e. as $\mathfrak{r}$-representations:
\begin{align*}
	\mathfrak{h}(\gamma;u) \cong H^{*}_{dR}(\mathcal{M}_{\st}^{Q}(\gamma;Z); \mathbb{C}),
\end{align*}
where we have identified the ring of harmonic forms with the de-Rham cohomology ring.  Now, in the Lefschetz representation, the generator of a Cartan subalgebra $J_{3}$ is
\begin{align*}
	J_{3} = \frac{1}{2} \left(\deg - \dim \left[\mathcal{M}_{\st}^{Q}(\gamma;Z) \right] \right) \mathbbm{1}_{H^{*}_{dR}(\mathcal{M}_{\st}^{Q}(\gamma;Z); \mathbb{C})}
\end{align*}
where $\deg :H^{*}_{dR} (\mathcal{M}_{\st}^{Q}(\gamma;Z); \mathbb{C} ) \rightarrow \mathbb{Z}$ maps a cohomology class to its degree.  As a result, the BPS index \eqref{eq:BPS_superdim} is (up to a sign) the same as an Euler characteristic:
\begin{align*}
	\Omega(\gamma;u) = (-1)^{\dim \left[\mathcal{M}_{\st}^{Q}(\gamma;Z) \right]} \chi \left(\CM_{\st}^{Q}(\gamma;Z) \right).
\end{align*}

For non-primitive $\gamma$, the moduli space $\mathcal{M}_{\st}^{Q}(\gamma;Z)$ is properly contained in the (possibly singular) projective variety $\CM_{\sst}^{Q}(\gamma;Z)$ and the argument above is no longer valid.  However, the BPS index may always be computed in terms of quiver DT invariants (which reduce to Euler characteristics of stable moduli in the case of primitive dimension vectors).

\section{Useful properties of \texorpdfstring{$m$-wild}{m-wild} spectra} \label{app:m_wild}
Recall that, as mentioned in the introduction (\S \ref{sec:intro_BPS_indices}), $m$-herds can appear as the spectral network representing slope-1 BPS states at an $m$-wild point on the Coulomb branch.  In order to study generalizations of $m$-herds associated to other states in the $m$-wild spectrum, we provide the following recollections for some properties of the $m$-wild spectrum, all of which can be found in \cite{weist:kron}, \cite{weist:loc}, and \cite{dhkk:dyn_cat}.

\subsection{Equivalences between moduli spaces} \label{app:mod_equivs}

\begin{definition}
	For notational convenience, letting $K_{m}$ denote the Kronecker $m$-quiver,
	\begin{align*}
		\CM^{m}_{\st}(a,b) := \CM^{K_{m}}_{\st}(a q_{1} + b q_{2}; \Theta_{3})
	\end{align*}

\end{definition}

\begin{proposition}[c.f. Prop. 4.4 of \cite{weist:loc}] \label{prop:mod_equiv}
		 There exists isomorphisms of varieties\footnote{This equivalence of moduli spaces holds for any choice of stability condition.} :
		\begin{equation}
			\begin{aligned}
				 \CM^{m}_{\st}(a,b) &\cong \CM^{m}_{\st}(b,a)
			\end{aligned}
			\label{eq:dim_flip}
		\end{equation}
		and
		\begin{equation}
			\begin{aligned}
				\CM^{m}_{\st}(a,b) &\cong \CM^{m}_{\st}(mb - a, b).
			\end{aligned}
			\label{eq:dim_reflect}
		\end{equation}
%
		
\end{proposition}
Hence, we have a corresponding equivalence between Euler characteristics of stable moduli; using Reineke's functional equation (c.f. \S \ref{sec:reineke_func}), this also implies a corresponding equivalence between DT invariants.\footnote{The proof of Prop.~\ref{prop:mod_equiv}, sketched below, is a corollary of the existence of two autofunctors $\mathsf{Rep}_{\mathbb{C}}^{\Theta}(K_{m}) \rightarrow \mathsf{Rep}_{\mathbb{C}}^{\Theta}(K_{m})$; from this the equivalence of DT invariants can be seen more directly.}
\begin{corollary}
We have the equivalence of DT invariants
	\begin{align*}
		d(a,b,m) &= d(b,a,m)\\
		d(a,b,m) &= d(mb - a,b,m).
	\end{align*}
\end{corollary}

Hence, at an $m$-wild point on the Coulomb branch, if we let $\gamma_{1}$ and $\gamma_{2}$ be BPS hypermultiplet charges that satisfy \eqref{eq:cond_wild_centralcharge} and generate a Kronecker $m$-subquiver of the BPS-quiver, then we have
\begin{equation}
	\begin{aligned}
		\Omega\left(a \gamma_{1} + b \gamma_{2} \right) &= \Omega \left(b \gamma_{1} + a \gamma_{2} \right)\\
		\Omega \left(a \gamma_{1} + b \gamma_{2} \right) &= \Omega\left((mb - a) \gamma_{1} + b \gamma_{2} \right).
	\end{aligned}
	\label{eq:reflection_omega}
\end{equation}

The proof of \eqref{eq:dim_flip} follows by application of the transposition functor that takes representations of a quiver $Q$ to representations of the quiver $Q^{\mathrm{op}}$ given by reversing all arrows of $Q$: each vector space is taken to its dual and each morphism $f: V_{q_{i}} \rightarrow V_{q_{j}}$ is taken to its induced action on dual-spaces: $f^{*}; V_{q_{j}}^{*} \rightarrow V_{q_{i}}^{*}$.  When $Q$ is the Kronecker $m$-quiver, then $Q^{\mathrm{op}}$ is the $m$-Kronecker quiver again with a relabelling of its vertices: $q_{1}^{\mathrm{op}} = q_2$ and $q_{2}^{\mathrm{op}} = q_{1}$; thus, $\Theta_{3} = q_{2}^{*}$ (semi)-stable representations of $Q$ are taken to $(q_{2}^{\mathrm{op}})^{*}$ (semi)stable representations for $Q^{\mathrm{op}}$.

The proof of \eqref{eq:dim_reflect} is a consequence of the application of another type of functor: Weist's ``reflection functor", between quiver representation categories (c.f. \cite[Thm. 2.6]{weist:loc} and \cite{bgp:gabriel_thm}).  In the case of the $m$-Kronecker quiver, this functor maps representations of the $m$-Kronecker quiver to itself. 

From a physical perspective, the equivalences \eqref{eq:reflection_omega} are expected to arise as a corollary of the existence of a $\mathbb{Z}/(2\mathbb{Z}) * \mathbb{Z}/(2\mathbb{Z})$ subgroup of monodromy transformations on the local system $\widehat{\Gamma} \rightarrow \CB^{*}$ \cite[\S 6.2]{wwc}.  In particular, let us fix an $m$-wild point $u \in \CB^{*}$ and focus our attention on the rank 2-sublattice of $\Gamma=\widehat{\Gamma}_{u}$ spanned by the primitive hypermultiplet charges $\gamma_{1}$ and $\gamma_{2}$; then, in the basis $(\gamma_{1}, \gamma_{2})$ we expect to see a $\mathbb{Z}/2\mathbb{Z} * \mathbb{Z}/2\mathbb{Z}$ monodromy subgroup generated by the $\text{Sp}\left(2;\mathbb{Z} \right)$ matrices (where $*$ denotes the free-product of groups):
\begin{align*}
	T &= \left(
	\begin{array}{cc}
		0 & 1\\
		1 & 0
	\end{array}
	\right)\\	
	R &= \left(
	\begin{array}{cc}
		-1 & m\\
		0 & 1
	\end{array}
	\right).
\end{align*}
As the BPS spectrum must be invariant under all such monodromies, then the equivalences \eqref{eq:reflection_omega} follow.

\subsection{Schur Roots and The Dense Arc} \label{app:dense_arc}
Let us now return to the case of a general acyclic quiver $Q$; we begin with the question: ``\textit{Given some slope-stability condition $\Theta \in \mathrm{Hom}(\Lambda, \mathbb{Z})$, for which dimension vectors are there non-trivial stable moduli?}".  A partial answer is given by the Schur roots.

\begin{definition}\footnote{This definition was taken from the notes \cite{faenzi:quiv_tour}.}
	A dimension vector $d \in \mathbb{Z}_{\geq 0} Q$ is a \textit{Schur root} if $d = \underline{\dim}(V)$ for some $V \in \mathsf{Rep}_{\mathbb{C}}(Q)$ such that $\mathrm{End}_{\mathsf{Rep}_{\mathbb{C}}(Q)}(V) \cong \mathbb{C}$. Moreover, $d$ is a \textit{real} Schur root if $V$ is unique up to isomorphism and is an \textit{imaginary} Schur root if there exists infinitely many non-isomorphic choices for $V$.
\end{definition}

Typically Schur roots associated to a particular quiver are defined in terms of the positive roots of a Kac-Moody algebra associated to $Q$; defined in this manner, Schur roots can be calculated combinatorially.  However, the definition above is equivalent to the usual definition by a theorem of Kac \cite{kac:inf_root_1}.


The following is a result of King.

\begin{proposition}[\cite{king:quiv_rep}] \label{prop:king_schur}
	A dimension vector $d$ is a Schur root if and only if there exists a stability condition $\Theta \in \mathrm{Hom}(\Lambda, \mathbb{R})$ such that $\CM_{\st}^{Q}(d;\Theta) \neq \emptyset$.
\end{proposition}

We return to the $m$-Kronecker quiver and ask for its Schur roots.  For a concise description of the following results see \cite{derksen_wayman} or the notes \cite{faenzi:quiv_tour}; these references expand upon the original reference \cite[Page 159, Example (c)]{kac:inf_root_2}.  For the remainder of this section we assume $m \geq 3$.

Define a ``generalized Fibbonaci" sequence $(a_{k})_{k=0}^{\infty}$ via $a_{k} = m a_{k-1} - a_{k-2}$ with $a_{0} = 0,\, a_{1} = 1$; this has the closed-form solution:
\begin{align}
		a_{k} = \frac{1}{2^{k} \sqrt{m^2 - 4}} \left[ \left(m + \sqrt{m^2 - 4} \right)^{k} -  \left(m - \sqrt{m^2 - 4} \right)^{k} \right].
		\label{eq:gen_fib}
\end{align}
Then the real Schur roots are given by
\begin{align*}
	\Delta^{\mathrm{Re}}_{m} &:= \left\{a_{k} q_{1} + a_{k+1} q_{2}: k \in \mathbb{Z}_{\geq 0} \right\} \cup \left\{a_{k+1} q_{1} + a_{k} q_{2}: k \in \mathbb{Z}_{\geq 0} \right\}
\end{align*}
It is an exercise in induction to show that $\gcd(a_{k}, a_{k+1}) = 1$ for all $k \geq 0$; hence, all real Schur roots are primitive dimension vectors.  On the other hand, the imaginary Schur roots need not be primitive, and are given by
\begin{align*}
	\Delta^{\mathrm{Im}}_{m} &:= \left\{a q_{1} + b q_{2} \in \mathbb{Z}_{>0}^{2}: \frac{a}{b} \in [m_{-},m_{+}] \right\}.
\end{align*}
where
\begin{align*}
	m_{\pm} := \frac{1}{2} \left(m \pm \sqrt{m^2 - 4} \right).
\end{align*}
See Fig.~\ref{fig:3_kron_schur} for a depiction of small Schur roots for the 3-Kronecker quiver.

\begin{figure}[t!]
	\begin{center}
		 \includegraphics[scale=0.6]{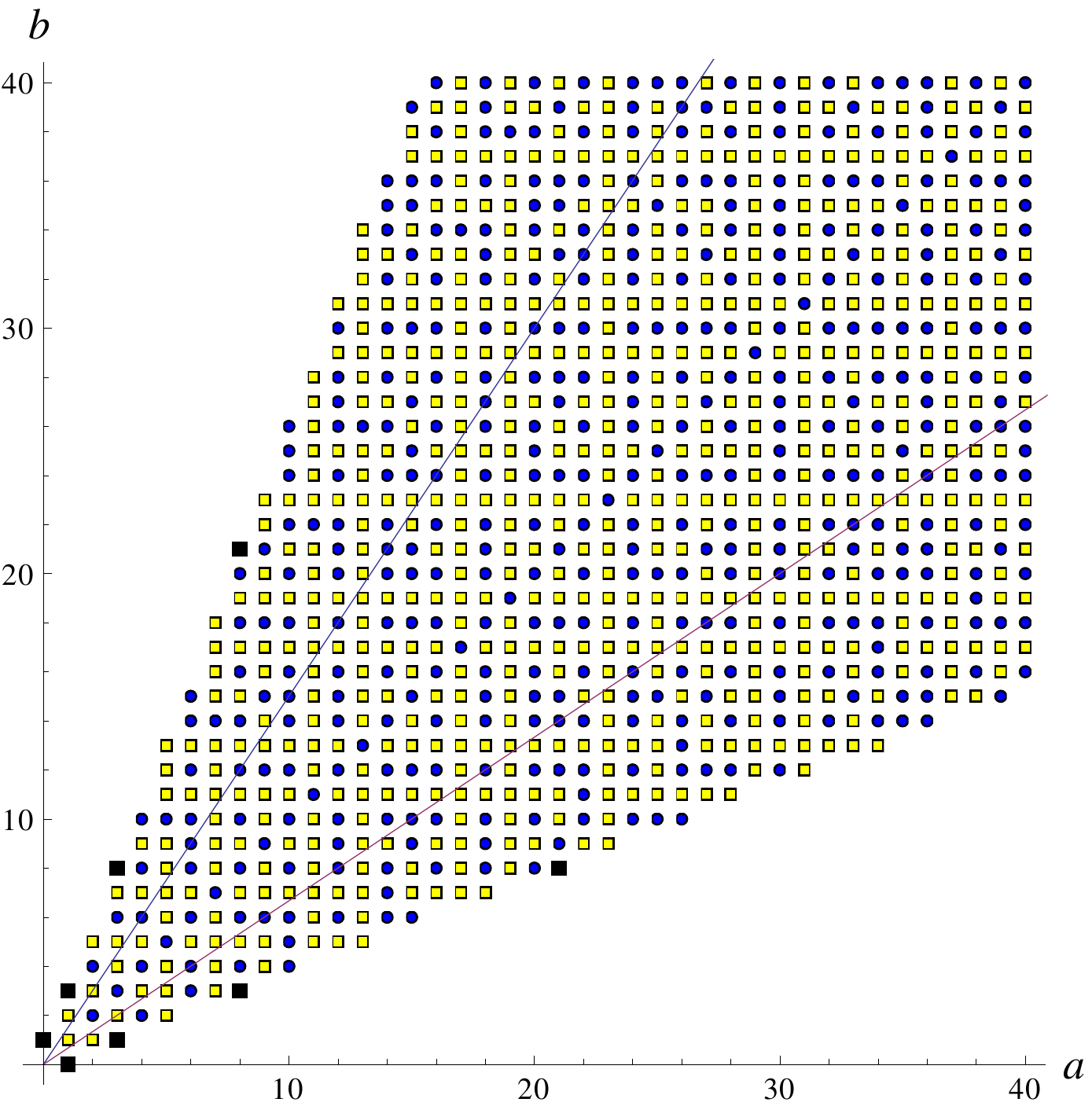}
		\caption{Schur roots for the 3-Kronecker quiver lying in the set $\{a q_{1} + b q_{1} : (a,b) \in \{1,\cdots,40\} \}$.  Real Schur roots are denoted by black boxes, primitive imaginary schur roots denoted by yellow boxes, and non-primitive imaginary schur roots are denoted by blue circles.  Lines of slope $2/3$ and $3/2$ are drawn. \label{fig:3_kron_schur}}
	\end{center}
\end{figure}

Now, via Prop.~\ref{prop:king_schur} every Schur root must correspond to a stable representation with respect to \textit{some} stability condition.  As the previous sections have shown, however, there are only three choices of stability conditions for the $m$-Kronecker quiver up to equivalence; from the discussion at the end of \S \ref{app:stab_m_kron}, we have the following.

\begin{theorem} \label{thm:schur_roots_m_wild}
	All Schur roots are the dimension vector of a stable representation with respect to the stability condition $\Theta_3$.
\end{theorem}

As a consequence of our definition for a real Schur root, we have the following corollary.
\begin{corollary}
	The moduli space $\CM_{\st}^{m}(a,b)$ is a point if $(a,b) \in \Delta^{\mathrm{Re}}_{m}$
\end{corollary}
This yields the immediate physical corollary (see \S \ref{app:BPS_index_quiver}): ``real Schur roots correspond to BPS hypermultiplets at $m$-wild points", or more precisely:

%
%
	
\begin{corollary}[Physical Corollary]
	Let $u$ be an $m$-wild point with charges $\gamma_1,\, \gamma_2 \in \widehat{\Gamma}_{u}$ such that $\langle \gamma_1, \gamma_2 \rangle = m \geq 3	$ are nodes of an $m$-Kronecker BPS subquiver satisfying \eqref{eq:cond_wild_centralcharge}.  If $a \gamma_{1} + b \gamma_{2} \in \Delta^{\mathrm{Re}}_{m}$, then the subspace $\Hilb_{\mathrm{BPS}}^{\mathrm{rest}}(a \gamma_{1} + b \gamma_{2})$ is a half-hypermultiplet representation.
\end{corollary}	

More generally, one can calculate the dimension of $\CM_{\st}^{m}(a,b)$ for both real and imaginary Schur roots:
\begin{align*}
	\dim(\CM_{\st}^{m}(a,b)) = mab - a^2 - b^2 + 1
\end{align*}
for any $(a,b) \in \Delta^{\mathrm{Re}}_{m} \cup \Delta^{\mathrm{Im}}_{m}$.  It follows that $\dim(\CM_{\st}^{m}(a,b)) > 1$ for any $a$ and $b$ satisfying $a/b \in (m_{-}, m_{+})$, i.e. any imaginary Schur root.

%

One may be curious about the central-charge phases of BPS states---particularly when computing spectral networks.  Note that for any stability condition $Z: \Lambda \rightarrow \mathbb{C}$, the phase function
\begin{align*}
	\begin{array}{llccl}
	\arg \left[Z \right] & : & \Lambda^{+} \backslash \{0\} & \longrightarrow & [0,\pi)\\
		 {} & :  &  aq_{1} + bq_{2} & \longmapsto & \arg \left[ Z \left( aq_{1} + b q_{2} \right) \right]
	\end{array}
\end{align*}
depends only on the ratio\footnote{In the appendix we will refrain from referring to this as the ``slope" to prevent confusion with the slope-function in slope-stability conditions.} $a/b$.  Hence, if we extend $Z$ linearly to a map $\Lambda \otimes_{\mathbb{Z}} \mathbb{R} \rightarrow [0,\pi)$ we may define a function
\begin{align*}
	\Phi: \mathbb{R}_{>0} \rightarrow  [0,\pi)
\end{align*}
via
\begin{align*}
	r \mapsto \arg\left[Z(\alpha q_{1} + \beta q_{2}) \right],
\end{align*}
for any $\alpha, \beta \in \mathbb{R}_{\geq 0}$ such that $\alpha/\beta = r$.  With this bit of notation, the following corollary of Thm.~\ref{thm:schur_roots_m_wild} describes the structure of the set of central-charge phases at an $m$-wild point.

\begin{corollary}(see also \cite[Corollary 3.20]{dhkk:dyn_cat})
	Let $Z: \Lambda \rightarrow \mathbb{C}$ be a stability condition such that $\arg[Z(q_{1})] < \arg[Z(q_{2})]$, then
	\begin{align}
		\arg \left[Z \left(\Delta^{\mathrm{Re}}_{m} \cup \Delta^{\mathrm{Im}}_{m} \right) \right] = F_{-} \cup D \cup F_{+}
		\label{eq:phase_arc}
	\end{align}
	where 
	\begin{itemize}
	\item $D = \arg \left[Z \left(\Delta^{\mathrm{Im}}_{m} \right) \right] = \Phi \left([m_{-}, m_{+}] \cap \mathbb{Q} \right)$
	
	\item and $F_{-}$ (resp. $F_{+}$) is a set consisting of a sequence of points (monotonically) converging to $\Phi(m_{-})$ (resp. $\Phi(m_{+})$).\footnote{The limit points $\Phi(m_{-})$ and $\Phi(m_{+})$ are not contained in $F_{-} \cup D \cup F_{+}$.}
	\end{itemize}
\end{corollary}

Because the set $D$ is dense in the arc $\Phi([m_{-},m_{+}])$ we will refer to it as the \textit{dense arc}.  A depiction of the set \eqref{eq:phase_arc} for $m=3$ and $Z = - \Theta_{3}  + i \dim$ is shown in Fig.~\ref{fig:3_kron_arc}.

\begin{figure}[t!]
	\begin{center}
		 \includegraphics[scale=0.55]{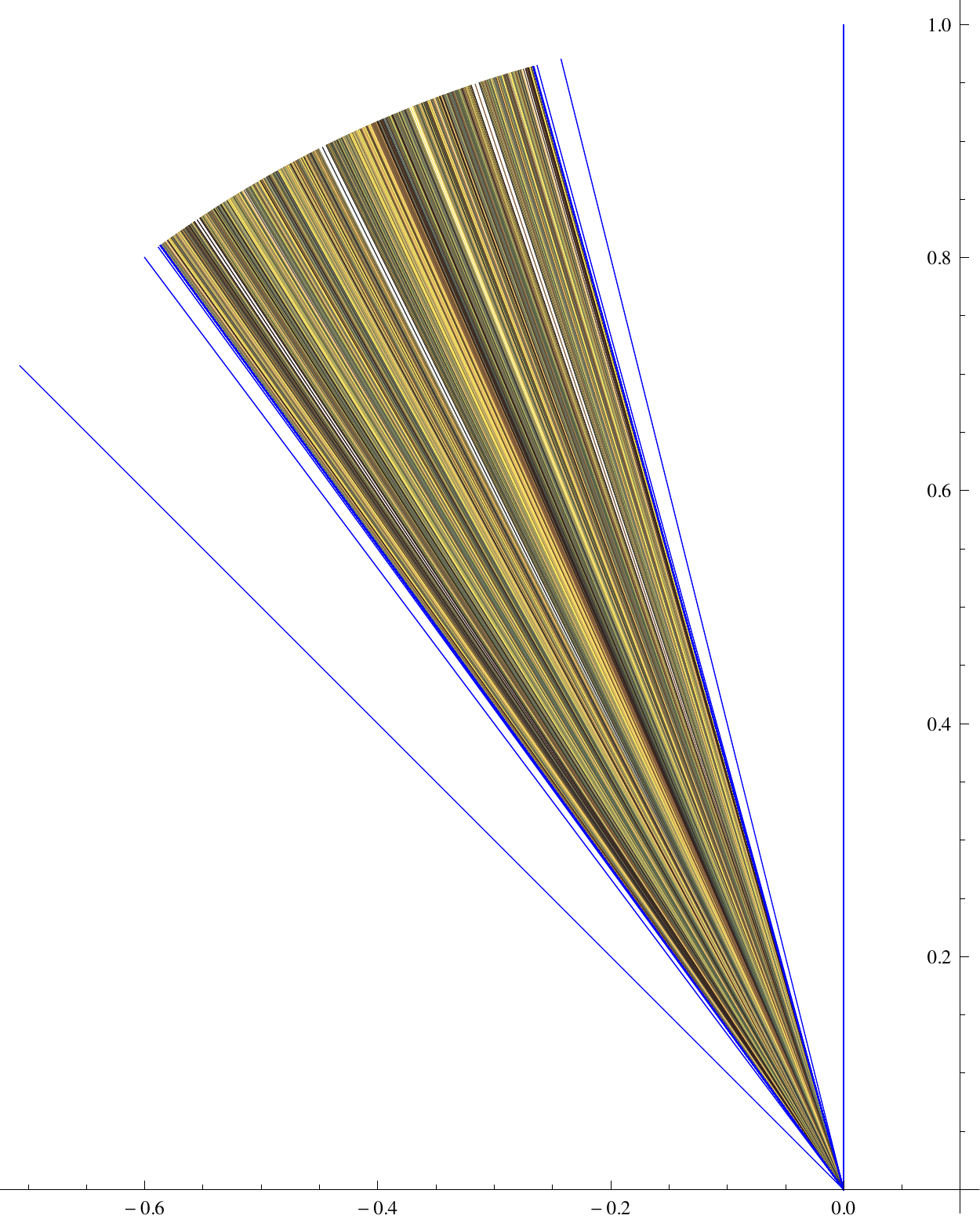}
		\caption{Depiction of the set of phases \eqref{eq:phase_arc} for $m=3$ and $Z= -\Theta_{3} + i \dim$.  For each element $\theta \in F_{-} \cup D \cup F_{+}$, a unit length ray is drawn with slope defined by the angle $\theta$.  Phases of real Schur roots (elements of $F_{-} \cup F_{+}$) are represented by blue rays, while phases of imaginary schur roots are represented by brownish/tan rays.  There are 202 real Schur root phases depicted: 101 elements of $F_{-}$ and 101 elements of $F_{+}$, corresponding to the first $101$ elements of the sequence \eqref{eq:gen_fib}.  One can see how points of $F_{\pm}$ accumulate to $\Phi(m_{\pm})$ where $m_{\pm} = \frac{1}{2}(3 \pm \sqrt{5})$ (given $m = 3$).  Further, there are 1161 imaginary Schur root phases (elements of the ``dense arc" $D$) depicted, corresponding to the subset of coprime integers in the set $\{(a,b) \in \{1, \cdots, N_{\mathrm{cut}} \}^{\times 2}: a/b \in (m_{-}, m_{+}) \}$ where $N_{\mathrm{cut}} = 55$. \label{fig:3_kron_arc}}
	\end{center}
\end{figure}

\section{Derivation of \texorpdfstring{\eqref{eq:func_eqs}}{The algebraic equations for M,V, and W}} \label{sec:derivation}
In the following we provide a sketch of the methods used to obtain \eqref{eq:func_eqs}, which we restate here (with each individual equation numbered) for referential convenience.
\begin{align}
	M &= 1 + z M^{4} \left\{ (1 + V) (1 + V - W)^2 [V^2(1+W) - 1]^{3} \right\} \label{eq:func_eqs_1}\\
	0 &= (-1 + V) (1 + V)^2 + (1 + V^3) W - V (M + V) W^2 \label{eq:func_eqs_2}\\
	0 &=  V \left(V^2 -1 \right) - \left[M (V + 1) + V(V-2) - 1 \right] W \label{eq:func_eqs_3}.
\end{align}

  Unless noted otherwise all notation will be drawn from \cite{wwc}.  The reader is warned that these derivations are ad-hoc, and rely on numerical order-by-order observations to establish particular identities.  Nevertheless we present them here in the spirit of full transparency, and hopes that some reader may improve upon the technique.

\begin{definition}
	We will write $\twid{\Gamma_{ij}}(p)$ to denote the set of collection relative homology classes (``charges") representing solitons of type $ij$ that propagate along the street $p$. 
\end{definition}

\begin{remark}
	In \cite{wwc} the notation $\twid{\Gamma}_{ij}(\twid{x},-\twid{x})$ was used to denote the set of soliton charges of type $ij$ over some \textit{point} $x \in p$ (and with $\twid{x} \in UT \Sigma$ a lift of $x$ to the unit tangent bundle of $\Sigma$, using the orientation on $p$).  Dropping the explicit point $x$ from the notation is only be a mild abuse as all such sets attached to various points on $p$ are related by parallel transport along $p$.
\end{remark}

First, we recall a major part of the spectral network machinery: to each (two-way) street $p$ of type $ij \in \{12,23,13\}$ we attach two formal series: $\Upsilon(p) \in \formal{\mathbb{Z}}{\twid{\Gamma}_{ij}(p)},\, \Delta(p) \in \formal{\mathbb{Z}}{\twid{\Gamma}_{ji}(p)}$, called \textit{soliton generating series}, where $\formal{\mathbb{Z}}{\twid{\Gamma}_{ij}(p)}$ and $\formal{\mathbb{Z}}{\twid{\Gamma}_{ji}(p)}$ are certain (mildly) non-commutative $\formal{\mathbb{Z}}{\twid{\Gamma}}$-algebras.  Less formally speaking, the terms of these formal series are formal variables of the form $X_{a}$ for $a \in \twid{\Gamma}_{ij}$; furthermore, one can order the terms of soliton generating series according to the mass of each term (see the proof of the claim below).  

Now, assume our spectral network has the property that for every branch point of type $ij$, there is at most one two-way street with an endpoint on that branchpoint.\footnote{$m$-herds and all of their generalizations in this paper satisfy this property.}  From a spectral network we develop a collection of equations on the soliton generating series:

\begin{enumerate}
	\item At each joint of a spectral network, write down the six-way junction equations \cite[App. B]{wwc}.
	
	\item For each two-way street $p$ emanating from a branch point of type $ij$ (according to our assumptions there is at most one), 
\end{enumerate}

As we are often only interested in generating series on the two-way streets, we use the observation of \cite[App. C.3]{wwc}: the resulting equations on two-way streets form a closed system of equations, i.e. there is a way in which we can ignore all one-way streets.\footnote{In particular, the two-way ``skeleton networks" shown throughout this paper and \cite{wwc}--which are not full spectral networks but only a specification of their two-way streets---suffice to understand the BPS state counts.}

\begin{figure}[t!]
	\begin{center}
		 \includegraphics[scale=0.55]{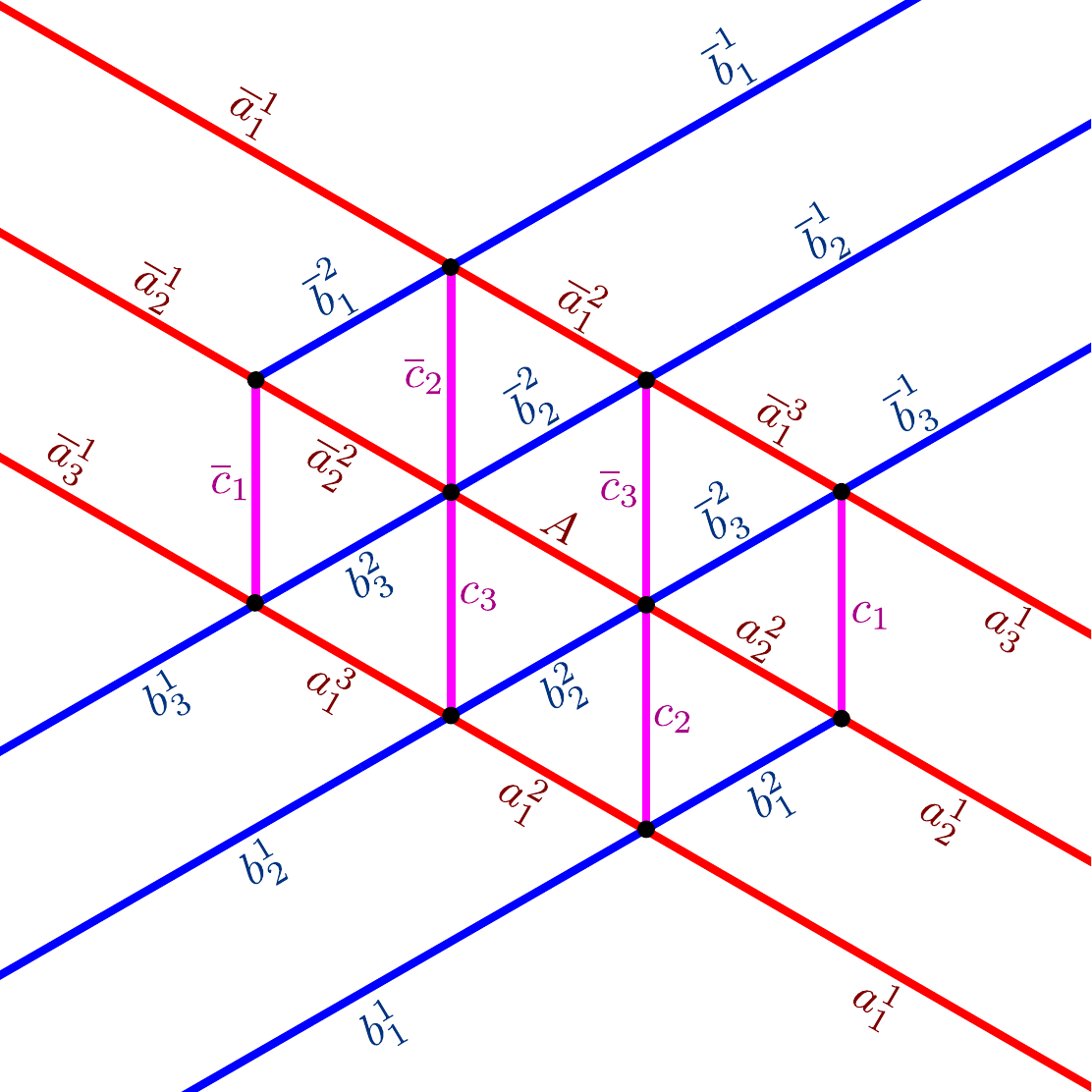}
		\caption{The central building-block of a $(3,2|3)$-herd with streets labeled.  For the computations in App. \ref{sec:derivation} this diagram fits into the central region of Fig.~\ref{fig:(3_2)_diagram}.  The labelling is chosen in a way that reflects the 180 degree rotational symmetry of the full $(3,2|3)$-herd (which can also be seen on the central building block). \label{fig:(3_2)_building_block_labeled}}
	\end{center}
\end{figure}

For each pair of soliton generating series $\Upsilon(p)$ and $\Delta(p)$ we make the following claim.

\begin{itclaim} \label{claim:abelianization}
Let $\mathpzc{N}$ be a spectral network supporting integer multiples of a \textit{single} BPS charge $\gamma_{c}$ and let $p \in \mathrm{str}(\mathpzc{N})$ be a street of type $ij$, then there exists a unique decomposition
\begin{align}
	\Upsilon(p) &= X_{s} U(p) \in \formal{\mathbb{Z}}{\twid{\Gamma}_{ij}(p)}\\
	\Delta(p) &= X_{\conj{s}} D(p) \in \formal{\mathbb{Z}}{\twid{\Gamma}_{ji}(p)}
	\label{eq:ab_decomp}
\end{align}
where
\begin{itemize}
	\item  $s \in \twid{\Gamma}_{ij}(p)$ and $\conj{s} \in \twid{\Gamma}_{ji}(p)$ are two pairs of solitons.
	\item $U(p), D(p) \in \formal{\mathbb{Z}}{\twid{\Gamma_{c}}} = \formal{\mathbb{Z}}{z}$ are formal series with non-vanishing degree zero terms..
\end{itemize}
\end{itclaim}
\begin{proof}
		To show existence, first note that $\twid{\Gamma}_{ij}(\twid{x}, -\twid{x})$, for any point $x \in p$, is a $\twid{\Gamma}$-torsor; so, in particular we can make sense of differences $s - s' \in \twid{\Gamma}$ for any $s,s' \in \twid{\Gamma}_{ij}(\twid{x}, -\twid{x})$.  In particular, as $N$ supports only integer multiples of a single primitive charge $\twid{\gamma}_{c}$, any two solitons supported on $p$ must differ by a multiple of $\gamma_{c}$.  Let $S$ denote the (countable) set of solitons that appear in the formal variables in the expansion of $\Upsilon(p)$.  We now impose a total ordering on $S$: $s < s'$ if the mass of the soliton $s$ is less than or equal to the mass of $s'$; this is a total order as $\text{Mass}(s) < \text{Mass}(s + \twid{\gamma}_{c})$; furthermore, as the mass is always strictly positive, then by the well-ordering principle there exists a smallest element $s_{*} \in S$.  Now, as any soliton in $S$ can be obtained by adding a multiple of $\twid{\gamma}_{c}$, it follows that there exists $U(p) \in \formal{\mathbb{Z}}{\twid{\Gamma_{c}}} $ such that $\Upsilon(p) = X_{s_{*}} U(p)$.  Furthermore, by construction the degree zero term of $U(p)$, must be non-vanishing (otherwise $s_{*} \notin S$).

	Uniqueness is immediate from the condition of non-vanishing degree zero terms: if there exists two decomposition $\Upsilon(p) = X_{s} U(p) = X_{s'} U'(p)$, extracting degree zero terms from $U(p)$ and $U'(p)$ requires $X_{s} = X_{s'}$; hence, $U(p) = U'(p)$ must also follow.
\end{proof}

\begin{remark}
	$U(p), D(p) \in \formal{\mathbb{Z}}{z}$ are invertible as elements of the ring $\formal{\mathbb{Q}}{z}$: an element of $\formal{\mathbb{Q}}{z}$ is invertible if and only if it has non-vanishing degree zero term.  This observation is computationally convenient.
\end{remark}

\begin{definition}[Terminology]
	We will occasionally refer to the commutative formal series $U(p)$ and $D(p)$ of \eqref{eq:ab_decomp} as \textit{abelian parts}.
\end{definition}

In practice, this decomposition allows one to transform equations in the (mildly) non-commutative soliton generating series $\Upsilon(p),\, \Delta(p)$ into equations in the commutative (and invertible) formal series $U(p)$ and $D(p)$.

\begin{remark}
	Note that in the decomposition \eqref{eq:ab_decomp}, $s$ and $\conj{s}$ must be solitons such that, $\cl \left[s + \conj{s} \right] = \pm n \twid{\gamma}_{c}$ under the closure map $\cl: \bigcup_{i} \twid{\Gamma}_{ii} \mapsto \twid{\Gamma}$
\end{remark}

We will denote particular street factors via
\begin{equation}
	\begin{aligned}
		M &:= \oQ{c_{1}} = \oQ{\conj{c}_{1}}\\
		V &:= \oQ{c_{3}} = \oQ{\conj{c}_{3}}\\
		W &:= \oQ{c_{2}} = \oQ{\conj{c}_{2}}.
	\end{aligned}
	\label{eq:MVW_defs}
\end{equation}
where the second equalities follow via the 180 degree rotation symmetry of the diagram.

 In the following, $a$ and $c$ represent the simpleton charges of types 12 and 21 emitted from the red streets at the top and bottom (respectively) of the red streets of Fig.~\ref{fig:(2_1)_diagram}.  Similarly, $b$ and $d$ represent the simpleton charges of type $23$ and $32$.  We will make two rather abusive, but notationally convenient omissions:
\begin{enumerate}
	\item We will explicitly omit the various parallel transport maps applied to the $a,\,b,\,c,\,$ and $d$;
	
	\item We will explicitly omit applications of the closure map $\cl: \bigcup_{i} \twid{\Gamma}_{ii} \rightarrow \twid{\Gamma}$.
\end{enumerate}
In practice these omissions simplify notation greatly without an unrecoverable loss of rigour. Indeed, for the first type of omission: applying the six-way street equations appropriately ensures that the appropriate parallel transport maps have been applied whenever one has a product of two formal variables in soliton-charges $X_{s} X_{s'}$; hence, there exists a well-defined manner in which this product is non-vanishing\footnote{Technically speaking, according to the groupoid addition rule, at least one of $s + s'$ and $s' + s$ is ill-defined (assuming $s$ and $s'$ are not both charges of type $ii$), i.e. one of $X_{s + s'}$ and $X_{s' + s}$ vanishes; if one keeps track of the correct order in the six-way street equations, one will not have such non-vanishing sums.} and can be written as $X_{s + s'}$.  With this in mind, we can freely identify factors of $z$ in equations via:
\begin{align*}
	z = X_{3a + 2b + 3c + 2d}.
\end{align*}
Note that, letting $\widehat{\gamma}_{1}$ and $\widehat{\gamma}_{2}$ be the charges defined\footnote{Please be aware of the following notational conflict between this section and App. \ref{app:signs}: in the former the letters $a, b$ denote soliton charges, in the latter $a$ and $b$ are integers.} in App. \ref{app:signs},
\begin{align*}
	\widehat{\gamma}_{1} &= a + b\\
	\widehat{\gamma}_{2} &= c + d
\end{align*}
assuming the appropriate parallel transport maps have been applied to the right hand side.  Thus,
\begin{align*}
	z &= X_{3 \widehat{\gamma}_{1} + 2 \widehat{\gamma}_{2}}.
\end{align*}

\subsection{Numerical Observations}
We first remark that, given the decomposition of soliton generating series as stated in the claim, if we define $z := X_{\widehat{\gamma_{c}}}$ we must have
\begin{align*}
	\oQ{p} &= 1 + z^{k} U(p) D(p).
\end{align*}
for some $k \in \mathbb{Z}$.   For the $(3,2|3)$-herd, it can be checked that for any two-way street $p$
\begin{align*}
	\Upsilon(p) \Delta(p) &= z \oU{p} \oD{p}.
\end{align*}
i.e. given $\oU{p}$ and $\oD{p}$ we can recover the street-factors via
\begin{align}
	Q(p) &= 1 + z \oU{p} \oD{p}.
	\label{eq:street_ident}
\end{align}

 Now, for $m$-herds, all soliton generating series were of the form $X_{s} P^{k}$ for some soliton charge $s,\,$ street-factor $P,\,$ and positive integer $k$.  Now, it would be a wonderful situation if something analogous held for $(3,2|3)$-herds: if $\oU{p}$ and $\oD{p}$ could be entirely expressed as monomials in terms of the street factors $M,\,V,\,$ and $W$ defined in Fig.~\ref{fig:(3_2)_diagram}.  A bit of numerical exploration shows that this dream fails, but in a mild way.  Namely, the $\oU{p}$ and $\oD{p}$ turn out to be rational functions in the three street factors $M,\,V,\,W \in \formal{\mathbb{Z}}{z}$, the simplest of which are monomials in $M,V,$ and $W$.  
 
Order-by-order exploration suggest the following identities (a partial list among the $29 \times 2 =58$ soliton generating series associated to the $29$ two-way streets of Fig.~\ref{fig:(3_2)_building_block_labeled}).

\begin{equation}
\begin{array}{ll}
   \oup{A} = X_{2a + b + c + d} \oU{A} &  \odown{A} = X_{a + b + 2c + d} \oU{A}\\
	\oup{a_{1}^{1}} = X_{a} M W &  \odown{\conj{a}_{1}^{1}} = X_{c} M W\\
	\oup{a_{2}^{1}} = X_{2a + c + b + d} Y &  \odown{\conj{a}_{1}^{3}} = X_{c} M^2 V W\\
	\oup{a_{2}^{2}} = X_{2a + b + c + d} Y & \oup{\conj{a}_{2}^{1}} = X_{2a + b + c + d} \oU{b_{2}^{2}}\\
	\odown{a_{2}^{2}} = X_{a + b + 2c + d} \oD{a_{2}^{2}} &  \odown{\conj{b}_{2}^{2}} = X_{2a + b + 2c + 2d} \oU{b_{2}^{2}} \\
	\oup{b_{1}^{1}} = X_{b} &  \oup{\conj{b}_{3}^{2}} = X_{2a + 2b + 2c + d} \oU{\conj{b}_{3}^{2}}\\
	\oup{b_{2}^{2}} = X_{2a + 2b + 2c + d} \oU{b_{2}^{2}} & \odown{\conj{b}_{3}^{2}} = X_{a + c + d} (MW)^2 V\\
	\odown{b_{2}^{2}} = X_{a + c + d} \oD{b_{2}^{2}} &  \oup{\conj{c}_{3}} = \eta^{-1} X_{a + b} \oU{\conj{c}_{3}}\\
	\odown{b_{2}^{1}} = X_{a + c + d} Y & \odown{\conj{c}_{3}} = \eta X_{2a + b + 3c + 2d} \oD{\conj{c}_{3}}\\
	\oup{c_{2}} = \eta^{-1} X_{a + b} M W & \\
	\odown{c_{2}} = \eta X_{2a + b + 3c + 2d} \oD{c_{2}} &.
	 \end{array}
	 \label{eq:numerical_results}
\end{equation}

Where $Y \in \formal{\mathbb{Z}}{z}$ is a series whose first few terms are given by
\begin{align*}
	Y &=  2 + 72 z + 6438 z^2 + 752640 z^3 + 100221708 z^4 + 14415715416 z^5 +  \mathcal{O}(z^{6}).
\end{align*}
Eventually we will find that $Y$ can be expressed as a rational function in $M,V$ and $W$.

\begin{remark}[Observation]
Let $p$ be any street in Fig.~\ref{fig:(3_2)_building_block_labeled}, i.e. $p \in \{a_{i}^{j},\, \conj{a}_{i}^{j},\, b_{i}^{j},\, \conj{b}_{i}^{j}\}_{i,j = 1}^{3} \cup \{c_{i},\, \conj{c}_{i}\}_{i=1}^{3} \cup \{A\}$.  With the convention that $\conj{\conj{p}} = p$ and $\conj{A} = A$, then performing a 180-degree rotation of the $(3,2|3)$-herd, we have a map $p \mapsto \conj{p}$; furthermore, the $(3,2|3)$-herd's symmetry under this rotation implies that, under this rotation we have an involution that maps soliton generating series to soliton generating series:
\begin{align*}
	\oup{p} &\rightsquigarrow \odown{\conj{p}}\\
	\odown{p} &\rightsquigarrow \oup{\conj{p}},
\end{align*}
which is induced by the replacements
\begin{align*}
	\begin{array}{lcr}
		 a \mapsto c, &  c \mapsto a\\
		 b \mapsto d, & d \mapsto b
	\end{array}
\end{align*}
and 
\begin{align*}
	\eta \mapsto \eta^{-1},
\end{align*}
while fixing the abelian part of the generating series.  This is reflected in some of the numerical observations above (e.g. $\oup{a_{1}^{1}} = X_{a} MW$ and $\odown{\conj{a}_{1}^{1}} = X_{c} M W$).
\end{remark}

In the following section we show that we may recover the algebraic relations \eqref{eq:func_eqs} from: the collection of results in \eqref{eq:numerical_results}, the observation above, the six-way junction equations, the abelian six-way junction equations\footnote{Which follow as a corollary of the six-way junction equations.} \eqref{eq:abelian_rules}, and the relation between abelian parts and street factors \eqref{eq:street_ident}.


\subsection{Manipulations}
We now develop a functional equation through applications of the six-way junction equations, combined with the numerical observations of the previous section.  First, one can show through the six-way junction equations.
\begin{align}
	\oup{b_{1}^{2}} &= \inup{b_{1}^{1}} \label{eq:b12eb11}\\
	\oup{b_{2}^{2}} &= \inup{b_{2}^{1}} \nonumber\\
	\odown{b_{1}^{2}} &=  \oQ{\conj{b}_{2}^{2}} \times \nonumber\\
		&\times \frac{\indown{\conj{a}_{1}^{1}} \indown{\conj{b}_{3}^{1}} \inup{a_{2}^{1}} \left[1 + \indown{\conj{a}_{2}^{1}} \indown{\conj{b}_{1}^{1}} \inup{a_{1}^{1}} \inup{b_{3}^{1}} \oQ{b_{2}^{2}} \right]}{1 - \indown{\conj{a}_{1}^{1}} \indown{\conj{a}_{2}^{1}} \indown{\conj{b}_{1}^{1}} \indown{\conj{b}_{3}^{1}} \inup{a_{1}^{1}} \inup{a_{2}^{1}} \inup{b_{3}^{1}} \inup{b_{1}^{1}} \oQ{b_{2}^{2}} \oQ{\conj{b}_{2}^{2}} }. \label{eq:db12_sol}
\end{align}
Defining,
\begin{align*}
	\Pi &:= \indown{\conj{a}_{1}^{1}} \indown{\conj{a}_{2}^{1}} \indown{\conj{b}_{1}^{1}} \indown{\conj{b}_{3}^{1}} \inup{a_{1}^{1}} \inup{a_{2}^{1}} \inup{b_{3}^{1}} \inup{b_{1}^{1}},
\end{align*}
then \eqref{eq:db12_sol} implies (using \eqref{eq:b12eb11} and $\oQ{p} = 1 + \oup{p} \odown{p}$ for any street $p$)
\begin{align}
	\oQ{b_{1}^{2}} &=  1 +  \frac{\oQ{\conj{b}_{2}^{2}} \oup{b_{1}^{2}} \indown{\conj{a}_{1}^{1}} \indown{\conj{b}_{3}^{1}} \inup{a_{2}^{1}} \left[1 + \indown{\conj{a}_{2}^{1}} \indown{\conj{b}_{1}^{1}} \inup{a_{1}^{1}} \inup{b_{3}^{1}} \oQ{b_{2}^{2}} \right]}{1 - \indown{\conj{a}_{1}^{1}} \indown{\conj{a}_{2}^{1}} \indown{\conj{b}_{1}^{1}} \indown{\conj{b}_{3}^{1}} \inup{a_{1}^{1}} \inup{a_{2}^{1}} \inup{b_{3}^{1}} \inup{b_{1}^{1}} \oQ{b_{2}^{2}} \oQ{\conj{b}_{2}^{2}} } \nonumber\\
	&= 1 + \frac{\oQ{\conj{b}_{2}^{2}}  \inup{b_{1}^{1}} \indown{\conj{a}_{1}^{1}} \indown{\conj{b}_{3}^{1}} \inup{a_{2}^{1}} + \Pi \oQ{b_{2}^{2}} \oQ{\conj{b}_{2}^{2}}}{1 - \Pi \oQ{b_{2}^{2}} \oQ{\conj{b}_{2}^{2}}} \nonumber \\
	&= \frac{1 + \oQ{\conj{b}_{2}^{2}} \inup{b_{1}^{1}} \indown{\conj{b}_{3}^{1}} \inup{a_{2}^{1}} \indown{\conj{a}_{1}^{1}} }{1 - \Pi \oQ{b_{2}^{2}} \oQ{\conj{b}_{2}^{2}}} \label{eq:Q_b12_reduce}\\
\end{align}
From the abelian rules \eqref{eq:abelian_rules} and the identifications \eqref{eq:MVW_defs}, we have
\begin{align*}
	\oQ{b_{1}^{2}} &= M\\
	\oQ{b_{2}^{2}} &= V = \oQ{\conj{b}_{2}^{2}}.
\end{align*}
Furthermore, using the numerical results of \eqref{eq:numerical_results}, we find
\begin{align*}
		\Pi &= \left[z (M W)^{3} V Y \right]^2
\end{align*}
and
\begin{align*}
		\inup{b_{1}^{1}} \indown{\conj{b}_{3}^{1}} \inup{a_{2}^{1}} \indown{\conj{a}_{1}^{1}} &= z (MW)^3 V Y.
\end{align*}
With these identities, \eqref{eq:Q_b12_reduce} reduces to:
\begin{align*}
	M &= \left[1 - z (MW)^3 V^2 Y \right]^{-1},
\end{align*}
or in a slightly different form:
\begin{align}
	M &=1 + z M^{4} W^3 V^2 Y
	\label{eq:M_eq_prelim}
\end{align}
which is highly reminiscent of \eqref{eq:func_eqs_1} (i.e. the top equation of \eqref{eq:func_eqs}).  Indeed, as mentioned previously, further analysis will show $Y$ can be written as a rational function in $M,V,$ and $W$; so \eqref{eq:M_eq_prelim} will, indeed, reduce to \eqref{eq:func_eqs_1}.

By applying appropriate six-way street equations, and replacing $X_{3a + 2b + 3c + 2d}$ with $z$, one can obtain the following nine equations:

\begin{align}
	\oD{b_{2}^{2}} &= (MW)^2 V + z \oD{\conj{c}_{3}} \oU{A} \label{eq:aD_b22}\\
	\oD{a_{2}^{2}} &= \oU{A} + z \oU{b_{2}^{2}} \oD{c_{2}} \label{eq:aD_a22}\\
	\oD{c_{2}} &= \oD{\conj{c}_{3}} + \oD{b_{2}^{2}} \oU{A} \label{eq:aD_c2}\\
	\oU{\conj{c}_{3}} &= MW + z \oU{\conj{b}_{3}^{2}} Y \label{eq:aU_c3}\\
	\oU{\conj{b}_{3}^{2}} &= \oU{b_{2}^{2}} + MW \oD{a_{2}^{2}} \label{eq:aU_b32}\\
	\oU{A} &= Y + (M W)^2 V \oU{\conj{c}_{3}} \label{eq:aU_A}\\
	\oU{b_{2}^{2}} &= \oD{a_{2}^{2}} M + M^4 V^2 W^3 \label{eq:aU_b22}\\
	Y &= \oD{b_{2}^{2}} + \oU{\conj{c}_{3}} M^2 V W \label{eq:Y}\\
	\oD{\conj{c}_{3}} &= M^2 W \oU{b_{2}^{2}} \label{eq:aD_c3}.
\end{align}
Through various manipulations of these nine equations we can express $Y$ in terms of $M,V,$ and $W$, and derive \eqref{eq:func_eqs_1} and \eqref{eq:func_eqs_2}(the bottom two equations of \eqref{eq:func_eqs}), i.e. the equations not containing the variable $z$, but force further constraints on the relationships between $M,V,$ and $W$.  In the spirit of unattractive transparency, we present a derivation of one such constraint equation (which, by all means, may not be the most efficient way of deriving such an equation).

First, we play with (\ref{eq:aD_b22}):
\begin{align*}
	\oD{b_{2}^{2}} &= (MW)^2 V + z \oD{\conj{c}_{3}} \oU{A}\\
	&= (MW)^2 V + z \oD{\conj{c}_{3}} \left\{  Y + (M W)^2 V \oU{\conj{c}_{3}} \right\}\\
	&= (MW)^2 V + z \oD{\conj{c}_{3}}Y + z (MW)^2 V \oD{\conj{c}_{3}} \oU{\conj{c}_{3}}\\
	&= (MW)^2 V \left[ 1 + z\oD{\conj{c}_{3}} \oU{\conj{c}_{3}} \right] + z \oD{\conj{c}_{3}}Y\\
	&= (MW)^2 V \oQ{\conj{c}_{3}} + z \oD{\conj{c}_{3}} Y\\
	&= (M V W)^2 + z \oD{\conj{c}_{3}} Y.
\end{align*}
Thus,
\begin{align*}
	\oQ{b_{2}^{2}} &= 1 + z \oU{b_{2}^{2}} \oD{b_{2}^{2}}\\
	&= 1 + \frac{z \oD{\conj{c}_{3}}}{M^2 W} \left[Y - \oU{\conj{c}_{3}} M^2 V W \right]\\
	&= 1 + \frac{z \oD{\conj{c}_{3}} Y}{M^2 W} - z \oD{\conj{c}_{3}} \oU{\conj{c}_{3}} V\\
	V &= 1 +  \frac{z \oD{\conj{c}_{3}} Y}{M^2 W} - \left(V - 1 \right)V;
\end{align*}
so
\begin{align}
	z \oD{\conj{c}_{3}} Y &= M^2 W \left(V^2 - 1 \right).
	\label{eq:aDxy_c3}
\end{align}
This yields,
\begin{align}
	\oD{b_{2}^{2}} &= \left(M V W \right)^2 + M^2 W \left(V^2 -1 \right). \label{eq:aD_b22_sol}
\end{align}
Using $\oD{\conj{b}_{3}^{2}} = (MW)^2 V$ (c.f. \eqref{eq:numerical_results}), $\oQ{\conj{b}_{3}^{2}} = W$, and (\ref{eq:aU_c3})
\begin{align*}
	W &= 1 + z (MW)^2 V \oU{\conj{b}_{3}^{2}}\\
	&= 1 + z (MW)^2 V \left[\frac{\oU{\conj{c}_{3}} - MW}{z Y} \right]\\
	&= 1 + (MW)^2 V \left[\frac{\oU{\conj{c}_{3}} - MW}{Y} \right];
\end{align*}
we rewrite this as
\begin{align}
	Y &= \frac{(MW)^2 V \left[\oU{\conj{c}_{3}} - MW \right]}{W - 1}.
	\label{eq:Y_partsol}
\end{align}
Hence, (\ref{eq:Y}) requires, along with (\ref{eq:aD_b22_sol}),
\begin{align*}
	(M V W)^2 + M^2 W(V^2 - 1) + \oU{\conj{c}_{3}} M^2 V W &= \frac{(MW)^2 V \left[\oU{\conj{c}_{3}} - MW \right]}{W - 1};
\end{align*}
solving for $\oU{\conj{c}_{3}}$, we have,
\begin{align}
	\oU{\conj{c}_{3}} &= \frac{1 - V^2 - W + M V W^2 + V^2 W^2}{V}
	\label{eq:aU_c3_sol}.
\end{align}
Substituting (\ref{eq:aU_c3_sol}) back into (\ref{eq:Y_partsol}), and simplifying we have
\begin{align}
	Y &= M^2 W^2 \left[ V^2 + V \left( M + V \right) W -1 \right] \nonumber \\
	&= (M V W)^2 + M^3 V W^3 + M^2 V ^2 W^3 - M^2 W^2.
	\label{eq:Y_sol}
\end{align}
Finally, from the identity
\begin{align*}
	V = \oQ{\conj{c}_{3}} = 1 + \left(z  \oD{\conj{c}_{3}} \right) \oU{\conj{c}_{3}} 
\end{align*}
and the results (\ref{eq:aDxy_c3}), (\ref{eq:aU_c3_sol}), and (\ref{eq:Y_sol}), we have
\begin{align*}
	V &= 1+\left(\frac{1-V^2-W+M V W^2+V^2W^2}{V}\right) \frac{M^2W\left(V^2-1\right)}{(M V W)^2+M^3V; W^3+M^2V^2W^3-M^2W^2}
\end{align*}
this simplifies to
\begin{align*}
	0 &= (-1 + V) \left[(-1 + V) (1 + V)^2 + (1 + V^3) W - V (M + V) W^2 \right].
\end{align*}
But as $V \neq 1$, then we must have
\begin{align}
	0 &= (-1 + V) (1 + V)^2 + (1 + V^3) W - V (M + V) W^2;
	\label{eq:rel_1}
\end{align}
this is \eqref{eq:func_eqs_2}.

Now, solving \eqref{eq:rel_1} for $M$, and substituting this expression into (\ref{eq:Y_sol}) gives
\begin{align}
	Y &= \frac{(1+V) (1+V-W)^2 \left(V^2 (1+W) - 1 \right)^3}{V^2 W^3}.
	\label{eq:Y_red}
\end{align}
Combining this with \eqref{eq:M_eq_prelim}, we arrive at \eqref{eq:func_eqs_1}:
\begin{align*}
	M &= 1 + z M^{4} \left\{ (1 + V) (1 + V - W)^2 [V^2(1+W) - 1]^{3} \right\}.
\end{align*}

To derive the second relation between $M,\, V,\,$ and $W$ we study $\oQ{c_{2}},\,\oQ{a_{2}^{2}},\,$ or $\oQ{b_{2}^{1}}$, among various other generating series.  In particular, we can write
\begin{align*}
	V &= \oQ{a_{2}^{2}}\\
	&= 1 + x Y \oD{a_{2}^{2}}\\
	&= 1 + \left(\frac{M^2 W (V^2 - 1)}{\oD{\conj{c}_{3}}} \right) \oD{a_{2}^{2}}\\
	&= 1 + \frac{(V^2 -1) \oD{\conj{a_{2}^{2}}}}{\oU{b_{2}^{2}}}
\end{align*}
or
\begin{align*}
	W &= \oQ{c_{2}}\\
	&= 1 + z M W \oD{\conj{c}_{2}}\\
	&= 1 + MW \left[\frac{\oD{a_{2}^{2}} - \oU{A}}{\oU{b_{2}^{2}}} \right].
\end{align*}
Using the finalized expressions above, the relation \eqref{eq:rel_1} (which can be used to eliminate $M$), and the condition that $V \neq 1$, and the assistance of Mathematica, we consistently arrive at the relation \eqref{eq:func_eqs_3}:
\begin{align}
	0 &= (-1 + V) (1 + V)^3 + (1 + V) (1 + V^2) W - V (1 + 3 V) W^2.
\label{eq:rel_2}
\end{align}

%

\section{The Standard Lift} \label{app:standard_lift}
We will define a map  $\twid{\left( \cdot \right)}: \Gamma: \rightarrow \twid{\Gamma}$. First, represent $\gamma$ as a sum of $k$ smooth closed curves $\{\beta_m:S^{1} \rightarrow \Sigma\}_{k=1}^{m}$ on $\Sigma$: 
\begin{align*}
	\gamma = \sum_{k=1}^{m}[\beta_{m}] \in \Gamma
\end{align*}
Now, each $\beta_{m}$ has a canonical lift $\widehat{\beta}_m: S^{1} \rightarrow UT\Sigma$ given by its tangent framing. 
\begin{equation}
\twid{\gamma} := \sum_{m=1}^k ([\widehat{\beta}_m] + H) + \sum_{m \le n} \# (\beta_m \cap \beta_n) H \in \twid{\Gamma}
 \label{eq:lift_def}
\end{equation}
where $\# (\beta_m \cap \beta_n)$ is the mod 2 intersection number of the closed curves $\beta_{m}$ and $\beta_{n}$; it can be calculated by perturbing the $\beta_{m}$ such that each intersection is transverse.\footnote{Because we are working mod 2, it agrees with the intersection pairing on homology:
\begin{align*}
	\# (\beta_m \cap \beta_n) = \langle [\beta_m], [\beta_n] \rangle_{\Gamma}  \text{ mod 2}.
\end{align*}
}
As shown in \cite[App. E]{wwc}, $\twid{\gamma}$ is independent of our choices of $\{\beta_{m}\}_{m=1}^{k}$; hence, $\twid{\left( \cdot \right)}: \Gamma: \rightarrow \twid{\Gamma}$ is a well-defined map.

\begin{remark}
It should be noted that $\twid{\left( \cdot \right)}$ is not a homomorphism; indeed, its failure to be a homomorphism is encapsulated via the identity:
	\begin{align*}
		\twid{\gamma} + \twid{\gamma}' = \twid{\left(\gamma + \gamma'\right)} + \langle \gamma, \gamma' \rangle_{\Gamma} H.
	\end{align*}
where $\langle \cdot, \cdot \rangle_{\Gamma}$ denotes the skew-symmetric pairing on $\Gamma$
\end{remark}

\section{Signs in the definition of \texorpdfstring{$z$}{z}} \label{app:signs}
We now elaborate on the reasons for the sign $(-1)^{m a b - a^2 - b^2}$ that appears in the definition of the variable $z$ for $(a,b|m)$-herds; the reasoning follows along the lines of the reason for the sign $(-1)^{m}$ in the definition of the formal variable in in the $m$-herd (c.f. Prop. \ref{prop_Q}).  Indeed, as in the proof for $m$-herds, we expect that the street factors (generating series associated to each street) are derived most naturally (using the same techniques to derive them as used with $m$-herds) as series in the formal variable $z = X_{m \widehat{\gamma}_{1} + (m-1) \widehat{\gamma}_{2}}$; where $\widehat{\gamma}_{1}, \, \widehat{\gamma}_{2} \in \twid{\Gamma}$ are defined in the following manner:


\begin{enumerate}
	\item Starting from an $(a,b|m)$-herd, shrink all fuchsia streets (streets of type 13) to points; the resulting diagram should appear as the superposition of two saddle connections with transverse intersections at $m$ points (c.f. Fig.~\ref{fig:saddle_conns}).
	
	\item Each saddle-connection lifts to a closed, oriented loop on the spectral cover $\Sigma$; let $\alpha$ denote the loop whose homology class is $\gamma_{1}$ (using the red-blue colour-coding used throughout this paper: the lift of the blue saddle connection), and $\beta$ the lift of the loop whose homology class is $\gamma_{2}$ (the lift of the red saddle-connection).  
	
	\item The loops $\alpha$ and $\beta$ have natural tangent-framing lifts to the unit tangent bundle $UT\Sigma$; denote them by $\widehat{\alpha}$ and $\widehat{\beta}$ respectively.  Then by taking the homology classes of these tangent framing lifts, we define $\widehat{\gamma}_{1} := [\widehat{\alpha}] \in \twid{\Gamma}$ and $\widehat{\gamma}_{2} := [\widehat{\beta}] \in \twid{\Gamma}$.  The explicit loops $\widehat{\alpha}$ and $\widehat{\beta}$ depend on the choice of shrinking procedure in step 1---however, it is clear that these loops are unique up to homotopy.  Hence, $\widehat{\gamma}_{1}$ and $\widehat{\gamma}_{2}$ are well-defined homology classes.
	
\end{enumerate}

Now, let $a,b \in \mathbb{Z}$, then the homology class $a \gamma_{1} + b \gamma_{2}$ can be represented by the sum of $a$ copies of $\widehat{\alpha}$ and $b$ copies of $\widehat{\beta}$; thus, from \eqref{eq:lift_def}, it follows that
\begin{align*}
	\twid{\left(a \gamma_{1} + b \gamma_{2} \right)} = a \widehat{\gamma}_{1} + b \widehat{\gamma}_{2} + (a + b + m a b) H
\end{align*}
Hence,
\begin{align*}
	z := X_{a \widehat{\gamma}_{1} + b \widehat{\gamma}_{2}} = (-1)^{mab + a + b} X_{\twid{\left(a \gamma_{1} + b \gamma_{2} \right)}}.
\end{align*}
Because $mab + a + b$ has the same parity as $mab - a^2 - b^2$, we may write (if we so wish)
\begin{align*}
	z = (-1)^{mab - a^2 - b^2} X_{\twid{\left(a \gamma_{1} + b \gamma_{2} \right)}}.
\end{align*}

\section{Amusingly Large Polynomials} \label{app:3_2_polys}
For the reader's convenience, we rewrite some of the polynomials that determine important generating series in a form that allows for a simple copy and paste into \textit{Mathematica} (or other computer-algebra software).
The absolutely irreducible polynomial (in two-variables) with root (as a polynomial with coefficients in $\mathbb{Z}[z]$) the series $V$ (see \S \ref{sec:2_3_herds} and App. \ref{sec:derivation}) is stated below.
\begin{verbatim}
PolyV[z_, v_] := -z - 
   z^2 + (-11*z - 20*z^2)*v + (-24*z - 170*z^2)*v^2 + (-1 + 177*z - 
      761*z^2)*v^3 + (-14 + 1053*z - 1620*z^2)*v^4 + (-75 + 1285*z + 
      549*z^2)*v^5 + (-180 - 3866*z + 11973*z^2)*v^6 + (-135 - 
      10062*z + 24543*z^2)*v^7 + (162 + 7192*z - 
      5700*z^2)*v^8 + (243 + 35757*z - 102404*z^2)*v^9 + (-11615*z - 
      129829*z^2)*v^10 + (-97765*z + 124394*z^2)*v^11 + (4563*z + 
      462849*z^2)*v^12 + (180593*z + 198599*z^2)*v^13 + (8205*z - 
      711637*z^2)*v^14 + (-215927*z - 932750*z^2)*v^15 + (27968*z + 
      378160*z^2)*v^16 + (184698*z + 1536505*z^2)*v^17 + (-81597*z + 
      527710*z^2)*v^18 + (-100800*z - 1365795*z^2)*v^19 + (74492*z - 
      1276293*z^2)*v^20 + (10798*z + 539545*z^2)*v^21 + (-28991*z + 
      1270220*z^2)*v^22 + (11257*z + 200132*z^2)*v^23 + (-1475*z - 
      718560*z^2)*v^24 - 408046*z^2*v^25 + 204359*z^2*v^26 + 
   259484*z^2*v^27 + 10647*z^2*v^28 - 87687*z^2*v^29 - 
   31977*z^2*v^30 + 13506*z^2*v^31 + 11541*z^2*v^32 + 725*z^2*v^33 - 
   1718*z^2*v^34 - 578*z^2*v^35 + 33*z^2*v^36 + 58*z^2*v^37 + 
   13*z^2*v^38 + z^2*v^39
\end{verbatim}
The absolutely irreducible polynomial (in two-variables) with root (as a polynomial with coefficients in $\mathbb{Z}[z]$) $\gdt_{3/2}$ (the polynomial \eqref{eq:3_2_3_poly}) is stated below.
\begin{verbatim}
PolyT[z_,t_] := -(1 + z) + t*(4 - 5*z) +  t^2*(-6 + z) + t^3*(4 + 21*z) +
	t^4*(-1 - 34*z) - t^5*(7*z) + t^6*(76*z + z^2) + t^7*(-64*z - 13*z^2) +
	t^8*(6*z - 114*z^2) + t^9*(7*z - 80*z^2) + t^(10)*(6*z^2) +
 	t^(11)*(119*z^2) + t^(12)*(53*z^2 + z^3) + t^(13)*(-55*z^2 + 44*z^3) +
 	t^(14)*(-21*z^2 - 38*z^3) + t^(15)*(77*z^3) - t^(16)*(382*z^3) +
 	t^(17)*(270*z^3) + t^(18)*(80*z^3 - z^4) + t^(19)*(35*z^3 + 7*z^4) +
 	t^(20)*(39*z^4) - t^(21)*(367*z^4) - t^(22)*(173*z^4) -
	t^(23)*(30*z^4) - t^(24)*(35*z^4) +  t^(25)*(3*z^5) - t^(26)*(17*z^5) -
	t^(27)*(77*z^5) - t^(28)*(14*z^5) + t^(29)*(21*z^5) - t^(32)*(3*z^6) +
	t^(33)*(9*z^6) - t^(34)*(7*z^6) + t^(39)*z^7
\end{verbatim}

The absolutely irreducible polynomial (in two-variables) with root (as a polynomial with coefficients in $\mathbb{Z}[z]$) $\geul_{3/2}$ is stated below.
\begin{verbatim}
PolyE[z_,e_] := -z + e*(-5*z + z^2) + e^2*(z - 13*z^2 + z^3) + 
	e^3*(21*z - 114*z^2 + 44*z^3 - z^4) + 
	e^4*(-34*z - 80*z^2 - 38*z^3 + 7*z^4) + 
	e^6*(4 + 76*z + 119*z^2 - 382*z^3 - 367*z^4 - 17*z^5) + 
	e^5*(-1 - 7*z + 6*z^2 + 77*z^3 + 39*z^4 + 3*z^5) + 
	e^7*(-6 - 64*z + 53*z^2 + 270*z^3 - 173*z^4 - 77*z^5 - 3*z^6) + 
	e^8*(4 + 6*z - 55*z^2 + 80*z^3 - 30*z^4 - 14*z^5 + 9*z^6) + 
	e^9*(-1 + 7*z - 21*z^2 + 35*z^3 - 35*z^4 + 21*z^5 - 7*z^6 + z^7)
\end{verbatim}

\printbibliography

\end{document}